\documentclass[11pt]{article}
\usepackage{amsmath}
\usepackage{amsthm}
\usepackage{amssymb}
\usepackage{algorithm}
\usepackage{subfig}
\usepackage{color}
\usepackage[english]{babel}
\usepackage{graphicx}

\usepackage{wrapfig,epsfig}
\usepackage{epstopdf}
\usepackage{url}
\usepackage{graphicx}
\usepackage{color}
\usepackage{epstopdf}
\usepackage{algpseudocode}
\usepackage{scrextend}
\usepackage[T1]{fontenc}
\usepackage{bbm}
\usepackage{comment}

\usepackage{tikz}
\usepackage{hyperref}  
\hypersetup{colorlinks=true,citecolor=blue,linkcolor=blue} 
\usetikzlibrary{arrows}
\usepackage[margin=1in]{geometry}
\usepackage{tikz}
\usetikzlibrary{positioning}
\definecolor{processblue}{cmyk}{0.96,0,0,0}

\author{
  Alexandr Andoni\thanks{Research supported in part by Simons Foundation (\#491119 to
Alexandr Andoni), NSF (CCF-1617955, CCF-1740833), and Google
Research Award.}\\
  \texttt{andoni@cs.columbia.edu}\\
  Columbia University
  \and
  Clifford Stein\thanks{Research supported in part by NSF grants CCF-1421161 and CCF-1714818.} \\
  \texttt{cliff@cs.columbia.edu}\\
  Columbia University
  \and
  Zhao Song\thanks{Part of the work done while visiting IBM Almaden (hosted by David P. Woodruff) and Harvard University (hosted by Jelani Nelson). UT-Austin.}\\
  \texttt{zhaos@g.harvard.edu}\\
  Harvard University
  \and
  Zhengyu Wang \\
  \texttt{zhengyuwang@g.harvard.edu}\\
  Harvard University
  \and
  Peilin Zhong\thanks{Part of the work done while visiting IBM Almaden (hosted by David P. Woodruff), and supported in part by Simons Foundation, and NSF CCF-1617955.}\\
  \texttt{peilin.zhong@columbia.edu}\\
  Columbia University
}

\date{}
\title{Parallel Graph Connectivity in Log Diameter Rounds}
\newtheorem{theorem}{Theorem}[section]
\newtheorem{lemma}[theorem]{Lemma}
\newtheorem{definition}[theorem]{Definition}

\newtheorem{corollary}[theorem]{Corollary}

\newtheorem{fact}[theorem]{Fact}
\newtheorem{remark}[theorem]{Remark}
\newtheorem{claim}[theorem]{Claim}

\newcommand{\wh}{\widehat}
\newcommand{\wt}{\widetilde}

\newcommand{\eps}{\epsilon}

\renewcommand{\varepsilon}{\epsilon}
\renewcommand{\tilde}{\wt}

\newcommand{\nul}{{\rm{null}}}

\DeclareMathOperator*{\E}{{\bf {E}}}

\DeclareMathOperator*{\Var}{{\bf {Var}}}

\DeclareMathOperator{\poly}{poly}

\DeclareMathOperator{\p}{par}

\DeclareMathOperator{\rank}{rank}

\DeclareMathOperator{\dist}{dist}

\DeclareMathOperator{\dep}{dep}
\DeclareMathOperator{\pos}{pos}
\DeclareMathOperator{\diam}{diam}

\DeclareMathOperator{\col}{col}
\DeclareMathOperator{\lca}{lca}
\DeclareMathOperator{\child}{child}
\DeclareMathOperator{\leaves}{leaves}
\DeclareMathOperator{\st}{st}
\DeclareMathOperator{\ed}{ed}
\DeclareMathOperator{\MPC}{MPC}

\makeatletter
\newcommand*{\RN}[1]{\expandafter\@slowromancap\romannumeral #1@}
\makeatother

\usepackage{lineno}

\usepackage{etoolbox}

\makeatletter

\newcommand{\define}[4][ignore]{%
  \ifstrequal{#1}{ignore}{}{
  \@namedef{thmtitle@#2}{#1}}%
  \@namedef{thm@#2}{#4}%
  \@namedef{thmtypen@#2}{lemma}%
  \newtheorem{thmtype@#2}[theorem]{#3}%
  \newtheorem*{thmtypealt@#2}{#3~\ref{#2}}%
}

\newcommand{\state}[1]{%
  \@namedef{curthm}{#1}
  \@ifundefined{thmtitle@#1}{
  \begin{thmtype@#1}
    }{
  \begin{thmtype@#1}[\@nameuse{thmtitle@#1}]
  }
    \label{#1}
    \@nameuse{thm@#1}
  \end{thmtype@#1}
  \@ifundefined{thmdone@#1}{
  \@namedef{thmdone@#1}{stated}%
  }{}
}

\newcommand{\restate}[1]{%
  \@namedef{curthm}{#1}
  \@ifundefined{thmtitle@#1}{
    \begin{thmtypealt@#1}
    }{
  \begin{thmtypealt@#1}[\@nameuse{thmtitle@#1}]
  }
    \@nameuse{thm@#1}
  \end{thmtypealt@#1}
  \@ifundefined{thmdone@#1}{
  \@namedef{thmdone@#1}{stated}%
  }{}
}

\newcommand{\thmlabel}[1]{
  \@ifundefined{thmdone@\@nameuse{curthm}}{\label{#1}
    }{\tag*{\eqref{#1}}}
}
\makeatother

\begin{document}

\begin{titlepage}
  \maketitle
  \begin{abstract}
Many modern parallel systems, such as MapReduce, Hadoop and Spark, can be modeled well by the MPC model.
The MPC model captures well coarse-grained computation on large data --- data is distributed to processors,
each of which has a sublinear (in the input data) amount of memory and we alternate between rounds of
computation and rounds of communication, where each machine can communicate an amount of data as large
as the size of its memory.  This model is stronger than the classical PRAM model, and it is an intriguing
question to design algorithms whose running time is smaller than in the PRAM model.

One fundamental graph problem is connectivity.  On an undirected graph with $n$ nodes and $m$ edges,  $O(\log n)$ round connectivity algorithms
have been known for over 35 years.  However, no algorithms with better complexity bounds were known.
In this work, we give  {\bf fully scalable, faster algorithms for the
connectivity problem}, by parameterizing the time complexity as a
function of the {\em diameter} of the graph. Our main result is a
$O(\log D \log\log_{m/n} n)$ time connectivity algorithm for
diameter-$D$ graphs, using $\Theta(m)$ total memory.
If our algorithm can use more memory, it can terminate in fewer rounds, and there is no
lower bound on the memory per processor.

We extend our results to
related graph problems such as spanning forest, finding a DFS
sequence,  exact/approximate minimum spanning forest, and bottleneck spanning forest.
We also show that achieving similar bounds
for reachability in {\em directed graphs} would imply faster boolean
matrix multiplication algorithms.

We introduce several new algorithmic ideas.  We describe a general
technique called {\em double exponential speed problem size reduction}
which roughly means that if we can use total memory $N$ to reduce a
problem from size $n$ to $n/k$, for $k=(N/n)^{\Theta(1)}$ in one
phase, then we can solve the problem in $O(\log\log_{N/n} n)$ phases.
In order to achieve this fast reduction for graph connectivity, we use
a multistep algorithm.  One key step is a carefully constructed
truncated broadcasting scheme where each node broadcasts neighbor sets
to its neighbors in a way that limits the size of the resulting neighbor sets.
Another key step is {\em random leader contraction}, where we choose a
smaller set of leaders than many previous works do.

  \end{abstract}
  \thispagestyle{empty}
\end{titlepage}



\newpage
{\hypersetup{linkcolor=black}
\tableofcontents
}
\newpage








\section{Introduction}

Recently, several parallel systems, including MapReduce
\cite{DG04-mapreduce, DG08-mapreduce}, Hadoop \cite{white2012hadoop},
Dryad \cite{isard2007dryad}, Spark \cite{zaharia2010spark}, and
others, have become successful in practice. This success has sparked a
renewed interest in algorithmic ideas for these parallel systems.

One important theoretical direction has been to develop
good models of these modern systems and to relate them
to classic models
such as PRAM. The work of \cite{FMSSS10-mad, ksv10,
  gsz11, bks13, ANOY14-parallel} have led to the
model of {\em Massive Parallel Computing} (MPC) that balances
accurate modeling  with theoretical elegance. MPC is a variant of the Bulk
Synchronous Parallel (BSP) model \cite{Val90-BSP}.
In particular, MPC allows
$N^\delta$ space per machine (processor), where $\delta\in(0,1)$ and
$N$ is the input size, with alternating rounds of unlimited local computation, and communication
of up to $N^\delta$ data per processor.
An MPC
algorithm can equivalently be seen as a small circuit, with
arbitrary, $N^\delta$-fan-in gates; the depth of the circuit is the
parallel time. Any PRAM algorithm can be simulated on MPC in the same
parallel time \cite{ksv10, gsz11}. However, MPC is
in fact more powerful than the PRAM: even computing the XOR of $N$
bits requires near-logarithmic parallel-time on the most powerful
CRCW PRAMs \cite{BH89-CRCWlb}, whereas it takes constant,
$O(1/\delta)$, parallel time on the MPC model.

The main algorithmic question of this area is then: for which problems
can we design MPC algorithms that are {\em faster} than the best PRAM
algorithms?  Indeed, this question has been the focus of several
recent papers, see, e.g., \cite{ksv10, LMSV11-filtering,
  EIM11-clustering, ANOY14-parallel, ahn2018access, assadi17-matching,
  Im17-parallel, czumaj18-parallel}.
Graph problems have been particularly well studied and one fundamental problem is {\em connectivity in
  a graph}. While this problem has a standard logarithmic time PRAM
algorithm \cite{SV82-connect}, we do not know whether we can solve it
faster in the MPC model.

While we would like {\em fully scalable} algorithms---which work for
any value of $\delta>0$---there have been graph algorithms that use
space close to the number of vertices $n$ of the graph. In particular,
the result of \cite{LMSV11-filtering} showed a faster algorithm for
the setting when the space per machine is polynomially larger than the
number of vertices, i.e., $s\ge n^{1+\Omega(1)}$, and hence the number
of edges is necessarily $m\ge n^{1+\Omega(1)}$.  In fact, similar
space restrictions are pervasive for all known sub-logarithmic time
graph algorithms, which require $s=\Omega(\tfrac{n}{\log^{O(1)}n})$
\cite{LMSV11-filtering, ahn2018access, assadi17-matching,
  czumaj18-parallel} (the only exception is \cite{ANOY14-parallel} who
consider geometric graphs). We highlight the work of
\cite{czumaj18-parallel}, who manage to obtain {\em slightly
  sublinear} space of $n/\log^{\Omega(1)} n$ in $\log^{O(1)}\log n$
parallel time, for the approximate matching problem and
\cite{AssadiBBMSarxiv2017} who obtain {\em slightly sublinear} space
of $n/\log^{\Omega(1)} n$ in $O(\log\log n)$ parallel time.
We note that the space of $\sim n$ also coincides with the space
barrier of the semi-streaming model: essentially no graph problems are
solvable in less than $n$ space in the streaming model, unless we have
many more passes; see e.g. the survey \cite{McG-survey}.


It remains a major open question whether there exist fully scalable
connectivity MPC algorithms with sub-logarithmic time (e.g., for
sparse graphs).  There are strong indications that such algorithms do
{\em not} exist: \cite{bks13} show logarithmic lower
bounds for restricted algorithms. Alas, showing an unconditional lower
bound may be hard to prove, as that would imply circuit lower bounds
\cite{DBLP:conf/spaa/RoughgardenVW16}.

In this work, we show {\bf faster, fully scalable algorithms for the
connectivity problem}, by parameterizing the time complexity as a
function of the {\em diameter} of the graph. The diameter of the graph is the largest diameter of its connected components. Our main result is an
$O(\log D \log\log_{m/n} n)$ time connectivity algorithm for
diameter-$D$ graphs with $m$ edges. Parameterizing as a
function of $D$ is standard, say, in the distributed computing
literature \cite{prs16,hhw18}. In fact, some previous MPC
algorithms for connectivity in the applied communities have been {\em
  conjectured} to obtain $O(\log D)$ time \cite{rmcs13};
alas, we show in Section \ref{sec:rastogi} the algorithm of \cite{rmcs13} has
a lower bound of $\Omega(\log n)$ time.

Our algorithms exhibit a tradeoff between the total amount of memory
available and the number of rounds of computation needed.  For
example, if the total space is $\Omega(n^{1+\gamma'})$ for some
constant $\gamma' > 0$, then our algorithms run in $O(\log D)$ rounds
only.

\subsection{The MPC model}

Before stating our full results, we briefly recall the $\MPC$ model
\cite{bks13}. A detailed discussion appears in
Section~\ref{sec:MPCmodel}, along with some core primitives
implementable in the $\MPC$ model.

\begin{definition}[$(\gamma,\delta)-\MPC$ model]\label{def:MPCmodel}
Fix parameters $\gamma,\delta>0$, and suppose $N\ge 1$ is the input
size. There are $p\ge 1$ machines (processors) each with local memory
size $s=\Theta(N^{\delta})$, such that $p\cdot s=O(N^{1+\gamma})$. The
space size is measured by words, each of $\Theta(\log(s\cdot p))$
bits.  The input is distributed on the local memory of $\Theta(N/s)$
input machines.  The computation proceeds in rounds. In each round,
each machine performs computation on the data in its local memory, and
sends messages to other machines at the end of the round.  The total
size of messages sent or received by a machine in a round is bounded
by $s$.  In the next round, each machine only holds the received
messages in its local memory.  At the end of the computation, the
output is distributed on the output machines.  Input/output machines
and other machines are identical except that input/output machine can
hold a part of the input/output.  The parallel time of an algorithm is
the number of rounds needed to finish the computation.
\end{definition}

In this model,  the space per machine is sublinear in $N$, and the total
space is only an $O(N^{\gamma})$ factor more than the input size $N$.  In
this paper, we consider the case when $\delta$ is an arbitrary
constant in $(0,1).$ Our results are for both the most restrictive
case of $\gamma=0$ (total space is linear in the input size), as well
as $\gamma>0$ (for which our algorithms are a bit faster).  The model from
Definition~\ref{def:MPCmodel} matches the model $\MPC(\varepsilon)$
from~\cite{bks13} with $\eps=\gamma/(1+\gamma-\delta)$ and the number
of machines $p=O(N^{1+\gamma-\delta})$.

\subsection{Our Results}



While our main result is a $\sim \log D$ time connectivity MPC
algorithm, our techniques extend to
related graph problems, such as spanning forest, finding a DFS
sequence, and exact/approximate minimum spanning forest.
We also prove a lower bound showing that, achieving similar bounds
for reachability in {\em directed graphs} would imply faster boolean
matrix multiplication algorithms.

We now state our results formally. For all results below, consider an
input graph $G=(V,E)$, with $n=|V|$, $N=|V|+|E|$, and $D$ being the
upper bound on the diameter of any connected component of $G$.

\newcommand{\mypara}[1]{\vspace{0.5ex}\noindent{\bf #1}}

\mypara{Connectivity:} In the connectivity problem,  the goal is to output the connected
components of an input graph $G$, i.e. at the end of the computation,
$\forall v\in V,$ there is a unique tuple $(x,y)$ with $x=v$ stored on
an output machine, where $y$ is called the color of $v$. Any two
vertices $u,v$ have the same color if and only if they are in the same
connected component.
\begin{theorem}[Connectivity in MPC, restatement of
    Theorem~\ref{thm:formal_statement_of_connectivity}]
\label{thm:introConn}
For any $\gamma\in[0,2]$ and any constant $\delta\in(0,1),$ there is a
randomized $(\gamma,\delta)-\MPC$ algorithm (see
Algorithm~\ref{alg:batch_algorithm2}) which outputs the connected
components of the graph $G$ in
$O(\min(\log D\cdot
\log \tfrac{\log n}{\log (N^{1+\gamma}/n)},\log n))$
parallel time.
The success probability is at least
$0.98.$ In addition, if the algorithm fails, then it returns {\rm
  FAIL}.
\end{theorem}

Notice that in the most restrictive case of $\gamma=0$ and $m=n$, we
obtain $O(\min(\log D\cdot \log\log n ,\log n))$ time. When the total
space is slightly larger, or the graph is slightly
denser---i.e. $\gamma>c$ or $\log_n m>c$, where $c>0$ is an
arbitrarily small constant---then we obtain $O(\log D)$ time.

\begin{remark}
We note the concurrent and independent work
of~\cite{asw18}, who also give a connectivity algorithm in the $\MPC$
model but with different guarantees. In particular, their
runtime is parameterized as a function of $\lambda$, which is a lower bound on
the spectral gap\footnote{The spectral gap of a graph $G$ is the second smallest eigenvalue of the normalized Laplacian of $G$.} of the
connected components of $G$.  For a graph $G$ with $n$ vertices and
$m=\wt{O}(n)$ edges, their algorithm runs in $O(\log\log
n+\log(1/\lambda))$ parallel time and uses $\wt{O}(n/\lambda^2)$ total
space. In contrast, our algorithm has a runtime of $O(\log
D\cdot\log\log_{N/n}n),$ where $D$ is the largest diameter of
a connected component of $G$, and $N=\Omega(m)$ is the total space
available. To compare the two runtimes, we note that: 1) $D\le
O(\frac{\log n}{\lambda})$ for any undirected graph $G$; and 2) there
exist sparse graphs $G$\footnote{We can construct $G$ as the following: a bridge connects two $3$-regular expanders where each expander has $n/2$ vertices.} with $n$ vertices and $O(n)$ edges such that
$\frac{1}{\lambda}\geq D\cdot n^{\Omega(1)}$ and $D\leq O(\log n).$
Thus, our results subsume \cite{asw18} in the case when total space is
$N=n^{1+\Omega(1)}$, but are incomparable otherwise.
\end{remark}

\mypara{Spanning forest problem:} In the spanning forest problem, the
goal is to output a subset of edges of an input graph $G$ such that
the output edges together with the vertices of $G$ form a spanning
forest of the graph $G$.  In the rooted spanning forest problem, in
addition to the edges of the spanning forest, we are also required to
orient the edge from child to parent, so that the parent-child pairs
form a rooted spanning forest of the input graph $G$.

\begin{theorem}[Spanning Forest, restatement of Theorem~\ref{thm:formal_statement_of_spanning_tree}]\label{thm:intro_spanning_forest}
For any $\gamma\in[0,2]$ and any constant $\delta\in(0,1),$ there is a
randomized $(\gamma,\delta)-\MPC$ algorithm (see
Algorithm~\ref{alg:batch_algorithm3} and
Algorithm~\ref{alg:batch_algorithm4}) which outputs the rooted
spanning forest of the graph $G$ in 
$O(\min(\log D\cdot
\log \tfrac{\log n}{\log (N^{1+\gamma}/n)},\log n))$
parallel time.
The success probability is at least
$0.98.$ In addition, if the algorithm fails, then it returns {\rm
  FAIL}.
\end{theorem}

Our spanning forest algorithm can also output an approximation to the
diameter, as follows.

\begin{theorem}[Diameter Estimator, restatement of Theorem~\ref{thm:formal_statement_of_diameter}]
For any $\gamma\in[0,2]$ and any constant $\delta\in(0,1),$ there is a
randomized $(\gamma,\delta)-\MPC$ algorithm which outputs a
diameter estimator $D'$ of the input graph $G$ in 
$O(\min(\log D\cdot
\log \tfrac{\log n}{\log (N^{1+\gamma}/n)},\log n))$
parallel time such that $D\leq D'\leq D^{O(\log
  ( 1/\gamma' ))},$ where $\gamma'=\tfrac{\log (N^{1+\gamma}/n)}{\log n}.$ The success probability is at least
$0.98.$ In addition, if the algorithm fails, then it returns {\rm
  FAIL}.
\end{theorem}

\mypara{Depth-First-Search sequence:} If the input graph $G$ is a
tree, then we are able to output a Depth-First-Search sequence of that
tree in $O(\log D)+T$ parallel time, where $T$ is parallel time to
compute a rooted tree (see Theorem~\ref{thm:intro_spanning_forest} for
our upper bound of $T$) for $G$. (See Section~\ref{sec:data_organize}
for a discussion how to represent a sequence in the $\MPC$ model.)

\begin{theorem}[DFS Sequence of a Tree, restatement of
    Theorem~\ref{thm:formal_statement_of_DFS}]
Suppose the graph $G$ is a tree.
For any $\gamma\in[\beta,2]$ and any constant $\delta\in(0,1),$ there
is a randomized $(\gamma,\delta)-\MPC$ algorithm
(Algorithm~\ref{alg:dfs_sequence}) that outputs a Depth-First-Search
sequence for the input graph $G$ in $O(\min(\log D\cdot \log
(1/\gamma),\log n))$ parallel time, where
$\beta=\Theta(\log\log n/\log n).$ 
The success probability is at least $0.98$. In addition, if the algorithm fails, then it returns {\rm FAIL}.
\end{theorem}
Applications of DFS sequence of a tree include lowest common ancestor,
tree distance oracle, the size of every subtree, and others.  See
Section~\ref{sec:dfs_app} for a more detailed discussion of the  DFS sequence of a
tree.


\mypara{Minimum Spanning Forest:} In the minimum spanning forest problem, the goal is to compute the minimum
spanning forest of a weighted graph $G$.

\begin{theorem}[Minimum Spanning Forest, restatement of Theorem~\ref{thm:fomal_statement_of_exact_MST}]
Consider a weighted graph $G$ with weights
$w:E\rightarrow \mathbb{Z}$ such that $\forall e\in E,|w(e)|\leq
\poly(n)$.
For any $\gamma\in[0,2]$ and any constant $\delta\in(0,1),$ there is a
randomized $(\gamma,\delta)-\MPC$ algorithm which outputs a minimum
spanning forest of $G$ in $O(\min(\log D_{MSF}\cdot \log (\tfrac{\log
  n}{1+\gamma\log n}),\log n)\cdot\tfrac{\log n}{1+\gamma\log n})$
parallel time, where $D_{MSF}$ is the diameter (with respect to the
number of edges/hops) of a minimum spanning
forest of $G$.
The success probability is at least $0.98.$ In addition, if the
algorithm fails, then it returns {\rm FAIL}.
\end{theorem}

We note that we require the bounded weights condition merely to ensure
that each weight is described by one word.

\begin{theorem}[Approximate Minimum Spanning Forest, restatement of Theorem~\ref{thm:formal_statement_of_approximate_MST}]
Consider a weighted graph $G$ with weights
$w:E\rightarrow \mathbb{Z}_{\geq 0}$ such that
$\forall e\in E,|w(e)|\leq \poly(n).$
For any $\varepsilon\in(0,1),$ $\gamma\in[\beta,2]$ and any constant $\delta\in(0,1),$ there is a randomized $(\gamma,\delta)-\MPC$ algorithm which can output a $(1+\varepsilon)$ approximate minimum spanning forest for $G$ in 
$O(\min(\log D_{MSF}\cdot \log (\tfrac{\log n}{\log(N^{1+\gamma}/(\varepsilon^{-1}n\log n))}),\log n))$
parallel time, where 
$\beta=\Theta(\log( \varepsilon^{-1} \log n )/\log n),$ and $D_{MSF}$ is the diameter (with respect to the
number of edges/hops) of a minimum spanning
forest of $G.$ 
The success probability is at least $0.98.$ In addition, if the algorithm fails, then it returns {\rm FAIL}.
\end{theorem}

\begin{theorem}[Bottleneck Spanning Forest, restatement of Theorem~\ref{thm:formal_statement_of_bottleneck_spanning_tree}]
Consider a weighted graph $G$ with weights
$w:E\rightarrow \mathbb{Z}$ such that $\forall e\in E,|w(e)|\leq
\poly(n)$.
For any $\gamma\in[0,2]$ and any constant $\delta\in(0,1),$ there is a randomized $(\gamma,\delta)-\MPC$ algorithm which can output a bottleneck spanning forest for $G$ in $O(\min(\log D_{MSF}\cdot \log (\tfrac{\log n}{1+\gamma\log n}),\log n)\cdot \log (\tfrac{\log n}{1+\gamma\log n}))$ parallel time, where 
$D_{MSF}$ is the diameter (with respect to the
number of edges/hops) of a minimum spanning
forest of $G.$ 
 The success probability is at least $0.98.$ In addition, if the algorithm fails, then it returns {\rm FAIL}.
\end{theorem}

\mypara{Conditional hardness for directed reachability.} We also
consider the reachability question in the directed graphs, for which
we show similar to the above results are unlikely.  In particular, we
show that if there is a fully scalable multi-query directed
reachability $(0,\delta)-\MPC$ algorithm with $n^{o(1)}$ parallel time
and polynomial local running time, then we can compute the Boolean
Matrix Multiplication in $n^{2+\varepsilon+o(1)}$ time for arbitrarily
small constant $\varepsilon>0$. We note that the equivalent problem
for {\em undirected graphs} can be solved in $O(\log D\log\log n)$
parallel time via Theorem~\ref{thm:introConn}.

\begin{theorem}[Directed Reachability vs. Boolean Matrix
    Multiplication, restatement of
    Theorem~\ref{thm:formal_statement_of_digraph_reachability}]
  Consider a directed graph $G=(V,E)$.  If there is a polynomial local
  running time, fully scalable $(\gamma,\delta)-\MPC$ algorithm that
  can answer $|V|+|E|$ pairs of reachability queries simultaneously
  for $G$ in $O(|V|^{\alpha})$ parallel time, then there is a
  sequential algorithm which can compute the multiplication of two
  $n\times n$ boolean matrices in $O(n^2\cdot
  n^{2\gamma+\alpha+\varepsilon})$ time, where $\varepsilon>0$ is a
  constant which can be arbitrarily small.
\end{theorem}


Finally, in Section \ref{sec:rastogi} we show hard instances for the algorithm~\cite{rmcs13}.




\subsection{Our Techniques}

In this section, we give an overview of the various techniques that we
use in our algorithms.  More details, as well as some of the low level
details of the implementation in the MPC model, are defered to later
sections.

Before getting into our techniques, we mention two standard tools to
help us build our $\MPC$ subroutines.  The first one is sorting: while
in the PRAM model it takes $\sim\log N$ parallel time, sorting takes
only constant parallel time in the MPC model~\cite{g99,gsz11}.  The
second tool is indexing/predecessor search~\cite{gsz11}, which also
has a constant parallel time in $\MPC$.  Furthermore, these two tools are
fully scalable, and hence all the subroutines built on these two tools
are also fully scalable.  See Section~\ref{sec:MPCmodel} for how to
use these two tools to implement the $\MPC$ operations needed for our
algorithms.


\paragraph{Graph Connectivity:}
A natural approach to the graph connectivity problem is via the
classic primitive of contracting to leaders: select a number of leader
verteces, and contract every vertex (or most vertices) to a leader
from its connected component (this is usually implemented by labeling
the vertex by the corresponding leader). Indeed, many previous works
(see e.g.~\cite{ksv10,rmcs13,klmrv14}) are based on this
approach. There are two general questions to address in this approach:
1) how to choose leader vertices, and 2) how to label each vertex by
its leader.  For example, the algorithm in~\cite{ksv10} randomly
chooses half of the vertices as leaders, and then contracts each
non-leader vertex to one of its neighbor leader vertex.  Thus, in each
round of their algorithm, the number of vertices drops by a constant
fraction.  At the same time, half of the vertices are leaders, and
hence their algorithm still needs at least $\Omega(\log n)$ rounds to
contract all the vertices to one leader.  Note that a constant
fraction of leaders is needed to ensure that there is a constant
fraction of non-leader vertices who are adjacent to at least one
leader vertex and hence are contracted.  This leader selection method
appears optimal for some graphs, e.g. path graphs.

To improve the runtime to $\ll \log n$, one would have to choose a
much smaller fraction of the vertices to be leaders.  Indeed, for a
graph where every vertex has a large degree, say at least $d \gg \log
n$, we can choose fewer leaders: namely, we can choose each vertex to
be a leader with probability $p=\Theta((\log n)/d)$. Then the number
of leaders will be about $\wt{O}(n/d)$, while each non-leader vertex
has at least one leader neighbor with high probability.  After
contracting non-leader vertices to leader vertices, the number of
remaining vertices is only a $1/d$ fraction of original number of
vertices.

By the above discussion, the goal would now be to modify our input
graph $G$ so that every vertex has a uniformly large degree, without
affecting the connectivity of the graph. An obvious such modification
is to add edges between pairs of vertices that are already in the same
connected component.  In particular, if a vertex $v$ learns of a large
number of vertices which are in the same connected component as $v$,
then we can add edges between $v$ and those vertices to increase the
degree of $v$.  A na\"ive way to implement the latter is via
broadcasting: each vertex $v$ first initializes a set $S_v$ which
contains all the neighbors of $v$, and then, in each {\em round},
every vertex $v$ updates the set $S_v$ by adding the union of the sets
$S_u$ over all neighbors $u$ of $v$ (old and new).  This approach takes log-diameter
number of rounds, and each vertex learns all vertices which are in the
same connected component at the end of the procedure.  However, in a
single round, the total communication needed may be as huge as
$\Omega(n^3)$ since each of $n$ vertices may have $\Omega(n)$ neighbors,
each with a set of size $\Omega(n)$.

Since our goal of each vertex $v$ is to learn only $d$ vertices in the
same component (not necessarily the entire component), we can
therefore implement a ``truncated'' version of the above broadcasting
procedure:
\begin{enumerate}
\item If $S_v$ already had size $d$, then we do not need any further operation for $S_v$.
\item If $u$ is in $S_v$, and $S_u$ already has $d$ vertices, then we
  can just put all the elements from $S_u$ into $S_v$ and thus $S_v$
  becomes of size $d$.
\item If $|S_v|<d$, and for every $u\in S_v$, the set $S_u$ is also
  smaller than $d$, then we can implement one step of the broadcasting
  --- add the union of $S_u$'s, for all neighbors $u\in S_v$, to
  $S_v$.
\end{enumerate}
In the above procedure, if the number of vertices in $S_v$ is smaller
than $d$ after the $i^{\text{th}}$ round, then we expect $S_v$ to contain
all the vertices whose distance to $v$ is at most $2^i$.  Thus,
the above procedure also takes at most log-diameter rounds.
Furthermore, the total communication needed is at most $O(n\cdot
d^2).$

Our full graph connectivity algorithm implements the above ``truncated
broadcasting'' procedure iteratively, for values $d$ that follow a
certain ``schedule'', depending on the available space. At the
beginning of the algorithm, we have an $n$ vertex graph $G$ with
diameter $D$, and a total of $\Omega(m)$ space.  The algorithm
proceeds in phases, where each phase takes $O(\log D)$ rounds of
communication.  In the first phase, the starting number of vertices is
$n_1=n.$ We implement a truncated broadcasting procedure where the
target degree $d$ is $d_1=(m/n_1)^{1/2}$, using $O(\log D)$ rounds and
$O(m)$ total space. Then we can randomly select $\wt{O}(n_1/d_1)$
leaders, and contract all the non-leader vertices to leader vertices.
At the end of the first phase, the total number of remaining vertices
is at most $n_2=\wt{O}(n_1/d_1)=\wt{O}(n_1^{1.5}/m^{0.5}).$ In
general, suppose, at the beginning of the $i^{\text{th}}$ phase, the
number of remaining vertices is $n_i.$ Then we use the truncated
broadcasting procedure for value $d$ set to $d_i=(m/n_i)^{1/2},$ 
thus making each vertex have degree at least
$d_i=(m/n_i)^{1/2}$ in $O(\log D)$ number of communication rounds and
$O(m)$ total space.  Then we choose $\wt{O}(n_i/d_i)$ leaders, and,
after contracting non-leaders, the number $n_{i+1}$ of remaining
vertices is at most $\wt{O}(n_i^{1.5}/m^{0.5}).$ Let us look at the
progress of the value $d_i$. We have that
$d_{i+1}=\wt{\Omega}((m/n_{i+1})^{1/2})=\wt{\Omega}((m^{1.5}/n_{i}^{1.5})^{1/2})=\wt{\Omega}(d_i^{1.5}).$
Thus, we are making double exponential progress on $d_i$, which
implies that the total number of phases needed is at most
$O(\log\log_{m/n} n)$, and the total parallel time is thus $O(\log
D\cdot \log\log_{m/n} n).$

This technique of double-exponential progress is more general and
extends to other problems beyond connectivity.  In particular, for a
problem, suppose its size is characterized by a parameter $n$ (not
necessarily the input size---e.g. in connectivity problem, $n$ is the
number of vertices).  When $n$ is a constant, the problem can be
solved in $O(1)$ parallel time.  If there is a procedure that uses
total space $\Theta(m)$ to reduce the problem size to at most $n/k$
for $k=(m/n)^{c},c=\Omega(1)$, then we can repeat the procedure
$O(\log\log_{m/n} n)$ times to solve the overall problem. In
particular, after repeating the procedure $i$ times, the problem size
is $n_i\leq n_{i-1}/(m/n_{i-1})^c\leq n\cdot
(n/m)^{(1+c)^i-1}.$ We call this technique  {\em double-exponential
  speed problem size reduction}.

\begin{remark}\label{rem:double_exponential_progress}
For any problem characterized by a size parameter $n$, if we can use
parallel time $T$ and total space $\Theta(m)$ to reduce the problem
size such that the reduced problem size is $n/k$ for
$k=(m/n)^{\Omega(1)},$ then we can solve the problem in $O(m)$ total
space and $O(T\cdot \log\log_{m/n} n)$ parallel time.
\end{remark}

\paragraph{Spanning Forest and Diameter Estimator:} 
Extending a connectivity algorithm to a spanning forest
algorithm is usually straightforward.  For example, in ~\cite{ksv10}, they
only contract a non-leader vertex to an adjacent leader vertex, thus
their algorithm can also give a spanning forest, using the contracted
edges.  Here however, extending our connectivity algorithm to a
spanning forest algorithm requires several new ideas.  In our
connectivity algorithm, because of the added edges, we only ensure
that when a vertex $u$ is contracted to a vertex $v$, $u$ and $v$ must
be in the same connected component; but $u$ and $v$ may not be
adjacent in the original graph.  Thus, we need to record more
information to help us build a spanning forest.

We can represent a forest as a collection of parent pointers $\p(v)$, one for each vertex $v \in V$.
If $v$ is a root in the forest, then we let $\p(v)=v.$
We use $\dep_{\p}(v)$ to denote the depth of $v$ in the forest, i.e. $\dep_{\p}(v)$ is the distance from $v$ to its root.
Let $\dist_G(u,v)$ denote the distance between two vertices $u$ and $v$ in a graph $G$.

Our connectivity algorithm uses the ``neighbor increment'' procedure
described above.  We observed that if the set $S_v$ has fewer than $d$
vertices after the $i^{\text{th}}$ round, then $S_v$ should contain
all the vertices with distance at most $2^i$ to $v$.  This motivates
us to maintain a shortest path tree for $S_v,$ with root $v$.  In the
$i^{\text{th}}$ round, if we need to update $S_v$ to be $\bigcup_{u\in
  S_v} S_u,$ then we can update the shortest path tree of $S_v$ in the
following way:
\begin{enumerate}
\item For each $x\in S_u$ for some $u\in S_v,$ we can create a tuple
  $(x,u).$
\item Then, for each $x\in \left(\bigcup_{u\in S_v} S_u\right)\setminus S_v,$ we can sort all the tuples $(x,u_1),(x,u_2),\cdots,(x,u_k)$ such that $u_1$ minimizes $\min_{u\in S_v}\dist_G(v,u)+\dist_G(u,x).$
    Since $u$ is in $S_v,$ $x$ is in $S_u,$ it is easy to get the value of $\dist_G(v,u),\dist_G(u,x)$ by the information of shortest path tree for $S_v$ and $S_u.$
    Then we set the new parent of $x$ in the shortest path tree for $S_v$ to be the parent of $x$ in the shortest path tree for $S_{u_1}.$
\end{enumerate}
Since $S_v$ before the update contains all the vertices which have
distance to $v$ at most $2^{i-1},$ the union of the shortest path from
$x$ to $u_1$ and the shortest path from $u_1$ to $v$ must be the
shortest path from $x$ to $v$.  Then by induction, we can show that
the parent of $x$ in the shortest path tree for $S_{u_1}$ is also the
parent of $x$ in the shortest path tree for updated $S_v.$ Thus, this
modified ``neighbor increment'' procedure can find $n$ local shortest
path trees where there is a tree with root $v$ for each vertex $v$.
Furthermore, the procedure still takes $O(\log D)$ rounds.  And we can
still use $O(nd^2)$ total space to make each shortest path tree have
size at least $d$.  Next, we show how to use these $n$ local shortest
path trees to construct a forest with the roots in
the forest being the leaders.

As discussed in the connectivity algorithm, if every local shortest
path tree has size at least $d$, we can choose each vertex as a leader
with probability $p=\Theta((\log n)/d)$ and then every tree will
contain at least one leader with high probability.  Let $L$ be the set
of sampled leaders, and let $\dist_G(v,L)$ be defined as $\min_{u\in
  L}\dist_G(v,u).$ Let $v$ be a non-leader vertex, i.e. $v\in
V\setminus L.$ According to the shortest path tree for $S_v$
\footnote{ The construction of $S_v$ for spanning forest algorithm is
  slightly different from that described in the connectivity
  algorithm.  $S_v$ in spanning forest algorithm has a stronger
  property: $\forall u\in V\setminus S_v,$ $\dist_G(u,v)$ must be at
  least $\dist_G(u',v)$ for any $u'\in S_v.$ }, since $L\cap
S_v\not=\emptyset,$ we can find a child $u$ of the root $v$ such that
$\dist_G(v,L)>\dist_G(u,L)$; in this case we set $\p(v)=u.$ For vertex
$v\in L,$ we can set $\p(v)=v.$ We can see now that $\p$ denotes a
rooted forest where the roots are sampled leaders.  Furthermore, since
$\forall v\not\in L,$ $(v,\p(v))$ is from the shortest path tree for
$S_v,$ we know that $v$ and $\p(v)$ are adjacent in the original graph
$G$.  After doing the above for all nodes $v$, the forest denoted by
the resulting vector $\p$ must be a subgraph of the spanning forest
of $G$.  We then apply the standard doubling algorithm to contract all
the vertices to their leaders (roots), in $O(\log D)$ rounds.
Therefore, the problem is reduced to finding a spanning forest in the
contracted graph.  The number of vertices remaining in the contracted
graph is at most $\wt{O}(n/d),$ where $d=(m/n)^{\Theta(1)}.$ By
Remark~\ref{rem:double_exponential_progress}, we can output a spanning
forest in $O(\log D\cdot \log \log_{m/n} n)$ parallel time.

Although the above algorithm can output the edges of a spanning
forest, it cannot output a rooted spanning forest.  To output a rooted
spanning forest, we follow a top-down construction.  Suppose now we
have a rooted spanning forest of the contracted graph.  Since we have
all the information of how vertices were contracted, we know the
contraction trees in the original graph.  To merge these contraction
trees into the rooted spanning forest of the contracted graph, we only
need to change the root of each contraction tree to a proper vertex in
that tree.  This changing root operation can be implemented by the
doubling algorithm via a divide-and-conquer approach.

 Since the spanning forest algorithm needs $O(\log\log_{m/n} n)$ phases to contract all vertices to a single vertex, the total parallel time to compute a rooted spanning forest is $O(\log D\cdot \log\log_{m/n} n).$
 Furthermore, the depth of the rooted spanning forest will be at most $O(D^{O(\log\log_{m/n} n)}).$
 Thus, we can use the doubling algorithm to calculate the depth of the tree, and output this depth as an estimator of the diameter of the input graph.

 \paragraph{Depth-First-Search Sequence:}
Here, when the input graph $G$ is a tree, our goal is to output a DFS
sequence for this tree.  Once we have this sequence, it is easy to
output a rooted tree.  Thus, computing a DFS sequence is at least as
hard as computing a rooted tree, and all the previous algorithms need
$\Omega(\log n)$ parallel time to do so.

First of all, we use our spanning forest algorithm to compute a rooted
tree, reducing the problem to computing a DFS sequence for a rooted tree.
The idea is motivated by TeraSort~\cite{o08}.
If the size of the tree is small enough such that it can be handled by a single machine, then we can just use a single machine to generate its DFS sequence.
Otherwise, our algorithm can be roughly described as follows. (Recall that $\delta$ is the parameter such that each machine has $\Theta(n^\delta)$ local memory.)
\begin{enumerate}
\small
\item Sample $n^{\delta/2}$ leaves $l_1,l_2,\cdots,l_s.$
\item Determine the order of sampled leaves in the DFS sequence.
\item Compute the DFS sequence $\wt{A}$ of the tree which only
  consists of sampled leaves and their ancestors.
\item Compute the DFS sequence $A_v$ of every root-$v$ subtree which does not contain any sampled leaf.
\item Merge $\wt{A}$ and all the $A_v.$
\end{enumerate}

The first and second steps go as follows. Since we only sample
$n^{\delta/2}$ leaves, we can send them to a single machine.  We
generate queries for every pair of sampled leaves where each query
$(l_i,l_j)$ queries the lowest common ancestor of $(l_i,l_j).$ We have
$n^{\delta}$ such queries in total.  Since the input tree is rooted,
we can use a doubling algorithm to preprocess a data structure in
$O(\log D)$ parallel time and answer all the queries simultaneously in
$O(\log D)$ parallel time.  Thus, we know the lowest common ancestor
of any pair of sampled leaves, and we can store this all on a single
machine.  Based on the information of lowest common ancestors of each
pair of sampled leaves, we are able to determine the order of the
leaves.

For the third step,
suppose the sampled leaves have order $l_1,l_2,\cdots,l_s.$
Let $v$ be the root of the tree.
Then the DFS sequence $\wt{A}$ should be: the path from $v$ to $l_1,$ the path from $l_1$ to the lowest common ancestor of $(l_1,l_2),$ the path from the lowest common ancestor of $(l_1,l_2)$ to $l_2,$ the path from $l_2$ to the lowest common ancestor of $(l_2,l_3),$ ..., the path from $l_s$ to $v$.
We can find these paths simultaneously by a doubling algorithm together with a divide-and-conquer algorithm in $O(\log D)$ parallel time.

In the fourth step, we apply the procedure recursively.  Suppose the
total number of leaves in the tree is $q\leq n.$ Since we randomly
sampled $n^{\delta/2}$ number of leaves, with high probability, each
subtree which does not contain a sampled leaf will have at most
$O(q/n^{\delta/2})$ number of leaves.  Thus, the depth of the
recursion will be at most a constant, $O(1/\delta)$.

\paragraph{Minimum Spanning Forest and Bottleneck Spanning Forest.}
Recall that the input is a graph $G=(V,E=(e_1,e_2,\cdots,e_m))$ together with a weight function $w$ on $E$.
Without loss of generality, we only consider the case when all the weights of edges are different, i.e. $w(e_1)<w(e_2)<\cdots<w(e_m).$
Since the weights of edges are different, the minimum spanning forest of the graph is unique.
By Kruskal's algorithm, the diameter of the graph induced by the first $i$ edges for any $i\in[m]$ is at most the depth of the minimum spanning forest.
Now, let us use $D$ to denote the depth of the minimum spanning forest.

We first discuss the minimum spanning forest algorithm. A crucial
observation of Kruskal's algorithm is: if we want to determine which
edges in $e_i,e_{i+1},\cdots,e_j$ are in the minimum spanning forest,
we can always contract the first $i-1$ edges to obtain a graph $G'$,  run a minimum
spanning forest algorithm on the contracted graph $G'$, and observe whether an edge is included in the spanning forest of $G'$.
Thus, if the total space is
$\Theta(m^{1+\gamma}),$ we can have $m^{\gamma}$ copies of the graph, where
the $i^{\text{th}}$ copy contracts the first $(i-1)\cdot m^{1-\gamma}$
edges.  Thus, we are able to divide the edges into $m^{\gamma}$ groups
where each group has $m^{1-\gamma}$ number of edges.  We only need to
solve the minimum spanning forest problem for each group.  Then in the
second phase, we can divide the edges into $m^{2\gamma}$ groups where
each group has $m^{1-2\gamma}$ number of edges.  Thus, the total
number of phases needed is at most $O(1/\gamma).$ In each phase, we
just need to run our connectivity algorithm to contract the graph.

For the approximate minimum spanning forest algorithm, we use a similar idea.
If we want a $(1+\varepsilon)$ approximation, then we round each
weight to the closest value $(1+\varepsilon)^i$ for some integer $i$.
After rounding, there are only $O(1/\varepsilon\cdot \log n)$ edge groups.
Since our total space is at least $\Omega(m\log(n)/\varepsilon),$ we can make $O(1/\varepsilon\cdot \log n)$ copies of the graph.
The $i^{\text{th}}$ copy of the graph contracts all the edges in group $1,2,\cdots,i-1.$
Then, we only need to run our spanning forest algorithm on each copy to determine which edges should be chosen in each group.

Another application of our {\em double exponential speed problem size reduction} technique is bottleneck spanning forest.
For the bottleneck spanning forest, 
suppose we have $\Theta(km)$ total space.
We can have $k$ copies of the graph where the $i^{\text{th}}$ copy contracts the first $(i-1)\cdot m/k$ number of edges.
 We can determine the group of $O(m/k)$ edges which contains the bottleneck edge.
Thus, we reduce the problem to $O(m/k).$
According to Remark~\ref{rem:double_exponential_progress}, the number of phases is at most $O(\log\log_k m),$ and each phase needs $T$ parallel time, where $T$ is the parallel time for spanning forest.

\paragraph{Directed Reachability vs. Boolean Matrix Multiplication}
If there is a fully scalable multi-query directed reachability $\MPC$ algorithm with almost linear total space, we can simulate the algorithm in sequential model.
Thus, it will imply a good sequential multi-query directed reachability algorithm which implies a good sequential Boolean Matrix Multiplication algorithm.

\subsection{Roadmap}

The rest of the paper contains the technical details of our
algorithms.
In Section~\ref{sec:simple_connectivity}, we described a simplified connectivity algorithm.
In Section~\ref{sec:preli}, we describeed the notations.
In Sections \ref{sec:gc}, \ref{sec:sf}, and \ref{sec:dfs}, we give the
details of our main algorithms for connectivity, spanning forest and
depth first search sequence.  In these sections, we focus on the design of
the algorithms and the analysis of the number of rounds.  In Section
\ref{sec:MPCmodel}, we describe the MPC model in detail and discuss some known primitives in that model.  In
Section \ref{sec:implement}, we discuss how to implement the details of our algorithms in the MPC model to achieve the
bounds claimed in the previous sections.
In Section \ref{sec:mst_bst}, we show how to apply our connectivity and spanning forest algorithm in minimum spanning forest and bottleneck spanning forest problems.
In Section \ref{sec:rastogi}, we show hard instances for the algorithm~\cite{rmcs13}.
In Section \ref{sec:alternative_leader}, we show an alternative approach for random leader selection.


\vspace{-2mm}
\section{A Simplified Batch Algorithm for Connectivity}\label{sec:simple_connectivity}
In this section, we show a simplified version of our connectivity algorithm.

Firstly, let us describe the simplified version of truncated broadcasting procedure in the following.
Since $G'$ is obtained by adding edges between the vertices in the same component of $G$.
 $G'$ will preserve the connectivity of $G$.
The parallel time needed is at most $O(\log D)$ where $D$ is the diameter of $G$.
The procedure takes at most $O(nd^2+m)$ total space.
\ \\
\vspace{-2mm}

\noindent\fbox{
\begin{minipage}{\linewidth}
Truncated Broadcasting for Neighbor Increment:
\begin{itemize}
\item \textbf{Input}:
    \begin{itemize}

    \item A graph $G=(V,E)$ with $n=|V|$ vertices and $m=|E|$ number of edges.

    \item A parameter $d.$

    \end{itemize}

\item \textbf{Output}:
    \begin{itemize}

    \item A graph $G'=(V,E')$ such that $\forall v\in V,$ $|\Gamma(v)|\geq d.$ \Comment{$\Gamma(v)$ denotes the neighbors of $v$.}

    \end{itemize}

\item \textbf{While} $\exists x\in V$ such that $|\Gamma(x)|<d:$

    \begin{itemize}

    \item For each $v\in V$ with $|\Gamma(v)|< d:$

        \begin{itemize}

        \item If $\exists u\in \Gamma(v)$ which has $|\Gamma(u)|\geq d,$ then $\Gamma(v)\leftarrow \Gamma(v)\cup\Gamma(u).$

        \item Otherwise, $\Gamma(v)\leftarrow \Gamma(v)\cup\bigcup_{u\in \Gamma(v)} \Gamma(u). $

        \end{itemize}

    \end{itemize}

\end{itemize}
\end{minipage}

}
\ \\

 We can apply the above procedure to make each vertex have a large degree.
Next, let us briefly describe how to choose the leaders and implement the contraction operation for the graph where each vertex has a large degree.
The following procedure just needs $O(n+m)$ total space and $O(1)$ parallel time.
 If every vertex has degree at least $d$, then in the following procedure we can reduce the number of vertices to $\wt{O}(n/d)$ by contracting all the vertices to $\wt{O}(n/d)$ number of leaders.
\ \\

\noindent\fbox{
\begin{minipage}{\linewidth}
Random Leader Contraction:
\begin{itemize}
\item \textbf{Input}:
    \begin{itemize}

    \item A graph $G=(V,E)$ with $n=|V|$ vertices where each vertex has degree at least $d$.

    \end{itemize}

\item \textbf{Output}:
    \begin{itemize}

    \item A graph $G'=(V',E')$ with $\tilde{O}(n/d)$ vertices.

    \item A mapping $\p:V\rightarrow V',$ such that $\p(v)$ is the vertex that $v$ contracts to.

    \end{itemize}

\item \textbf{Leader Selection:}
    \begin{itemize}

    \item Let $L$ denote the set of leaders.

    \item For each $v\in V,$ with probability at least $\tilde{\Omega}(1/d),$ choose $v$ as a leader, i.e. $L\leftarrow L\cup\{v\}.$

    \end{itemize}

\item \textbf{Contraction:}
    \begin{itemize}

    \item For each $v\in L,$ let $\p(v)=v,$ and put $v$ into $V'.$

    \item For each $v\in V\setminus L,$ choose $u\in \Gamma(v)\cap L,$ and set $\p(u)=v.$

    \item For each $(u,v)\in E,$ if $\p(u)\not=\p(v),$ put the edge $(\p(u),\p(v))$ into $E'.$
    \end{itemize}

\end{itemize}
\end{minipage}

}
\ \\

Finally, we describe the simplified version of our connectivity algorithm in the following.
\ \\

\noindent\fbox{
\begin{minipage}{\linewidth}
Connectivity:
\begin{itemize}
\item \textbf{Input}:
    \begin{itemize}

    \item A graph $G=(V,E)$ with $n=|V|$ vertices and $m=|E|$ edges.

    \item Total space $N$ which is $\Theta(m)$.

    \end{itemize}

\item \textbf{Output}:
    \begin{itemize}

    \item A mapping $\col:V\rightarrow V$ satisfies $\forall u,v\in V,\col(u)=\col(v)$ if and only if $u$ and $v$ are connected.

    \end{itemize}

\item \textbf{Initialization:}
    \begin{itemize}

    \item Let $G_0\leftarrow G,n_0\leftarrow n.$

    \end{itemize}

\item \textbf{In phase $i$:}
    \begin{itemize}

    \item Compute $G'_{i-1}:$ Increase the degree of every vertex in $G_{i-1}$ to at least $d_i=\Theta\left((N/n_{i-1})^{1/2}\right).$

    \item Compute $G_i:$ Select $n_i=\wt{O}(n_{i-1}/d_i)$ leaders in $G'_{i-1}$ and contract all the vertices to the leaders.

    \item If $v$ is contracted to $u$, record $\p(v)=u.$

    \item If $G_i$ does not have any edges, then for every vertex $v$ in $G_i,$ set $\p(v)=v,$ and exit the loop.
    \end{itemize}

\item \textbf{Finding the root leader:}
    \begin{itemize}
        \item For each $v\in V,$ find the root of $v$ in $\p,$ i.e. find $u=\p(\p(\cdots\p(v)))$ such that $\p(u)=u.$
        \item Set $\col(v)=u.$
    \end{itemize}

\end{itemize}
\end{minipage}

}
\ \\

After phase $i$, the number of vertices survived is at most $\wt{O}(n_{i-1}/(N/n_{i-1})^{1/2}).$
By Remark~\ref{rem:double_exponential_progress}, there will be at most $O(\log\log_{N/n} n)$ phases.
For phase $i$, we need $O(\log D)$ parallel time to increase the neighbors of every vertex in $G_{i-1}.$
The total parallel time is thus $O(\log D\cdot \log\log_{N/n} n).$
The total space used in phase $i$ is at most $O(m+(N/n_{i-1})^{1/2}\cdot n_{i-1})=O(N).$

\newpage

\appendix

\section{Notations}\label{sec:preli}
$[n]$ denotes the set $\{1,2,\cdots,n\}.$ 
Let $G$ be an undirected graph with vertex set $V$ and edge set $E$.
For $v\in V,$ $\Gamma_G(v)$ denotes the set of neighbors of $v$ in $G$, i.e. $\Gamma_G(v)=\{u\in V\mid (v,u)\in E\}.$
For any $u,v\in V,$ $\dist_G(u,v)$ denotes the distance between $u,v$ in graph $G$.
If $u,v$ are not in the same connected component, then $\dist_G(u,v)=\infty.$
If $u,v$ are in the same connected component, then $\dist_G(u,v)<\infty.$
For $v\in V,$  $\{u\in V\mid\dist_G(u,v)<\infty\}$ is the set of all the vertices in the same connected component as $v$.
The diameter $\diam(G)$ of $G$ is the largest diameter of its components, i.e. $\diam(G)=\max_{u,v\in V:\dist_G(u,v)<\infty}\dist_G(u,v).$

\section{Graph Connectivity}
\label{sec:gc}

    \subsection{Neighbor Increment Operation}\label{sec:neighbor_incr}
In this section, we describe a procedure which can increase the number of neighbors of every vertex and preserve the connectivity at the same time. The input of the procedure is an undirected graph $G=(V,E)$ and a parameter $m$ which is larger than $|V|.$ The output is a graph $G'=(V,E')$ such that for each vertex $v$, either the connected component which contains $v$ is a clique or $v$ has at least $\lceil\left(m/|V|\right)^{1/2}\rceil-1$ neighbors. Furthermore, $|E'|\leq |E|+m.$ 
We use $\Gamma_G(v)$ to denote the neighbors of $v$ in graph $G$, i.e. $\Gamma_G(v)=\{u\in V\mid (u,v)\in E\}.$ Similarly, we let $\Gamma_{G'}(v)$ be the neighbors of $v$ in $G'$, i.e. $\Gamma_{G'}(v)=\{u\in V\mid (u,v)\in E'\}.$

\begin{lemma}\label{lem:properties_of_S}
Let $G=(V,E)$ be an undirected graph, $m\in\mathbb{Z}_{\geq 0}$ which has $m\geq 4|V|.$ Let $n=|V|.$ Let $r$ be the value at the end of the procedure $\textsc{NeighborIncrement}(m,G)$ (Algorithm~\ref{alg:neighbor_increment}.) Then $\forall i\in\{0,1,\cdots,r\},v\in V,$ $S_v^{(i)}$ satisfies the following properties:
\begin{enumerate}
\item $v\in S_v^{(i)}.$ \label{itm:neighbor_S_pro1}
\item $\forall u\in S_v^{(i)},$ $\dist_G(u,v)<\infty$. \label{itm:neighbor_S_pro2}
\item $|S_v^{(i)}|< \lceil(m/n)^{1/2}\rceil\Rightarrow S_v^{(i)}=\{u\in V\mid \dist_G(u,v)\leq 2^i\}.$ \label{itm:neighbor_S_pro4}
\item $|S_v^{(i)}|\leq m/n.$ \label{itm:neighbor_S_pro5}
\end{enumerate}
\end{lemma}
\begin{proof}
For property~\ref{itm:neighbor_S_pro1}, we can prove it by induction. When $i=0,$ due to line~\ref{sta:init}, we know $v\in S_v^{(0)}.$ Suppose property~\ref{itm:neighbor_S_pro1} holds for $S_v^{(i-1)}$ for all $v\in V.$ If $S_v^{(i)}$ is updated by line~\ref{sta:alreadygood2}, there are two cases: 1. if $u=v,$ then $v\in S_u^{(i-1)},$ and the condition of line~\ref{sta:alreadygood2} does not hold, thus $v\in S_v^{(i)};$ 2. if $u\not=v,$ then after implementing line~\ref{sta:alreadygood2}, $v$ will not be removed, thus $v\in S_v^{(i)}.$ If $S_v^{(i)}$ is updated by line~\ref{sta:stillbad}, then since $v\in S_v^{(i-1)},$ $v$ is also in the set $S_v^{(i)}.$ Thus, property~\ref{itm:neighbor_S_pro1} holds for every $S_v^{(i)}.$

For property~\ref{itm:neighbor_S_pro2}, we can prove it by induction. When $i=0,$ it is easy to see $S_v^{(0)}\subseteq \Gamma_G(v)\cup\{v\},$ thus property~\ref{itm:neighbor_S_pro2} holds for it. Suppose property~\ref{itm:neighbor_S_pro2} holds for $S_v^{(i-1)}$ for all $v\in V.$ If $S_v^{(i)}$ is updated by  line~\ref{sta:alreadygood2}, then since $u\in S_v^{(i-1)}$ and $S_v^{(i)}\subseteq S_u^{(i-1)}\cup\{v\},$ all the vertices from $S_v^{(i)}$ are in the same connected component as $u$ and $u$ is in the same connected component as $v$. Thus, property~\ref{itm:neighbor_S_pro2} holds in this case. If $S_v^{(i)}$ is updated by the line~\ref{sta:stillbad}, then $\forall p\in S_v^{(i)},$ there exists $u\in S_v^{(i-1)}$ such that $p\in S_u^{(i-1)}.$ We have $p$ is in the same connected component as $u$, and $u$ is in the same connected component as $v$. Thus, property~\ref{itm:neighbor_S_pro2} also holds in this case.


For property~\ref{itm:neighbor_S_pro4}, we can prove it by induction. When $i=0,$ due to 
line~\ref{sta:init_S2},
we have $|S_v^{(0)}|<\lceil(m/n)^{1/2}\rceil\rightarrow S^{(0)}_v=\Gamma_G(v)\cup\{v\}=\{u\in V\mid \dist_G(u,v)\leq 1\}.$ 
 Suppose property~\ref{itm:neighbor_S_pro4} holds for $S_v^{(i-1)}$ for all $v\in V.$ Since if 
 $|S_v^{(i)}|<\lceil(m/n)^{1/2}\rceil,$
 then $S_v^{(i)}$ can only be updated by line~\ref{sta:stillbad}, and $\forall u\in S_v^{(i-1)},$ it has $|S_u^{(i-1)}|<\lceil(m/n)^{1/2}\rceil.$ Thus, $S_v^{(i)}=\bigcup_{u\in S_v^{(i-1)}} S_u^{(i-1)}=\bigcup_{u\in V,\dist_G(u,v)\leq 2^{i-1}} \{p \in V\mid \dist_G(p,u)\leq 2^{i-1}\}=\{u\in V\mid\dist_G(u,v)\leq 2^i\}.$ Thus, property~\ref{itm:neighbor_S_pro4} holds.

For property~\ref{itm:neighbor_S_pro5}, we can prove it by induction. When $i=0,$ due to line~\ref{sta:init_S2}, $\forall v\in V,$ we have $|S_v^{(0)}|\leq \lceil(m/n)^{1/2}\rceil\leq m/n,$ where the last inequality follows by $m/n\geq 4.$
Now suppose property~\ref{itm:neighbor_S_pro5} holds for $S_v^{(i-1)}$ for all $v\in V.$ If $S_v^{(i)}$ is updated by line~\ref{sta:alreadygood2}, then $|S_v^{(i)}|=|S_u^{(i-1)}|\leq m/n.$ If $S_v^{(i)}$ is updated by line~\ref{sta:stillbad}, we know 
$\forall u\in S_v^{(i-1)},|S_u^{(i-1)}|< \lceil(m/n)^{1/2}\rceil.$ Notice that by property~\ref{itm:neighbor_S_pro1}, $v\in S_v^{(i-1)},$ so $|S_v^{(i-1)}|< \lceil(m/n)^{1/2}\rceil.$ Thus, $|S_v^{(i)}|=|\bigcup_{u\in S_v^{(i-1)}} S_u^{(i-1)}|\leq (m/n)^{1/2}\cdot(m/n)^{1/2}\leq m/n.$
\end{proof}

\begin{algorithm}[!t]
\caption{Neighbor Increment Operation}\label{alg:neighbor_increment}
\begin{algorithmic}[1]
\Procedure{\textsc{NeighborIncrement}}{$m,G=(V,E)$} \Comment{Lemma~\ref{lem:properties_of_S}, Lemma~\ref{lem:properties_of_neighbor_increment}}
\State \Comment{Output: $G'=(V,E')$}
\State Initially, $n=|V|,E'=\emptyset$ and let $S_v^{(0)}=\{v\}$ for all $v\in V$. \label{sta:init}
\For{$v\in V$}\Comment{Initially, let $S_v^{(0)}$ be the set (or subset) of direct neighbors}
    \For {$u\in\Gamma_G(v)$}
        \If {$|S_v^{(0)}|<\lceil(m/n)^{1/2}\rceil$}
            \State  $S_v^{(0)}\leftarrow S_v^{(0)}\cup\{u\}$. \label{sta:init_S2}
        \EndIf
    \EndFor
\EndFor
\State $r \leftarrow 1$.
\For {$\mathbf{true}$}
    \For{$v\in V$}
        \If{$\exists u\in S_v^{(r-1)},|S_u^{(r-1)}|\geq \lceil(m/n)^{1/2}\rceil$}\label{sta:large_condition}\Comment{neighbor $u$ has many neighbors}
                \State $S_v^{(r)}=S_u^{(r-1)}\cup \{v\}$.
                \If {$|S_v^{(r)}|>|S_u^{(r-1)}|$}
                    \State $S_v^{(r)}\leftarrow S_v^{(r)}\setminus\{u\}$. \label{sta:alreadygood2}
                \EndIf
        \Else
                \State $S_v^{(r)}=\bigcup_{u\in S_v^{(r-1)}} S_u^{(r-1)}$. \label{sta:stillbad}\Comment{ neighbors of $v$'s neighbors are $v$'s new neighbors.}
        \EndIf
    \EndFor
    \State \Comment{$S_v^{(r)}$ is large or is a component}
    \If{$\forall v\in V,$ either $|S_v^{(r)}|\geq \lceil(m/n)^{1/2}\rceil$ or $|S_v^{(r)}|=|S_v^{(r-1)}|$ } \label{sta:condition_break}
        \State Let $E'=E\cup\bigcup_{v\in V}\{(v,u)\mid v\in S_u^{(r)}\text{~or~}u\in S_v^{(r)},u\not=v\}$. \label{sta:add_new_edges}
        \State \Return $G'=(V,E')$
    \Else
        \State $r\leftarrow r+1$.
    \EndIf
\EndFor
\EndProcedure
\end{algorithmic}
\end{algorithm}

The following definition defines the number of iterations of Algorithm~\ref{alg:neighbor_increment}.
\begin{definition}\label{def:neighbor_incr_num_iter}
Given an undirected graph $G=(V,E)$ and a parameter $m\in\mathbb{Z}_{\geq 0},m\geq 4|V|,$ the number of iterations of $\textsc{NeighborIncrement}(m,G)$ (Algorithm~\ref{alg:neighbor_increment}) is the value of $r$ at the end of the procedure.
\end{definition}

In the following lemma, we characterize the properties of Algorithm~\ref{alg:neighbor_increment}.
\begin{lemma}\label{lem:properties_of_neighbor_increment}
Let $G=(V,E)$ be an undirected graph, $m\in\mathbb{Z}_{\geq 0}$ which has $m\geq 4|V|.$ 
Let $G'=(V,E')$ be the output of $\textsc{NeighborIncrement}(m,G).$ We have:
\begin{enumerate}
\item The number of iterations (Definition~\ref{def:neighbor_incr_num_iter}), 
    $r\leq \min(\lceil\log(\diam(G))\rceil,\lceil\log(m/n)\rceil)+1.$ \label{itm:neighbor_incr_pro1}
\item For all $u,v\in V,$ $\dist_G(u,v)<\infty\Leftrightarrow \dist_{G'}(u,v)<\infty.$ \label{itm:neighbor_incr_pro2}
\item $\forall v\in V,$ if $|\Gamma_{G'}(v)|< \lceil(m/n)^{1/2}\rceil-1,$ then the connected component in $G'$ which contains $v$ is a clique. It also implies that $\forall u,v\in V,$ if $|\Gamma_{G'}(v)|< \lceil(m/n)^{1/2}\rceil-1$ and $|\Gamma_{G'}(u)|\geq \lceil(m/n)^{1/2}\rceil-1,$ then $\dist_{G'}(u,v)=\infty.$ \label{itm:neighbor_incr_pro3}
\item $E\subseteq E',|E'|\leq |E|+m.$ \label{itm:neighbor_incr_pro4}
\end{enumerate}
\end{lemma}
\begin{proof}
For property~\ref{itm:neighbor_incr_pro1}, if $r>\lceil\log(\diam(G))\rceil+1,$ then let $i=\lceil\log(\diam(G))\rceil+1.$ Let $v\in V.$ By property~\ref{itm:neighbor_S_pro4} of Lemma~\ref{lem:properties_of_S}, if $|S_v^{(i)}|<\lceil(m/n)^{1/2}\rceil,$ then $S_v^{(i)}=\{u\in V\mid \dist_G(u,v)\leq 2^i\}=\{u\in V\mid\dist_G(u,v)\leq 2\cdot 2^{\lceil\log(\diam(G))\rceil}\}=\{u\in V\mid \dist_G(u,v)<\infty\}.$ Furthermore, if $|S_v^{(i)}|<\lceil(m/n)^{1/2}\rceil,$ then $|S_v^{(i-1)}|<\lceil(m/n)^{1/2}\rceil,$ which means that $S_v^{(i-1)}=\{u\in V\mid \dist_G(u,v)\leq 2^{i-1}\}=\{u\in V\mid\dist_G(u,v)\leq  2^{\lceil\log(\diam(G))\rceil}\}=\{u\in V\mid \dist_G(u,v)<\infty\}=S_v^{(i)}.$ Then due to the condition in line~\ref{sta:condition_break}, it will end the procedure in this round, which contradicts to $r>\lceil\log(\diam(G))\rceil+1=i.$ Similarly, if $r>\lceil\log(m/n)\rceil+1,$ then let $i=\lceil\log(m/n)\rceil+1.$ Furthermore, if $|S_v^{(i)}|<\lceil(m/n)^{1/2}\rceil,$ then $|S_v^{(i-1)}|<\lceil(m/n)^{1/2}\rceil,$ which means that $S_v^{(i-1)}=\{u\in V\mid \dist_G(u,v)\leq 2^{i-1}\}\not=\{u\in V\mid\dist_G(u,v)\leq  2^{i}\}=S_v^{(i)}.$ Thus, there exists $u\in S_v^{(i)}$ such that $|S_v^{(i)}|>\dist_G(u,v)>2^{i-1}>m/n\geq \lceil (m/n)^{1/2}\rceil$ which leads to a contradiction. 

For property~\ref{itm:neighbor_incr_pro2}, if $u,v$ are in the same connected component in $G,$ then since $E\subseteq E',$ $u,v$ are in the same connected component in $G'$. If $u,v$ are in the same connected component in $G',$ then there should be a path $u=u_1\rightarrow u_2\rightarrow \cdots \rightarrow u_p=v$ in $G',$ i.e. $\forall j\in[p-1],(u_{j},u_{j+1})\in E'.$ $(u_{j},u_{j+1})\in E'$ implies that either $(u_{j},u_{j+1})\in E$ or $u_{j}\in S_{u_{j+1}}^{(r)}$ or $u_{j+1}\in S_{u_j}^{(r)}.$ By property~\ref{itm:neighbor_S_pro2} of Lemma~\ref{lem:properties_of_S}, we know that $\forall j\in[p-1],$ $u_j$ and $u_{j+1}$ are in the same connected component in $G.$ Thus, $u$ and $v$ are in the same connected component in $G$.

For property~\ref{itm:neighbor_incr_pro3}, due to line~\ref{sta:add_new_edges}, if $|\Gamma_{G'}(v)|< \lceil(m/n)^{1/2}\rceil-1,$ then we have $|S_v^{(r)}|<\lceil(m/n)^{1/2}\rceil.$ 
By property~\ref{itm:neighbor_S_pro4} of Lemma~\ref{lem:properties_of_S}, and the condition in line~\ref{sta:condition_break}, we know $\{u\in V\mid \dist_G(u,v)\leq 2^{r}\}=S_v^{(r)}=S_v^{(r-1)}=\{u\in V\mid \dist_G(u,v)\leq 2^{r-1}\}.$ Thus, $S_v^{(r)}=\{u\in V\mid \dist_G(u,v)<\infty\}.$ Due to property~\ref{itm:neighbor_incr_pro2}, we have $\Gamma_{G'}(v)\cup \{v\}\subseteq \{u\in V\mid \dist_G(u,v)<\infty\}.$ Notice that $S_v^{(r)}\subseteq\Gamma_{G'}(v)\cup \{v\},$ thus, we have $\Gamma_{G'}(v)\cup \{v\}=\{u\in V\mid \dist_G(u,v)<\infty\}.$ Let $v'\in \{u\in V\mid \dist_G(u,v)<\infty\},$ then due to property~\ref{itm:neighbor_S_pro2}, $\Gamma_{G'}(v')\cup \{v'\}\subseteq \{u\in V\mid \dist_G(u,v')<\infty\}=\{u\in V\mid \dist_G(u,v)<\infty\},$ then we have $|\Gamma_{G'}(v')\cup \{v'\}|<\lceil(m/n)^{1/2}\rceil.$ Thus, $|S_{v'}^{(r)}|<\lceil(m/n)^{1/2}\rceil.$ 
By property~\ref{itm:neighbor_S_pro4} of Lemma~\ref{lem:properties_of_S}, and the condition in line~\ref{sta:condition_break}, we know $\{u\in V\mid \dist_G(u,v')\leq 2^{r}\}=S_{v'}^{(r)}=S_{v'}^{(r-1)}=\{u\in V\mid \dist_G(u,v')\leq 2^{r-1}\}.$ Thus, $S_{v'}^{(r)}=\{u\in V\mid \dist_G(u,v')<\infty\}.$ Thus, $\Gamma_{G'}(v')\cup \{v'\}=\{u\in V\mid \dist_G(u,v')<\infty\}=\{u\in V\mid \dist_G(u,v)<\infty\}.$ Thus, $\forall p,q\in \{u\in V\mid \dist_G(u,v)<\infty\},$ we have $(p,q)\in E',$ which means that $\{u\in V\mid \dist_G(u,v)<\infty\}$ is a clique in $G'.$

Now consider two vertices $u,v\in V.$ Suppose $|\Gamma_{G'}(v)|< \lceil(m/n)^{1/2}\rceil-1,$ then we have that $\{p\in V\mid \dist_G(p,v)<\infty\}$ is a clique in $G'.$ Thus, $\forall q\in \{p\in V\mid \dist_G(p,v)<\infty\},$ we have $|\Gamma_{G'}(q)|=|\Gamma_{G'}(v)|< \lceil(m/n)^{1/2}\rceil-1.$ If $|\Gamma_{G'}(u)|\geq \lceil(m/n)^{1/2}\rceil-1,$ then $\dist_{G'}(u,v)=\infty.$

For property~\ref{itm:neighbor_incr_pro4}, by line~\ref{sta:add_new_edges}, we have $E\subseteq E'$ and $|E'|\leq |E|+\sum_{v\in V} |S_v^{(r)}|\leq |E|+n\cdot m/n=|E|+m$ where the last inequality follows by the property~\ref{itm:neighbor_S_pro5} of Lemma~\ref{lem:properties_of_S}.
\end{proof}
    \subsection{Random Leader Selection}\label{sec:leader}
Given an undirected graph $G=(V,E),$ to design a connected component algorithm, a natural way is constantly contracting the vertices in the same component. One way to do the contraction is that we randomly choose some vertices as leaders, then contract non-leader vertices to the neighbor leader vertices.

In this section, we show that if $\forall v\in V,$ the number of neighbors of $v$ is large enough, then we can just sample a small number of leaders such that for each non-leader vertex $v\in V,$ there is at least one neighbor of $v$ which is chosen as a leader. A more generalized statement is stated in the following lemma.

\begin{lemma}\label{lem:random_leader_props}
Let $V$ be a vertex set with $n$ vertices. Let $0<\gamma\leq n,\delta\in(0,1).$ For each $v\in V,$ let $S_v$ be a subset of $V\setminus\{v\}$ with size at least $\gamma-1.$ Let $l:V\rightarrow \{0,1\}$ be a random hash function such that $\forall v\in V,l(v)$ are i.i.d. Bernoulli random variables, i.e.
\begin{align*}
l(v)= 
\begin{cases}
1 & \mathrm{~with~probability~} p; \\
0 & \mathrm{~otherwise}.
\end{cases}
\end{align*}
If $p\geq\min((10\log( 2n/\delta)) / \gamma,1),$ then, with probability at least $1-\delta,$
\begin{enumerate}
\item $\sum_{v\in V}l(v)\leq \frac{3}{2}pn$;
\item $\forall v\in V,\exists u\in S_v\cup\{v\}$ such that $l(u)=1.$
\end{enumerate}
\end{lemma}
\begin{proof}
For a fixed vertex $v\in V,$  we have
{\small
\begin{align*}
&\Pr\left(\sum_{u\in S_v\cup\{v\}}( \E(l(u))-l(u) )>\frac{1}{2}\sum_{u\in S_v\cup\{v\}} \E(l(u))\right)\\
\leq~&\exp\left(-\frac{\frac{1}{2} \left(\frac{1}{2}\sum_{u\in S_v\cup\{v\}} \E(l(u))\right)^2}{\sum_{u\in S_v\cup\{v\}}\Var(l(u))+\frac{1}{3}\cdot1\cdot\frac{1}{2}\sum_{u\in S_v\cup\{v\}} \E(l(u))}\right)\\
\leq~&\exp\left(-\frac{\frac{1}{2} \left(\frac{1}{2}\sum_{u\in S_v\cup\{v\}} \E(l(u))\right)^2}{\sum_{u\in S_v\cup\{v\}}\E(l(u))+\frac{1}{3}\cdot1\cdot\frac{1}{2}\sum_{u\in S_v\cup\{v\}} \E(l(u))}\right)\\
=~&\exp\left(-\frac{3}{28}\cdot \sum_{u\in S_v\cup\{v\}} \E(l(u))\right)
=\exp\left(-\frac{3}{28}\cdot p\cdot| S_v\cup\{v\}| \right)
\leq \frac{\delta}{2n},
\end{align*}}
where the first inequality follows by Bernstein inequality and $|l(u)-E(l(u))|\leq 1$, the second inequality follows by $\Var(l(u))\leq \E(l^2(u))=\E(l(u)).$ The last inequality follows by $| S_v\cup\{v\}|\geq \gamma,$ and $p\geq  \min((10\log( 2n/\delta)) / \gamma,1).$ Since $\frac{1}{2}\sum_{u\in S_v\cup\{v\}} \E(l(u))\geq 1,$ with probability at least $1-\delta/(2n),$ $\sum_{u\in S_v\cup\{v\}}l(v)\geq 1.$ By taking union bound over all $S_v,$ with probability at least $1-\delta/2,$ $\forall v\in V,\exists u\in S_v\cup\{v\},l(u)=1.$

Similarly, we have
{\small
\begin{align*}
&\Pr\left(\sum_{u\in V}( l(u)-\E(l(u)) )>\frac{1}{2}\sum_{u\in V} \E(l(u))\right)\\
\leq~&\exp\left(-\frac{\frac{1}{2} \left(\frac{1}{2}\sum_{u\in V} \E(l(u))\right)^2}{\sum_{u\in V}\Var(l(u))+\frac{1}{3}\cdot1\cdot\frac{1}{2}\sum_{u\in V} \E(l(u))}\right)\\
\leq~&\exp\left(-\frac{\frac{1}{2} \left(\frac{1}{2}\sum_{u\in V} \E(l(u))\right)^2}{\sum_{u\in V}\E(l(u))+\frac{1}{3}\cdot1\cdot\frac{1}{2}\sum_{u\in V} \E(l(u))}\right)\\
=~&\exp\left(-\frac{3}{28}\cdot \sum_{u\in V} \E(l(u))\right)\\
=~&\exp\left(-\frac{3}{28}\cdot p\cdot| V| \right)
\leq \frac{\delta}{2n}
\leq \frac{\delta}{2}.
\end{align*}
}
Since $\sum_{u\in V}\E(l(u))=p\cdot n,$ with probability at least $1-\delta/2,$ $\sum_{u\in V}l(u)\leq 1.5pn.$

By taking union bound, with probability at least $1-\delta,\sum_{u\in V}l(u)\leq 1.5pn$ and $\forall v\in V, \exists u\in S_v\cup\{v\},l(u)=1.$
\end{proof}

If the number of neighbors of each vertex is not large, then we can still have a constant fraction of vertices which can contract to a leader.

\begin{lemma}\label{lem:random_leader_low_prob}
Let $V$ be a vertex set with $n$ vertices. Let $S_v$ be a subset of $V\setminus\{v\}$ with size at least $1.$ Let $l:V\rightarrow \{0,1\}$ be a random hash function such that $\forall v\in V,l(v)$ are i.i.d. Bernoulli random variables, i.e.
\begin{align*}
l(v)=
\begin{cases}
1 & \mathrm{~with~probability~} \frac{1}{2} ; \\
0 & \mathrm{~otherwise}.
\end{cases}
\end{align*}
Let $L=\{v\in V\mid l(v)=1\}\cup\{v\in V\mid \forall u\in S_v\cup\{v\},l(u)=0\}.$ $\E(L)\leq 0.75n.$
\end{lemma}
\begin{proof}
For $v\in V,$ $\Pr(l(v)=1)=\frac{1}{2}.$ Let $u\in S_v.$ Then $\Pr(\forall x\in S_v\cup\{v\},l(x)=0)\leq \Pr(l(v)=0,l(u)=0)=0.25.$ $\E(|L|)=\sum_{v\in V}\Pr(v\in L)\leq 0.75n.$
\end{proof}
    \subsection{Tree Contraction Operation}
In this section, we introduce the contraction operation.
 Firstly, let us introduce the concept of the parent pointers which can define a rooted forest.
\begin{definition}\label{def:parent_pointers}
Given a set of vertices $V,$ let $\p:V\rightarrow V$ satisfy that $\forall v\in V,\exists i>0$ such that $\p^{(i)}(v)=\p^{(i+1)}(v),$ where $\forall v\in V,j>0,\p^{(j)}(v)$ is defined as $\p(\p^{(j-1)}(v)),$ and $\p^{(0)}(v)=v.$ Then, we call such $\p$ a set of parent pointers on $V$. For $v\in V,$ if $\p(v)=v,$ then we say $v$ is a root of $\p.$ $\p$ can have more than one root. The depth of $v\in V,$ $\dep_{\p}(v)$ is the smallest $i\in \mathbb{Z}_{\geq 0}$ such that $\p^{(i)}(v)=\p^{(i+1)}(v).$ The root of $v\in V,$ $\p^{(\infty)}(v)$ is defined as $\p^{(\dep_{\p}(v))}(v).$ The depth of $\p,$ $\dep(\p)$ is defined as $\max_{v\in V}\dep_{\p}(v).$
\end{definition}
It is easy to see that a set of parent pointers $\p$ on $V$ formed a rooted forest on $V$. For a vertex $v\in V,$ if $\p(v)=v,$ then $v$ is a root in the forest. Otherwise $\p(v)$ is the parent of $v$ in the forest.

In the following, we define the union operation of several sets of parent pointers.
\begin{definition}\label{def:union_of_parent_pointers}
Let $\p_1:V_1\rightarrow V_1,\p_2:V_2\rightarrow V_2,\cdots,\p_k:V_k\rightarrow V_k$ be $k$ sets of parent pointers on vertex sets $V_1,V_2,\cdots,V_k$ respectively, where $\forall i\not=j\in[k],V_i\cap V_j=\emptyset$. Then $\p=\p_1\cup \p_2\cup\cdots \cup \p_k$ is a set of parent pointers on the vertex set $V_1\cup V_2\cup \cdots\cup V_k$ such that $\forall i\in [k],v\in V_k,\p(v)=\p_i(v).$
\end{definition}

Now we focus on the parent pointers which can preserve the connectivity of the graph.
\begin{definition}\label{def:compatible_parent_pointers}
Given a graph $G=(V,E)$ and a set of parent pointers $\p$ on $V,$ if $\forall v\in V,$ we have $\dist_G(v,\p(v))<\infty,$ then $\p$ is \textit{compatible} with $G$.
\end{definition}
It is easy to show the following fact:
\begin{fact}\label{fac:same_root}
Given a graph $G=(V,E)$ and a set of parent pointers $\p$ which is compatible with $G,$ then $\forall u,v\in V$ with $\p^{(\infty)}(u)=\p^{(\infty)}(v),$ we have $\dist_G(u,v)<\infty.$
\end{fact}
\begin{proof}
By the definition of compatible, $\forall v\in V,\dist_G(v,\p(v))<\infty.$ By induction, $\forall l\in\mathbb{Z}_{>0},v\in V,$ we have $\dist_G(v,\p^{(l)}(v))\leq \dist_G(v,\p^{(l-1)}(v))+\dist_G(\p^{(l-1)}(v),\p^{(l)}(v))<\infty.$ Thus, for any pair of vertices $u,v\in V,$ if $\p^{(\infty)}(u)=\p^{(\infty)}(v),$ then $\dist_G(u,v)\leq\dist_G(u,\p^{(\infty)}(u))+\dist_G(\p^{(\infty)}(v),v)<\infty.$
\end{proof}

 In this section, we describe a procedure which can be used to reduce the number of vertices. The input of the procedure is an undirected graph $G=(V,E)$ and a set of parent pointers $\p:V\rightarrow V$, where $\p$ is compatible with $G$. The output of the procedure will be the root of each vertex in $V$ and an undirected graph $G'=(V',E')$ which satisfies $V'=\{v\in V\mid \p(v)=v\},E'=\{(u,v)\in V'\times V'\mid u\not=v,\exists (p,q)\in E,\p^{(\infty)}(p)=u,\p^{(\infty)}(q)=v\}.$ Notice that $V'$ only contains all the roots in the forest induced by $\p,$ and $|E'|\leq |E|.$

\begin{algorithm}[t]
\caption{Tree Contraction Operation}\label{alg:tree_contraction}
\begin{algorithmic}[1]
\Procedure{\textsc{TreeContraction}}{$G=(V,E),\p: V\rightarrow V$}\Comment{Lemma~\ref{lem:contraction_properties}, Corollary~\ref{cor:tree_contract_conn}}
\State \Comment{ Output: $G'=(V',E'),\p^{(\infty)}(v)$ for all $v\in V$}
\State Initially, for each $v\in V$ let $g^{(0)}(v)\leftarrow\p(v).$ Let $V'=\emptyset,E'=\emptyset.$
\State $l \leftarrow 0.$
\For{$\exists v\in V,\p(g^{(l)}(v))\not=g^{(l)}(v)$} \label{sta:tree_contraction_condition}
    \State $l\leftarrow l+1.$
    \State For each $v\in V,$ compute $g^{(l)}(v)=g^{(l-1)}(g^{(l-1)}(v)).$\Comment{$g^{(l)}$ is $\p^{(2^l)}$}
\EndFor
\State $r \leftarrow l.$ \label{sta:r_is_final_l}\Comment{$r$ is the number of iterations, and is used in the analysis.}
\State For $v\in V,$ if $\p(v)=v,$ let $V'\leftarrow V'\cup\{v\}.$\label{sta:Vprime}
\State For $(u,v)\in E,$ if $g^{(r)}(u)\not=g^{(r)}(v),$ let $E'\leftarrow E'\cup\{(g^{(r)}(u),g^{(r)}(v))\}.$ \label{sta:Eprime}

\Comment{$\forall v\in V,$ contract $v$ to $\p^{(\infty)}(v)$}
\State \Return $g^{(r)}(v)$ as $\p^{(\infty)}(v)$ for all $v\in V,$ and $G'=(V',E')$
\EndProcedure
\end{algorithmic}
\end{algorithm}

\begin{lemma}\label{lem:contraction_properties}
Let $G=(V,E)$ be an undirected graph, $\p:V\rightarrow V$ be a set of parent pointers (See Definition~\ref{def:parent_pointers}). 
Then $\textsc{TreeContraction}(G,\p)$ (See Algorithm~\ref{alg:tree_contraction}) will output output $(G',g^{(r)})$ with $r\leq \lceil\log \dep(\p)\rceil$ satisfies the following properties:
\begin{enumerate}
\item $\forall v\in V,$ $g^{(r)}(v)=\p^{(\infty)}(v).$\label{itm:tree_return_pro1}
\item $V'=\{v\in V\mid\p(v)=v\}.$\label{itm:tree_return_pro2}
\item $E'=\{(u,v)\in V'\times V'\mid u\not=v,\exists (p,q)\in E,\p^{(\infty)}(p)=u,\p^{(\infty)}(q)=v\}.$\label{itm:tree_return_pro3}
\end{enumerate}
\end{lemma}

\begin{proof}
One crucial observation is the following claim.
\begin{claim}\label{cla:exponential_boost_g}
$\forall l\in\{0,1,\cdots,r\},v\in V,$ we have $g^{(l)}(v)=\p^{(2^l)}(v).$
\end{claim}
\begin{proof}
The proof is by induction. When $l=0,$ $\forall v\in V, g^{(0)}(v)=\p(v)=\p^{(1)}(v),$ the claim is true. Suppose for $l-1,$ we have $\forall v\in V,g^{(l-1)}(v)=\p^{(2^{l-1})}(v),$ then $\forall v\in V,g^{(l)}(v)=g^{(l-1)}(g^{(l-1)}(v))=\p^{(2^{l-1})}(\p^{(2^{l-1})}(v))=\p^{(2^l)}(v).$ So the claim is true.
\end{proof}

If $r>\lceil \log \dep(\p)\rceil,$ then $r-1\geq \lceil\log \dep(\p)\rceil.$ Due to claim~\ref{cla:exponential_boost_g}, we have $\forall v\in V,g^{(r-1)}(v)=\p^{(2^{r-1})}(v)=\p^{(\infty)}(v).$ Due to the condition in line~\ref{sta:tree_contraction_condition}, the loop will stop when $l\leq r-1$ which leads to a contradiction to line~\ref{sta:r_is_final_l}. Thus, at the end of the algorithm, $r$ should be at most $\lceil \log \dep(\p)\rceil.$

Since we have $\forall v\in V,\p(g^{(r)}(v))=g^{(r)}(v)$ at the end of the Algorithm~\ref{alg:tree_contraction}, $\forall v\in V,g^{(r)}(v)$ must be $\p^{(\infty)}(v)$.
Then due to line~\ref{sta:Vprime} and line~\ref{sta:Eprime}, we have $V'=\{v\in V\mid \p(v)=v\},E'=\{(u,v)\in V'\times V'\mid u\not=v,\exists (p,q)\in E,\p^{(\infty)}(p)=u,\p^{(\infty)}(q)=v\}.$
\end{proof}

\begin{definition}\label{def:num_it_tree_contract}
Let $G=(V,E)$ be an undirected graph, $\p:V\rightarrow V$ be a set of parent pointers (See Definition~\ref{def:parent_pointers}). 
Then the number of iteration of $\textsc{TreeContraction}(G,\p)$ is defined as the value of $r$ at the end of the procedure.
\end{definition}

\begin{corollary}[Preserved connectivity and diameter]\label{cor:tree_contract_conn}
Let $G=(V,E)$ be an undirected graph, $\p:V\rightarrow V$ be a set of parent pointers (See Definition~\ref{def:parent_pointers}) which is compatible (See Definition~\ref{def:compatible_parent_pointers}) with $G$. Then at the end of the Algorithm~\ref{alg:tree_contraction}, $r\leq \lceil\log \dep(\p)\rceil$ and the output $(G',g^{(r)})$ will satisfy the following properties:
\begin{enumerate}
\item $\diam(G')\leq \diam(G).$ \label{itm:tree_contracted_pro1}
\item $\forall u,v\in V,\dist_G(u,v)<\infty\Rightarrow\dist_{G'}(\p^{(\infty)}(u),\p^{(\infty)}(v))<\infty.$ \label{itm:tree_contracted_pro2}
\item $\forall u,v\in V,\dist_G(u,v)<\infty\Leftarrow\dist_{G'}(\p^{(\infty)}(u),\p^{(\infty)}(v))<\infty.$ \label{itm:tree_contracted_pro3}
\end{enumerate}
\end{corollary}
\begin{proof}
By Lemma~\ref{lem:contraction_properties}, we have $r\leq \lceil\log \dep(\p)\rceil,$ $V'=\{v\in V\mid \p(v)=v\}$ and $E'=\{(u,v)\in V'\times V'\mid u\not=v,\exists (p,q)\in E,\p^{(\infty)}(p)=u,\p^{(\infty)}(q)=v\}.$

For any two vertices $u,v\in V$ which are in the same connected component in $G,$ then there should be a path $u=u_1\rightarrow u_2\rightarrow\cdots\rightarrow u_p=v$ in graph $G$. So $\forall i\in [p-1],(u_i,u_{i+1})\in E$ which means that either $\p^{(\infty)}(u_i)=\p^{(\infty)}(u_{i+1})$ or $(\p^{(\infty)}(u_i),\p^{(\infty)}(u_{i+1}))\in E'.$ Thus, $\p^{(\infty)}(u_1)\rightarrow \p^{(\infty)}(u_2)\rightarrow\cdots\rightarrow \p^{(\infty)}(u_p)$ is a valid path in $G'$, and the length of this path in $G'$ is at most $p.$ Thus, the properties~\ref{itm:tree_contracted_pro1} and~\ref{itm:tree_contracted_pro2} are true.

For any two vertices $u,v\in V$ which are not in the same connected component in $G,$ but there is a path $\p^{(\infty)}(u)=u'_1\rightarrow u'_2\rightarrow \cdots\rightarrow u'_p=\p^{(\infty)}(v)$ in $G',$ then it means that there exists vertices $u_{1,1},u_{1,2},u_{2,1},u_{2,2},\cdots,u_{p,1},u_{p,2}\in V$ which satisfies
\begin{enumerate}
\item[(a)] $\forall i\in[p-1],(u_{i,2},u_{i+1,1})\in E,\p^{(\infty)}(u_{i,2})=u'_i,\p^{(\infty)}(u_{i+1,1})=u'_{i+1}.$
\item[(b)] $u_{1,1}=u,u_{p,2}=v.$
\item[(c)] $\forall i\in[p],\p^{(\infty)}(u_{i,1})=\p^{(\infty)}(u_{i,2}).$ By Fact~\ref{fac:same_root}, we have $\dist_{G}(u_{i,1},u_{i,2})<\infty.$
\end{enumerate}
Thus, there exists a path from $u$ to $v$. This contradicts to that $u,v$ are not in the same connected component. Therefore, property~\ref{itm:tree_contracted_pro3} is also true.
\end{proof}
    \subsection{Connectivity Algorithm}\label{sec:batch_leader}
In this section, we described a batch algorithm for graph connectivity/connected components problem. The input is an undirected graph $G=(V,E),$ a space/rounds trade-off parameter $m$, and the rounds parameter $r\leq |V|.$ The output is a function $\col:V\rightarrow V$ such that $\forall u,v\in V,\dist_G(u,v)<\infty\Leftrightarrow \col(u)=\col(v).$

The algorithm is described in Algorithm~\ref{alg:batch_algorithm2}. The following theorem shows the correctness of Algorithm~\ref{alg:batch_algorithm2}.

\begin{algorithm}[h]
\caption{Graph Connectivity}\label{alg:batch_algorithm2}
\begin{algorithmic}[1]
\small
\Procedure{\textsc{Connectivity}}{$G=(V,E),m,r$} \Comment{Theorem~\ref{thm:correct_rand_leader_alg}, Theorem~\ref{thm:alg2_time_prob}}
\State Output: FAIL or $\col:V\rightarrow V.$
\State $n \leftarrow |V|$
\State $\forall v\in V,$ $h_0(v)\leftarrow  \nul .$
\State $G_0=(V_0,E_0)=G,$ i.e. $V_0=V,E_0=E.$
\State $n_0=n.$
\For{$i=1\rightarrow r$}
\State $\forall v\in V,$ $h_i(v)\leftarrow  \nul .$\label{sta:alg2_init_null} \Comment{$h_i(v)$ is the vertex that $v$ contracts to}
\State $G'_i=(V'_i,E'_i)=\textsc{NeighborIncrement}(m,G_{i-1}).$  \label{sta:alg2_neighbor_incr} \Comment{Algorithm~\ref{alg:neighbor_increment}}
\State Compute $V''_i=\{v\in V'_i\mid |\Gamma_{G'_i}(v)|\geq \lceil(m/n_{i-1})^{1/2}\rceil-1\}.$\label{sta:alg2_Vdoubleprime_create}
\State Compute $E''_i=\{(u,v)\in E_{i-1} \mid u\in V''_i,v\in V''_i\}.$ \label{sta:alg2_Edoubleprime_create}
\State $G''_i=(V''_i,E''_i).$ \Comment{$G''_i$ is obtained by removing all the small components of $G_i$}
\State Let $\gamma_i=\lceil(m/n_{i-1})^{1/2}\rceil,p_i=\min((30\log(n)+100 )/ \gamma_i,1/2).$\label{sta:alg2_start_to_select_leader}
\State Let $l_i:V''_i\rightarrow \{0,1\}$ be a random hash function such that $\forall v\in V''_i,l_i(v)$ are i.i.d. Bernoulli random variables, and $\Pr(l_i(v)=1)=p_i$.\label{sta:alg2_sampling}
\State Let $L_i=\{v\in V''_i\mid l_i(v)=1\}\cup\{v\in V''_i\mid \forall u\in \Gamma_{G'_i}(v)\cup \{v\},l_i(u)=0\}.$\label{sta:alg2_leader_set}\Comment{$L_i$ are leaders}
\State $\forall v\in V''_i$ with $v\in L_i,$ let $\p_i(v)=v.$\label{sta:alg2_assign_pointer1}
\State $\forall v\in V''_i$ with $v\not\in L_i,$ let $\p_i(v)=\min_{u\in L_i\cap (\Gamma_{G'_i}(v)\cup \{v\})} u.$\label{sta:alg2_assign_pointer2}\Comment{Non-leader finds a leader.}
\State Let $((V_i,E_i),g_i^{(r'_i)})=\textsc{TreeContraction}(G''_i,\p_i).$ \label{sta:alg2_tree_contract} \Comment{Algorithm~\ref{alg:tree_contraction}}
\State $G_i=(V_i,E_i).$
\State $n_i=|V_i|.$
\State For each $v\in V'_i\setminus V''_i,$ let $h_i(v)\leftarrow \min_{u\in \Gamma_{G'_i}(v)\cup\{v\}} u.$ \label{sta:VspminusVdp}\Comment{Contract small component to one vertex}
\State For each $v\in V''_i\setminus V_i,$ let $h_i(v)\leftarrow g_i^{(r'_i)}(v).$ \label{sta:VdpminusV}\Comment{Contract non-leader to leader}
\State For each $v\in V,$ if $h_{i-1}(v)\not= \nul ,$ then let $h_i(v)=h_{i-1}(v).$ \label{sta:hlasttime}
\EndFor
\State If $n_r\not=0,$ return FAIL.
\State $((\wh{V},\wh{E}),\col)=\textsc{TreeContraction}(G,h_r).$  \label{sta:final_output_alg2} \Comment{Algorithm~\ref{alg:tree_contraction}}
\State \Return $\col.$
\EndProcedure
\end{algorithmic}
\end{algorithm}

\begin{theorem}[Correctness of Algorithm~\ref{alg:batch_algorithm2}]\label{thm:correct_rand_leader_alg}
Let $G=(V,E)$ be an undirected graph, $m\geq 4|V|$, and $r\leq |V|$ be the rounds parameter. If $\textsc{Connectivity}(G,m,r)$ (Algorithm~\ref{alg:batch_algorithm2}) does not output FAIL, then $\forall u,v\in V,$ we have $\dist_G(u,v)<\infty \Leftrightarrow \col(u)=\col(v).$
\end{theorem}
\begin{proof}
Firstly, we show that the input of line~\ref{sta:alg2_tree_contract} is valid.
\begin{claim}\label{cla:alg2_tree_valid}
$\forall i\in[r],$ $\p_i$ is a set of parent pointers on $V''_i$, (See Definition~\ref{def:parent_pointers}) and is compatible (See Definition~\ref{def:compatible_parent_pointers}) with $G''_i.$
\end{claim}
\begin{proof}
$\forall v\in V''_i,$ if $v\in L_i,$ then $\p_i(v)=v.$ For $v\in V''_i\setminus L_i,$  due to property~\ref{itm:neighbor_incr_pro3} of Lemma~\ref{lem:properties_of_neighbor_increment}, we have $\p_i(v)\in V''_i.$ Since $\p_i(v)\in L_i,$ we have $\p_i(\p_i(v))=\p_i(v).$ Thus, $\p_i:V''_i\rightarrow V''_i$ is a set of parent pointers on $V''_i.$ Due to property~\ref{itm:neighbor_incr_pro2} of Lemma~\ref{lem:properties_of_neighbor_increment} and $\dist_{G'_{i}}(\p_i(v),v)<\infty$, we know that $\dist_{G_{i-1}}(\p_i(v),v)<\infty.$ Thus, $\dist_{G''_i}(\p_i(v),v)<\infty.$ It implies that $\p_i$ is compatible with $G''_i.$
\end{proof}
The following claim shows that the number of the remaining vertices cannot increase after each round.
\begin{claim}\label{cla:V_always_subset}
If $\textsc{Connectivity}(G,m,r)$ does not output FAIL, then $\forall i\in[r],V_i\subseteq V''_i\subseteq V'_i=V_{i-1}.$
\end{claim}
\begin{proof}
Let $i\in[r].$ 
Due to Claim~\ref{cla:alg2_tree_valid}, the input of line~\ref{sta:alg2_tree_contract} is valid. Then, we can apply property~\ref{itm:tree_return_pro2} of Lemma~\ref{lem:contraction_properties} to get $V_i\subseteq V''_i.$ By the construction of $V''_i$ we have $V''_i\subseteq V'_i.$ Since the procedure $\textsc{NeighborIncrement}(m,G_{i-1})$ (Algorithm~\ref{alg:neighbor_increment}) does not change the vertex set, we have $V'_i=V_{i-1}.$
\end{proof}
Now, we show that $\forall u,v\in V_i,\dist_{G_i}(u,v)<\infty \Leftrightarrow \dist_{G}(u,v)<\infty$.
\begin{claim}\label{cla:alg2_always_connect}
If $\textsc{Connectivity}(G,m,r)$ does not output {\rm FAIL}, then $\forall i\in [r],\forall u,v\in V_i,$ we have $\dist_{G_i}(u,v)<\infty \Leftrightarrow \dist_{G}(u,v)<\infty.$
\end{claim}
\begin{proof}
The proof is by induction. Suppose $\forall u,v\in V_{i-1},\dist_{G_{i-1}}(u,v)<\infty \Leftrightarrow \dist_{G}(u,v)<\infty.$ $\forall w,z\in V_i,$ according to Claim~\ref{cla:V_always_subset}, $w,z\in V''_{i}$. By property~\ref{itm:tree_contracted_pro2},\ref{itm:tree_contracted_pro3} of Lemma~\ref{cor:tree_contract_conn}, and property~\ref{itm:tree_return_pro2} of Lemma~\ref{lem:contraction_properties}, $\dist_{G_i}(w,z)<\infty\Leftrightarrow \dist_{G''_i}(w,z)<\infty.$
Due to property~\ref{itm:neighbor_incr_pro2},\ref{itm:neighbor_incr_pro3} of Lemma~\ref{lem:properties_of_neighbor_increment}, there is no edge in $E_{i-1}$ between $V''_i$ and $V'_i\setminus V''_i.$ According to Claim~\ref{cla:V_always_subset}, $w,z\in V_{i-1}$. Thus, $\dist_{G''_i}(w,z)<\infty\Leftrightarrow \dist_{G_{i-1}}(w,z)<\infty.$  By induction hypothesis, we have $\forall w,z\in V_i, \dist_{G_i}(w,z)<\infty\Leftrightarrow \dist_G(w,z).$
\end{proof}
The following claim states that once a vertex $v\in V$ is contracted to an another vertex, it will never be operated.
\begin{claim}\label{cla:whether_in_v_i}
Suppose $\textsc{Connectivity}(G,m,r)$ does not output {\rm FAIL}. $\forall i\in \{0,1,\cdots,r\},$ $v\in V,$ we have $h_i(v)= \nul \Leftrightarrow v\in V_i.$ Furthermore, $\forall v\in V,\exists j\in[r]$ such that $h_0(v)=h_1(v)=\cdots=h_{j-1}(v)= \nul $ and $h_j(v)=h_{j+1}(v)=\cdots=h_r(v)\not= \nul ,{\dist}_G(v,h_r(v))<\infty.$
\end{claim}
\begin{proof}
When $i=0,$ $\forall v\in V,$ $h_0(v)= \nul ,v\in V_0=V.$ Suppose it is true that $\forall v\in V,h_{i-1}(v)= \nul \Leftrightarrow v\in V_{i-1}.$ If $v\not\in V_i,$ according to Claim~\ref{cla:V_always_subset}, there are three cases: $v\in V''_i\setminus V_i,v\in V'_i\setminus V''_i,v\not\in V_{i-1}.$ In the first case, due to line~\ref{sta:VdpminusV}, $h_i(v)\not= \nul .$ In the second case, due to line~\ref{sta:VspminusVdp}, $h_i(v)\not= \nul ,$ In the third case, due to line~\ref{sta:hlasttime}, $h_i(v)\not= \nul .$ If $h_i(v)= \nul ,$ then $h_i(v)$ cannot be updated by line~\ref{sta:VspminusVdp}, line~\ref{sta:VdpminusV} or line~\ref{sta:hlasttime} which implies that $v\in V_{i-1},v\not\in V'_i\setminus V''_i,v\not\in V''_i\setminus V_i.$ Thus, $v\in V_i.$

Since the procedure does not FAIL, we have $n_r=0$ which means that $\forall v\in V,h_r(v)\not= \nul .$
Notice that by line~\ref{sta:hlasttime}, if $h_{i-1}(v)\not= \nul ,$ then $h_{i}(v)=h_{i-1}(v).$ Thus, $\forall v\in V,\exists j\in[r]$ such that $h_0(v)=h_1(v)=\cdots=h_{j-1}(v)= \nul $ and $h_j(v)=h_{j+1}(v)=\cdots=h_r(v)\not= \nul .$

For $v\in V,$ if $h_j(v)\not= \nul $ and $h_{j-1}(v)= \nul ,$ then $h_j(v)$ can only be updated by ~\ref{sta:VspminusVdp} or line~\ref{sta:VdpminusV}. In both cases, $\dist_{G_{j-1}}(v,h_j(v))<\infty.$ By Claim~\ref{cla:alg2_always_connect}, we have that $\dist_{G}(v,h_j(v))<\infty.$
\end{proof}

In the following, we show that $h_r$ is a rooted tree such that $\dist_G(u,v)<\infty\Leftrightarrow u,v$ have the same root. Due to Claim~\ref{cla:whether_in_v_i}, if $\textsc{Connectivity}(G,m,r)$ does not output FAIL, then $n_r=0$ which implies that $\forall v\in V,h_r(v)\not= \nul .$ Thus, we can define $h_r^{(k)}(v)$ for $k\in\mathbb{Z}_{>0}$ as applying $h_r$ on $v$ $k$ times. $\forall v\in V,$ by Claim~\ref{cla:whether_in_v_i}, let $j\in[r]$ satisfy that $h_j(v)\not= \nul $ and $h_{j-1}(v)= \nul .$ If $h_j(v)$ is updated by~line~\ref{sta:VdpminusV}, then $h_j(h_j(v))= \nul .$ If $h_j(v)$ is updated by~line~\ref{sta:VspminusVdp}, then $h_j(h_j(v))=h_j(v).$ In both cases, $h_j$ cannot create a cycle. Thus, we can define $h_r^{(\infty)}(v)=h_r^{(k)}(v)$ for some $k$ which satisfies $h_r(h_r^{(k)}(v))=h_r^{(k)}(v).$

\begin{claim}\label{cla:root_equiv}
Suppose $\textsc{Connectivity}(G,m,r)$ does not output {\rm FAIL}. Then $\forall u,v\in V,$ we have $\dist_G(u,v)<\infty\Leftrightarrow h_r^{(\infty)}(u)=h_r^{(\infty)}(v).$
\end{claim}
\begin{proof}
Let $u,v\in V.$ By Claim~\ref{cla:whether_in_v_i}, if $h_r^{\infty}(u)=h_r^{\infty}(v)$ we have $\dist_G(u,v)<\infty.$

If $\dist_G(u,v)<\infty,$ then let $u'=h_r^{(\infty)}(u),v'=h_r^{(\infty)}(v).$ By Claim~\ref{cla:whether_in_v_i}, $\dist_G(u',v')\leq \dist_G(u,u')+\dist_G(u,v)+\dist_G(v,v')<\infty,$ and we can find $j\in[r]$ such that $h_j(u')\not= \nul ,h_{j-1}(u')= \nul .$ Without loss of generality, we can assume $h_{j-1}(v')= \nul $ (otherwise we can swap $u'$ and $v'$). Due to Claim~\ref{cla:whether_in_v_i}, $u',v'\in V_{i-1}.$ Since $h_j(u')=h_r(u')=u',$ $h_j(u')$ can be only updated by line~\ref{sta:VspminusVdp}, and $u'\in V'_j\setminus V''_j$. Then due to property~\ref{itm:neighbor_incr_pro3} of Lemma~\ref{lem:contraction_properties}, $v'$ should be in $\Gamma_{G'_i}(u)\cup\{u\}.$ Since $h_j(v')=h_r(v')=v',$ we can conclude that $u'=v'.$
\end{proof}

If $\textsc{Connectivity}(G,m,r)$ does not output FAIL, then in line~\ref{sta:final_output_alg2}, $\col$ is exactly $h_r^{(\infty)}.$ By Claim~\ref{cla:root_equiv}, we have $\forall u,v\in V,\dist_G(u,v)<\infty\Leftrightarrow \col(u)=\col(v).$

\end{proof}

Now let us consider the number of iterations of Algorithm~\ref{alg:batch_algorithm2} and the success probability.

\begin{definition}[Total iterations]\label{def:total_iter_connectivity}
Let $G=(V,E)$ be an undirected graph, $\poly(n)\geq m> 4n,$ and $r\leq n$ be the rounds parameter where $n$ is the number of vertices in $G$.
The total number of iterations of $\textsc{Connectivity}(G,m,r)$ (Algorithm~\ref{alg:batch_algorithm2}) is defined as $\sum_{i=1}^r (k_i+r'_i),$
where $k_i$ denotes the number of iterations (See Definition~\ref{def:neighbor_incr_num_iter}) of $\textsc{NeighborIncrement}(m,G_{i-1})$ (see line~\ref{sta:alg2_neighbor_incr}), and $r_i'$ denotes the number of iterations (See Definition~\ref{def:num_it_tree_contract}) of $\textsc{TreeContraction}(G''_i,\p_i)$ (see line~\ref{sta:alg2_tree_contract}).
\end{definition}

\begin{theorem}[Success probability and total iterations]\label{thm:alg2_time_prob}
Let $G=(V,E)$ be an undirected graph, $\poly(n)\geq m> 4n,$ and $r\leq n$ be the rounds parameter where $n=|V|$. Let $c>0$ be a sufficiently large constant. If $r\geq c\log \log_{m/n} (n)$, then with probability at least $0.98$, $\textsc{Connectivity}(G,m,r)$ (Algorithm~\ref{alg:batch_algorithm2}) will not return {\rm FAIL}.
If $\textsc{Connectivity}(G,m,r)$ succeeds,
 let $k_i$ denote the number of iterations (See Definition~\ref{def:neighbor_incr_num_iter}) of $\textsc{NeighborIncrement}(m,G_{i-1})$ (see line~\ref{sta:alg2_neighbor_incr}), and let $r'_i$ denote the number of iterations of (See Definition~\ref{def:num_it_tree_contract}) of $\textsc{TreeContraction}(G''_i,\p_i)$ (see line~\ref{sta:alg2_tree_contract}),
then
\begin{enumerate}
\item $\forall i\in[r],$ $r_i'= 0$. \label{itm:alg2_suc_pro1}
\item $\forall i\in[r],$ $k_i$ is at most $\lceil\log(\diam(G))\rceil+1.$ \label{itm:alg2_suc_pro2}
\item The number of iterations of line~\ref{sta:final_output_alg2} is at most $\lceil\log r\rceil.$ \label{itm:alg2_suc_pro3}
\item $\sum_{i=1}^{r} k_i\leq O(r\log(\diam(G))).$ \label{itm:alg2_suc_pro4}
\end{enumerate}
Let $c_1>0$ be a sufficiently large constant. If $m\geq c_1 n\log^4 n,$ then with probability at least $0.99,$ $\sum_{i=1}^{r} k_i\leq O(\log(\diam(G))\log\log_{\diam(G)}(n)).$ If $m< c_1 n\log^4 n$, then with probability at least $0.98,$ $\sum_{i=1}^{r} k_i\leq O(\log(\diam(G))\log\log_{\diam(G)}(n)+(\log\log(n))^2).$
\end{theorem}
\begin{proof}
Suppose $\textsc{Connectivity}(G,m,r)$ succeeds. 
Property~\ref{itm:alg2_suc_pro1} follows by $\forall v\in V''_i,\p_i(\p_i(v))=\p_i(v)$ and Lemma~\ref{lem:contraction_properties}. Property~\ref{itm:alg2_suc_pro2} follows by $\diam(G_r)\leq \diam(G''_r)\leq \diam(G'_r)\leq \diam(G_{r-1})\leq \diam(G''_{r-1})\leq \diam(G'_{r-1})\leq \cdots \leq \diam(G_0)=\diam(G)$ and property~\ref{itm:neighbor_incr_pro1} of Lemma~\ref{lem:properties_of_neighbor_increment}. Property~\ref{itm:alg2_suc_pro3} follows by the depth of $h_r$ is at most $r$ and Lemma~\ref{lem:contraction_properties}.
Property~\ref{itm:alg2_suc_pro4} follows by property~\ref{itm:alg2_suc_pro2}.

Now let us prove the success probability. Let $i\in[r].$
If $p_i<0.5,$ then we can apply Lemma~\ref{lem:random_leader_props} on vertex set $V''_i,$ parameter $\gamma_i,$ and hash function $l_i.$ Notice that the set $S_v$ in the statement of Lemma~\ref{lem:random_leader_props} is $\Gamma_{G'_i}(v)$ in the algorithm. Notice that $|V''_i|\leq n.$ Then in the $i^{\text{th}}$ round, if $p_i<0.5,$ then with probability at most $1/(100n^2),$ $L_i$ will be $\{v\in V''_i\mid l_i(v)=1\},$ and $n_i=|L_i|\leq 1.5p_in_{i-1}.$ By taking union bound over all $i\in[r],$ we have that with probability at least $0.99,$ event $\mathcal{E}$ happens: for all $i\in[r],$ if $p_i<0.5,$ then  $n_i\leq 1.5p_i n_{i-1}\leq 0.75n_{i-1}.$ Suppose $\mathcal{E}$ happens. For $i\in[r],p_i=0.5,$ if we apply Lemma~\ref{lem:random_leader_low_prob}, then condition on $n_{i-1},$ we have $\E(n_i)\leq 0.75 n_{i-1}.$ Thus, we know $\forall i\in[r],\E(n_i)\leq 0.75 \E(n_{i-1})\leq 0.75^i n.$

Next, we discuss the case for $p_0=0.5$ and the case for $p_0<0.5$ separately.

If $p_0=0.5,$ then $m\leq n\cdot (600\log n)^4.$
By Markov's inequality, when $i^*\geq 4\log_{4/3}(6000\log n),$ with probability at least $0.99,$ $n_{i^*}\leq n/(600\log n)^4$ and thus $p_{i^*}<0.5.$ Condition on this event and $\mathcal{E}$, we have
{\small
\begin{align*}
n_r&\leq \frac{\left(\frac{\left(\frac{n_{i^*}^{1.5}}{m^{0.5}}(45\log n+150)\right)^{1.5}}{m^{0.5}}(45\log n+150)\right)^{\cdots}}{\cdots}&~\text{(Apply $r'=r-i^*$ times)}\\
&= \frac{n_{i^*}^{1.5^{r'}}}{m^{1.5^{r'}-1}}(45\log n+150)^{2\cdot(1.5^{r'}-1)}\\
&=n_{i^*}/(m/n_{i^*})^{1.5^{r'}-1}\cdot (45\log n+150)^{2\cdot(1.5^{r'}-1)}\\
&\leq n/\left(m/\left(n_{i^*}(45\log n+150)^2\right)\right)^{1.5^{r'}-1}\\
&\leq n/\left(m/\left(n_{i^*}(45\log n+150)^2\right)\right)^{1.5^{r'/2}}\\
&\leq n/\left(m/n\right)^{1.5^{r'/2}}
\leq \frac{1}{2},
\end{align*}
}
where the second inequality follows by $n_{i^*}\leq n,$ the third inequality follows by $r'\geq 5,$ the forth inequality follows by $n_{i^*}\leq n/(600\log n)^4,$ and the last inequality follows by $r'\geq \frac{2}{\log 1.5}\log\log_{m/n}(2n).$
Since $4n\leq m\leq n\cdot (600\log n)^4,\log\log_{m/n} n=\Theta(\log\log n).$ Let $c>0$ be a sufficiently large constant.
Thus, when $r\geq c\log\log_{m/n} n\geq  i^*+r'=4\log(6000\log n)/\log(4/3)+\frac{2}{\log 1.5}\log\log_{m/n}(2n),$ with probability at least $0.98,$ $\textsc{Connectivity}(G,m,r)$ will not fail. 

Since property~\ref{itm:neighbor_incr_pro1} of Lemma~\ref{lem:properties_of_neighbor_increment}, we have $k_i\leq O(\log (\min(m/n_{i-1},\diam(G)))).$ Thus,
{\small
\begin{align*}
 & ~ \sum_{i=1}^r k_i 
 = \sum_{i=1}^{i^*} k_i +  \sum_{i=i^*+1}^{r} k_i
\leq  O\left((\log \log n)^2\right)+ \sum_{i=i^*+1}^{r} k_i\\
\leq & ~ O\left((\log \log n)^2\right)+ \sum_{i:i\geq i^*+1,m/n_{i-1}\leq \diam(G)} k_i+\sum_{i:i\leq r,m/n_{i-1}> \diam(G)} k_i\\
\leq & ~ O\left((\log \log n)^2\right)+ O\left(\sum_{i=0}^{\lceil\log_{1.25}\log_2 (\diam(G))\rceil}\log( 2^{1.25^i})\right)+ O\left(\sum_{i=0}^{\lceil\log_{1.25}\log_{\diam(G)} (m)\rceil}\log( \diam(G))\right)\\
\leq & ~ O\left((\log \log n)^2\right)+O(\log(\diam(G)))+O(\log(\diam(G))\log\log_{\diam(G)} (n))\\
\leq & ~ O(\log(\diam(G))\log\log_{\diam(G)}(n)+(\log\log(n))^2),
\end{align*}
}
where the first inequality follows by $i^*=O(\log \log n)$ and $\forall i\leq[i^*], m/n_{i-1}\leq \poly(\log n),$ the third inequality follows by $m/n_{i+1}\geq (m/n_i)^{1.5}/(45\log n+150)\geq (m/n_i)^{1.25}.$

If $m> n\cdot (600\log n)^4,$ then $\forall i\in\{0\}\cup[r-1],$ we have $p_i<0.5.$ Since $\mathcal{E}$ happens. We have:
{\small
\begin{align*}
n_r&\leq \frac{\left(\frac{\left(\frac{n^{1.5}}{m^{0.5}}(45\log n+150)\right)^{1.5}}{m^{0.5}}(45\log n+150)\right)^{\cdots}}{\cdots}&~\text{(Apply $r$ times)}\\
&= \frac{n^{1.5^r}}{m^{1.5^{r}-1}}(45\log n+150)^{2\cdot(1.5^{r}-1)}\\
&=n/(m/n)^{1.5^{r}-1}\cdot (45\log n+150)^{2\cdot(1.5^{r}-1)}\\
&= n/\left(m/\left(n(45\log n+150)^2\right)\right)^{1.5^{r}-1}\\
&\leq n/\left(m/\left(n(45\log n+150)^2\right)\right)^{1.5^{r/2}}\\
&\leq n/\left(m/\left(  n(200\log n)^2\right)\right)^{1.5^{r/2}}\\
&\leq \frac{1}{2},
\end{align*}}
where the second inequality follows by $r\geq 5,$ the third inequality follows by $45\log n+150\leq 200\log n,$ and the last inequality follows by
$$ r\geq c\log\log_{m/n}n\geq 2\log_{1.5} \log_{(m/n)^{1/2}} 2n \geq 2\log_{1.5} \log_{m/(n(200\log n)^2)} 2n$$ for a sufficiently large constant $c>0.$ 

By property~\ref{itm:neighbor_incr_pro1} of Lemma~\ref{lem:properties_of_neighbor_increment}, we have $k_i\leq O(\log (\min(m/n_{i-1},\diam(G)))).$ Thus,
{\small
\begin{align*}
\sum_{i=1}^r k_i&\leq \sum_{m/n_{i-1}\leq \diam(G)} k_i+\sum_{m/n_{i-1}> \diam(G)} k_i\\
&\leq O\left(\sum_{i=0}^{\lceil\log_{1.25}\log_2 (\diam(G))\rceil}\log( 2^{1.25^i})\right)+O\left(\sum_{i=0}^{\lceil\log_{1.25}\log_{\diam(G)} (m)\rceil}\log( \diam(G))\right)\\
&\leq O(\log(\diam(G)))+O(\log(\diam(G))\log\log_{\diam(G)} (n)),
\end{align*}
}
where the first inequality follows by $m/n_{i+1}\geq (m/n_i)^{1.5}/(45\log n+150)\geq (m/n_i)^{1.25}.$

Since $n_r$ is an integer, $n_r$ must be $0$ when $n_r\leq 1/2.$ Let $c>0$ be a sufficiently large constant. For all $m\geq 4n,$ if $r\geq c\log\log_{m/n} n$ then $\textsc{Connectivity}(G,m,r)$ will succeed with probability at least $0.98.$

\end{proof}

\section{Spanning Forest}
\label{sec:sf}

    \subsection{Local Shortest Path Tree}\label{sec:local_short_tree}

In this section, we introduce an important procedure which will be used in the spanning tree algorithm. Roughly speaking, our procedure can merge several local shortest path trees into a larger local shortest path tree. Before we describe the details of the procedure, let us look at some concepts.

\begin{definition}[Local shortest path tree (LSPT)]\label{def:local_short_tree}
Let $V'$ be a set of vertices, $v$ be a vertex in $V',$ and $\p:V'\rightarrow V'$ be a set of parent pointers (See Definition~\ref{def:parent_pointers}) on $V'$ which satisfies that $v$ is the only root of $\p.$ Let $T=(V',\p).$ Given an undirected graph $G=(V,E),$ if $V'\subseteq V$ and $\forall u\in V'\setminus\{v\},(u,\p(u))\in E,\dep_{\p}(u)=\dist_G(u,v),$ then we say $T$ is a local shortest path tree (LSPT) in $G$, and $T$ has root $v$. The vertex set ($V'$ in the above) in $T$ is denoted as $V_T$. The set of parent pointers ($\p$ in the above) in $T$ is denoted as $\p_T$. For short, $\dep_{\p_{T}}$ is denoted as $\dep_{T}, $ and $\dep(\p(T))$ is denoted as $\dep(T).$
\end{definition}

\begin{definition}\label{def:centered_ball}
Given an undirected graph $G=(V,E),$ a vertex $v\in V,$ and $s\in\mathbb{Z}_{\geq 0},$ we define the ball centered at $v$ with radius $s$ as the set $B_{G,s}(v)=\{u\in V\mid \dist_G(u,v)\leq s\}.$
\end{definition}
If in the context graph $G$ is clear, then we use $B_s(v)$ to denote $B_{G,s}(v).$

\begin{definition}[Local complete shortest path tree (LCSPT)]\label{def:local_complete_tree}
Given an undirected graph $G=(V,E),$ $s\in \mathbb{Z}_{\geq 0}$ and a local shortest path tree $T=(V_T,\p_T)$ in $G$ where $T$ has root $v\in V$. If $V_T=B_{G,s}(v),$ then we call $T$ a local complete shortest path tree (LCSPT) in $G$. The root of $T$ is $v$. The radius of $T$ is $s$.
\end{definition}

Let $\wt{T}=(V_{\wt{T}},\p_{\wt{T}})$ with radius $s_1\in \mathbb{Z}_{\geq 0}$ and root $v$ be a local complete shortest path tree in some graph $G=(V,E)$. For $s_2\in\mathbb{Z}_{\geq 0},$ if for every $u\in V_{\wt{T}},$ we have a local complete shortest path tree $T(u)=(V_{T(u)},\p_{T(u)})$ with root $u$ and radius $s_2,$ then we can compute a larger local complete shortest path tree $\wh{T}$ with root $v$ and radius $s_1+s_2.$ The procedure is described in Algorithm~\ref{alg:merge_short_tree}.

\begin{algorithm}[h]
\caption{Local Complete Shortest Path Tree Expansion}\label{alg:merge_short_tree}
\begin{algorithmic}[1]
\small
\Procedure{\textsc{TreeExpansion}}{$\wt{T},\dep_{{\wt{T}}},\{T(u)\mid u\in V_{\wt{T}}\},\{\dep_{{T(u)}}\mid u\in V_{\wt{T}}\}$} \Comment{Lemma~\ref{lem:merge_short_tree}}

\Comment{$\wt{T}=(V_{\wt{T}},\p_{\wt{T}})$ with root $v$ and radius $s_1$ is a LCSPT in graph $G=(V,E).$}

\Comment{$\dep_{{\wt{T}}}:V_{\wt{T}}\rightarrow\mathbb{Z}_{\geq 0}$ records the depth of every vertex in $\wt{T}$.}

\Comment{$\forall u\in V_{\wt{T}}, T(u)=(V_{T(u)},\p_{T(u)})$ with root $u$ and radius $s_2$ is a LCSPT in $G$.}

\Comment{$\forall u\in V_{\wt{T}},\dep_{T(u)}:V_{T(u)}\rightarrow\mathbb{Z}_{\geq 0}$ records the depth of every vertex in $T(u).$ }
\State Output: $\wh{T}=(V_{\wh{T}},\p_{\wh{T}}), \dep_{{\wh{T}}}.$


\State Let $V_{\wh{T}}=\bigcup_{u\in \wt{T}} V_{T(u)}.$\label{sta:lcspt_merge_vertex}
\State $\forall x\in V_{\wt{T}},\p_{\wh{T}}(x)\leftarrow \p_{\wt{T}}(x).$\label{sta:lcspt_original_pointer}
\State $\forall x\in V_{\wt{T}},h(x)\leftarrow\dep_{{\wt{T}}}(x).$\label{sta:lcspt_original_depth}
\State $\forall x\in V_{\wh{T}}\setminus V_{\wt{T}},u_x\leftarrow \underset{u:u\in V_{\wt{T}},x\in V_{T(u)}}{\arg\min}\dep_{{\wt{T}}}(u)+\dep_{{T(u)}}(x),\p_{\wh{T}}(x)\leftarrow \p_{T(u_x)}(x).$\label{sta:lcspt_find_ux}

\Comment{$u_x$ is on the shortest path from $x$ to $v$.}
\State $\forall x\in V_{\wh{T}}\setminus V_{\wt{T}},h(x)\leftarrow \dep_{{\wt{T}}}(u_x)+\dep_{{T(u_x)}}(x).$\label{sta:lcspt_compute_depth}
\State \Return $\wh{T}=(V_{\wh{T}},\p_{\wh{T}}),$ and return $h:V_{\wh{T}}:\rightarrow \mathbb{Z}_{\geq 0}$ as $\dep_{{\wh{T}}}.$
\EndProcedure
\end{algorithmic}
\end{algorithm}

\begin{lemma}\label{lem:merge_short_tree}
Let $G=(V,E)$ be an undirected graph, $s_1,s_2\in\mathbb{Z}_{\geq 0},$ and $v\in V.$ Let $\wt{T}=(V_{\wt{T}},\p_{\wt{T}})$ with root $v$ and radius $s_1$ be a local complete shortest path tree in $G,$ and $\dep_{{\wt{T}}}:V_{\wt{T}}\rightarrow\mathbb{Z}_{\geq 0}$ be the depth of every vertex in $\wt{T}$. $\forall u\in V_{\wt{T}},$ let $T(u)$ with root $u$ and radius $s_2$ be a local complete shortest path tree in $G,$ and $\dep_{{T(u)}}:V_{T(u)}\rightarrow \mathbb{Z}_{\geq 0}$ be the depth of every vertex in $T(u)$. Let $(\wh{T}=(V_{\wh{T}},\p_{\wh{T}}),\dep_{{\wh{T}}})=\textsc{TreeExpansion}(\wt{T},\dep_{{\wt{T}}},\{T(u)\mid u\in V_{\wt{T}}\},\{\dep_{{T(u)}}\mid u\in V_{\wt{T}}\})$ (Algorithm~\ref{alg:merge_short_tree}), then $\wh{T}$ is a local complete shortest path tree with root $v$ and radius $s_1+s_2$ in $G$. In addition, $\dep_{{\wh{T}}}$ records the depth of every vertex in $\wh{T}$.
\end{lemma}
\begin{proof}
 If $x\in B_{s_1+s_2}(v),$ then there must exist $u\in V$ such that $\dist_G(v,u)\leq s_1$ and $\dist_G(u,x)\leq s_2.$ Thus, $V_{\wh{T}}=\bigcup_{u\in \wt{T}}V_{T(u)}=\bigcup_{u\in B_{s_1}(v)}B_{s_2}(u)=B_{s_1+s_2}(v).$

 Now we want to prove that $\p_{\wh{T}}:V_{\wh{T}}\rightarrow V_{\wh{T}}$ also satisfies the condition that $\wh{T}$ is a local shortest path tree. We can prove it by induction. If $\dist_G(u,v)=0,$ then it means $u=v.$ In this case, $\p_{\wh{T}}(u)=\p_{\wt{T}}(u)=v,$ and $h(u)=\dep_{{\wt{T}}}(u)=0.$ Let $s\in[s_1+s_2].$ Suppose $\forall x\in B_{s-1}(v),$ we have $h(x)=\dep_{{\wh{T}}}(x)=\dist_G(x,v).$ If $B_s(v)=B_{s-1}(v),$ then we are already done. Otherwise, let $x$ be the vertex which has $\dist_G(x,v)=s.$ If $x\in B_{s_1}(v),$ then $h(x)=\dep_{{\wt{T}}}(x)=\dist_G(x,v).$ Additionally, we have $\p_{\wh{T}}(x)=\p_{\wt{T}}(x).$ Therefore, $\dep_{{\wh{T}}}(x)=\dep_{{\wh{T}}}(\p_{\wt{T}}(x))+1=\dist_G(v,\p_{\wt{T}}(x))+1=\dist_G(v,x).$ If $x\in B_{s_2}(v)\setminus B_{s_1}(v),$ then $h(x)=\min_{u:\dist_G(v,u)\leq s_1,\dist_G(u,x)\leq s_2}\dep_{{\wt{T}}}(u)+\dep_{{T(u)}}(x)=\min_{u:\dist_G(v,u)\leq s_1,\dist_G(u,x)\leq s_2} \dist_G(v,u)+\dist_G(u,x)=\dist_G(v,x)=s.$ And we have $\dist_G(v,x)=\dist_G(v,u_x)+\dist_G(u_x,x).$ Notice that $\dist_G(v,\p_{T(u_x)}(x))=\dist_G(v,u_x)+\dist_G(u_x,\p_{T(u_x)}(x))=\dist_G(v,x)-1=s-1.$ Thus,
 \begin{align*}
 \dep_{{\wh{T}}}(x)=\dep_{{\wh{T}}}(\p_{\wh{T}}(x))+1=\dep_{{\wh{T}}}(\p_{T(u_x)}(x))+1=s.
 \end{align*}

 To conclude, $\wh{T}$ is a local complete shortest path tree with root $v$ and radius $s_1+s_2$ in $G$. In addition, $\dep_{{\wh{T}}}$ records the depth of every vertex in $\wh{T}$.
\end{proof}

    \subsection{Multiple Local Shortest Path Trees}\label{sec:multi_local_short_tree}

\begin{algorithm}[t]
\small
\caption{Doubling Algorithm for Local Complete Shortest Path Trees}\label{alg:doubling_expansion}
\begin{algorithmic}[1]
\small
\Procedure{\textsc{MultiRadiusLCSPT}}{$G=(V,E),m$} \Comment{Lemma~\ref{lem:multi_radius_trees}}
\State Output: $r,\{T_i(v)\mid i\in\{0\}\cup[r],v\in V\},\{\dep_{T_i(v)}\mid i\in\{0\}\cup[r],v\in V,T_i(v)\not= \nul \}$

\State \textbf{Initialization:}
\State  $\forall v\in V,$ if $|\{v\}\cup\Gamma_G(v)|<\lceil(m/n)^{1/4}\rceil,$ then let $T_0(v)\leftarrow(\{v\}\cup\Gamma_G(v),\p_{T_0(v)}),$ \label{sta:alg_exp_init1}
\State where $\p_{T_0(v)}:\{v\}\cup\Gamma_G(v)\rightarrow\{v\}\cup\Gamma_G(v),$ and $\forall u\in \{v\}\cup\Gamma_G(v),\p_{T_0(v)}(u)=v.$ \label{sta:alg_exp_init2}
\State $\forall v\in V,$ if $|\{v\}\cup\Gamma_G(v)|\geq \lceil(m/n)^{1/4}\rceil,$ then let $T_0(v)\leftarrow  \nul .$\label{sta:alg_exp_init2point5}
\State $\forall v\in V,$ if $T_0(v)\not= \nul ,$ let $\dep_{T_0(v)}:V_{T_0(v)}\rightarrow \mathbb{Z}_{\geq 0}$ s.t. $\dep_{T_0(v)}(v)=0,$ $\forall u\in \Gamma_G(v),\dep_{T_0(v)}(u)=1.$\label{sta:alg_exp_init3}
\State $r=1.$
\State \textbf{Main Loop:}
\For{\textbf{true}}
\For{$v\in V$} \Comment{If $T_r(v)\not= \nul ,$ $T_r(v)$ is a local complete shortest path tree with radius $2^r.$}
\If{$T_{r-1}(v)$ is $ \nul $} $T_{r}(v)\leftarrow \nul .$\label{sta:alg_exp_set_null1}
\ElsIf{$\exists u\in V_{T_{r-1}(v)},$ $T_{r-1}(u)$ is $ \nul $} $T_{r}(v)\leftarrow \nul .$\label{sta:alg_exp_set_null2}
\Else
\State
{\tiny $\left(T_{r}(v),\dep_{{T_r(v)}}\right)=\textsc{TreeExpansion}\left(T_{r-1}(v),\dep_{{T_{r-1}(v)}},\underset{u\in V_{T_{r-1}(v)}}{\bigcup}\left\{T_{r-1}(u)\right\},\underset{u\in V_{T_{r-1}(v)}}{\bigcup}\left\{\dep_{{T_{r-1}(u)}}\right\}\right).$} \label{sta:alg_exp_double}

\Comment{Algorithm~\ref{alg:merge_short_tree}}
\State {\small If $|V_{T_r(v)}|\geq \lceil(m/n)^{1/4}\rceil,$ let $T_r(v)\leftarrow \nul .$ }\label{sta:alg_exp_set_null}
\EndIf
\EndFor
\If {$\forall v\in V$ either $T_r(v)= \nul $ or $|V_{T_r(v)}|=|V_{T_{r-1}(v)}|$}\label{sta:alg_exp_condition}
\State  \Return $r,\{T_i(v)\mid i\in\{0\}\cup[r],v\in V\},\{\dep_{T_i(v)}\mid i\in\{0\}\cup[r],v\in V,T_i(v)\not= \nul \}$
\EndIf
\State $r\leftarrow r+1.$
\EndFor
\EndProcedure
\end{algorithmic}
\end{algorithm}

In this section, we show a procedure which is a generalization of neighbor increment procedure shown in Section~\ref{sec:neighbor_incr}. The input of the procedure is an undirected graph $G=(V,E)$ and a parameter $m$ which is larger than $|V|=n.$ The output will be $n$ local shortest path trees (See Definition~\ref{def:local_short_tree}) such that $\forall v\in V,$ there is a shortest path tree with root $v$. Furthermore, the size of each shortest path tree is at least $\left\lceil\left( m / |V| \right)^{1/4}\right\rceil$ and at most $\left\lceil\left( m / |V| \right)^{1/2}\right\rceil$. The algorithm is described in Algorithm~\ref{alg:maximal_short_tree}. The high level idea is that we firstly use doubling technique and the algorithm described in Section~\ref{sec:local_short_tree} to get local complete shortest path trees rooted at every vertex with multiple radius, and then use these LCSPTs to find large enough local shortest path trees rooted at every vertex. The doubling algorithm is described in Algorithm~\ref{alg:doubling_expansion}.

\begin{definition}\label{def:num_iter_multi_radius_lcspt}
Given a graph $G=(V,E)$ and a parameter $m\in\mathbb{Z}_{\geq 0},m\geq |V|,$ the number of iterations of $\textsc{MultiRadiusLCSPT}(G,m)$ (Algorithm~\ref{alg:doubling_expansion}) is the value of $r$ at the end of the procedure.
\end{definition}

\begin{lemma}\label{lem:multi_radius_trees}
Let $G=(V,E)$ be an undirected graph, and $m$ be a parameter which is at least $|V|.$ Let $(r,\{T_i(v)\mid i\in\{0\}\cup[r],v\in V\},\{\dep_{T_i(v)}\mid i\in\{0\}\cup[r],v\in V,T_i(v)\not= \nul \})=\textsc{MultiRadiusLCSPT}(G,m)$ (Algorithm~\ref{alg:doubling_expansion}).
We have following properties.
\begin{enumerate}
\item $\forall i\in\{0\}\cup[r],v\in V,$ if $T_i(v)\not= \nul ,$ then $T_i(v)$ is a LCSPT (See Definition~\ref{def:local_complete_tree}) with root $v$ and radius $2^i$ in $G$. Furthermore, $\dep_{T_i(v)}$ records the depth of every vertex in $T_i(v).$ \label{itm:multi_pro1}
\item $\forall i\in\{0\}\cup[r],v\in V,$ $|B_{G,2^i}(v)|\geq \lceil(m/n)^{1/4}\rceil\Leftrightarrow T_i(v)= \nul .$\label{itm:multi_pro2}
\item For $v\in V,$ if $T_r(v)\not= \nul ,$ then $V_{T_r(v)}=\{u\in V\mid \dist_G(u,v)<\infty\}.$ \label{itm:multi_pro3}
\item The number of iterations (see Definition~\ref{def:num_iter_multi_radius_lcspt}) $r\leq \min(\lceil\log(\diam(G))\rceil,\lceil\log(m/n)\rceil)+1.$ \label{itm:multi_pro4}
\end{enumerate}
\end{lemma}
\begin{proof}
For property~\ref{itm:multi_pro1}, we can prove it by induction. If $i=0,$ the property holds by line~\ref{sta:alg_exp_init1}, line~\ref{sta:alg_exp_init2} and line~\ref{sta:alg_exp_init3}. Now suppose $\forall v\in V,$ if $T_{i-1}(v)$ is not $ \nul ,$ then $T_{i-1}(v)$ is a LCSPT with root $v$ and radius $2^{i-1}$ in $G,$ and $\dep_{T_{i}(v)}$ records the depth of every vertex in $T_i(v).$ For $v\in V,$ notice that the only place that will make $T_i(v)$ not $ \nul $ is line~\ref{sta:alg_exp_double}, and if the procedure run line~\ref{sta:alg_exp_double}, any of $T_{i-1}(v)$ and $T_{i-1}(u)$ with $u\in V_{T_{i-1}(v)}$ cannot be $ \nul .$ By Lemma~\ref{lem:merge_short_tree}, since the radius of $T_{i-1}(v)$ is $2^{i-1},$ and $\forall u\in V_{T_{i-1}(v)}, T_{i-1}(u)$ has radius $2^{i-1},$ $T_i(v)$ is a LCSPT with root $v$ and radius $2^i.$ Furthermore $\dep_{T_i(v)}$ records the depth of every vertex in $T_i(v).$

For property~\ref{itm:multi_pro2}, if $i=0,$ then this property holds by line~\ref{sta:alg_exp_init1} to line~\ref{sta:alg_exp_init3}. For $i\in[r],$ our proof is by induction. Suppose the property holds for $i-1$. Now consider $T_i(v)$ for $v\in V$. The only way to make $T_i(v)$ not $ \nul $ is line~\ref{sta:alg_exp_double}. If the procedure invokes line~\ref{sta:alg_exp_double}, then any of $T_{i-1}(v)$ and $T_{i-1}(u)$ with $u\in V_{T_{i-1}(v)}$ cannot be $ \nul .$ By property~\ref{itm:multi_pro1} and Lemma~\ref{lem:merge_short_tree}, $T_i(v)$ will be a LCSPT with root $v$ and radius $2^i$ in line~\ref{sta:alg_exp_double}. If $|B_{G,2^i}(v)|\geq \lceil(m/n)^{1/4}\rceil,$ then $T_i(v)$ is set to be $ \nul $ in line~\ref{sta:alg_exp_set_null}. Thus, we already got $|B_{G,2^i}(v)|\geq \lceil(m/n)^{1/4}\rceil\Rightarrow T_i(v)= \nul .$ Now we want to show $|B_{G,2^i}(v)|\geq \lceil(m/n)^{1/4}\rceil\Leftarrow T_i(v)= \nul .$ If $T_i(v)= \nul ,$ then there are three cases. The first case is that $T_i(v)$ is set at line~\ref{sta:alg_exp_set_null1}. In this case, $T_{i-1}(v)= \nul $ implies $|B_{G,2^i}(v)|\geq|B_{G,2^{i-1}}(v)|\geq \lceil(m/n)^{1/4}\rceil.$ The second case is that $T_i(v)$ is set at line~\ref{sta:alg_exp_set_null2}. In this case, $\exists u\in V_{T_{i-1}(v)}=B_{G,2^{i-1}}(v)$ such that $|B_{G,2^{i-1}}(u)|\geq \lceil(m/n)^{1/4}\rceil$ which implies $|B_{G,2^i}(v)|\geq \lceil(m/n)^{1/4}\rceil.$ In the final case, $T_i(v)$ is set at line~\ref{sta:alg_exp_set_null}, and thus, $|B_{G,2^i}(v)|\geq \lceil(m/n)^{1/4}\rceil.$

For property~\ref{itm:multi_pro3}, if $T_r(v)\not= \nul ,$ then by property~\ref{itm:multi_pro1}, we know $V_{T_r(v)}=B_{G,2^r}(v).$ By the condition in line~\ref{sta:alg_exp_condition}, we know $V_{T_r(v)}=V_{T_{r-1}(v)}$ which implies $B_{G,2^r}(v)=B_{G,2^{r-1}}(v).$ Thus, $V_{T_r(v)}=\{u\in V\mid \dist_G(u,v)<\infty\}.$

For property~\ref{itm:multi_pro4}, we can prove it by contradiction. If $r>\lceil\log(\diam(G))\rceil+1,$ then let $i=\lceil\log(\diam(G))\rceil+1.$ By the condition in line~\ref{sta:alg_exp_condition}, we know there is a vertex $v\in V$ such that $T_i(v)\not= \nul $ and $V_{T_i(v)}\not=V_{T_{i-1}(v)}.$ It means that $B_{G,2^i}(v)\not=B_{G,2^{i-1}}(v),$ i.e. $\exists u\in V,\dist_G(v,u)>2^{i-1}.$ But this contradicts to $i=\lceil\log(\diam(G))\rceil+1.$ Similarly, if $r>\lceil\log(m/n)\rceil+1,$ then let $i=\lceil\log(m/n)\rceil+1.$  By the condition in line~\ref{sta:alg_exp_condition}, we know there is a vertex $v\in V$ such that $T_i(v)\not= \nul $ and $V_{T_i(v)}\not=V_{T_{i-1}(v)}.$ If $2^{i-1}\leq \diam(G),$ then we have $V_{T_i(v)}\not=V_{T_{i-1}(v)}$ which leads to a contradiction. If $2^{i-1}\geq \diam(G),$ then $|B_{G,2^{i-1}}(v)|\geq 2^{i-1}\geq m/n\geq \lceil(m/n)^{1/4}\rceil$ which contradicts to property~\ref{itm:multi_pro2}.
\end{proof}

Next, we show how to use Algorithm~\ref{alg:doubling_expansion} to design an algorithm which can output $|V|$ number of local shortest path trees rooted at every vertex in $V$. The details of the algorithm is described in Algorithm~\ref{alg:maximal_short_tree}, and the guarantees of the algorithm is stated in the following lemma.

\begin{algorithm}[h!]
\small
\caption{Large Local Shortest Path Trees}\label{alg:maximal_short_tree}
\begin{algorithmic}[1]
\Procedure{\textsc{MultipleLargeTrees}}{$G=(V,E),m$} \Comment{Lemma~\ref{lem:maximal_large_short_tree}, Lemma~\ref{lem:num_it_multiplelargetree}}
\State Output: $\{\wt{T}(v)\mid v\in V\},\{\dep_{\wt{T}(v)}\mid v\in V\}.$ 
\State {\tiny $\left(r,\{T_i(v)\mid i\in\{0\}\cup [r],v\in V\},\{\dep_{T_i(v)}\mid i\in\{0\}\cup[r],v\in V,T_i(v)\not= \nul \}\right)=\textsc{MultiRadiusLCSPT}(G,m).$}\label{sta:alg_maximal_invoke_multi} \Comment{Algorithm~\ref{alg:doubling_expansion}}
\State $\forall v\in V$ with $T_r(v)\not= \nul $ let $\wt{T}(v)=T_r(v)$ and $\dep_{\wt{T}(v)}\leftarrow \dep_{T_r(v)}.$\label{sta:alg_maximal_init0}
\State $\forall v\in V$ with $T_r(v)= \nul ,$ let $\wt{T}_{0}(v)=(\{v\},\p_{\wt{T}_0(v)}),s_0(v)= 0,$ and $\dep_{{\wt{T}_0(v)}}:\{v\}\rightarrow \mathbb{Z}_{\geq 0},$ \label{sta:alg_maximal_init1}
\State where $\p_{\wt{T}_0(v)}:\{v\}\rightarrow\{v\}$ satisfies $\p_{\wt{T}_0(v)}(v)=v,$ and $\dep_{\wt{T}_0(v)}(v)=0.$ \label{sta:alg_maximal_init2}
\For{$i=1\rightarrow r$}\label{sta:alg_maximal_loop_start}
\For{$v\in \{u\in V\mid T_r(u)= \nul \}$}
\If{$\forall u\in V_{\wt{T}_{i-1}(v)},$ $T_{r-i}(u)\not= \nul $} \label{sta:alg_maximal_not_null_condition}
\State  {\tiny
$\left(\wt{T}_{i}(v),\dep_{{\wt{T}_{i}(v)}}\right)=\textsc{TreeExpansion}\left(\wt{T}_{i-1}(v),\dep_{{\wt{T}_{i-1}(v)}},\underset{u\in V_{\wt{T}_{i-1}(v)}}{\bigcup}\left\{T_{r-i}(u)\right\},\underset{u\in V_{\wt{T}_{i-1}(v)}}{\bigcup}\left\{\dep_{{T_{r-i}(u)}}\right\}\right).$}\label{sta:alg_maximal_tree_expansion1}

\Comment{Algorithm~\ref{alg:merge_short_tree}}
\State If $|V_{\wt{T}_{i}(v)}|<\lceil(m/n)^{1/4}\rceil,$ then let $s_i(v)= s_{i-1}(v)+2^{r-i}.$ \label{sta:alg_maximal_update_s}
\State Otherwise, let $s_i(v)=s_{i-1}(v),\wt{T}_i(v)\leftarrow\wt{T}_{i-1}(v),\dep_{{\wt{T}_{i}(v)}}\leftarrow \dep_{{\wt{T}_{i-1}(v)}}.$
\Else
\State Let $s_i(v)=s_{i-1}(v),\wt{T}_i(v)=\wt{T}_{i-1}(v),\dep_{{\wt{T}_{i}(v)}}\leftarrow \dep_{{\wt{T}_{i-1}(v)}}.$
\EndIf
\EndFor
\EndFor\label{sta:alg_maximal_loop_end} \Comment{ $\wt{T}_r(v)$ is a LCSPT with root $v$ and the largest radius s.t. $|V_{\wt{T}_r(v)}|<\lceil(m/n)^{1/4}\rceil.$}
\State $\forall v\in V,$ if $|\Gamma_G(v)\cup\{v\}|\leq \lceil(m/n)^{1/4}\rceil,$ then let $N(v)=\Gamma_G(v)\cup\{v\}.$\label{sta:alg_maximal_create_N1}
\State Otherwise arbitrarily choose $N(v)\subseteq \Gamma_G(v)\cup\{v\}$ with $|N(v)|=\lceil(m/n)^{1/4}\rceil.$\label{sta:alg_maximal_create_N2}
\For{$v\in \{u\in V\mid T_r(u)= \nul \}$}\Comment{Expand $\wt{T}_r(v)$ a little bit to get large enough $\wt{T}$.}
\If{$\forall u\in V_{{\wt{T}_r(v)}},T_0(u)\not= \nul $}
\State
{\tiny
$\left(\wt{T}(v),\dep_{{\wt{T}(v)}}\right)=\textsc{TreeExpansion}\left(\wt{T}_r(v),\dep_{{\wt{T}_r(v)}},\underset{u\in V_{\wt{T}_r(v)}}{\bigcup}\left\{T_{0}(u)\right\},\underset{u\in V_{\wt{T}_r(v)}}{\bigcup}\left\{\dep_{{T_{0}(u)}}\right\}\right).$}\label{sta:alg_maximal_tree_expansion2}
\Comment{Algorithm~\ref{alg:merge_short_tree}}
\Else
\State Select an arbitrary $u_v\in V_{{\wt{T}_r(v)}}$ with $T_0(u_v)= \nul .$\label{sta:alg_maximal_wT_assignment}
\State Let $V_{\wt{T}(v)}=N(u_v)\cup V_{\wt{T}_r(v)}.$
\State $\forall x\in V_{\wt{T}_r(v)},$ let $\p_{\wt{T}(v)}(x)=\p_{\wt{T}_r(v)}(x),\dep_{\wt{T}(v)}(x)=\dep_{\wt{T}_r(v)}(x).$
\State $\forall x\in N(u_v),x\not\in V_{\wt{T}_r(v)}, $ let $\p_{\wt{T}(v)}(x)=u_v,\dep_{\wt{T}(v)}(x)=\dep_{\wt{T}_r(v)}(u_v)+1.$
\State Let $\wt{T}(v)=(V_{\wt{T}(v)},\p_{\wt{T}(v)}).$\label{sta:alg_wT_assignment2}
\EndIf
\EndFor
\State \Return $\{\wt{T}(v)\mid v\in V\},\{\dep_{\wt{T}(v)}\mid v\in V\}.$
\EndProcedure
\end{algorithmic}
\end{algorithm}

\begin{lemma}\label{lem:maximal_large_short_tree}
Let $G=(V,E)$ be an undirected graph, and $m$ be a parameter which is at least $16|V|.$ Let $\left(\{\wt{T}(v)\mid v\in V\},\{\dep_{\wt{T}(v)}\mid v\in V\}\right)=\textsc{MultipleLargeTrees}(G,m).$ (Algorithm~\ref{alg:maximal_short_tree}) Then, the output satisfies the following properties.
\begin{enumerate}
\item $\forall v\in V,$ $\wt{T}(v)$ is a LSPT (See Definition~\ref{def:local_short_tree}) with root $v$, and $\dep_{\wt{T}(v)}$ records the depth of every vertex in $\wt{T}(v).$\label{itm:large_tree_pro1}
\item $\forall v\in V,u\in V_{\wt{T}(v)},w\in V\setminus V_{\wt{T}(v)},$ it satisfies $\dist_G(v,u)\leq \dist_G(v,w).$\label{itm:large_tree_pro2}
\item $\forall v\in V,$ either $|V_{\wt{T}(v)}|\geq \lceil(m/n)^{1/4}\rceil$ or $V_{\wt{T}(v)}=\{u\in V\mid \dist_G(u,v)<\infty\}.$\label{itm:large_tree_pro3}
\item $\forall v\in V,$ $|V_{\wt{T}(v)}|\leq \lfloor(m/n)^{1/2}\rfloor.$\label{itm:large_tree_pro4}
\end{enumerate}
\end{lemma}
\begin{proof}
Before we prove above properties, we first show some crucial observations.
\begin{claim}\label{cla:maximal_Tiv_is_good}
$\forall v\in V,i\in\{0\}\cup[r],$ if $T_i(v)\not= \nul ,$ then $T_i(v)$ is a LCSPT with root $v$ and radius $2^i$ in graph $G$. Furthermore, $\dep_{T_i(v)}:V_{T_i(v)}\rightarrow\mathbb{Z}_{\geq 0}$ records the depth of every vertex in $T_i(v).$ If $T_i(v)= \nul ,$ then $|B_{G,2^i}(v)|\geq \lceil(m/n)^{1/4}\rceil.$
\end{claim}
\begin{proof}
Follows by property~\ref{itm:multi_pro1} and property~\ref{itm:multi_pro2} of Lemma~\ref{lem:multi_radius_trees} directly.
\end{proof}

\begin{claim}\label{cla:maximal_shrink_range}
Let $v\in V$ be a vertex with $T_r(v)= \nul .$ Then, $\forall i\in\{0\}\cup[r],$ $\wt{T}_i(v)$ is a {\rm LCSPT} (See Definition~\ref{def:local_complete_tree}) with root $v$ and radius $s_i(v)$ in $G,$ and $\dep_{\wt{T}_i(v)}$ records the depth of every vertex in $\wt{T}_i(v)$. Furthermore, we have $|B_{G,s_i(v)}(v)|<\lceil(m/n)^{1/4}\rceil,|B_{G,s_i(v)+2^{r-i}}(v)|\geq \lceil(m/n)^{1/4}\rceil.$
\end{claim}
\begin{proof}
 Let $v\in V$ be a vertex with $T_r(v)= \nul .$ When $i=0,$ then due to line~\ref{sta:alg_maximal_init1} and line~\ref{sta:alg_maximal_init2}, $\wt{T}_0(v)$ is a {\rm LCSPT} (See Definition~\ref{def:local_complete_tree}) with root $v$ and radius $0=s_0(v)$ in $G$. According to property~\ref{itm:multi_pro2} of Lemma~\ref{lem:multi_radius_trees}, since $T_r(v)= \nul ,$ we know $|B_{G,0+2^r}(v)|\geq \lceil(m/n)^{1/4}\rceil.$

For $i\in[r],$ we prove it by induction. Suppose the claim is true for $i-1.$ By Claim~\ref{cla:maximal_Tiv_is_good}, Lemma~\ref{lem:merge_short_tree} and the condition in line~\ref{sta:alg_maximal_not_null_condition}, if the procedure executes line~\ref{sta:alg_maximal_tree_expansion1}, then we know $\wt{T}_i(v)$ is a LCSPT with root $v$ and radius $s_{i-1}(v)+2^{r-i}$ in $G$ at the end of the execution of line~\ref{sta:alg_maximal_tree_expansion1}, and $\dep_{\wt{T}_i(v)}$ records the depth of every vertex in $\wt{T}_i(v)$. If $|V_{\wt{T}_i(v)}|<\lceil(m/n)^{1/4}\rceil,$ then $|B_{G,s_{i-1}(v)+2^{r-i}}(v)|<\lceil(m/n)^{1/4}\rceil.$ The procedure will execute line~\ref{sta:alg_maximal_update_s}, and thus $s_i(v)$ is the radius of $\wt{T}_i(v).$ In addition, since $s_i(v)=s_{i-1}(v)+2^{r-i}$ and $|B_{G,s_{i-1}(v)+2^{r-i+1}}(v)|\geq \lceil(m/n)^{1/4}\rceil,$ we have $|B_{G,s_{i}(v)+2^{r-i}}(v)|\geq \lceil(m/n)^{1/4}\rceil.$ If at the end of line~\ref{sta:alg_maximal_tree_expansion1}, $|V_{\wt{T}_i(v)}|\geq\lceil(m/n)^{1/4}\rceil,$ then we know $|B_{G,s_{i-1}(v)+2^{r-i}}(v)|\geq \lceil (m/n)^{1/4}\rceil.$ In this case, $\wt{T}_i(v),s_i(v)$ and $\dep_{\wt{T}_i(v)}$ will be set to be $\wt{T}_{i-1}(v),s_{i-1}(v)$ and $\dep_{\wt{T}_i(v)}$ respectively, and thus $|B_{G,s_i(v)}(v)|<\lceil(m/n)^{1/4}\rceil.$ If the condition in line~\ref{sta:alg_maximal_not_null_condition} does not hold, then we know $|B_{G,s_{i-1}(v)+2^{r-i}}(v)|\geq \lceil (m/n)^{1/4}\rceil$ by claim~\ref{cla:maximal_Tiv_is_good}. In this case, $\wt{T}_i(v),s_i(v)$ and $\dep_{\wt{T}_i(v)}$ will also be set to be $\wt{T}_{i-1}(v),s_{i-1}(v)$ and $\dep_{\wt{T}_i(v)}$ respectively, and thus $|B_{G,s_i(v)}(v)|<\lceil(m/n)^{1/4}\rceil.$
\end{proof}

Claim~\ref{cla:maximal_shrink_range} shows that for each vertex $v\in V$, we know $\wt{T}_r(v)$ is a LCSPT with root $v$ and radius $s_r(v)$ in $G$ such that $|B_{G,s_r(v)}(v)|<\lceil(m/n)^{1/4}\rceil$ and $|B_{G,s_r(v)+1}(v)|\geq\lceil(m/n)^{1/4}\rceil.$

Now, let us prove property~\ref{itm:large_tree_pro1} and property~\ref{itm:large_tree_pro2}. For $v\in V,$ if $T_r(v)\not= \nul ,$ then $\wt{T}(v),\dep_{\wt{T}(v)}$ will be set to be $T_r(v),\dep_{T_r(v)}$ respectively. By Claim~\ref{cla:maximal_Tiv_is_good}, the properties holds. Let $v$ be a vertex in $V$ with $T_r(v)= \nul .$ If $\wt{T}(v)$ is assigned at line~\ref{sta:alg_maximal_tree_expansion2}, then by Lemma~\ref{lem:merge_short_tree}, we know $\wt{T}(v)$ is a LCSPT with root $v$, and $\dep_{\wt{T}(v)}$ records the depth of every vertex in $\wt{T}(v).$ Thus, both properties hold. If $\wt{T}(v)$ is assigned at line~\ref{sta:alg_wT_assignment2}, then there are two cases for the vertices in $V_{\wt{T}(v)}:$
\begin{enumerate}
\item If $x$ is in $V_{\wt{T}_r(v)},$ then since Claim~\ref{cla:maximal_shrink_range} shows $\wt{T}_r(v)$ is a LCSPT with root $v,$ it is easy to show $\dep_{\wt{T}(v)}(x)=\dep_{\wt{T}_r(v)}(x)=\dist_G(v,x),$ and $\p_{\wt{T}(v)}(x)=\p_{\wt{T}_r(v)}(x)\in E.$
\item if $x$ is in $N(u_v)$ but not in $V_{\wt{T}_r(v)},$ then since $\wt{T}_r(v)$ is a LCSPT with root $v$ and radius $s_r(v)$, $\dist_G(v,x)\geq s_r(v)+1.$ Also notice that $\dist_G(v,x)\leq \dist_G(v,u_v)+\dist_G(u_v,x)\leq s_r(v)+1.$ Therefore, $\dist_G(v,x)=s_r(v)+1,\dep_{\wt{T}(v)}(x)=\dep_{\wt{T}(v)}(u_v)+1=\dist_G(v,x)=s_r(v)+1.$ Since $x\in N(u_v),$ $(\p_{\wt{T}(v)}(x),x)=(u_v,x)\in E.$
\end{enumerate}
Thus, $\wt{T}$ is a LSPT with root $v$, and it proves property~\ref{itm:large_tree_pro1}. Due to above both cases, we know $B_{G,s_r(v)}(v)\subseteq V_{\wt{T}},$ and $\forall x\in V_{\wt{T}(v)},$ $\dist_G(v,x)\leq s_r(v)+1.$ Thus, property~\ref{itm:large_tree_pro2} holds.

For property~\ref{itm:large_tree_pro3} and property~\ref{itm:large_tree_pro4}, we have two cases. The first case is when $v$ satisfies $T_r(v)\not= \nul .$ In this case $\wt{T}(v)=T_r(v),$ due to property~\ref{itm:multi_pro3},~\ref{itm:multi_pro4} of Lemma~\ref{lem:multi_radius_trees}, we have $V_{\wt{T}(v)}=\{u\in V\mid \dist_G(v,u)<\infty\},$ and $|V_{\wt{T}(v)}|<\lceil(m/n)^{1/4}\rceil\leq \lfloor (m/n)^{1/2}\rfloor.$ The second case is $T_r(v)= \nul .$ In this case, if $\wt{T}(v)$ is assigned at line~\ref{sta:alg_maximal_tree_expansion2}, then $V_{\wt{T}(v)}=B_{G,s_r(v)+1}(v).$ Then by Claim~\ref{cla:maximal_shrink_range}, we can get $|V_{\wt{T}(v)}|\geq \lceil(m/n)^{1/4}\rceil.$ Because $|V_{\wt{T}_r(v)}|<\lceil(m/n)^{1/4}\rceil$ and $\forall u\in V_{\wt{T}_r(v)},|V_{T_0(u)}|<\lceil(m/n)^{1/4}\rceil,$ we know $|V_{\wt{T}(v)}|\leq \lfloor(m/n)^{1/2}\rfloor.$ If $\wt{T}(v)$ is assigned at line~\ref{sta:alg_wT_assignment2}, then $|V_{\wt{T}(v)}|\geq |N(u_v)|\geq \lceil (m/n)^{1/4}\rceil,$ and $|V_{\wt{T}(v)}|\leq |V_{\wt{T}_r(v)}|+|N(u_v)|<2\cdot \lceil(m/n)^{1/4}\rceil\leq \lfloor(m/n)^{1/2}\rfloor.$
\end{proof}

\begin{definition}\label{def:num_it_multiplelargetree}
Let graph $G=(V,E)$, and let $m$ be a parameter which is at least $16|V|.$ The number of iterations of $\left(\{\wt{T}(v)\mid v\in V\},\{\dep_{\wt{T}(v)}\mid v\in V\}\right)=\textsc{MultipleLargeTrees}(G,m)$ (Algorithm~\ref{alg:maximal_short_tree}) is defined as the value of $r$ in the procedure.
\end{definition}

\begin{lemma}[Number of iterations of Algorithm~\ref{alg:maximal_short_tree}]\label{lem:num_it_multiplelargetree}
Let $G=(V,E)$ be an undirected graph, and let $m$ be a parameter which is at least $16|V|.$  The number of iterations (see Definition~\ref{def:num_it_multiplelargetree}) of $\left(\{\wt{T}(v)\mid v\in V\},\{\dep_{\wt{T}(v)}\mid v\in V\}\right)=\textsc{MultipleLargeTrees}(G,m)$ (Algorithm~\ref{alg:maximal_short_tree}) is at most  $\min(\lceil\log(\diam(G))\rceil,\lceil\log(m/n)\rceil)+1.$
\end{lemma}
\begin{proof}
It follows by property~\ref{itm:multi_pro4} of Lemma~\ref{lem:multi_radius_trees} directly.
\end{proof}

    \subsection{Path Generation and Root Changing}
In this section, we show a procedure which can output a path from a certain vertex to the root in a rooted tree. Then we show how to use the procedure to change the root of a rooted tree to a certain vertex in the tree. To output the vertex-root path, we have two stages. The first stage is using doubling method to compute the depth and the $2^i$th (for all $i\in\{0,1,\cdots,\log(\dep)\}$) ancestor of each vertex. The second stage is using divide-and-conquer technique to split the path into segments, and recursively find the path for each segment. Once we have the procedure to find the vertex-root path, then we can use it to implement root-changing. The idea is very simple, if we want to change the root to a certain vertex, we just need to find the path from that vertex to the root, and reverse the parent pointers of every vertex on the path. The path finding procedure is described in Algorithm~\ref{alg:path_find}. The root changing procedure is described in Algorithm~\ref{alg:root_change}.

\begin{algorithm}[t]
\caption{Depth and Ancestors of Every Vertex}\label{alg:ancestors}
\begin{algorithmic}[1]
\Procedure{ \textsc{FindAncestors} }{$\p:V\rightarrow V$} \Comment{Lemma~\ref{lem:depthandancestor}}
\State For $v\in V$ let $g_0(v)=\p(v).$ If $\p(v)=v,$ let $h_0(v)= 0.$ Otherwise, let $h_0(v)=  \nul .$
\State Let $l= 0.$
\For{$\exists v\in V,h_l(v)= \nul $}
    \State $l\leftarrow l+1.$
    \For{$v\in V$}
        \State Let $g_l(v)=g_{l-1}(g_{l-1}(v)).$ \Comment{$g_l$ is $\p^{(2^l)}.$}
        \If{$h_{l-1}(v)\not= \nul $} $h_l(v)=h_{l-1}(v).$
        \ElsIf{$h_{l-1}(g_{l-1}(v))\not= \nul $} $h_l(v)=h_{l-1}(g_{l-1}(v))+2^{l-1}.$
        \Else ~$h_l(v)= \nul .$
        \EndIf
    \EndFor
\EndFor
\State Let $r=l,\dep_{\p}\leftarrow h_r.$
\State \Return $r, \dep_{\p}, \{g_i:V\rightarrow V\mid i\in\{0\}\cup[r]\}.$ \Comment{$\dep_{\p}:V\rightarrow \mathbb{Z}_{\geq 0}$}
\EndProcedure
\end{algorithmic}
\end{algorithm}

\begin{definition}\label{def:num_iter_ancestors}
Let $\p:V\rightarrow V$ be a set of parent pointers (See Definition~\ref{def:parent_pointers}) on a vertex set $V.$ The number of iterations of $\textsc{FindAncestors}(\p)$  is defined as the value of $r$ at the end of the procedure.
\end{definition}

\begin{lemma}\label{lem:depthandancestor}
Let $\p:V\rightarrow V$ be a set of parent pointers (See Definition~\ref{def:parent_pointers}) on a vertex set $V.$ Let $(r,\dep_{\p},\{g_i\mid i\in\{0\}\cup[r]\})= \textsc{FindAncestors}(\p)$ (Algorithm~\ref{alg:ancestors}).  Then the number of iterations (see Definition~\ref{def:num_iter_ancestors}) $r$ should be at most $\lceil\log(\dep(\p)+1)\rceil,$ $\dep_{\p}:V\rightarrow \mathbb{Z}_{\geq 0}$ records the depth of every vertex in $V,$ and $\forall i\in \{0\}\cup[r],v\in V$ $g_i(v)=\p^{(2^i)}(v).$
\end{lemma}
\begin{proof}
 $h_l$ and $g_l$ will satisfies the properties in the following claim.

\begin{claim}\label{cla:find_path_rounds}
$\forall i\in \{0\}\cup[r],v\in V$ $g_i(v)=\p^{(2^i)}(v),$ and if $\dep_{\p}(v)\leq 2^i-1$ then $h_i(v)=\dep_{\p}(v).$ Otherwise $\dep_{\p}(v)= \nul .$
\end{claim}
\begin{proof}
The proof is by induction. The claim is obviously true when $i=0.$ Suppose the claim is true for $i-1.$ We have $g_i(v)=g_{i-1}(g_{i-1}(v))=\p^{(2^{i-1})}(\p^{(2^{i-1})}(v))=\p^{(2^i)}(v).$ If $h_i(v)\not= \nul ,$ then there are two cases. In the first case, we have $h_i(v)=h_{i-1}(v).$ By induction we know $h_i(v)=\dep_{\p}(v).$ In the second case, we have $h_i(v)=h_{i-1}(g_{i-1}(v))+2^{i-1}=\dep_{\p}(\p^{(2^{i-1})}(v))+2^{i-1}.$ Notice that in this case $h_{i-1}(v)= \nul ,$ thus by the induction, $\dep_{\p}(v)\geq 2^{i-1}.$ Therefore, $\dep_{\p}(v)=\dep_{\p}(\p^{(2^{i-1})}(v))+2^{i-1}=h_i(v).$ If $h_i(v)= \nul ,$ then it means that $h_{i-1}(\p^{(2^{i-1})}(v))= \nul $ which implies that $\dep_{\p}(v)\geq 2^i.$
\end{proof}
Due to the above claim, we know that if $i\geq \lceil\log(\dep_{\p}(v)+1)\rceil$ then $h_i(v)\not= \nul .$ Thus, we have $r\leq \lceil\log(\dep(\p)+1)\rceil.$ Since the procedure returns $h_r$ as $\dep_{\p},$ the returned $\dep_{\p}$ is correct.
\end{proof}

\begin{algorithm}[t]
\caption{Path in a Tree}\label{alg:path_find}
\begin{algorithmic}[1]
\Procedure{ \textsc{FindPath} }{$\p:V\rightarrow V,q\in V$} \Comment{Lemma~\ref{lem:find_path}}
\State Output: $\dep_{\p}:V\rightarrow\mathbb{Z}_{\geq 0},P\subseteq V,w\in V\cup\{ \nul \}.$
\State $(r,\dep_{\p},\{g_i\mid i\in\{0\}\cup[r]\}) = \textsc{FindAncestors}(\p)$ \Comment{Algorithm~\ref{alg:ancestors}}
\State Let $S_0=\{(q,g_r(q))\},k=\lceil\log(\dep_{\p}(q))\rceil.$\Comment{$S_0$ contains $(q,\text{the~root~of~}q)$} \label{sta:alg_find_init}
\For{$i=1\rightarrow k$} \Comment{$S_i$ is a set of segments partitioned the path from $q$ to the root of $q$}
    \State Let $S_i\leftarrow \emptyset.$
    \For{$(x,y)\in S_{i-1}$}
        \If{$\dep_{\p}(x)-\dep_{\p}(y)> 2^{k-i}$} $S_i\leftarrow S_i\cup\{(x,g_{k-i}(x)),(g_{k-i}(x),y)\}.$\label{sta:alg_find_split}
        \Else~ $S_i\leftarrow S_i\cup\{(x,y)\}.$\label{sta:alg_find_not_split}
        \EndIf
    \EndFor
\EndFor \Comment{$S_k$ only contains segments with length at most $1$}
\State Let $P\leftarrow\{q\}$
\For{$(x,y)\in S_k$}
\State Let $P\leftarrow P\cup\{y\}$
\EndFor
\State Find $w\in P$ with $\dep_{\p}(w)=1.$ If $w$ does not exist, let $w\leftarrow  \nul .$
\State \Return $(\dep_{\p},P,w)$
\EndProcedure
\end{algorithmic}
\end{algorithm}

\begin{lemma}\label{lem:find_path}
Let $\p:V\rightarrow V$ be a set of parent pointers (See Definition~\ref{def:parent_pointers}) on a vertex set $V.$ Let $q$ be a vertex in $V$. Let $(\dep_{\p},P,w) = \textsc{FindPath}(\p,q)$ (Algorithm~\ref{alg:path_find}). Then $\dep_{\p}:V\rightarrow \mathbb{Z}_{\geq 0}$ records the depth of every vertex in $V$ and $P\subseteq V$ is the set of all vertices on the path from $q$ to the root of $q$, i.e. $P=\{v\in V\mid \exists k\geq 0,v=\p^{(k)}(q)\}.$ If $\dep_{\p}(q)\geq 1,$ then $w=\p^{(\dep_{\p}(q)-1)}(q).$ Furthermore, $k$ should be at most $\lceil\log(\dep(\p))\rceil$. 
\end{lemma}
\begin{proof} %
By Lemma~\ref{lem:depthandancestor}, since $(r,\dep_{\p},\{g_i\mid i\in\{0\}\cup[r]\})=\textsc{FindAncestors}(\p),$ we know $r$ should be at most $\lceil\log(\dep(\p)+1)\rceil,$ $\dep_{\p}:V\rightarrow \mathbb{Z}_{\geq 0}$ records the depth of every vertex in $V,$ and $\forall i\in \{0\}\cup[r],v\in V$ $g_i(v)=\p^{(2^i)}(v).$ Thus $k=\lceil\log(\dep_{\p}(q))\rceil\leq \lceil\log(\dep(\p)+1)\rceil$

Now let us prove that $P$ is the vertex set of all the vertices on the path from $q$ to the root of $q$. We use divide-and-conquer to get $P$. The following claim shows that $S_i$ is a set of segments which is a partition of the path, and each segment has length at most $2^{k-i}.$
\begin{claim}\label{cla:property_path}
$\forall i\in\{0\}\cup[k],$ $S_i$ satisfies the following properties:
\begin{enumerate}
\item $\exists (x,y)\in S_i$ such that $x=q.$\label{itm:path_pro1}
\item $\exists (x,y)\in S_i$ such that $y=g_r(q).$\label{itm:path_pro2}
\item $\forall (x,y)\in S_i,$ $\dep_{\p}(y)-\dep_{\p}(x)\leq 2^{k-i}.$\label{itm:path_pro3}
\item $\forall (x,y)\in S_i,$ if $y\not=g_r(q),$ then $\exists (x',y')\in S_i,x'=y.$\label{itm:path_pro4}
\item $\forall (x,y)\in S_i,$ $\exists j\in \mathbb{Z}_{\geq 0}, \p^{(j)}(x)=y.$\label{itm:path_pro5}
\end{enumerate}
\end{claim}
\begin{proof}
Our proof is by induction. According to line~\ref{sta:alg_find_init}, all the properties hold when $i=0.$ Suppose all the properties hold for $i-1.$ For property~\ref{itm:path_pro1}, by induction we know there exists $(x,y)\in S_{i-1}$ such that $x=q.$ Then by line~\ref{sta:alg_find_split} and line~\ref{sta:alg_find_not_split}, there must be an $(x,y')$ in $S_i.$ For property~\ref{itm:path_pro2}, by induction we know there exists $(x,y)\in S_{i-1}$ such that $y=g_r(q).$ Thus, there must be an $(x',y)$ in $S_i.$ For property~\ref{itm:path_pro3}, if $(x,y)$ is added into $S_i$ by line~\ref{sta:alg_find_not_split}, then $\dep_{\p}(x)-\dep_{\p}(y)\leq 2^{k-i}.$ Otherwise, in line~\ref{sta:alg_find_split}, we have $\dep_{\p}(x)-\dep_{\p}(g_{k-i}(x))\leq 2^{k-i},\dep_{\p}(g_{k-i}(x))-\dep_{\p}(y)\leq 2^{k-i+1}-2^{k-i}=2^{k-i}.$
For property~\ref{itm:path_pro4}, if $(x,y)$ is added into $S_i$ by line~\ref{sta:alg_find_not_split}, then by induction there is $(y,y')\in S_{i-1},$ and thus by line~\ref{sta:alg_find_not_split} and line~\ref{sta:alg_find_split}, there must be $(y,y'')\in S_i.$ Otherwise, in line~\ref{sta:alg_find_split} will generate two pairs $(x,g_{k-i}(x)),(g_{k-i}(x),y).$ For $(x,g_{k-i}(x)),$ the property holds. For $(g_{k-i}(x),y),$ there must be $(y,y')\in S_{i-1}$ and thus there should be $(y,y'')\in S_i.$ For property~\ref{itm:path_pro5}, since $g_{k-i}(x)=\p^{(k-i)}(x),$ for all pairs generated by line~\ref{sta:alg_find_split} and line~\ref{sta:alg_find_not_split}, the property holds.
\end{proof}

By Claim~\ref{cla:property_path}, we know
\begin{align*}
S_k=\{(q,\p(q)),(\p(q),\p^{(2)}(q)),(\p^{(2)}(q),\p^{(3)}(q)),\cdots,(\p^{(\dep_{\p}(q)-1)}(q),\p^{(\dep_{\p}(q))}(q))\}.
\end{align*}
Thus, $P$ is the set of all the vertices on the path from $q$ to the root of $q$. And $w=\p^{(\dep_{\p}(q)-1)}(q)$ when $\dep_{\p}(q)\geq 1.$
\end{proof}

\begin{algorithm}[t]
\caption{Root Changing}\label{alg:root_change}
\begin{algorithmic}[1]
\Procedure{RootChange}{$\p:V\rightarrow V,q\in V$} \Comment{Lemma~\ref{lem:root_change}}
\State Output: $\wh{\p}:V\rightarrow V.$
\State $(\dep_{\p},P,w) = \textsc{FindPath}(\p,q).$ \Comment{Algorithm~\ref{alg:path_find}}
\State $\forall v\in V\setminus P,$ let $\wh{\p}(v)=\p(v).$
\State Let $\wh{\p}(q)=q.$
\State Let $h:\{0\}\cup[\dep_{\p}(q)]\rightarrow P$ such that $\forall i\in\{0\}\cup[\dep_{\p}(q)], h(i)=x$  where $\dep_{\p}(x)=i.$
\For{$v\in P\setminus\{q\}$} \Comment{Reverse $\p$ of all the vertices on the path from $q$ to the root of $q$.}
\State Let $\wh{\p}(v)=h(\dep_{\p}(v)+1).$
\EndFor
\State \Return $\wh{\p}.$
\EndProcedure
\end{algorithmic}
\end{algorithm}

\begin{lemma}\label{lem:root_change}
Let $\p:V\rightarrow V$ be a set of parent pointers (See Definition~\ref{def:parent_pointers}) on a vertex set $V.$ Let $q$ be a vertex in $V$. Let $\wh{\p}=\textsc{RootChange}(\p,q)$ (Algorithm~\ref{alg:root_change}). Then $\wh{\p}:V\rightarrow V$ is still a set of parent pointers (See Definition~\ref{def:parent_pointers}) on $V.$ $\forall v\in V,$ if $\p^{(\infty)}(v)=\p^{(\infty)}(q)$ then $\wh{\p}^{(\infty)}(v)=q.$ Otherwise $\wh{\p}^{(\infty)}(v)=\p^{(\infty)}(v).$ $\forall u\not=v\in V,\p(v)=u\Leftrightarrow$ either $\wh{\p}(v)=u$ or $\wh{\p}(u)=v.$ Furthermore, $\dep(\wh{\p})\leq 2\dep(\p).$
\end{lemma}
\begin{proof}
For a vertex $v\in V,$ if $\{u\mid i\in\mathbb{Z}_{\geq 0},u=\p^{(i)}(v)\}\cap P=\emptyset,$ then we have $\forall i\in\mathbb{Z}_{\geq 0},\p^{(i)}(v)=\wh{\p}^{(i)}(v).$ According to Lemma~\ref{lem:find_path}, $P=\{u\in V\mid i\in \mathbb{Z}_{\geq 0},\p^{(i)}(q)=u\}.$ Then for all $v\in P\setminus\{q\},$ if $\p(u)=v$ then $\wh{\p}(v)=u.$ Thus, $\forall u\in P,$ $\wh{\p}^{(\infty)}(u)=q.$ Let $i^*$ be the smallest number such that $\p^{(i^*)}(v)\in P.$ Then $\wh{\p}^{(i^*)}(v)\in P.$ Thus, $\wh{\p}^{(\infty)}(v)=\wh{\p}^{(\infty)}(\wh{\p}^{(i^*)}(v))=q.$ Furthermore, we have $\forall v\in V,\dep(\wh{\p})\leq \dep(\p)+\dep_{\p}(q)\leq 2\dep(\p).$
\end{proof}

    \subsection{Spanning Forest Expansion}
In this section, we give the definition of spanning forest. If we are given a spanning forest of a contracted graph and spanning trees of each contracted component, then we show a procedure which can merge them to get a spanning forest of the original graph. Before go to the details, let us formally define the spanning forest.

\begin{definition}[Rooted Spanning Forest]\label{def:spanning_forest}
Let $G=(V,E)$ be an undirected graph. Let $\p:V\rightarrow V$ be a set of parent pointers which is compatible (Definition~\ref{def:compatible_parent_pointers}) with $G$. If $\forall u,v\in V,\dist_G(u,v)<\infty\Rightarrow \p^{(\infty)}(u)=\p^{(\infty)}(v),$ and $\forall v\in V,\p(v)\not=v\Rightarrow (v,\p(v))\in E,$ then we call $\p$ a rooted spanning forest of $G$.
\end{definition}

The Algorithm~\ref{alg:spanning_forest_expansion} shows how to combine the spanning forest in the contracted graph with local spanning trees to get a spanning forest in the graph before contraction. Figure~\ref{fig:spanforestexp} shows an example.

\begin{figure}[!t]
  \centering
  \includegraphics[width=0.48\textwidth]{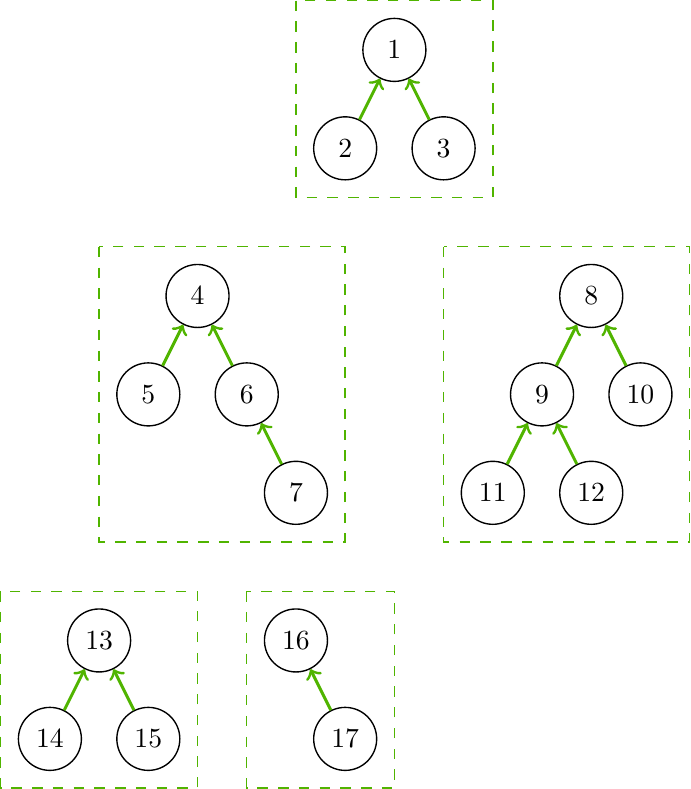}
  \hspace{1mm}
  \includegraphics[width=0.48\textwidth]{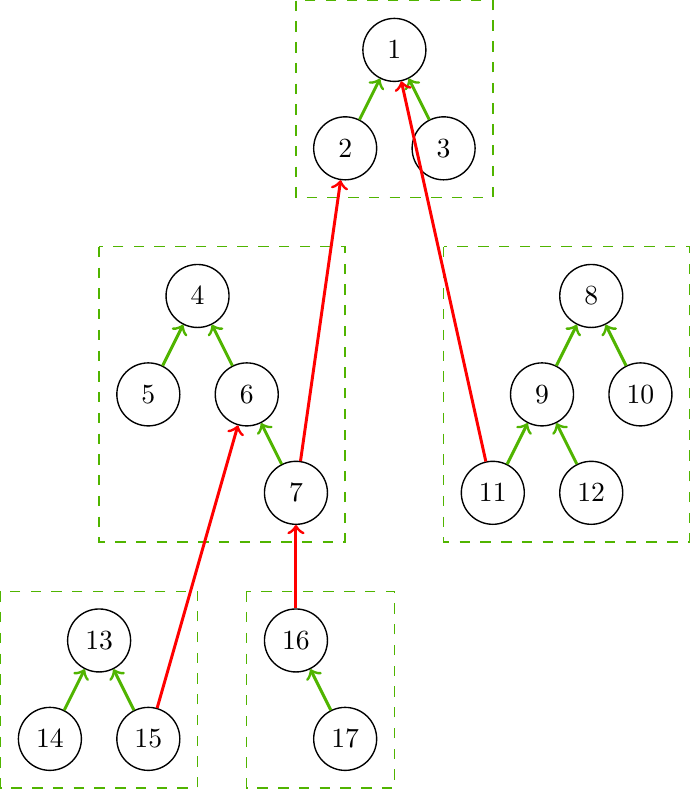}\\
  \vspace{2mm}
  \includegraphics[width=0.48\textwidth]{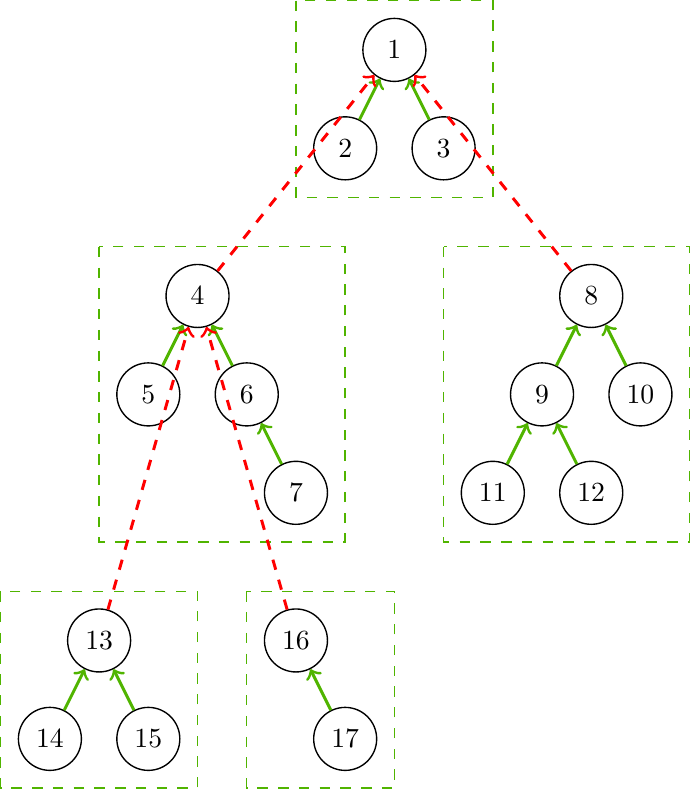}
  \hspace{1mm}
  \includegraphics[width=0.48\textwidth]{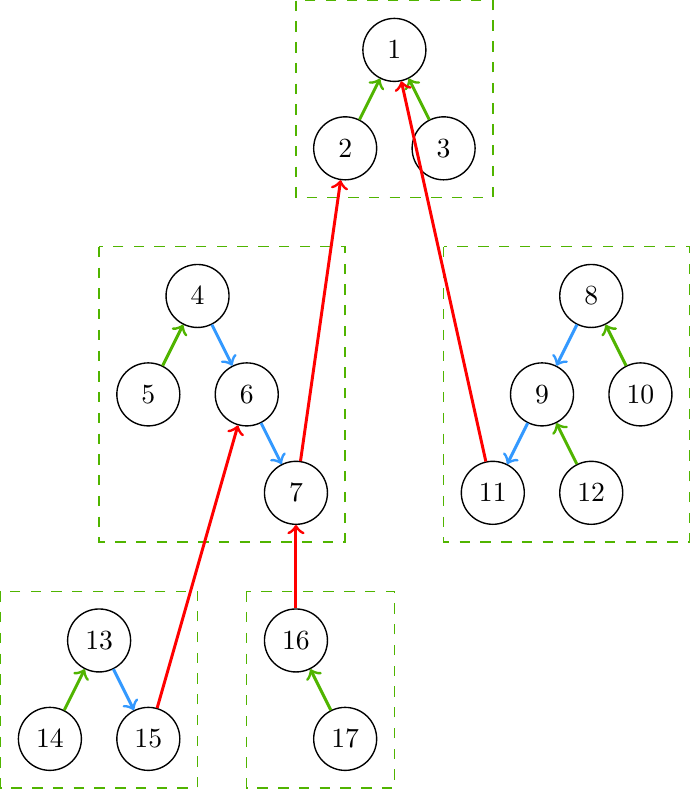}
  \caption{Each tree with green edges on the top-left is a rooted tree of each contracted component. For example, there are five components $\{1,2,3\}, \{4,5,6,7\}, \{ 8,9,10,11,12\}, \{ 13,14,15\}, \{16,17\}$. The dashed edges in the bottom-left figure is a root spanning tree of five components. The red edges in the top-right figure correspond to the dashed edges in the bottom-left figure before contraction. In bottom-right figure, by changing (see blue edges) the root of each contracted tree, we get a rooted spanning tree in the original graph}\label{fig:spanforestexp}
\end{figure}

\begin{lemma}\label{lem:expandrootedspanningforest}
Let $G_2=(V_2,E_2)$ be an undirected graph. Let $\wt{\p}:V_2\rightarrow V_2$ be a set of parent pointers (See Definition~\ref{def:parent_pointers}) which satisfies that $\forall v\in V_2$ with $\wt{\p}(v)\not=v$, $(v,\wt{\p}(v))$ must be in $E_2$. Let $G_1=(V_1,E_1)$ be an undirected graph satisfies $V_1=\{v\in V_2\mid \wt{\p}(v)=v\},E_1=\{(u,v)\in V_1\times V_1\mid u\not=v,\exists (x,y)\in E_2,\wt{\p}^{(\infty)}(x)=u,\wt{\p}^{(\infty)}(y)=v\}.$ Let $\p:V_1\rightarrow V_1$ be a rooted spanning forest (See Definition~\ref{def:spanning_forest}) of $G_1.$ Let $f:V_1\times V_1\rightarrow\{ \nul \}\cup \left(V_2\times V_2\right)$ satisfy the following property:
for $u\not=v\in V_1,$ if $\p(u)=v,$ then $f(u,v)\in\{(x,y)\in E_2\mid \wt{\p}^{(\infty)}(x)=u,\wt{\p}^{(\infty)}(y)=v\},$ and $f(v,u)\in \{(x,y)\in E_2\mid \wt{\p}^{(\infty)}(x)=v,\wt{\p}^{(\infty)}(y)=u\}.$
Let $\wh{\p} = \textsc{ForestExpansion}(\p,\wt{\p},f).$ Then $\wh{\p}:V_2\rightarrow V_2$ is a rooted spanning forest of $G_2.$ In addition, $\dep(\wh{\p})\leq (2\cdot\dep(\wt{\p})+1)(\dep(\p)+1).$
\end{lemma}
\begin{proof}
Let $x,y\in V_2,u=\wt{\p}^{(\infty)}(x),v=\wt{\p}^{(\infty)}(y)\in V_1$ if $\dist_{G_2}(x,y)<\infty,$ then since $E_1=\{(u',v')\in V_1\times V_1\mid u'\not=v',\exists (x',y')\in E_2,\wt{\p}^{(\infty)}(x')=u',\wt{\p}^{(\infty)}(y')=v'\},$ it must be true that $\dist_{G_1}(u,v)<\infty.$ Since $\p$ is a spanning forest of $G_1,$ we have $\p^{(\infty)}(u)=\p^{(\infty)}(v).$ It suffices to say $\forall x\in V_2,\wh{\p}^{(\infty)}(x)=\p^{(\infty)}(\wt{\p}^{(\infty)}(x)).$ We can prove it by induction on $\dep_{\p} (\wt{\p}^{(\infty)}(x)).$ Let $u=\wt{\p}^{(\infty)}(x).$ If $\dep_{\p}(u)=0,$ then $\p^{(\infty)}(u)=u.$ In this case, we have $\wh{\p}^{(\infty)}(x)=\wt{\p}^{(\infty)}(x)=u=\p^{(\infty)}(u)=\p^{(\infty)}(\wt{\p}^{(\infty)}(x)),$ and also we have $\dep_{\wh{\p}}(x)=\dep_{\wt{\p}}(x).$ Now suppose for all $x\in V_2$ with $\dep_{\p}(\wt{\p}^{(\infty)}(x))\leq i-1,$ it has $\wh{\p}^{(\infty)}(x)=\p^{(\infty)}(\wt{\p}^{(\infty)}(x))$ and $\dep_{\wh{\p}}(x)\leq i\cdot (2\dep(\wt{\p})+1).$ Let $y\in V_2$ satisfy $\dep_{\p}(\wt{\p}^{(\infty)}(y))=i.$ Let $v=\wt{\p}^{(\infty)}(y).$ By line~\ref{sta:alg_spanexp_f} and the properties of $f$, we know $\wt{\p}^{(\infty)}(x_v)=v,$ and $\wt{\p}^{(\infty)}(y_v)=\p(v).$ By line~\ref{sta:alg_spanexp_change}, line~\ref{sta:alg_spanexp_root_up} and Lemma~\ref{lem:root_change}, we have $\wh{\p}_v^{(\infty)}(y)=x_v,\wh{\p}(x_v)=y_v.$ Thus, there must be $k\leq 2\dep_{\wt{\p}}(y)$ such that $\wh{\p}^{(k)}(y)=x_v.$ Since $\wh{\p}^{(\infty)}(y_v)=\p^{(\infty)}(v)$ and $\dep_{\wh{\p}}(y_v)\leq i\cdot(2\dep(\wt{\p})+1),$ we have $\wh{\p}^{(\infty)}(y)=\p^{(\infty)}(v)=\p^{(\infty)}(\wt{\p}^{(\infty)}(y))$ and $\dep_{\wh{\p}}(y)\leq (i+1)\cdot(2\dep(\wt{\p})+1).$

In addition, by the properties of $f$ and Lemma~\ref{lem:root_change}, $\forall v\in V_2$ with $\wh{\p}(v)\not =v,$ we have $(v,\wh{\p}(v))\in E_2.$ To conclude, $\wh{\p}: V_2\rightarrow V_2$ is a spanning forest of $G_2,$ and $\dep(\wh{\p})\leq (\dep(\p)+1)(2\dep(\wt{\p})+1).$
\end{proof}

\begin{algorithm}[t]
\caption{Spanning Forest Expansion}\label{alg:spanning_forest_expansion}
\begin{algorithmic}[1]
\small
\Procedure{\textsc{ForestExpansion}}{$\p:V_1\rightarrow V_1,\wt{\p}:V_2\rightarrow V_2,f:V_1\times V_1\rightarrow \{ \nul \}\cup \left(V_2\times V_2\right)$}

\Comment{Lemma~\ref{lem:expandrootedspanningforest}}
\State Output: $\wh{\p}:V_2\rightarrow V_2.$
\State $((V'_2,\emptyset),\wt{\p}^{(\infty)}) = \textsc{TreeContraction}((V_2,\emptyset),\wt{\p}).$\label{sta:alg_span_exp_tree_contract} \Comment{Algorithm~\ref{alg:tree_contraction}}
\For{$v\in V_1$}
\State Let $V_2(v)=\{u\in V_2\mid \wt{\p}^{(\infty)}(u)=v\}.$
\State Let $\wt{\p}_v:V_2(v)\rightarrow V_2(v)$ such that $\forall u\in V_2(v),\wt{\p}_v(u)=\wt{\p}(u).$
\If{$\p(v)\not=v$}
\State Let $(x_v,y_v)=f(v,\p(v)).$\label{sta:alg_spanexp_f}
\State $\wh{\p}_v = \textsc{RootChange}(\wt{\p}_v,x_v).$\label{sta:alg_spanexp_change} \Comment{Algorithm~\ref{alg:root_change}}
\State Let $\wh{\p}(x_v)=y_v,$ and $\forall u\in V_2(v)\setminus\{x_v\},$ $\wh{\p}(u)=\wh{\p}_v(u).$\label{sta:alg_spanexp_root_up}
\Else~$\forall u\in V_2(v)$ let $\wh{\p}(u)=\wt{\p}_v(u).$
\EndIf
\EndFor
\State \Return $\wh{\p}.$
\EndProcedure
\end{algorithmic}
\end{algorithm}
    \subsection{Spanning Forest Algorithm}
In this section, we show how to apply the ideas shown in connectivity algorithm to get an spanning forest algorithm. Algorithm~\ref{alg:batch_algorithm3} can output a spanning forest of a graph $G$, but the edges are not orientated. Then in the Algorithm~\ref{alg:batch_algorithm4}, we assign each forest edge an direction thus it is a rooted spanning forest.

Before we prove the correctness of the algorithms, let us briefly introduce the meaning of each variables appeared in the algorithms.

In Algorithm~\ref{alg:batch_algorithm3}, $G_0$ is the original input graph, for $i\in\{0\}\cup[r-1],$ $G'_i$ is obtained by deleting all the small size connected components in $G_i,$ and $G_{i+1}$ is obtained by contracting some vertices of $G'_i.$ For a vertex $v$ in graph $G_i,$ if $h_i(v)= \nul ,$ then it means that the connected component which contains $v$ is deleted when obtaining $G'_i.$ If $h_i(v)\not= \nul ,$ it means that the vertex $v$ is contracted to the vertex $h_i(v)$ when obtaining $G_{i+1}.$ $\p_i$ is a rooted forest (may not be spanning) in graph $G_i,$ if a tree from the forest is spanning in $G_i,$ then all the vertex in that tree will be deleted when obtaining $G'_i.$ Otherwise all the vertices in that tree will be contracted to the root, and the root will be one of the vertex in $G_{i+1}.$ Since each connected component in $G_{i+1}$ is obtained by contraction of some vertices in a connected component in $G_i,$ each edge in $G_{i+1}$ must correspond to an edge in $G_{i}$ where the end vertices of the edge are contracted to different vertices. Thus, each edge in $G_i$ should correspond to an edge in $G,$ and $g_i:E_i\rightarrow E$ records the such correspondence. $D_i$ records the edges added to the spanning forest $F$ in the $i^{\text{th}}$ round. For each vertex $v$ in graph $G_i,$ $\wt{T}_i(v)$ is a local shortest path tree (See definition~\ref{def:local_short_tree}) which is either with a large size or is a spanning tree in the component of $v$. $L_i$ is a set of random leaders in $G'_i$ such that in each local shortest path tree $\wt{T}_i(v),$ there is at least one leader shown in the tree. The following lemmas formally state the properties of the algorithm.

\begin{algorithm}[t]
\caption{Undirected Graph Spanning Forest}\label{alg:batch_algorithm3}
\begin{algorithmic}[1]
\small
\Procedure{SpanningForest}{$G=(V,E),m,r$} 
\Comment{Corollary~\ref{cor:correctness_batch_alg3}, Theorem~\ref{thm:success_prob_batch_alg3}}
\State Output: FAIL or $\{V_i\subseteq V\mid i\in\{0\}\cup [r]\},\{\p_i:V_i\rightarrow V_i\mid i\in\{0\}\cup [r-1]\},\{h_i:V_{i}\rightarrow V_{i+1}\cup\{ \nul \}\mid i\in\{0\}\cup[r-1]\},F\subseteq E.$
\State $n_0=n=|V|,G_0=(V_0,E_0)=(V,E).$
\State Let $g_0:E_0\rightarrow E$ be an identity map.
\State Let $n'_0=n_0.$
\For{$i=0\rightarrow r-1$}
\State $D_i\leftarrow \emptyset.$
\State $\left(\{\wt{T}_i(v)\mid v\in V_i\},\{\dep_{\wt{T}_i(v)}\mid v\in V_i\}\right)=\textsc{MultipleLargeTrees}(G_i,m).$\label{sta:alg_span_multilargetree}
\Comment{Algorithm~\ref{alg:maximal_short_tree}}
\State Let $V'_i=\{v\in V_i\mid |V_{\wt{T}_i(v)}|\geq \lceil(m/n_i)^{1/4}\rceil\},E'_i=\{(u,v)\in E_i \mid u,v\in V'_i\},G'_i=(V'_i,E'_i).$\label{sta:alg_span_create_veprime}
\State $\forall v\in V_i\setminus V'_i,$ let $h_i(v)= \nul ,u_v=\min_{u\in V_{\wt{T}_i(v)}} u.$ Let $ \p_i(v)=\p_{\wt{T}_i(u_v)}(v).$ \label{sta:alg_span_spanningtree}
\State $\forall v\in V_i\setminus V'_i,$ if $\p_i(v)\not=v,$ then $D_i\leftarrow D_i\cup\{g_i(\p_i(v),v),g_i(v,\p_i(v))\}.$\label{sta:alg_span_in_vinotprime}
\State Let $\gamma_i=\lceil(m/n_{i})^{1/4}\rceil,p_i=\min((30\log(n)+100 )/ \gamma_i,1/2).$\label{sta:alg_span_samp_prob}
\State Let $l_i:V'_i\rightarrow \{0,1\}$ be chosen randomly s.t. $\forall v\in V'_i,l_i(v)$ are i.i.d. Bernoulli random variables with $\Pr(l_i(v)=1)=p_i$.
\State Let $L_i=\{v\in V'_i\mid l_i(v)=1\}\cup\{v\in V'_i\mid \forall u\in V_{\wt{T}_i(v)},l_i(u)=0\}.$\label{sta:alg_span_leaders}
\State For $v\in V'_i,$ let $z_i(v)=\arg\min_{u\in L_i\cap V_{\wt{T}_i(v)}}\dep_{\wt{T}_i(v)}(u).$ If $z_i(v)=v,$ let $\p_i(v)=v.$\label{sta:alg_span_assign_to_leader}
\State Otherwise, $(\dep_{\wt{T}_i(v)},P_i(v),w_i(v)) = \textsc{FindPath}(\p_{\wt{T}_i(v)},z_i(v)),$ and let $\p_i(v)=w_i(v).$\label{sta:alg_span_find_path}

\Comment{Algorithm~\ref{alg:path_find}}
\State Let $((V_{i+1},E_{i+1}),\p_i^{(\infty)})=\textsc{TreeContraction}(G'_i,\p_i:V'_i\rightarrow V'_i).$ \label{sta:alg_span_tree_contraction} \Comment{Algorithm~\ref{alg:tree_contraction}}
\State $G_{i+1}=(V_{i+1},E_{i+1}),n_{i+1}=|V_{i+1}|.$
\State $\forall v\in V'_i,$ $h_i(v)=\p_i^{(\infty)}(v).$ If $\p_i(v)\not= v,$  then $D_i\leftarrow D_i\cup\{g_i(\p_i(v),v),g_i(v,\p_i(v))\}.$\label{sta:alg_span_hi_notnull}
\State Let $g_{i+1}:E_{i+1}\rightarrow E$ satisfy $g_{i+1}(u,v)=\min_{(x,y)\in E_i,h_i(x)=u,h_i(y)=v} g_i(x,y).$\label{sta:alg_span_choose_one_edge}
\State Let $n'_{i+1}=n'_i+n_{i+1}.$ If $n'_{i+1}>40n,$ then return FAIL. \label{sta:alg_span_tree_fail_on_total_n}
\EndFor
\State If $n_r\not=0,$ return FAIL.
\State Let $F=\bigcup_{i\in\{0\}\cup[r-1]} D_i.$
\State \Return $\{V_i\mid i\in\{0\}\cup [r]\},\{\p_i\mid i\in\{0\}\cup [r-1]\},\{h_i\mid i\in\{0\}\cup[r-1]\},F.$
\EndProcedure
\end{algorithmic}
\end{algorithm}

\begin{lemma}\label{lem:span_diameter}
Let $G=(V,E)$ be an undirected graph, $m$ be a parameter which is at least $16|V|,$ and $r$ be a rounds parameter. If $\textsc{SpanningForest}(G,m,r)$ (Algorithm~\ref{alg:batch_algorithm3}) does not return FAIL, then $\diam(G)=\diam(G_0)\geq \diam(G'_0)\geq \diam(G_1)\geq \diam(G'_1)\geq \cdots \geq \diam(G_r).$
\end{lemma}
\begin{proof}
By property~\ref{itm:large_tree_pro3} of Lemma~\ref{lem:maximal_large_short_tree}, $\forall [i]\in \{0\}\cup[r-1], $ there is no edge between $V_i\setminus V'_i$ and $V'_i.$ Thus, $\diam(G'_i)\leq \diam(G_i).$ Then due to property~\ref{itm:tree_contracted_pro1} of Corollary~\ref{cor:tree_contract_conn}, we have $\diam(G_{i+1})\leq \diam(G'_i).$
\end{proof}

\begin{lemma}\label{lem:dep_p_i}
Let $G=(V,E)$ be an undirected graph, $m$ be a parameter which is at least $16|V|,$ and $r$ be a rounds parameter. If $\textsc{SpanningForest}(G,m,r)$ (Algorithm~\ref{alg:batch_algorithm3}) does not return FAIL, then $\forall i\in \{0\}\cup[r-1],$ $\dep(\p_i)\leq \min(\diam(G),\lfloor(m/n_i)^{1/2}\rfloor).$
\end{lemma}
\begin{proof}
Let $v\in V_i.$
If $v\in V_i\setminus V'_i,$ then due to property~\ref{itm:large_tree_pro3} of Lemma~\ref{lem:maximal_large_short_tree}, we have $V_{\wt{T}_i(v)}=V_{\wt{T}_i(u_v)}.$
Due to Lemma~\ref{lem:span_diameter} and Lemma~\ref{lem:maximal_large_short_tree}, we have $\dep_{\p_i}(v)\leq\dep(\wt{T}_i(u_v))\leq \min(\diam(G),\lfloor(m/n_i)^{1/2}\rfloor).$
For $v\in V_i,$ we define $\dist_{G_i}(v,L_i)=\min_{u\in L_i} \dist_{G_i}(v,u).$
By Lemma~\ref{lem:maximal_large_short_tree}, we know $\dist_{G_i}(v,L_i)=\dist_{G_i}(v,z_i(v)).$
Since $\wt{T}_i(v)$ is a LSPT (See Definition~\ref{def:local_short_tree}), by applying Lemma~\ref{lem:find_path}, we know $\dist_{G_i}(v,L_i)=\dist_{G_i}(w_i(v),L_i)+1,$ and $(v,w_i(v))\in E_i.$
Thus, by induction on $\dist_{G_i}(v,L_i),$ we can get $\dep_{\p_i}(v)\leq \dist_{G_i}(v,L_i).$
By Lemma~\ref{lem:span_diameter} and Lemma~\ref{lem:maximal_large_short_tree}, we can conclude $\dep(\p_i)(v)\leq \min(\diam(G),\lfloor(m/n_i)^{1/2}\rfloor).$
\end{proof}

\begin{lemma}\label{lem:spanning_tree}
Let $G=(V,E)$ be an undirected graph, $m$ be a parameter which is at least $16|V|,$ and $r$ be a rounds parameter. If $\textsc{SpanningForest}(G,m,r)$ (Algorithm~\ref{alg:batch_algorithm3}) does not return FAIL, then $\forall i\in\{-1,0\}\cup[r-1],$ we can define
\begin{align*}
\forall v\in V,h^{(i)}(v)=
\begin{cases}
 v & i=-1  ;\\
 h_i(h^{(i-1)}(v)) & h^{(i-1)}(v)\not= \nul ;\\
  \nul  & \mathrm{otherwise~}.
 \end{cases}
\end{align*}
Then we have following properties:
\begin{enumerate}
\item If $h^{(i)}(v)\not= \nul ,$ then $h^{(i)}(v)\in V_{i+1}.$\label{itm:span_pro1}
\item $\forall u,v\in V,h^{(i)}(u)\not=h^{(i)}(v),(u,v)\in E,$ we have $(h^{(i)}(u),h^{(i)}(v))\in E_{i+1}.$ \label{itm:span_pro2}
\item $\forall (x,y)\in E_{i+1},(u,v)=g_{i+1}(x,y),$ we have $(u,v)\in E,h^{(i)}(u)=x,h^{(i)}(v)=y.$\label{itm:span_pro3}
\end{enumerate}
\end{lemma}
\begin{proof}
For property~\ref{itm:span_pro1}, we can prove it by induction. It is true for $i=-1.$
If $h_0(v)\not= \nul ,$ we know $h_0(v)$ must be assigned at line~\ref{sta:alg_span_hi_notnull}.
Due to property~\ref{itm:tree_return_pro2} of Lemma~\ref{lem:contraction_properties}, $h_0(v)\in V_1.$
Suppose $\forall v\in V,h^{(i-1)}(v)\not= \nul ,$ we have $h^{(i-1)}(v)\in V_i.$
For a vertex $v$ with $h^{(i)}(v)\not= \nul ,$ according to the definition of $h^{(i)}(v),$ we know $h^{(i-1)}(v)\not= \nul .$
Let $u=h^{(i-1)}(v).$
$u$ must be a vertex in $G_i$ by the induction hypothesis.
Since $h^{(i)}(v)\not= \nul ,$ we know $h_i(u)\not= \nul .$
Thus, $h_i(u)$ must be assigned at line~\ref{sta:alg_span_hi_notnull}.
Due to property~\ref{itm:tree_return_pro2} of Lemma~\ref{lem:contraction_properties}, $h_i(u)$ must be in $G_{i+1},$ which implies $h^{(i)}(v)\in V_{i+1}.$

For property~\ref{itm:span_pro2}, we can also prove it by induction. It is true for $i=-1.$
If $(u,v)\in E,$ then due to property~\ref{itm:large_tree_pro3} of Lemma~\ref{lem:maximal_large_short_tree}, either both $u,v$ are in $V'_0$ or both $u,v$ are in $V_0\setminus V'_0.$
If both $u,v$ are in $V_0\setminus V'_0,$ then $h_0(u)=h_0(v)= \nul .$
Otherwise, if $h_0(u)\not=h_0(v),$ then due to property~\ref{itm:tree_return_pro3} of Lemma~\ref{lem:contraction_properties}, $(h_0(u),h_0(v))\in E_1.$
Now suppose we have $\forall u,v\in V,$ if $h^{(i-1)}(u)\not=h^{(i-1)}(v),(u,v)\in E,$ then $(h^{(i-1)}(u),h^{(i-1)}(v))\in E_i.$
Let $(u,v)\in E,h^{(i)}(u)\not=h^{(i)}(v).$
Let $x=h^{(i-1)}(u),y= h^{(i-1)}(v).$
Due to property~\ref{itm:large_tree_pro3} of Lemma~\ref{lem:maximal_large_short_tree}, either both $x,y$ are in $V'_i$ or both are in $V_i\setminus V'_i.$
If $x,y\in V_i\setminus V'_i,$ then $h_i(x)=h_i(y)= \nul $ which contradicts to $h^{(i)}(u)\not=h^{(i)}(v).$
Thus, both of $x,y\in V'_i.$
Then due to property~\ref{itm:tree_return_pro3} of Lemma~\ref{lem:contraction_properties}, $(h_i(x),h_i(v))\in E_{i+1}.$
Thus, $(h^{(i)}(u),h^{(i)}(v))\in E_{i+1}.$

For property~\ref{itm:span_pro3}, we can prove it by induction.
It is true for $i=-1.$
Let us consider the case when $i=0.$
Due to property~\ref{itm:tree_return_pro3} of Lemma~\ref{lem:contraction_properties} and the definition of $g_0,g_1,$ we have $\forall (x,y)\in E_1,$ $(u,v)=g_1(x,y),$ $h_0(u)=x,h_0(v)=y,(u,v)\in E.$
Now suppose the property holds for $i-1.$
Let $(x,y)\in E_{i+1}.$
Then $g_{i+1}(x,y)=g_i(x',y')$ for some $(x',y')\in E_i,h_i(x')=x,h_i(y')=y.$
Let $(u,v)=g_i(x',y').$
By the induction hypothesis $(u,v)\in E,h^{(i-1)}(u)=x',h^{(i-1)}(v)=y'.$
Thus, $h^{(i)}(u)=x,h^{(i)}(v)=y.$
\end{proof}

\begin{lemma}\label{lem:recursive_spanning_tree}
Let $G=(V,E)$ be an undirected graph, $m$ be a parameter which is at least $16|V|,$ and $r$ be a rounds parameter. If $\textsc{SpanningForest}(G,m,r)$ (Algorithm~\ref{alg:batch_algorithm3}) does not return FAIL, then $\forall i\in\{-1,0\}\cup[r-1],$ we can define
\begin{align*}
\forall v\in V,h^{(i)}(v)=
\begin{cases}
 v & i=-1\\
 h_i(h^{(i-1)}(v)) & h^{(i-1)}(v)\not= \nul \\
 \nul  & \mathrm{otherwise}.
\end{cases}
\end{align*}
Let $\forall i\in\{0\}\cup[r],$ $\wh{G}_i=(V_i,\wh{E}_i=\{(x,y)\mid (u,v)\in \bigcup_{j=i}^{r-1} D_j,h^{(j-1)}(u)=x,h^{(j-1)}(v)=y\}).$ Then $\wh{G}_i$ is a spanning forest of $G_i.$
\end{lemma}
\begin{proof}
The proof is by induction.
When $i=r,$ since $V_r=\emptyset,$ $\wh{G}_r=(\emptyset,\emptyset)$ is a spanning forest of $G_r.$
Now suppose $\wh{G}_{i+1}$ is a spanning forest of $G_{i+1}.$
Let $u,v\in V_i.$
By property~\ref{itm:span_pro2},~\ref{itm:span_pro3} of Lemma~\ref{lem:spanning_tree}, we have $\wh{E}_i\subseteq E_i.$
Thus, if $\dist_{G_i}(u,v)=\infty,$ then $\dist_{\wh{G}_i}(u,v)=\infty.$
If $\dist_{G_i}(u,v)<\infty,$ there are several cases:
\begin{enumerate}
\item If $h_i(u)=h_i(v)= \nul ,$ then due to line~\ref{sta:alg_span_spanningtree}, we know $u_u=u_v,$ and $\wt{T}_{i}(u_v)$ is a spanning tree of the component which contains $u,v.$ Thus, $\wh{G}_i$ has a spanning tree of the component which contains $u,v$.
\item If $h_i(u)=h_i(v)\not= \nul ,$ then $\p_i:\{x\in V_i\mid h_i(x)=h_i(v)\}\rightarrow \{x\in V_i\mid h_i(x)=h_i(v)\}$ is a tree, and $\forall y\in\{x\in V_i\mid h_i(x)=h_i(v)\},$ if $\p_i(y)\not=y,$ then $(y,\p_i(y))\in \wh{E}_i.$
    Since $\wh{G}_{i+1}$ does not have any cycle, there is a unique path from $u$ to $v$ in $\wh{G}_i$.
\item If $h_i(u)\not=h_i(v),$ then neither of them can be $ \nul .$
    Since $\wh{G}_{i+1}$ is a spanning forest on $G_{i+1},$ there must be a unique path from $h_i(u)$ to $h_i(v)$ in $\wh{G}_{i+1}.$
    Suppose the path in $\wh{G}_{i+1}$ is $h_i(u)=p_1-p_2-\cdots-p_t=h_i(v).$ Then there must be a sequence of vertices in $G_i,$ $u=p_{1,1},p_{1,2},p_{2,1},p_{2,2},\cdots,p_{t,1},p_{t,2}=v$ such that $h_i(p_{j,1})=h_i(p_{j,2})=p_j$ and $(p_{j-1,2},p_{j,1})\in \wh{E}_i.$
    Thus, there is a unique path from $u$ to $v$.
\end{enumerate}
Thus, $\wh{G}_i$ is a spanning forest of $G_i.$
\end{proof}

\begin{corollary}[Correctness of Algorithm~\ref{alg:batch_algorithm3}]\label{cor:correctness_batch_alg3}
Let $G=(V,E)$ be an undirected graph, $m$ be a parameter which is at least $16|V|,$ and $r$ be a rounds parameter. If $\textsc{SpanningForest}(G,m,r)$ (Algorithm~\ref{alg:batch_algorithm3}) does not return FAIL, then $\wh{G}_0=(V,F)$ is a spanning forest of $G.$
\end{corollary}
\begin{proof}
Just apply Lemma~\ref{lem:recursive_spanning_tree} for $i=0$ case.
\end{proof}

\begin{algorithm}[t]
\caption{Rooted Spanning Forest}\label{alg:batch_algorithm4}
\begin{algorithmic}[1]
\Procedure{Orientate}{$\{V_i\mid i\in\{0\}\cup [r]\},\{\p_i\mid i\in\{0\}\cup [r-1]\},\{h_i\mid i\in\{0\}\cup[r-1]\},F$}

\Comment Takes the output of Algorithm~\ref{alg:batch_algorithm3} as input.

\Comment{Theorem~\ref{thm:batch_algorithm4}}
\State Output:  $\p:V_0\rightarrow V_0.$
\State Let $F_0=F.$
\For{$i=0\rightarrow r-1$}\label{sta:batch_alg4_first_loop_st}
\State Initialize $F_{i+1}\leftarrow \emptyset, f_{i+1}:V_{i+1}\times V_{i+1}\rightarrow \{ \nul \}.$
\State $\forall (u,v)\in F_i,h_i(u)\not=h_i(v),$ let $F_{i+1}\leftarrow F_{i+1}\cup\{(h_i(u),h_i(v))\},f_{i+1}(h_i(u),h_i(v))\leftarrow(u,v).$
\EndFor\label{sta:batch_alg4_first_loop_ed}
\State $\wh{\p}_r:\emptyset\rightarrow \emptyset.$
\For{$i=r\rightarrow 1$}\Comment{$\wh{\p}_i$ is the spanning forest of $G_i.$}
\State Let $\wt{V}_i=V_i\cup \{v\in V_{i-1}\mid h_{i-1}(v)= \nul ,\p_{i-1}(v)=v\}.$\label{batch_alg4_second_loop1}
\State Let $\wt{\p}_i:\wt{V}_i\rightarrow \wt{V}_i$ satisfy $\forall v\in V_i,$ $\wt{\p}_i(v)=\wh{\p}_i(v),$ and $\forall v\in \wt{V}_i\setminus V_i,\wt{\p}_i(v)=v.$\label{batch_alg4_second_loop2}
\State Let $\wh{\p}_{i-1}=\textsc{ForestExpansion}(\wt{\p}_i,\p_{i-1},f_i).$ \label{batch_alg4_second_loop3} \Comment{Algorithm~\ref{alg:spanning_forest_expansion}}
\EndFor
\State Return $\wh{\p}_0$ as $\p.$
\EndProcedure
\end{algorithmic}
\end{algorithm}

\begin{lemma}\label{lem:F_corresponds_to_p}
Let $G=(V,E)$ be an undirected graph, $m$ be a parameter which is at least $16|V|,$ and $r$ be a rounds parameter. If $\textsc{SpanningForest}(G,m,r)$ (Algorithm~\ref{alg:batch_algorithm3}) does not return FAIL, then $\forall i\in\{-1,0\}\cup[r-1],$ we can define
\begin{align*}
\forall v\in V,h^{(i)}(v)=
\begin{cases}
 v & i=-1\\ 
 h_i(h^{(i-1)}(v)) & h^{(i-1)}(v)\not= \nul \\ 
 \nul  & \mathrm{otherwise.}
\end{cases}
\end{align*}
$\forall i\in\{0\}\cup [r-1],v\in V_i$ with $\p_i(v)\not=v,$ there exists $(x,y)\in F$ such that $h^{(i-1)}(x)=v,h^{(i-1)}(y)=\p_i(v).$
\end{lemma}
\begin{proof}
By line~\ref{sta:alg_span_in_vinotprime}, line~\ref{sta:alg_span_hi_notnull}, $\forall i\in\{0\}\cup[r-1], v\in V_i,\p_i(v)\not=v,$ we have $g_i(v,\p_i(v)),g_i(\p_i(v),v)\in D_i\subseteq F.$
Since $(\p_i(v),v)\in E_i,$ by property~\ref{itm:span_pro3} of Lemma~\ref{lem:spanning_tree}, $(x,y)=g_i(v,\p_i(v))$ satisfies $h^{(i-1)}(x)=v,h^{(i-1)}(y)=\p_i(v).$
\end{proof}

\begin{theorem}[Correctness of Algorithm~\ref{alg:batch_algorithm4}]\label{thm:batch_algorithm4}
Let $G=(V,E)$ be an undirected graph, $m$ be a parameter which is at least $16|V|,$ and $r$ be a rounds parameter. If $\textsc{SpanningForest}(G,m,r)$ (Algorithm~\ref{alg:batch_algorithm3}) does not return FAIL, then let the output be the input of $\textsc{Orientate}(\cdot),$ (Algorithm~\ref{alg:batch_algorithm4}) and the output $\p:V\rightarrow V$ of $\textsc{Orientate}(\cdot)$ will be a rooted spanning forest (See Definition~\ref{def:spanning_forest}) of $G$. Furthermore, $\dep(\p)\leq O(\diam(G))^r.$
\end{theorem}
\begin{proof}
The proof is by induction.
We want to show $\wh{\p}_i$ is a rooted spanning forest of $G_i.$
When $i=r,$ since $V_r=\emptyset,$ the claim is true.
Now suppose we have $\wh{\p}_{i+1}$ is a spanning forest of $G_{i+1}.$
Let $\wt{G}_{i+1}=(\wt{V}_{i+1},E_{i+1}).$
It is easy to see $\wt{\p}_{i+1}:\wt{V}_{i+1}\rightarrow\wt{V}_{i+1}$ is a spanning forest of $\wt{G}_{i+1}.$
An observation is $\wt{V}_{i+1}=\{v\in V_i\mid \p_i(v)=v\}.$
Thus, $\wt{\p}_{i+1},\p_{i}$ satisfies the condition in Lemma~\ref{lem:expandrootedspanningforest} when invoking
$\textsc{ForestExpansion}(\wt{\p}_{i+1},\p_{i},f_{i+1}).$
By Lemma~\ref{lem:F_corresponds_to_p}, we know $f_{i+1}$ also satisfies the condition in Lemma~\ref{lem:expandrootedspanningforest} when we invoke $\textsc{ForestExpansion}(\wt{\p}_{i+1},\p_{i},f_{i+1}).$
Thus, $\wh{\p}_i$ is a rooted spanning forest of $G_i$ due to Lemma~\ref{lem:expandrootedspanningforest}.

By Lemma~\ref{lem:expandrootedspanningforest}, we have $\dep(\wh{\p}_i)\leq 16\dep(\wh{\p}_{i+1})\diam(G).$ By induction, we have $\dep(\p)\leq O(\diam(G))^r.$
\end{proof}

\begin{lemma}\label{lem:sum_of_ni_is_not_large}
Let $G=(V,E)$ be an undirected graph, $m$ be a parameter which is at least $16|V|,$ and $r\leq n$ be a round parameter. If $\textsc{SpanningForest}(G,m,r)$ (Algorithm~\ref{alg:batch_algorithm3}) does not return FAIL, then with probability at least $0.89,$ $\sum_{i=0}^r n_i\leq 40n.$
\end{lemma}
\begin{proof}
Since $V_r\subseteq V_{r-1}\subseteq \cdots \subseteq V_0=V,$ we have $n_r\leq n_{r-1}\leq n_{r-2}\leq \cdots \leq n.$
Due to line~\ref{sta:alg_span_leaders}, line~\ref{sta:alg_span_assign_to_leader} and line~\ref{sta:alg_span_tree_contraction}, we know $\forall i\in\{0\}\cup[r-1],$ $V_{i+1}=L_i.$
If $p_i<1/2,$ we know $p_i=(30\log(n)+100 )/ \gamma_i.$
Since $|V_{\wt{T}_i}(v)|\geq \gamma_i,$ we can apply Lemma~\ref{lem:random_leader_props} to get $\Pr(|L_i|\leq 1.5p_i n_i)\geq \Pr(|L_i|\leq 0.75n_i)\geq 1-1/(100n)$.
By taking union bound over all $i\in \{0\}\cup[r-1],$ with probability at least $0.99,$ if $p_i<0.5,$ then $n_{i+1}\leq 0.75n_i.$
By applying Lemma~\ref{lem:random_leader_low_prob}, condition on $n_i$ and $p_i=\frac{1}{2},$ we have $\E(n_{i+1})\leq 0.75 n_i.$
 By Markov's inequality, with probability at $0.89,$ we have $\sum_{i=0}^r n_i\leq 40 n.$
\end{proof}

Now let us define the total iterations of Algorithm~\ref{alg:batch_algorithm3} as the following:
\begin{definition}[Total iterations]\label{def:num_total_it_spanning_tree}
Let graph $G=(V,E)$, $m\leq \poly(|V|)$ be a parameter which is at least $16|V|,$ and $r$ be a rounds parameter. The total number of iterations of $\textsc{SpanningForest}(G,m,r)$ (Algorithm~\ref{alg:batch_algorithm3}) is defined as $\sum_{i=0}^{r-1} (k_i+k'_i),$ where $\forall i\in\{0\}\cup[r-1],k_i$ denotes the number of iterations (See Definition~\ref{def:num_it_multiplelargetree}) of $\textsc{MultipleLargeTrees}(G_i,m)$ (see line~\ref{sta:alg_span_multilargetree}), and $k'_i$ denotes the number of iterations (See Definition~\ref{def:num_it_tree_contract}) of $\textsc{TreeContraction}(G'_i,\p_i)$ (see line~\ref{sta:alg_span_tree_contraction}).
\end{definition}

\begin{theorem}[Success probability of Algorithm~\ref{alg:batch_algorithm3}]\label{thm:success_prob_batch_alg3}
Let $G=(V,E)$ be an undirected graph. Let $m\leq \poly(n)$ and $m\geq 16|V|.$ Let $r$ be a rounds parameter. Let $c>0$ be a sufficiently large constant. If $r\geq c\log\log_{m/n}n,$ then with probability at least $0.79,$ $\textsc{SpanningForest}(G,m,r)$ (Algorithm~\ref{alg:batch_algorithm3}) does not return {\rm FAIL}. Furthermore, let $\forall i\in\{0\}\cup[r-1],k_i$ be the number of iterations (See Definition~\ref{def:num_it_multiplelargetree}) of $\textsc{MultipleLargeTrees}(G_i,m)$ and $k'_i$ be the number of iterations (See Definition~\ref{def:num_it_tree_contract}) of $\textsc{TreeContraction}(G'_i,\p_i:V'_i\rightarrow V'_i).$  Let $c_1>0$ be a sufficiently large constant. If $m\geq c_1 n\log^8 n,$ then with probability at least $0.99,$ $\sum_{i=0}^{r-1} k_i+k'_i\leq O(\min(\log(\diam(G))\log\log_{\diam(G)}(n),$ $r\log(\diam(G)))).$ If $m< c_1 n\log^8 n$, then with probability at least $0.98,$ $\sum_{i=0}^{r-1} k'_i+k_i\leq O(\min(\log(\diam(G))$ $\log\log_{\diam(G)}(n)+(\log\log n)^2,r\log(\diam(G)))).$
\end{theorem}
\begin{proof}
Due to Lemma~\ref{lem:sum_of_ni_is_not_large}, with probability at last $0.89,$ we have $\forall i\in[r],n'_i\leq 40n.$
Thus, we can condition on that $\textsc{SpanningForest}(G,m,r)$ will not fail on line~\ref{sta:alg_span_tree_fail_on_total_n}.

Due to Lemma~\ref{lem:num_it_multiplelargetree}, $k_i\leq O(\log (\diam(G_i)))\leq O(\log(\diam(G))).$ Due to Corollary~\ref{cor:tree_contract_conn} and Lemma~\ref{lem:dep_p_i}, $k'_i\leq O(\log(\diam(G))).$ Thus,
$\sum_{i\in\{0\}\cup[r-1]} k'_i+k_i\leq O(r\log(\diam(G))).$

Since $V_r\subseteq V_{r-1}\subseteq V_{r-2}\subseteq \cdots \subseteq V_0=V,$ we have $n_r\leq n_{r-1}\leq n_{r-2}\leq \cdots \leq n.$
Due to line~\ref{sta:alg_span_leaders}, line~\ref{sta:alg_span_assign_to_leader} and line~\ref{sta:alg_span_tree_contraction}, we know $\forall i\in\{0\}\cup[r-1],$ $V_{i+1}=L_i.$
If $p_i<1/2,$ we know $p_i=(30\log(n)+100 )/ \gamma_i.$
Since $|V_{\wt{T}_i}(v)|\geq \gamma_i,$ we can apply Lemma~\ref{lem:random_leader_props} to get $\Pr(|L_i|\leq 1.5p_i n_i)\geq 1-1/(100n)$.
By taking union bound over all $i\in \{0\}\cup[r-1],$ with probability at least $0.99,$ if $p_i<0.5,$ then $n_{i+1}\leq 1.5p_in_i\leq 0.75n_i.$
Let $\mathcal{E}$ be the event that $\forall i\in \{0\}\cup[r-1],$ if $p_i<0.5,$ then $n_{i+1}\leq 1.5p_in_i.$
Now, we suppose $\mathcal{E}$ happens.

If $p_0=0.5,$ then $m\leq n\cdot (600\log n)^8.$
By applying Lemma~\ref{lem:random_leader_low_prob}, $\E(n_{i+1})=\E(|L_i|)\leq 0.75\E(n_i)\leq \cdots \leq 0.75^{i+1} n.$
By Markov's inequality, when $i^*\geq 8\log(6000\log n)/\log(4/3),$ with probability at least $0.99,$ $n_{i^*}\leq n/(600\log n)^8$ and thus $p_{i^*}<0.5.$ Condition on this event and $\mathcal{E}$, we have
{
\small
\begin{align*}
n_r&\leq \frac{\left(\frac{\left(\frac{n_{i^*}^{1.25}}{m^{0.25}}(45\log n+150)\right)^{1.25}}{m^{0.25}}(45\log n+150)\right)^{\cdots}}{\cdots}&~\text{(Apply $r'=r-i^*$ times)}\\
&=n_{i^*}/(m/n_{i^*})^{1.25^{r'}-1}\cdot (45\log n+150)^{4\cdot(1.25^{r'}-1)}\\
&\leq n/\left(m/\left(n_{i^*}(45\log n+150)^4\right)\right)^{1.25^{r'}-1}\\
&\leq n/\left(m/\left(n_{i^*}(45\log n+150)^4\right)\right)^{1.25^{r'/2}}
\leq n/\left(m/n\right)^{1.25^{r'/2}}\leq 1/2,
\end{align*}
}
where the second inequality follows by $n_{i^*}\leq n,$ the third inequality follows by $r'\geq 5,$ the forth inequality follows by $n_{i^*}\leq n/(600\log n)^8,$ and the last inequality follows by $r'\geq \frac{2}{\log 1.25}\log\log_{m/n}(2n).$
Since $16n\leq m\leq n\cdot (600\log n)^8,\log\log_{m/n} n=\Theta(\log\log n).$ Let $c>0$ be a sufficiently large constant.
Thus, when $r\geq c\log\log_{m/n} n\geq  i^*+r'=8\log(6000\log n)/\log(4/3)+\frac{2}{\log 1.25}\log\log_{m/n}(2n),$ with probability at least $0.98,$ $n_r=0$ implies that $\textsc{SpanningForest}(G,m,r)$ will not fail.
Due to Lemma~\ref{lem:num_it_multiplelargetree}, we have $k_i\leq O(\min(\log (m/n_{i}),\log(\diam(G)))).$ Thus,
\begin{align*}
& ~ \sum_{i=0}^{r-1} k_i \\
= & ~ \sum_{i=0}^{i^*} k_i +  \sum_{i=i^*+1}^{r-1} k_i\\
\leq & ~ O\left((\log \log n)^2\right)+ \sum_{i=i^*+1}^{r-1} k_i\\
\leq & ~ O\left((\log \log n)^2\right)+ \sum_{i:i\geq i^*+1,m/n_{i}\leq \diam(G)} k_i+\sum_{i:i\leq r,m/n_{i}> \diam(G)} k_i\\
\leq & ~ O\left((\log \log n)^2\right)+ O\left(\sum_{i=0}^{\lceil\log_{1.25}\log_2 (\diam(G))\rceil}\log( 2^{1.25^i})\right)+ O\left(\sum_{i=0}^{\lceil\log_{1.25}\log_{\diam(G)} (m)\rceil}\log( \diam(G))\right)\\
\leq & ~ O\left((\log \log n)^2\right)+O(\log(\diam(G)))+O(\log(\diam(G))\log\log_{\diam(G)} (n))\\
\leq & ~ O(\log(\diam(G))\log\log_{\diam(G)}(n)+(\log\log(n))^2),
\end{align*}
where the first inequality follows by $i^*=O(\log \log n)$ and $\forall i\leq[i^*], m/n_{i}\leq \poly(\log n),$ the third inequality follows by $m/n_{i+1}\geq (m/n_i)^{1.25}/(45\log n+100)\geq (m/n_i)^{1.125}.$
Due to Corollary~\ref{cor:tree_contract_conn} and Lemma~\ref{lem:dep_p_i}, we also have $k'_i\leq O(\min(\log (m/n_{i}),\log(\diam(G)))).$ Then, by the same argument, we have $\sum_{i=0}^{r-1}k'_i=O(\log(\diam(G))\log\log_{\diam(G)}(n)+(\log\log(n))^2).$

If $m> n\cdot (600\log n)^8,$ then $\forall i\in\{0\}\cup[r-1],$ we have $p_i<0.5.$ Since $\mathcal{E}$ happens. We have:
\begin{align*}
n_r&\leq \frac{\left(\frac{\left(\frac{n^{1.25}}{m^{0.25}}(45\log n+150)\right)^{1.25}}{m^{0.25}}(45\log n+150)\right)^{\cdots}}{\cdots}&~\text{(Apply $r$ times)}\\
&= \frac{n^{1.25^r}}{m^{1.25^{r}-1}}(45\log n+150)^{4\cdot(1.25^{r}-1)}\\
&=n/(m/n)^{1.25^{r}-1}\cdot (45\log n+150)^{4\cdot(1.25^{r}-1)}\\
&= n/\left(m/\left(n(45\log n+150)^4\right)\right)^{1.25^{r}-1}\\
&\leq n/\left(m/\left(n(45\log n+150)^4\right)\right)^{1.25^{r/2}}\\
&\leq n/\left(m/\left(  n(200\log n)^4\right)\right)^{1.25^{r/2}}\\
&\leq \frac{1}{2},
\end{align*}
where the second inequality follows by $r\geq 5,$ the third inequality follows by $45\log n+150\leq 200\log n,$ and the last inequality follows by
$$ r\geq c\log\log_{m/n}n\geq 2\log_{1.25} \log_{(m/n)^{1/2}} 2n \geq 2\log_{1.25} \log_{m/(n(200\log n)^4)} 2n,$$ for a sufficiently large constant $c>0.$
Since $n_r$ is an integer, $n_r$ must be $0$ when $n_r\leq 1/2.$
$\textsc{SpanningForest}(G,m,r)$ will succeed with probability at least $0.99.$
Due to Lemma~\ref{lem:num_it_multiplelargetree}, we have $k_i\leq O(\min(\log (m/n_{i}),\log(\diam(G)))).$ Thus,
\begin{align*}
\sum_{i=0}^{r-1} k_i&\leq \sum_{m/n_{i}\leq \diam(G)} k_i+\sum_{m/n_{i}> \diam(G)} k_i\\
&\leq O\left(\sum_{i=0}^{\lceil\log_{1.25}\log_2 (\diam(G))\rceil}\log( 2^{1.25^i})\right)+O\left(\sum_{i=0}^{\lceil\log_{1.25}\log_{\diam(G)} (m)\rceil}\log( \diam(G))\right)\\
&\leq O(\log(\diam(G)))+O(\log(\diam(G))\log\log_{\diam(G)} (n)),
\end{align*}
where the second inequality follows by $m/n_{i+1}\geq (m/n_i)^{1.25}/(45\log n+100)\geq (m/n_i)^{1.125}.$
Due to Corollary~\ref{cor:tree_contract_conn} and Lemma~\ref{lem:dep_p_i}, we also have $k'_i\leq O(\min(\log (m/n_{i}),\log(\diam(G)))).$ Then, by the same argument, we have $\sum_{i=0}^{r-1}k'_i=O(\log(\diam(G)))+O(\log(\diam(G))\log\log_{\diam(G)} (n)).$


\end{proof}
\section{Depth-First-Search Sequence for Tree and Applications}
\label{sec:dfs}

    \subsection{Lowest Common Ancestor and Multi-Paths Generation}\label{sec:lca_multipath}
Given a rooted forest induced by $\p:V\rightarrow V$ which is a set of parent pointers (See Definition~\ref{def:parent_pointers}) on $V,$ and a set of $q$ queries $Q=\{(u_1,v_1),(u_2,v_2),\cdots,(u_q,v_q)\mid u_i,v_i\in V\},$ we show an algorithm which can return a mapping $\lca:Q\rightarrow (V\cup\{ \nul \})\times (V\cup\{ \nul \})\times (V\cup\{ \nul \})$ such that $\forall (u_i,v_i)\in Q,(p,p_{u_i},p_{v_i})=\lca(u_i,v_i)$ satisfies the following properties:
\begin{enumerate}
\item If $\p^{(\infty)}(u_i)=\p^{(\infty)}(v_i),$ then $p$ is the lowest ancestor of $u_i$ and $v_i$. Otherwise $p=p_{u_i}=p_{v_i}= \nul .$
\item Suppose $p\not= \nul .$
    If $p\not=u_i,$ then $p_{u_i}$ is an ancestor of $u_i$ and $\p(p_{u_i})=p.$ Otherwise, $p_{u_i}= \nul .$
\item Suppose $p\not= \nul .$
    If $p\not=v_i,$ then $p_{v_i}$ is an ancestor of $v_i$ and $\p(p_{v_i})=p.$ Otherwise, $p_{v_i}= \nul .$
\end{enumerate}
Before we describe the algorithms, let us formally define \textit{ancestor} and the \textit{lowest common ancestor}.
\begin{definition}[Ancestor]\label{def:ancestor}
Let $\p:V\rightarrow V$ be a set of parent pointers (See Definition~\ref{def:parent_pointers}) on a vertex set $V$. For $u,v\in V,$ if $\exists k\in \mathbb{Z}_{\geq 0}$ such that $u=\p^{(k)}(v),$ then $u$ is an ancestor of $v$.
\end{definition}

\begin{definition}[Common ancestor and the lowest common ancestor]\label{def:lca}
$\p:V\rightarrow V$ be a set of parent pointers (See Definition~\ref{def:parent_pointers}) on a vertex set $V$. For $u,v\in V,$ if $w$ is an ancestor of $u$ and is also an ancestor of $v,$ then $w$ is a common ancestor of $(u,v).$ If a common ancestor $w$ of $(u,v)$ satisfies $\dep_{\p}(w)\geq \dep_{\p}(x)$ for any common ancestor $x$ of $(u,v),$ then $w$ is the lowest common ancestor (LCA) of $(u,v).$
\end{definition}

\begin{definition}[Path between two vertices]\label{def:path_between_two_vertices}
$\p:V\rightarrow V$ be a set of parent pointers (See Definition~\ref{def:parent_pointers}) on a vertex set $V$. For $u,v\in V,$ if $\p^{(\infty)}(u)=\p^{(\infty)}(v),$ then the path from $u$ to $v$ is a sequence $(x_1,x_2,\cdots,x_j,x_{j+1},\cdots x_k)$ such that $\forall i\not=i'\in[k],x_i\not=x_{i'},x_1=u,x_k=v,x_j$ is the lowest common ancestor of $(u,v),$ $\forall i\in[j-1],\p(x_i)=x_{i+1},$ and $\forall i\in\{j+1,j+2,\cdots,k\},\p(x_i)=x_{i-1}.$
\end{definition}

The algorithm which can compute the lowest common ancestor is described in Algorithm~\ref{alg:lca}.

\begin{algorithm}[t]
\caption{Lowest Common Ancestor}\label{alg:lca}
\begin{algorithmic}[1]
\small
\Procedure{\textsc{LCA}}{$\p: V\rightarrow V,Q=\{(u_1,v_1),(u_2,v_2),\cdots,(u_q,v_q)\}$} \Comment{Lemma~\ref{lem:lca}}
\State Output: $\lca:Q\rightarrow (V\cup\{ \nul \})\times (V\cup\{ \nul \})\times (V\cup\{ \nul \})$
\State $(r,\dep_{\p},\{g_i\mid i\in\{0\}\cup[r]\})=\textsc{FindAncestors}(\p).$\label{sta:alg_lca_find_ancestor} \Comment{Algorithm~\ref{alg:ancestors}}
\State $\forall (u,v)\in Q,$ if $u=v$ then let $\lca(u,v)=(u, \nul , \nul ),Q\leftarrow Q\setminus\{(u,v)\}.$
\State $\forall (u,v)\in Q,g_r(u)\not=g_r(v),$ let $\lca(u,v)=( \nul , \nul , \nul ).$ \label{sta:not_in_the_same_tree}
\State Let $Q'=\emptyset.$
\State $\forall (u,v)\in Q,g_r(u)=g_r(v)$, if $\dep_{\p}(u)\geq \dep_{\p}(v),$ then let $Q'\leftarrow Q'\cup\{(u,v)\};$ Otherwise let $Q'\leftarrow Q'\cup \{(v,u)\}.$
\State Let $h_r:Q'\rightarrow Q'$ be an identity mapping.
\For{$i=r-1\rightarrow 0$} \Comment{Move $u$ to the almost same depth as $v.$}
\State For each $(u,v)\in Q',$ let $(x,v)=h_{i+1}(u,v).$ If $\dep_{\p}(x)-2^i>\dep_{\p}(y),$ then let $h_i(u,v)=(g_i(x),v);$ Otherwise let $h_i(u,v)=(x,v).$
\EndFor
\State For each $(u,v)\in Q'$ with $\p(h_0(u))=v,$ if $(u,v)\in Q,$ then let $\lca(u,v)=(v,h_0(u), \nul );$ Otherwise $\lca(v,u)=(v, \nul ,h_0(u)).$
\State Let $Q''=\emptyset.$
\State For each $(u,v)\in Q'$ with $\p(h_0(u))\not=v,\dep_{\p}(h_0(u))>\dep_{\p}(v)$ let $Q''\leftarrow Q''\cup\{(u,v)\},$ $h'_r(u,v)\leftarrow(\p(h_0(u)),v).$
\State For each $(u,v)\in Q'$ with $\dep_{\p}(u)=\dep_{\p}(v)$ let $Q''\leftarrow Q''\cup\{(u,v)\},$ $h'_r(u,v)\leftarrow(u,v).$
\For{$i=r-1\rightarrow 0$}\Comment{Move $u,v$ to the lowest common ancestor.}
\State For each $(u,v)\in Q'',$ let $(x,y)=h'_{i+1}(u,v).$ If $g_i(x)\not=g_i(y),$ then let $h'_i(u,v)=(g_i(x),g_i(y));$ Otherwise let $h'_i(u,v)=(x,y).$
\EndFor
\State For each $(u,v)\in Q'',$ if $(u,v)\in Q,$ then let $\lca(u,v)=(\p(h'_0(u)),h'_0(u),h'_0(v));$ Otherwise $\lca(v,u)=(\p(h'_0(v)),h'_0(v),h'_0(u)).$
\State \Return $\lca.$
\EndProcedure
\end{algorithmic}
\end{algorithm}

\begin{lemma}\label{lem:lca}
Let $\p:V\rightarrow V$ be a set of parent pointers (See Definition~\ref{def:parent_pointers}) on a vertex set $V$. Let $Q=\{(u_1,v_1),(u_2,v_2),\cdots,(u_q,v_q)\}$ be a set of $q$ pairs of vertices, and $\forall i\in [q],u_i\not=v_i$. Let $\lca=\textsc{LCA}(\p,Q)$ (Algorithm~\ref{alg:lca}). Then for any $(u,v)\in Q,$ $(p,p_u,p_v)=\lca(u,v)$ satisfies the following properties:
\begin{enumerate}
\item If $\p^{(\infty)}(u)\not=\p^{(\infty)}(v),$ then $p=p_u=p_v= \nul .$\label{itm:lca_pro1}
\item If $u$ (or $v$) is the lowest common ancestor of $(u,v),$ then $p=u,p_u= \nul ,p_v\not=u$ is an ancestor of $v$ such that $\p(p_v)=u$ (or $p=v,p_v= \nul ,p_u\not=v$ is an ancestor of $u$ such that $\p(p_u)=v.$)\label{itm:lca_pro2}
\item If neither $u$ nor $v$ is the lowest common ancestor of $(u,v)$ and $\p^{(\infty)}(u)=\p^{(\infty)}(v),$ then $p$ is the lowest common ancestor of $(u,v),$ $p_u\not=p$ is an ancestor of $u$, $p_v\not=p$ is an ancestor of $v,$ and $\p(p_u)=\p(p_v)=p.$\label{itm:lca_pro3}
\end{enumerate}
\end{lemma}
\begin{proof}
According to Lemma~\ref{lem:depthandancestor}, $r$ should be at most $\lceil\log(\dep(\p)+1)\rceil,$ $\dep_{\p}:V\rightarrow \mathbb{Z}_{\geq 0}$ records the depth of every vertex in $V,$ and $\forall i\in \{0\}\cup[r],v\in V$ $g_i(v)=\p^{(2^i)}(v).$
Then property~\ref{itm:lca_pro1} follows by line~\ref{sta:not_in_the_same_tree} directly.

Then for all $(u,v)\in Q$ with $\p^{(\infty)}(u)=\p^{(\infty)}(v),$ either $(u,v)\in Q'$ or $(v,u)\in Q'.$
For each $(u,v)\in Q',$ we have $\dep_{\p}(u)\geq \dep_{\p}(v).$
For all $(u,v)\in Q',$ with $\dep_{\p}(u)>\dep_{\p}(v),$ by induction we can prove that $\forall i\in\{0\}\cup[r-1],(x,y)=h_i(u,v)$ satisfies that $x$ is an ancestor of $u,$ $y=v,$ $\dep_{\p}(x)>\dep_{\p} (v)$ and $\p^{(2^i)}(x)$ is an ancestor of $v$.
Thus, for $(p,p_u,p_v)=\lca(u,v),$ if $v$ is the lowest common ancestor of $u,$ then we have $p=v,p_u=h_0(u),p_v= \nul .$
In this case, $h_0(u)$ is an ancestor of $u,$ and $\dep_{\p}(u)=\dep_{\p}(v)+1,\p(h_0(u))=v.$
Thus, property~\ref{itm:lca_pro2} holds.

For all $(u,v)\in Q$ with $\p^{(\infty)}(u)=\p^{(\infty)}(v),$ if neither $u$ nor $v$ is the lowest common ancestor of $(u,v),$ then we know either $(u,v)$ or $(v,u)$ is in $Q''.$
Now let $(u,v)\in Q''.$
 We have $\dep_{\p}(h'_r(u))=\dep_{\p}(h'_r(v)),h'_r(u)\not=h'_r(v),$ and $h'_r(u),h'_r(v)$ are ancestors of $u,v$ respectively.
We can prove by induction to get $\forall i\in\{0\}\cup[r],h'_i(u)\not=h'_i(v)$ and $\p^{(2^i)}(h'_i(u))=\p^{(2^i)}(h'_i(v))$ is a common ancestor of $(u,v).$
Thus, $p=\p(h'_0(u))=\p(h'_0(v))$ is the lowest common ancestor of $(u,v),$ and $\dep_{\p}(h'_0(u))=\dep_{\p}(h'_0(v))=\dep_{\p}(p)+1.$
Since $p_u=h'_0(u),p_v=h'_0(v),$ property~\ref{itm:lca_pro3} holds.
\end{proof}

In Algorithm~\ref{alg:multi_path}, we show a generalization of Algorithm~\ref{alg:path_find} such that we can find multiple vertex-to-ancestor paths simultaneously.

\begin{algorithm}[t]
\caption{Multiple Paths}\label{alg:multi_path}
\begin{algorithmic}[1]
\Procedure{MultiPath}{$\p: V\rightarrow V,Q=\{(u_1,v_1),(u_2,v_2),\cdots,(u_q,v_q)\}$} \Comment{ Lemma~\ref{lem:multi_path} }
\State Output: $\dep_{\p}:V\rightarrow \mathbb{Z}_{\geq 0},\{P_i\subseteq V\mid i\in[q]\}.$
\State $(r,\dep_{\p},\{g_i\mid i\in\{0\}\cup[r]\}) = \textsc{FindAncestors}(\p).$\label{sta:alg_multipath_findancestor} \Comment{Algorithm~\ref{alg:ancestors}}
\State $\forall j\in [q],$ let $S^{(0)}_j=\{(u_j,v_j)\mid (u_j,v_j)\in Q\}.$ \label{sta:alg_multipath_init}
\For{$i=1\rightarrow r$}
    \For{$j=1\rightarrow q$} \Comment{$S_j^{(i)}$ is a set of segments partitioned the path from $u_j$ to $v_j.$}
        \State Let $S^{(i)}_j\leftarrow \emptyset.$
        \For{$(x,y)\in S_j^{(i-1)}$}
            \If{$\dep_{\p}(x)-\dep_{\p}(y)> 2^{r-i}$} $S^{(i)}_j\leftarrow S^{(i)}_j\cup\{(x,g_{r-i}(x)),(g_{r-i}(x),y)\}.$ \label{sta:alg_multipath_split}
            \Else~$S_j^{(i)}\leftarrow S_j^{(i)}\cup\{(x,y)\}.$ \label{sta:alg_multipath_not_split}
            \EndIf
        \EndFor
    \EndFor
\EndFor\Comment{$S_j^{(r)}$ only contains segments with length $1$}
\State Let $\forall j\in[q],P_j\leftarrow\{u_j\}.$
\For{$j=1\rightarrow q$}
\For{$(x,y)\in S^{(r)}_j$}
\State Let $P_j\leftarrow P_j\cup\{y\}.$
\EndFor
\EndFor
\EndProcedure
\end{algorithmic}
\end{algorithm}

The following lemma claims the properties of the outputs of Algorithm~\ref{alg:multi_path}. And the proof is similar to the proof of Lemma~\ref{lem:find_path}.

\begin{lemma}\label{lem:multi_path}
Let $\p:V\rightarrow V$ be a set of parent pointers (See Definition~\ref{def:parent_pointers}) on a vertex set $V.$ Let $Q=\{(u_1,v_1),(u_2,v_2),\cdots,(u_q,v_q)\}\subseteq V\times V$ satisfy $\forall j\in[q],$ $v_j$ is an ancestor (See Definition~\ref{def:ancestor}) of $u_j$ in $\p$. Let $(\dep_{\p},\{P_j\mid j\in [q]\})=\textsc{MultiPath}(\p,Q)$ (Algorithm~\ref{alg:multi_path}). Then $\dep_{\p}:V\rightarrow \mathbb{Z}_{\geq 0}$ records the depth of every vertex in $V$ and $\forall j\in[q],P_j\subseteq V$ is the set of all vertices on the path from $u_j$ to $v_j$, i.e. $P_j=\{v\in V\mid \exists k_1,k_2\in\mathbb{Z}_{\geq 0},v=\p^{(k_1)}(u_j),v_j=\p^{(k_2)}(v)\}.$
Furthermore, $r$ should be at most $\lceil\log(\dep(\p)+1)\rceil.$
\end{lemma}

\begin{proof}
By Lemma~\ref{lem:depthandancestor}, since $(r,\dep_{\p},\{g_i\mid i\in\{0\}\cup[r]\})=\textsc{FindAncestors}(\p),$ we know $r$ should be at most $\lceil\log(\dep(\p)+1)\rceil,$ $\dep_{\p}:V\rightarrow \mathbb{Z}_{\geq 0}$ records the depth of every vertex in $V,$ and $\forall i\in \{0\}\cup[r],v\in V$ $g_i(v)=\p^{(2^i)}(v).$

For $j\in[q],$ let us prove that $P_j$ is the vertex set of all the vertices on the path from $u_j$ to its ancestor $v_j$. We use divide-and-conquer to get $P_j$. The following claim shows that $S_j^{(i)}$ is a set of segments which is a partition of the path from $u_j$ to $v_j$, and each segment has length at most $2^{r-i}.$
\begin{claim}\label{cla:property_multi_path}
$\forall i\in\{0\}\cup[r],j\in [q]$ $S^{(i)}_j$ satisfies the following properties:
\begin{enumerate}
\item $\exists (x,y)\in S^{(i)}_j$ such that $x=u_j.$\label{itm:multipath_pro1}
\item $\exists (x,y)\in S^{(i)}_j$ such that $y=v_j.$\label{itm:multipath_pro2}
\item $\forall (x,y)\in S^{(i)}_j,$ $\dep_{\p}(y)-\dep_{\p}(x)\leq 2^{r-i}.$\label{itm:multipath_pro3}
\item $\forall (x,y)\in S^{(i)}_j,$ if $y\not=v_j,$ then $\exists (x',y')\in S^{(i)}_j,x'=y.$\label{itm:multipath_pro4}
\item $\forall (x,y)\in S^{(i)}_j,$ $\exists k\in \mathbb{Z}_{\geq 0}, \p^{(k)}(x)=y.$\label{itm:multipath_pro5}
\end{enumerate}
\end{claim}
\begin{proof}
We fix a $j\in[q].$
Our proof is by induction. According to line~\ref{sta:alg_multipath_init}, all the properties hold when $i=0.$ Suppose all the properties hold for $i-1.$ For property~\ref{itm:multipath_pro1}, by induction we know there exists $(x,y)\in S^{(i-1)}_j$ such that $x=u_j.$ Then by line~\ref{sta:alg_multipath_split} and line~\ref{sta:alg_multipath_not_split}, there must be an $(x,y')$ in $S^{(i)}_j.$ For property~\ref{itm:multipath_pro2}, by induction we know there exists $(x,y)\in S^{(i-1)}_j$ such that $y=v_j.$ Thus, there must be an $(x',y)$ in $S^{(i)}_j.$ For property~\ref{itm:multipath_pro3}, if $(x,y)$ is added into $S^{(i)}_j$ by line~\ref{sta:alg_multipath_not_split}, then $\dep_{\p}(x)-\dep_{\p}(y)\leq 2^{r-i}.$ Otherwise, in line~\ref{sta:alg_multipath_split}, we have $\dep_{\p}(x)-\dep_{\p}(g_{r-i}(x))\leq 2^{r-i},\dep_{\p}(g_{r-i}(x))-\dep_{\p}(y)\leq 2^{r-i+1}-2^{r-i}=2^{r-i}.$
For property~\ref{itm:multipath_pro4}, if $(x,y)$ is added into $S^{(i)}_j$ by line~\ref{sta:alg_multipath_not_split}, then by induction there is $(y,y')\in S^{(i-1)}_j,$ and thus by line~\ref{sta:alg_multipath_not_split} and line~\ref{sta:alg_multipath_split}, there must be $(y,y'')\in S^{(i)}_j.$ Otherwise, in line~\ref{sta:alg_multipath_split} will generate two pairs $(x,g_{r-i}(x)),(g_{r-i}(x),y).$ For $(x,g_{r-i}(x)),$ the property holds. For $(g_{r-i}(x),y),$ there must be $(y,y')\in S_{i-1}$ and thus there should be $(y,y'')\in S^{(i)}.$ For property~\ref{itm:multipath_pro5}, since $g_{r-i}(x)=\p^{(r-i)}(x),$ for all pairs generated by line~\ref{sta:alg_multipath_split} and line~\ref{sta:alg_multipath_not_split}, the property holds.
\end{proof}
By Claim~\ref{cla:property_multi_path}, we know
\begin{align*}
S^{(r)}_j=\{ & ~\\
 & ~ (u_j,\p(u_j)),\\
 & ~ \left( \p(u_j),\p^{(2)}(u_j) \right), \\
 & ~ \left( \p^{(2)}(u_j),\p^{(3)}(u_j) \right), \\
 & ~ \cdots, \\
 & ~ \left( \p^{(\dep_{\p}(u_j)-\dep_{\p}(v_j)-1)}(u_j),\p^{(\dep_{\p}(u_j)-\dep_{\p}(v_j))}(u_j) \right) \\
 \} & ~ .
\end{align*}
Thus, $P_j$ is the set of all the vertices on the path from $u_j$ to an ancestor $v_j$.
\end{proof}

    \subsection{Depth-First-Search Sequence for a Tree}\label{sec:DFS_tree}
Since we can use our spanning tree algorithm to get a rooted tree, in this section, we only consider how to get a Depth-First-Search (DFS) sequence for a rooted tree. Before we go to the details, let us firstly give formal definitions of some useful concepts.

\begin{definition}[Children in the forest]\label{def:child}
Given a set of parent pointers (See Definition~\ref{def:parent_pointers}) $\p:V\rightarrow V$ on a vertex set $V$. $\forall u,v\in V,u\not=v$ if $\p(u)=v,$ then we say $u$ is a child of $v.$ $\forall v\in V,$ we can define $\child_{\p}(v)$ as the set of all children of $v,$ i.e. $\child_{\p}(v)=\{u\in V\mid u\not=v,\p(u)=v\}.$ Furthermore, if $u$ is the $k^{\text{th}}$ smallest vertex in the children set $\child_{\p}(v),$ then we say $\rank_{\p}(u)=k,$ or $u$ is the $k^{\text{th}}$ child of $v$. If $\p(v)=v,$ then $\rank_{\p}(v)=1.$ We use $\child_{\p}(v,k)$ to denote the $k^{\text{th}}$ child of $v$.
\end{definition}

For simplicity of the notation, if $\p:V\rightarrow V$ is clear in the context, we just use $\child(v),\rank(v)$ and $\child(v,k)$ to denote $\child_{\p}(v),\rank_{\p}(v)$ and $\child_{\p}(v,k)$ respectively.

\begin{definition}[Leaves in the forest]\label{def:leaf}
Given a set of parent pointers (See Definition~\ref{def:parent_pointers}) $\p:V\rightarrow V$ on a vertex set $V$. If $\child_{\p}(v)=\emptyset,$ then $v$ is called a leaf. The set of all the leaves of $\p$ is defined as $\leaves(\p)=\{v\mid \child_{\p}(v)=\emptyset\}.$
\end{definition}

\begin{definition}[Subtree]\label{def:subtree}
Let $\p:V\rightarrow V$ be a set of parent pointers (See Definition~\ref{def:parent_pointers}) on a vertex set $V.$ Let $v\in V,V'=\{u\in V\mid v\text{~is~an~ancestor~(Definition~\ref{def:ancestor})~of~}u\}.$ Let $\p':V'\rightarrow V'$ be a set of parent pointers on $V'$. If $\forall u\in V'\setminus\{v\},\p'(u)=\p(u),$ and $\p'(v)=v,$ then we say $\p'$ is the subtree of $v$ in $\p$. For $u\in V',$ we say $u$ is in the subtree of $v$.
\end{definition}

\begin{definition}[Depth-First-Search (DFS) sequence]\label{def:DFS_sequence}
Let $\p:V\rightarrow V$ be a set of parent pointers (See Definition~\ref{def:parent_pointers}) on a vertex set $V.$ Let $v$ be a vertex in $V$. If $v$ is a leaf (See Definition~\ref{def:leaf}) in $\p,$ then the DFS sequence of the subtree (See Definition~\ref{def:subtree}) of $v$ is $(v).$ Otherwise the DFS sequence of the subtree of $v$ in $\p$ is recursively defined as
\begin{align*}
(v,a_{1,1},a_{1,2},\cdots,a_{1,n_1},v,a_{2,1},a_{2,2},\cdots,a_{2,n_2},v,\cdots,a_{k,1},a_{k,2},\cdots,a_{k,n_k},v),
\end{align*}
where $k=|\child(v)|$ is the number of children (See Definition~\ref{def:child}) of $v,$ and $\forall i\in[k], (a_{i,1},\cdots,a_{i,n_i})$ is the DFS sequence of the subtree of $\child(v,i),$ i.e. the $i^{\text{th}}$ child of $v$.

If $\forall u\in V,\p^{(\infty)}(u)=v,$ then the subtree of $v$ is exactly $\p,$ and thus the DFS sequence of the subtree of $v$ is also called the DFS sequence of $\p.$
\end{definition}

Here are some useful facts of the above defined DFS sequence.
\begin{fact}\label{fac:dfs_sequence}
Let $\p:V\rightarrow V$ be a set of parent pointers (See Definition~\ref{def:parent_pointers}) on a vertex set $V$, and $\p$ has a unique root. Let $A=(a_1,a_2,\cdots,a_m)$ be the DFS sequence (See Definition~\ref{def:DFS_sequence}) of $\p.$ Then, $A$ satisfies the following properties:
\begin{enumerate}
\item $\forall v\in V,$ $v$ appears exactly $|\child(v)|+1$ times in $A$.\label{itm:dfs_fac1}
\item If $a_i$ is the $k^{\text{th}}$ time that $v$ appears, and $a_j$ is the $(k+1)^{\text{th}}$ time that $v$ appears. Then $(a_{i+1},a_{i+2},\cdots,a_{j-1})$ is the DFS sequence of the subtree of $\child(v,k)$ (See Definition~\ref{def:child}), the $k^{\text{th}}$ child of $v$. Furthermore, $a_{i+1}$ is the first time that $\child(v,k)$ appears, and $a_{j-1}$ is the last time of $\child(v,k)$ appears.\label{itm:dfs_fac2}
\item If $a_i$ is the first time that $v$ appears, and $a_j$ is the last time that $v$ appears. Then $(a_i,a_{i+1},\cdots,a_j)$ is the DFS sequence of the subtree of $v$.\label{itm:dfs_fac3}
\item $m=2|V|-1.$\label{itm:dfs_fac4}
\end{enumerate}
\end{fact}
\begin{proof}
The property~\ref{itm:dfs_fac1},~\ref{itm:dfs_fac2},~\ref{itm:dfs_fac3} directly follows by Definition~\ref{def:DFS_sequence}.

For property~\ref{itm:dfs_fac4}, notice that $\forall u\in V,\p(u)\not=u,$ $u$ can only be a child of $\p(u).$ Thus, $\sum_{v\in V}(|\child(v)|+1)=|V|-1+|V|=2|V|-1.$
\end{proof}

Due to the above fact, if $v$ is a leaf in $\p,$ then it will only once in the DFS sequence. Thus, we are able to determine the order of all the leaves.

\begin{definition}[The order of the leaves]\label{def:order_leaves}
Let $\p:V\rightarrow V$ be a set of parent pointers (See Definition~\ref{def:parent_pointers}) on a vertex set $V$, and $\p$ has a unique root. Let $A=(a_1,a_2,\cdots,a_{m})$ be the DFS sequence (See Definition~\ref{def:DFS_sequence}) of $\p.$ Let $u,v$ be two leaves (See definition~\ref{def:leaf}) of $\p.$ If $u$ appears before $v$ in $A$, then we say $u<_{\p} v.$
\end{definition}

\subsubsection{Leaf Sampling}
Given a set of rooted trees, our goal is to sample a set of leaves for each tree, and to give an order of those sampled leaves. The algorithm is shown in Algorithm~\ref{alg:sample_leaves}.

\begin{algorithm}[h]
\caption{Leaf Sampling}\label{alg:sample_leaves}
\begin{algorithmic}[1]
\Procedure{\textsc{LeafSampling}}{$\p:V\rightarrow V,m,\delta$} \Comment{Lemma~\ref{lem:properties_of_sampled_leaves}}
\State Output: $A=(a_1,a_2,\cdots,a_s).$
\State Let $t=\lceil m^{1/3}\rceil.$
\State Compute $L=\leaves(\p).$ \label{sta:alg_sample_leaf_compute_leaf}
\State Compute $\rank:V\rightarrow \mathbb{Z}_{\geq 0}$ such that $\forall v\in V,\rank(v)=\rank_{\p}(v).$\label{sta:alg_sample_leaf_compute_rank} \Comment{Definition~\ref{def:child}}
\State If $|V|\leq m,$ let $\{a_{1},a_{2},\cdots,a_{s}\}=L,$ and return $A=(a_1,a_2,\cdots,a_s)$ which satisfies $a_{1}<_{\p}a_{2}<_{\p}\cdots<_{\p}a_{s}.$\label{sta:when_V_is_small}\Comment{$<_{\p}$ follows Definition~\ref{def:order_leaves}}
\State If $|L|\leq 8t,$ let $S=L.$\label{sta:when_L_is_small}
\State Let $p=\min(1,640(1+\log(m)/\delta)t/|L|).$
\State If $|L|>t,$ sample each $v\in L$ with probability $p$ independently. let $S$ be the set of samples.\label{sta:when_L_needs_sample}
\State Compute $\p':V\rightarrow V$ such that $\forall v\in V,$ if $\child_{\p}(v)\not=\emptyset,$ then $\p'(v)=\child_{\p}(v,1);$ Otherwise let $\p'(v)=v.$\label{sta:point_to_first_leaf}\Comment{ $\p'(v)$ points to $v$'s first child in $\p$.}
\State $(r',\dep_{\p'}:V\rightarrow\mathbb{Z}_{\geq 0},\{g'_i:V\rightarrow V\mid i\in\{0\}\cup[r']\}) = \textsc{FindAncestors}(\p').$\label{sta:mid_find_the_first_leaf}
\State Find $w\in V$ with $\p(w)=w.$ \Comment{Find the root.}
\State Let $a_1=g'_{r'}(w),S\leftarrow S\cup \{a_1\}.$\label{sta:find_the_first_leaf} \Comment{Find the first leaf.}
\State Let $Q=\{(u,v)\mid (u,v)\in S\times S,u\not= v\}.$
\State Let $\lca= \textsc{LCA}(\p,Q).$\label{sta:alg_sample_leaf_lca} \Comment{Algorithm~\ref{alg:lca}}
\State Let $s=|S|.$
\State $(r,\dep_{\p}:V\rightarrow\mathbb{Z}_{\geq 0},\{g_i:V\rightarrow V\mid i\in\{0\}\cup[r]\}) = \textsc{FindAncestors}(\p).$\label{sta:alg_smaple_leaf_find_ancestor}
\For{$i=2\rightarrow s$}\label{sta:alg_sample_leaf_sequential_st} \Comment{Determine the order of sampled leaves.}
\State For all $x,y\in S\setminus\{a_1,a_2,\cdots,a_{i-1}\},$ let $(p_{x,y},p_{xy,x},p_{xy,y})=\lca(x,y).$\label{sta:using_lca}
\State Find $x^*\in S\setminus\{a_1,a_2,\cdots,a_{i-1}\}$ s.t. $\forall y\in S\setminus\{a_1,a_2,\cdots,a_{i-1},x^*\},$
$\rank(p_{x^*y,x^*})<\rank(p_{x^*y,y}).$
\State Let $a_i=x^*.$
\EndFor\label{sta:alg_sample_leaf_sequential_ed}
\State \Return $A=(a_1,a_2,\cdots,a_s).$
\EndProcedure
\end{algorithmic}
\end{algorithm}

\begin{lemma}\label{lem:properties_of_sampled_leaves}
Let $\p:V\rightarrow V$ be a set of parent pointers (See Definition~\ref{def:parent_pointers}) on a vertex set $V$, and $\p$ has a unique root. Let $m>0,\delta\in(0,1)$ be parameters, and let $|V|\leq m^{1/\delta}.$ Let $(a_1,a_2,\cdots,a_s) = \textsc{LeafSampling}(\p,m,\delta)$ (Algorithm~\ref{alg:sample_leaves}). Then it has following properties:
 \begin{enumerate}
 \item $a_1<_{\p}a_2<_{\p}\cdots<_{\p} a_s.$ \label{itm:sampled_leaves_pro1}
 \item If $|V|\leq m$ or $|\leaves(\p)|\leq 8\lceil m^{1/3}\rceil,$ then $\{a_1,a_2,\cdots,a_s\}=\leaves(\p).$ Otherwise, with probability at least $1-1/(100m^{5/\delta}),$ $\forall v\in \leaves(\p)\setminus\{a_1\},$ there is a vertex $w\in\{a_1,a_2,\cdots,a_s\}$ such that $w<_{\p} v$ and the number of leaves between $w$ and $v$ is at most $|\leaves(\p)|/\lceil m^{1/3}\rceil$, i.e. $|\{u\in\leaves(\p)\mid w<_{\p}u<_{\p}v\}|\leq |\leaves(\p)|/\lceil m^{1/3}\rceil.$ \label{itm:sampled_leaves_pro2}
 \item If $|V|> m$ and $|\leaves(\p)|> 8\lceil m^{1/3}\rceil,$ then with probability at least $1-1/(100m^{5/\delta}),$ $s=|S|=|\{a_1,a_2,\cdots,a_s\}|\leq 960\lceil m^{1/3}\rceil(1+\log(m)/\delta).$\label{itm:sampled_leaves_pro3}
 \end{enumerate}
\end{lemma}
\begin{proof}
Firstly, let us focus on property~\ref{itm:sampled_leaves_pro1}.
According to line~\ref{sta:mid_find_the_first_leaf} to line~\ref{sta:find_the_first_leaf} and Lemma~\ref{lem:depthandancestor}, we know $\forall k\in\mathbb{Z}_{\geq 0}\rank_{\p}(\p^{(k)}(a_1))=1,$ and $\p'(a_1)=a_1$ which implies that $a_1$ is a leaf.
Due to the definition of~\ref{def:DFS_sequence}, we know that $a_1$ must be the first leaf appeared in the DFS sequence of $\p.$
We can prove the property by induction. Suppose we already have $a_1<_{\p}a_2<_{\p}\cdots<_{\p}a_{i-1}.$
According to line~\ref{sta:using_lca} and Lemma~\ref{lem:lca},
$p_{a_{i-1},a_{i}}$ is the LCA of $(a_{i-1},a_i).$
$p_{a_{i-1}a_i,a_{i-1}}$ is a child of $p_{a_{i-1},a_{i}}$ and is an ancestor of $a_{i-1}.$
$p_{a_{i-1}a_i,a_i}$ is a child of $p_{a_{i-1},a_{i}}$ and is an ancestor of $a_{i}.$
By Fact~\ref{fac:dfs_sequence}, since $\rank(p_{a_{i-1}a_i,a_{i-1}})<\rank(p_{a_{i-1}a_i,a_{i}}),$ we have $a_{i-1}<_{\p} a_i.$
To conclude, we have $a_1<_{\p}a_2<_{\p}\cdots<_{\p} a_s.$

For property~\ref{itm:sampled_leaves_pro2}, if $|V|\leq m$ or $|\leaves(\p)|\leq 8\lceil m^{1/3}\rceil,$ then by line~\ref{sta:when_V_is_small} and line~\ref{sta:when_L_is_small}, we know $\{a_1,a_2,\cdots,a_s\}=\leaves(\p).$

Now consider the case when $|V|>m$ and $|\leaves(\p)|>8\lceil m^{1/3}\rceil.$
Let $t=\lceil m^{1/3}\rceil.$
Let $\leaves(\p)=\{u_1,u_2,\cdots,u_q\},$ and let $u_1<_{\p}u_2<_{\p}\cdots<_{\p} u_q.$
Let us partition $u_1,\cdots,u_q$ into $4\cdot t$ groups $G_1=\{u_1,u_2,\cdots,u_{\lfloor q/(4t)\rfloor}\},G_2=\{u_{\lfloor q/(4t)\rfloor+1},u_{\lfloor q/(4t)\rfloor+2},\cdots,u_{2\cdot \lfloor q/(4t)\rfloor}\},\cdots,G_{4t}=\{u_{(4t-1)\lfloor q/(4t)\rfloor+1},u_{(4t-1)\lfloor q/(4t)\rfloor+2},\cdots,u_q\}.$
Then each group has size at least $q/(8t)$ and at most $q/(2t).$
For a certain $G_i,$ by Chernoff bound, we have
\begin{align*}
&\Pr\left(|G_i\cap S|\leq \frac{1}{2}\cdot \frac{q}{8t}\cdot p\right)\\
\leq~&\exp\left(-\frac{1}{8}\cdot \frac{q}{8t}\cdot p\right) \\
\leq~&1/(100m^{10/\delta})
\end{align*}
where the last inequality follows by $p=\min(1,(10+10\log(m)/\delta)\cdot64t/q).$
Notice that $q\leq |V|\leq m^{1/\delta}.$
We can take union bound over all $G_i$.
Then with probability at least $1-1/(100m^{5/\delta}),$ $\forall i\in[4t],G_i\cap S\not=\emptyset.$
Thus, $\forall v\in \leaves(\p),$ there is a vertex $w\in\{a_1,a_2,\cdots,a_s\}$ such that $w<_{\p} v$ and the number of leaves between $w$ and $v$ is at most $|\leaves(\p)|/\lceil m^{1/3}\rceil$, i.e. $|\{u\in\leaves(\p)\mid w<_{\p}u<_{\p}v\}|\leq |\leaves(\p)|/\lceil m^{1/3}\rceil.$

For property~\ref{itm:sampled_leaves_pro3}, by applying Chernoff bound, we have
\begin{align*}
&\Pr\left(|S\cap \leaves(\p)|\geq \frac{3}{2} |\leaves(\p)|\cdot p\right)\\
\leq~&\exp\left(-\frac{1}{12}\cdot |\leaves(\p)|\cdot p\right)\\
\leq~&1/(100m^{10/\delta})
\end{align*}
where the last inequality follows by $p=\min(1,(10+10\log(m)/\delta)\cdot64t/|\leaves(\p)|).$

Since $\frac{3}{2} |\leaves(\p)|\cdot p\leq 960\lceil m^{1/3}\rceil(1+\log(m)/\delta),$ we complete the proof.
\end{proof}

\subsubsection{DFS Subsequence}
Let $\p:V\rightarrow V$ be a set of parent pointers on a vertex set $V$, and $\p$ has a unique root $v$. Let $\{u_1,u_2,\cdots,u_q\}=\leaves(\p),$ and $u_1<_{\p}u_2<_{\p}\cdots<_{\p}u_q.$ One observation is that the DFS sequence of $\p$ can be generated in the following way:
\begin{enumerate}
\item The first part of the DFS sequence is the path from $v$ to $u_1.$
\item Then it follows by the path from $\p(u_1)$ to the LCA of $(u_1,u_2),$ the path from one of the child of the LCA of $(u_1,u_2)$ to $u_2,$ the path from $\p(u_2)$ to the LCA of $(u_2,u_3),$ the path from one of the child of the LCA of $(u_2,u_3)$ to $u_3,\cdots,$ the path from one of the child of the LCA of $(u_{q-1},u_q)$ to $u_q$.
\item The last part of the DFS sequence is a path from $\p(u_q)$ to $v$.
\end{enumerate}

\begin{algorithm}[h!]
\caption{DFS subsequence}\label{alg:subdfs}
\begin{algorithmic}[1]
\Procedure{SubDFS}{$\p:V\rightarrow V,m,\delta$}\Comment{Lemma~\ref{lem:remove_trees}, Lemma~\ref{lem:A_is_subsequence}}
\State Output: $V'\subseteq V,A=(a_1,a_2,\cdots,a_s).$
\State If $V=\{v\},$ return $V'=V,A=(v).$
\State Let $v$ be the root in $\p,$ i.e. $\p(v)=v.$
\State $L=(l_1,l_2,\cdots,l_t) = \textsc{LeafSampling}(\p,m,\delta).$\label{sta:alg_subsequence_sample_leaf}\Comment{Algorithm~\ref{alg:sample_leaves}}
\State $Q=\{(l_i,l_{i+1})\mid i\in[t-1]\}.$
\State $\lca = \textsc{LCA}(\p,Q).$\label{sta:alg_subsequence_lca}\Comment{Algorithm~\ref{alg:lca}}
\State $\forall i\in[t-1],(p_{l_i,l_{i+1}},p_{i,l_i},p_{i,l_{i+1}})=\lca(l_i,l_{i+1}).$
\State $Q'=\{(l_1,v),(\p(l_1),p_{l_1,l_2}),(l_2,p_{1,l_2}),(\p(l_2),p_{l_2,l_3}),(l_3,p_{2,l_3}),\cdots,(l_t,p_{t-1,l_t}),(\p(l_t),v)\}.$\label{sta:alg_subsequence_Qprime}
\State $(\dep_{\p},\{P_i\mid i\in[2t]\}) = \textsc{MultiPath}(\p,Q').$\label{sta:alg_subsequence_multi_path}\Comment{Algorithm~\ref{alg:multi_path}}
\State $V'=\bigcup_{i=1}^{2t} P_i.$
\State Let $\p':V'\rightarrow V'$ satisfy $\forall v\in V',\p'(v)=\p(v).$
\For{$i\in\{1,3,5,\cdots,2t-1\}$}\label{sta:alg_subsequence_loop1}
\State Compute $A'_i=(u_1,u_2,\cdots,u_{|P_i|})$ such that $\{u_1,u_2,\cdots,u_{|P_i|}\}=P_i$ and $\dep_{\p}(u_1)<\dep_{\p}(u_2)<\cdots<\dep_{\p}(u_{|P_i|})$
\EndFor
\For{$i\in\{2,4,6,\cdots,2t\}$}\label{sta:alg_subsequence_loop2}
\State Compute $A'_i=(u_1,u_2,\cdots,u_{|P_i|})$ such that $\{u_1,u_2,\cdots,u_{|P_i|}\}=P_i$ and $\dep_{\p}(u_1)>\dep_{\p}(u_2)>\cdots>\dep_{\p}(u_{|P_i|})$
\EndFor
\State Let $A'=A'_1A'_2\cdots A'_{2t}.$\label{sta:concatentation_get_Aprime} \Comment{$A'$ is the concatenation of $A'_1,A'_2,\cdots,A'_{2t}.$}
\State $\forall u\in V',$ compute $\rank_{\p}(u)$ and $\rank_{\p'}(u).$\label{sta:alg_subsquence_compute_rank}
\State $\forall u\in V',i\in[|\child_{\p'}|+1]$ compute $\pos(u,i)=j$ such that the $j^{\text{th}}$ element in $A'$ is the $i^{\text{th}}$ time that $u$ appears. \label{sta:alg_subsequence_pos}
\State Let $b$ be the length of $A'.$
\State Initialize $c:[b]\rightarrow \mathbb{Z}_{\geq 0}.$\Comment{$c$ determine the number of copies needed for each element in $A'$}
\For{$u\in V'\setminus\{v\}$}\label{sta:alg_subsequence_final_loop}
\State If $u\in\leaves(\p'),$ let $c(\pos(u,1))=1.$\label{sta:just_copy_once}\Comment{A leaf should only have one copy.}
\State If $\rank_{\p'}(u)=1,$ let $c(\pos(\p'(u),1))=\rank_{\p}(u).$\label{sta:the_first_duplicate}
\State If $\rank_{\p'}(u)=|\child_{\p'}(\p'(u))|,$ let $c(\pos(\p'(u),\rank_{\p'}(u)+1))=|\child_{\p}(\p(u))|+1-\rank_{\p}(u).$\label{sta:final_duplicate}
\State If $1\leq \rank_{\p'}(u)<|\child_{\p'}(\p'(u))|,$ let $c(\pos(\p'(u),\rank_{\p'}(u)+1))=\rank_{\p}(\child_{\p'}(\p'(u),\rank_{\p'}(u)+1))-\rank_{\p}(u).$\label{sta:mid_duplicate}
\EndFor
\State For each $j\in[b],$ duplicate the $j^{\text{th}}$ element of $A'$ $c(j)$ times. Let $A$ be the obtained sequence.\label{sta:duplicate_elements}
\State \Return $V',A$.
\EndProcedure
\end{algorithmic}
\end{algorithm}

\begin{fact}\label{fac:dfs_generate_leaf_path}
Let $\p:V\rightarrow V$ be a set of parent pointers (See Definition~\ref{def:parent_pointers}) on a vertex set $V$, and $\p$ has a unique root $v$. Let $\{u_1,u_2,\cdots,u_q\}=\leaves(\p)$ (See Definition~\ref{def:leaf}), and $u_1<_{\p}u_2<_{\p}\cdots<_{\p}u_q.$ Let $A=(a_1,a_2,\cdots,a_m)$ be the DFS sequence (See Definition~\ref{def:DFS_sequence}) of $\p$. Then,
\begin{enumerate}
\item If $u_1$ appears at $a_i,$ then $(a_1,a_2,\cdots,a_i)$ is the path from $v$ to $u_1$.\label{itm:dfs_path_pro1}
\item $\forall i\in[q-1],$ if $u_i$ appears at $a_j,$ and $u_{i+1}$ appears at $a_k,$ then $\exists j<t<k$ such that $a_t$ is the LCA of $(u_i,u_{i+1}).$ In addition, $(a_j,a_{j+1},\cdots,a_t)$ is the path from $a_j$ to $a_t,$ and $(a_t,a_{t+1},\cdots,a_k)$ is the path from $a_t$ to $a_k.$\label{itm:dfs_path_pro2}
\item If $u_q$ appears at $a_i,$ then $(a_i,a_{i+1},\cdots,a_m)$ is the path from $u_q$ to $v$.\label{itm:dfs_path_pro3}
\end{enumerate}
\end{fact}
\begin{proof}
Property~\ref{itm:dfs_path_pro1},~\ref{itm:dfs_path_pro3} follows by the definition of DFS sequence (See Definition~\ref{def:DFS_sequence}) and a simple induction.

Now consider property~\ref{itm:dfs_path_pro2}. Since $A$ is a DFS sequence, $\forall l\in\{j,j+1,\cdots,k-1\},$ either $\p(a_l)=a_{l+1}$ or $\p(a_{l+1})=a_l.$ Thus, the path between $u_{i}$ and $u_{i+1}$ is a subsequence of $(a_j,a_{j+1},\cdots,a_k).$ If $\p(a_{l+1})=a_l$ but $a_{l+1}$ is not on the path between $u_i$ and $u_{i+1},$ then there must be a leaf $x$ in the subtree of $a_{l+1}$ which implies $u_i<_{\p}x<_{\p}u_{i+1},$ and thus leads to a contradiction. If $\p(a_l)=a_{l+1}$ but $a_{l+1}$ is not on the path between $u_i$ and $u_{i+1},$ then both $u_i$ and $u_{i+1}$ should be in the subtree of $a_l,$ and both of $u_i$ and $u_{i+1}$ should be in the DFS sequence of the subtree of $a_l.$ But we know $a_{l+1}$ cannot be in the DFS sequence of the subtree of $a_l.$ This leads to a contradiction.
\end{proof}

\begin{lemma}\label{lem:remove_a_tree_subdfs}
Let $\p:V\rightarrow V$ be a set of parent pointers (See Definition~\ref{def:parent_pointers}) on a vertex set $V$, with a unique root. Let $v\in V$. Let $V'=V\setminus\{u\in V\mid v\text{~is~an~ancestor~(See~Definition~\ref{def:ancestor})~of~}u\}.$ Let $\p':V'\rightarrow V'$ satisfy $\forall u\in V',\p'(u)=\p(u).$ Then the DFS sequence (See Definition~\ref{def:DFS_sequence}) of $\p'$ is a subsequence of the DFS sequence of $\p.$
\end{lemma}
\begin{proof}
The proof follows by the property~\ref{itm:dfs_fac3} of Fact~\ref{fac:dfs_sequence} directly.
\end{proof}

\begin{corollary}\label{cor:subdfs_remove_multi_subtree}
Let $\p:V\rightarrow V$ be a set of parent pointers (See Definition~\ref{def:parent_pointers}) on a vertex set $V$, and $\p$ has a unique root. Let $v_1,v_2,\cdots,v_t$ be $t$ vertices in $V$. Let $V'=V\setminus\{u\in V\mid \exists v\in \{v_1,\cdots,v_t\},v\text{~is~an~ancestor~(See~Definition~\ref{def:ancestor})~of~}u\}.$ Let $\p':V'\rightarrow V'$ satisfy $\forall u\in V',\p'(u)=\p(u).$ Then the DFS sequence (See Definition~\ref{def:DFS_sequence}) of $\p'$ is a subsequence of the DFS sequence of $\p.$
\end{corollary}
\begin{proof}
The proof is by induction on $t$. When $t=1,$ then the statement is true by Lemma~\ref{lem:remove_a_tree_subdfs}.
Suppose the statement is true for $t-1$.
Let $V''=V\setminus\{u\in V\mid \exists v\in \{v_1,\cdots,v_{t-1}\},v\text{~is~an~ancestor~of~}u\},$ and let $\p'':V''\rightarrow V''$ satisfy $\forall v\in V'',\p''(v)=\p(v).$
By induction hypothesis, the DFS sequence of $\p''$ is a subsequence of the DFS sequence of $\p$.
If one of the $v_1,\cdots,v_{t-1}$ is an ancestor of $v_t,$ then $\p'=\p'',$ thus, the DFS sequence of $\p'$ is a subsequence of the DFS sequence of $\p$.
Otherwise, we have $V'=V''\setminus \{u\in V''\mid v_t\text{~is~an~ancestor~of~}u\}.$
By Lemma~\ref{lem:remove_a_tree_subdfs}, the DFS sequence of $\p'$ is a subsequence of the DFS sequence of $\p''.$
Thus, the DFS sequence of $\p'$ is a subsequence of the DFS sequence of $\p$.
\end{proof}

\begin{lemma}[Removing several subtrees]\label{lem:remove_trees}
Let $\p:V\rightarrow V$ be a set of parent pointers (See Definition~\ref{def:parent_pointers}) on a vertex set $V$, and $\p$ has a unique root. Let $m>0,\delta\in(0,1)$ be parameters, and let $|V|\leq m^{1/\delta}.$ Let $(V',A) = \textsc{SubDFS}(\p,m,\delta)$ (Algorithm~\ref{alg:subdfs}). Then $\forall u\in V',$ we have $\p(u)\in V'.$ Furthermore, with probability at least $1-1/(100m^{5/\delta}),$ $\forall u\in V\setminus V',$ the number of leaves (See Definition~\ref{def:leaf}) in the subtree (See Definition~\ref{def:subtree}) of $u$ is at most $\lfloor |\leaves(\p)|/\lceil m^{1/3}\rceil\rfloor.$
\end{lemma}
\begin{proof}
By Lemma~\ref{lem:properties_of_sampled_leaves}, we know $L\subseteq \leaves(\p),$ and $l_1<_{\p}l_2<_{\p}\cdots<_{\p}l_t.$

We first prove $\forall u\in V',\p(u)\in V'.$
Our proof is by induction on the leaf $l_i.$
By Lemma~\ref{lem:lca}, we have that $\forall i\in[t-1],p_{l_i,l_{i+1}}$ is the LCA of $(l_i,l_{i+1}),$ $p_{i,l_{i+1}}$ is an ancestor of $l_{i+1},$ and $p_{i,l_{i+1}}\not=p_{l_i,l_{i+1}},\p(p_{i,l_{i+1}})=p_{l_i,l_{i+1}}.$
By Lemma~\ref{alg:multi_path}, $P_1$ contains all the vertices on the path from $l_1$ to the root $v$.
$P_1$ is the set of all the ancestors of $l_1.$
Thus, every ancestor $u$ of $l_1$ is in $P_1$ and satisfies $\p(u)\in V'.$
$P_2$ contains all the vertices on the path from $l_1$ to an ancestor of $l_1$.
Thus, $P_2\subseteq P_1.$
Suppose now $P_1\cup P_2\cup\cdots \cup P_{2i-2}=\{u\in V\mid\exists j\in[i-1],u\text{~is~an~ancestor~of~}l_j\}.$
Notice that $P_{2i-1}$ contains all the vertices on the path from $l_i$ to the ancestor $p_{i-1,l_i}.$
Since $\p(p_{i-1,l_i})=p_{l_{i-1},l_i}$ is also an ancestor of $l_{i-1},$ we have $P_1\cup P_2\cup\cdots \cup P_{2i-2}\cup P_{2i-1}=\{u\in V\mid\exists j\in[i],u\text{~is~an~ancestor~of~}l_j\}.$
Since $P_{2i}$ contains all the vertices on the path from $l_i$ to an ancestor of $l_i$, we have $P_{2i}\subseteq P_1\cup P_2\cup\cdots \cup P_{2i-2}\cup P_{2i-1}.$
To conclude, we have $V'=P_1\cup P_2\cup\cdots\cup P_{2t}=\{u\in V\mid\exists j\in[t],u\text{~is~an~ancestor~of~}l_j\}.$
Thus, $\forall u\in V',$ we have $\p(u)\in V'.$

By Lemma~\ref{lem:properties_of_sampled_leaves}, with probability at least $1-1/(100m^{5/\delta}),$ $\forall u\in \leaves(\p)\setminus L,$ there exists $w\in L,w<_{\p} u$ such that $|\{x\in\leaves(\p)\mid w<_{\p}x<_{\p}u\}|\leq \lfloor |\leaves(\p)|/\lceil m^{1/3}\rceil\rfloor.$
In the following, we condition on the above event happens.
Let $u\in V\setminus V'.$
Due to Fact~\ref{fac:dfs_sequence}, the DFS sequence of the subtree of $u$ in $\p$ must be a consecutive subsequence of the DFS sequence of $\p.$
Thus, $\exists x,y\in\leaves(\p),$ the leaves in the subtree of $u$ in $\p$ is the set $\{z\in\leaves(\p)\mid x<_{\p}z<_{\p}y\}\cup\{x\}\cup\{y\}.$
If the number of leaves in the subtree of $u$ is more than $\lfloor |\leaves(\p)|/\lceil m^{1/3}\rceil\rfloor,$ then $\exists l_i\in L,$ $u$ is an ancestor of leaf $l_i$.
But $l_i\in V'$ contradicts to $u\not\in V'.$
Thus, the number of leaves in the subtree of $u$ is at most $\lfloor |\leaves(\p)|/\lceil m^{1/3}\rceil\rfloor.$
\end{proof}

\begin{lemma}[$A$ is a subsequence]\label{lem:A_is_subsequence}
Let $\p:V\rightarrow V$ be a set of parent pointers (See Definition~\ref{def:parent_pointers}) on a vertex set $V$, and $\p$ has a unique root. Let $m>0,\delta\in(0,1)$ be parameters, and let $|V|\leq m^{1/\delta}.$ Let $(V',A) = \textsc{SubDFS}(\p,m,\delta)$ (Algorithm~\ref{alg:subdfs}). Then $A$ is a subsequence of the DFS sequence of $\p.$ Furthermore, $\forall u\in V',$ $u$ appears in $A$ exactly $|\child_{\p}(u)+1|$ times, and $\forall u\not\in V',$ $u$ does not appear in $A$.
\end{lemma}
\begin{proof}
We first show that $A'$ is the DFS sequence of $\p'.$
\begin{claim}\label{cla:Aprime_is_dfs_of_pprime}
$A'$ is the DFS sequence of $\p':V'\rightarrow V'.$
\end{claim}
\begin{proof}
By Lemma~\ref{lem:properties_of_sampled_leaves}, we know $\{l_1,l_2,\cdots,l_t\}=L\subseteq \leaves(\p),$ and $l_1<_{\p}l_2<_{\p}\cdots<_{\p}l_t.$
By Lemma~\ref{lem:lca}, we have that $\forall i\in[t-1],p_{l_i,l_{i+1}}$ is the LCA of $(l_i,l_{i+1}),$ $p_{i,l_{i+1}}$ is an ancestor of $l_{i+1},$ and $p_{i,l_{i+1}}\not=p_{l_i,l_{i+1}},\p(p_{i,l_{i+1}})=p_{l_i,l_{i+1}}.$
By Lemma~\ref{alg:multi_path}, $\forall i\in[t],$ $P_{2i-1}$ and $P_{2i}$ only contains some ancestors of $l_i.$
Thus, $\leaves(\p')=L.$

According to Lemma~\ref{lem:remove_trees} and Corollary~\ref{cor:subdfs_remove_multi_subtree}, the DFS sequence of $\p'$ is a subsequence of the DFS sequence of $\p.$
Thus, we still have $l_1<_{\p'}l_2<_{\p'}<_{\p'}\cdots<_{\p'} l_t.$
 Due to Lemma~\ref{alg:multi_path}, $P_1$ contains all the vertices on the path from $l_1$ to the root $v$, $P_{2t}$ contains all the vertices on the path from $l_t$ to the root $v$, $\forall i\in[t-1],P_{2i}$ contains all the vertices on the path from $\p'(l_i)$ to the LCA of $(l_i,l_{i+1}),$ and $P_{2i+1}$ contains all the vertices on the path from $l_{i+1}$ to $p_{i,l_{i+1}}.$
Thus, $A'_1$ is the path from the root $v$ to leaf $l_1$, $A'_{2t}$ is the path from $l_t$ to the root $v$, $\forall i\in[t-1],$ $A'_{2i}A'_{2i+1}$ is the path from $\p'(l_i)$ to $l_{i+1}.$
Due to Fact~\ref{fac:dfs_generate_leaf_path}, $A'=A'_1A'_2\cdots A'_{2t}$ is the DFS sequence of $\p'.$
\end{proof}

Let us define some notations. Let $\wt{A}=\{\wt{a}_1,\wt{a}_2,\cdots,\wt{a}_{\wt{s}}\}$ be the DFS sequence of $\p.$ $\forall u\in V,$ let $\st_{\wt{A}}(u)=j$ such that $\wt{a}_j$ is the first time that $u$ appears in $\wt{A}.$ We define $\ed_{\wt{A}}(u)$ be the position such that $\wt{a}_{\ed_{\wt{A}}(u)}$ is the last time that $u$ appears in $\wt{A}.$ Similarly, $\forall u\in V',$ we can define $\st_{A'}(u),\st_{A}(u),\ed_{A'}(u),\ed_{A}(u)$ to be the positions of the first time $u$ appears in $A',$ the first time $u$ appears in $A,$ the last time $u$ appears in $A',$ and the last time $u$ appears in $A$ respectively.

Since $v$ is the root (in both $\p$ and $\p'$), it suffices to prove that $(a_{\st_A(v)},a_{\st_A(v)+1},\cdots,a_{\ed_A(v)})$ is a subsequence of $(\wt{a}_{\st_{\wt{A}}(v)},\wt{a}_{\st_{\wt{A}}(v)+1},\cdots,\wt{a}_{\ed_{\wt{A}}(v)}).$
Our proof is by induction on $\dep_{\p}(u)$ for $u\in V'.$
If $\dep_{\p}(u)=\dep(\p),$ then $u$ must be a leaf in $\p'$ (or $\p$, since $\p'$ and $\p$ are the same on $V'$).
In this case, $\st_{A}(u)=\ed_{A}(u),\st_{\wt{A}}(u)=\ed_{\wt{A}}(u),$ and $(a_{\st_{A}(u)})=(\wt{a}_{\st_{\wt{A}}(u)})=(u).$
Suppose for all $u\in V'$ with $\dep_{\p}(u)>d,$ we have that $(a_{\st_A(u)},\cdots,a_{\ed_A(u)})$ is a subsequence of $(\wt{a}_{\st_{\wt{A}}(u)},\cdots,\wt{a}_{\ed_{\wt{A}}(u)}).$
Let $u$ be a vertex in $V'$ with $\dep_{\p}(u)=d.$
If $u$ is a leaf, then it is the same as the previous argument.
Now let us consider the case when $u$ is not a leaf.
According to Claim~\ref{cla:Aprime_is_dfs_of_pprime}, $A'$ is the DFS sequence of $\p'.$
Due to line~\ref{sta:duplicate_elements}, $A$ is obtained by duplicating each element of $A'$ several times.
Let $w_1,w_2,\cdots,w_k$ be the children of $u$ in $\p'$, and $\rank_{\p'}(w_1)=1,\rank_{\p'}(w_2)=2,\cdots,\rank_{\p'}(w_k)=|\child_{\p'}(u)|.$
Then, according to Fact~\ref{fac:dfs_sequence}, $(a_{\st_A(u)},\cdots,a_{\ed_A(u)})$ should look like:
\begin{align*}
(u,\cdots,u,a_{\st_A(w_1)},\cdots,a_{\ed_A(w_1)},u,\cdots,u,a_{\st_A(w_2)},\cdots,a_{\ed_A(w_2)},\cdots,a_{\st_A(w_k)},\cdots,a_{\ed_A(w_k)},u,\cdots,u)
\end{align*}
where the number of $u$ before $a_{\st_A(w_1)}$ is $\rank_{\p}(w_1)$ (see line~\ref{sta:the_first_duplicate}), the number of $u$ before $a_{\st_A(w_i)}$ for $i\in[k]\setminus \{1\}$ is $\rank_{\p}(w_i)-\rank_{\p}(w_{i-1})$ (see line~\ref{sta:mid_duplicate}), and the number of $u$ after $a_{\ed_A(w_k)}$ is $|\child_{\p}(u)|-\rank_{\p}(w_k)+1$ (see line~\ref{sta:final_duplicate}).
Since $\wt{A}$ is the DFS sequence of $\p,$ according to Fact~\ref{fac:dfs_sequence}, the number of $u$ in $\wt{A}$ before $\wt{a}_{\st_{\wt{A}}(w_1)}$ is $\rank_{\p}(w_1).$
By our induction hypothesis, $(a_{\st_A(w_1)},\cdots,a_{\ed_A(w_1)})$ is a subsequence of $(\wt{a}_{\st_{\wt{A}}(w_1)},\cdots,\wt{a}_{\ed_{\wt{A}}(w_1)}).$
Thus, $(a_{\st_A(u)},\cdots,a_{\ed_A(w_1)})$ is a subsequence of $(\wt{a}_{\st_{\wt{A}}(u)},\cdots,\wt{a}_{\ed_{\wt{A}(w_1)}}).$
According to Fact~\ref{fac:dfs_sequence}, $\forall i\in[k]\setminus\{1\},$ the number of $u$ in $\wt{A}$ between $\wt{a}_{\ed_{\wt{A}}(w_{i-1})}$ and $\wt{a}_{\st_{\wt{A}}(w_{i})}$ is $\rank_{\p}(w_i)-\rank_{\p}(w_{i-1}).$
By our induction hypothesis, for all $i\in[k]\setminus\{1\},$ $(a_{\st_A(w_i)},\cdots,a_{\ed_A(w_i)})$ is a subsequence of $(\wt{a}_{\st_{\wt{A}}(w_i)},\cdots,\wt{a}_{\ed_{\wt{A}}(w_i)}).$
Thus, $(a_{\st_A(u)},\cdots,a_{\ed_A(w_k)})$ is a subsequence of $(\wt{a}_{\st_{\wt{A}}(u)},\cdots,\wt{a}_{\ed_{\wt{A}(w_k)}}).$
According to Fact~\ref{fac:dfs_sequence}, the number of $u$ in $\wt{A}$ after $\wt{a}_{\ed_{\wt{A}}(w_{k})}$ is $|\child_{\p}(u)|-\rank_{\p}(w_k)+1.$
Thus, $(a_{\st_A(u)},\cdots,a_{\ed_A(u)})$ is a subsequence of $(\wt{a}_{\st_{\wt{A}}(u)},\cdots,\wt{a}_{\ed_{\wt{A}(u)}}).$
Furthermore, the number of $u$ appears in $A$ is $|\child_{\p}(u)|-\rank_{\p}(w_k)+1+\rank_{\p}(w_1)+\sum_{i=2}^k \rank_{\p}(w_i)-\rank_{\p}(w_{i-1})=|\child_{\p}(u)|+1.$

Since $A'$ is the DFS sequence of $\p',$ $\forall u\not\in  V',$ $u$ does not appear in $A'.$
Thus, $\forall u\not\in V',$ $u$ does not appear in $A$.

\end{proof}

\subsubsection{DFS Sequence}
In this section, we show how to use Algorithm~\ref{alg:subdfs} as a subroutine to output a DFS sequence. The high level idea is that we use Algorithm~\ref{alg:subdfs} to generate subsequences of the DFS sequence in each iteration, and we ensure that the miss part of the DFS sequence must be the DFS sequences of many subtrees. After the $i^{\text{th}}$ iteration, we should ensure that the number of leaves of each subtree which has missing DFS sequence is at most $n/m^i,$ where $m$ is some parameter depends on some computational resources (e.g. memory size of a machine). The description of the algorithm is shown in Algorithm~\ref{alg:dfs_sequence}. Figure~\ref{fig:dfs} shows one step in our algorithm.

\begin{figure}[t!]
  \centering
  \includegraphics[width=0.48\textwidth]{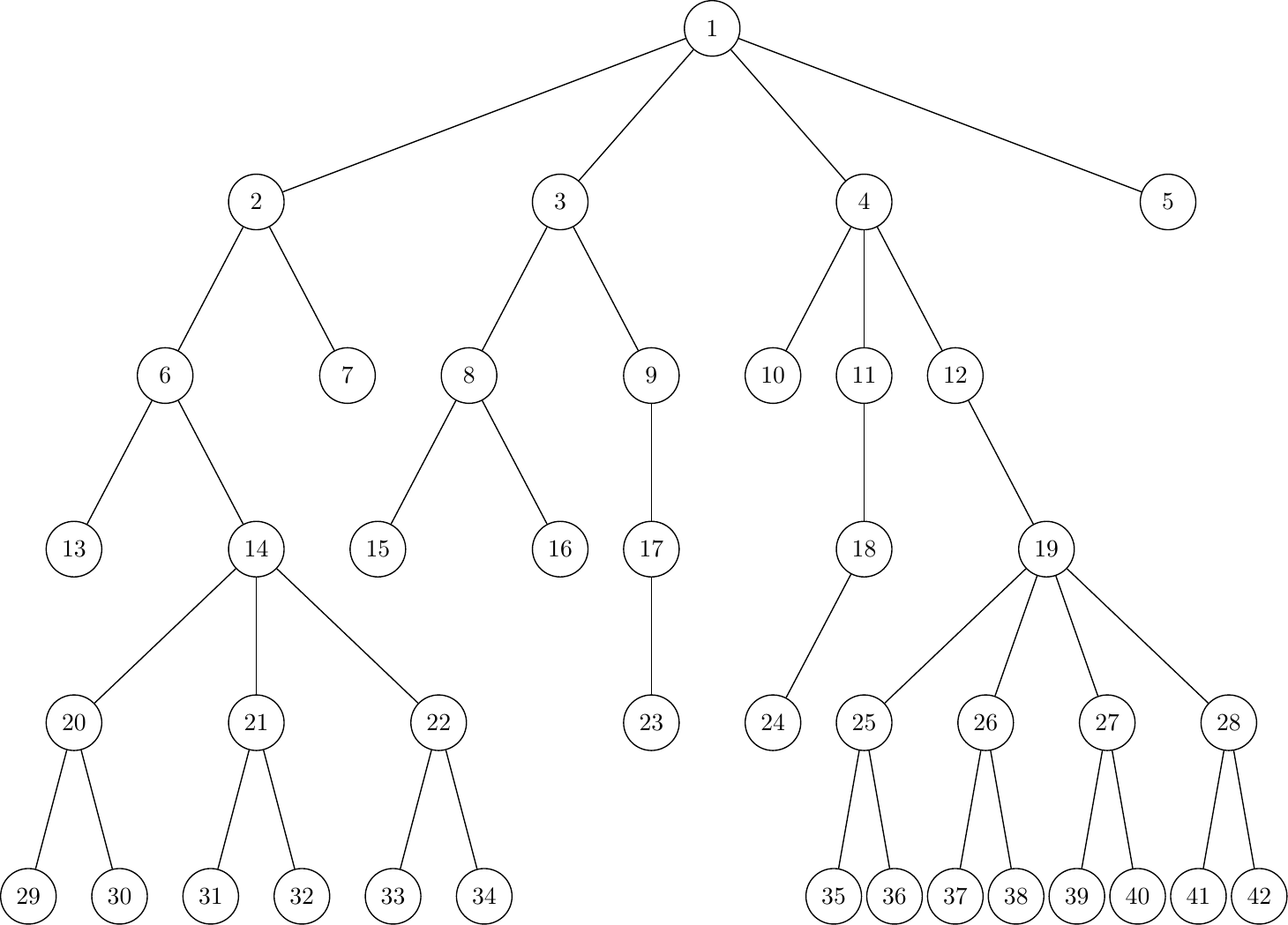}
  \includegraphics[width=0.48\textwidth]{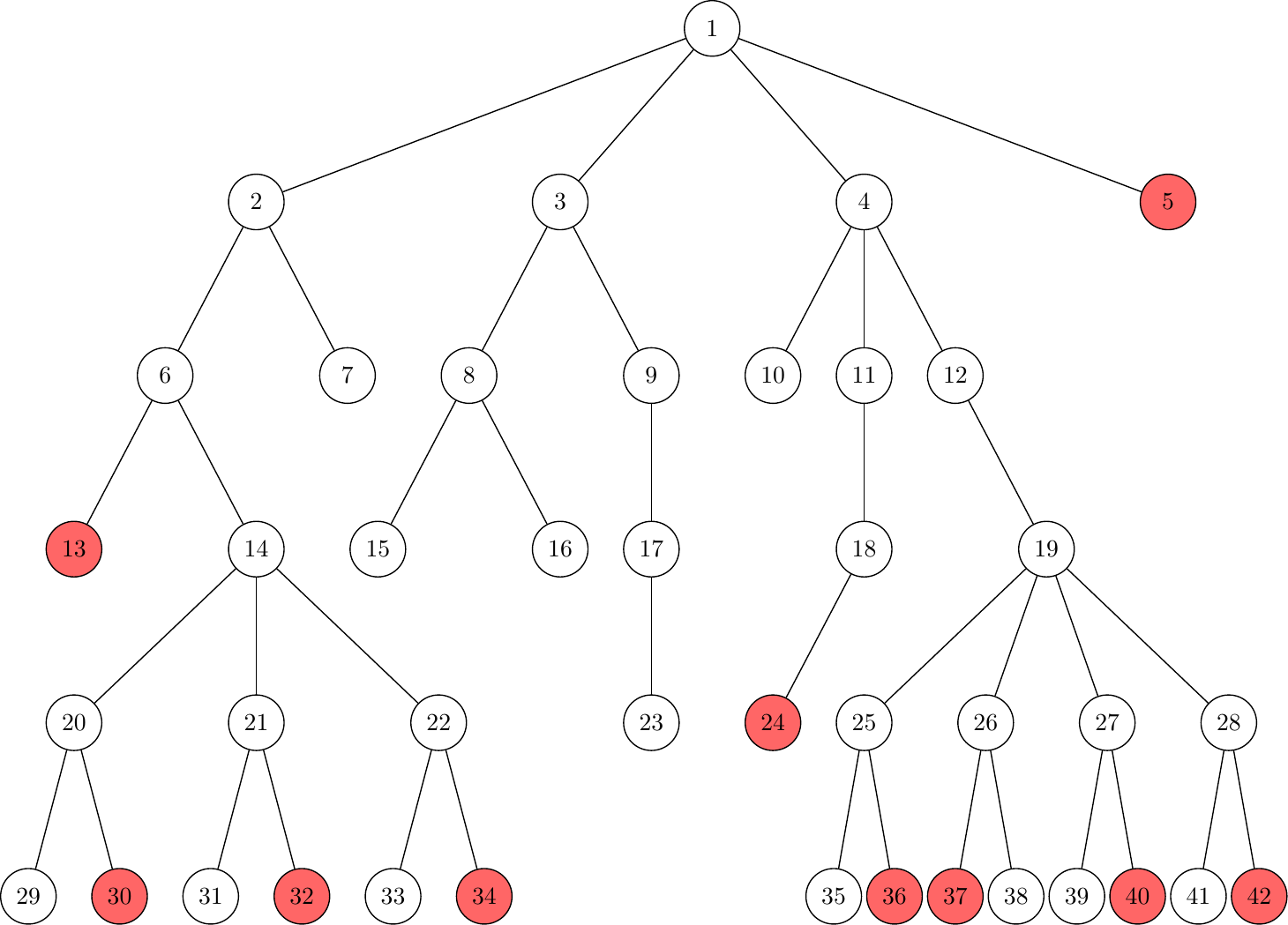}\\
  \includegraphics[width=0.48\textwidth]{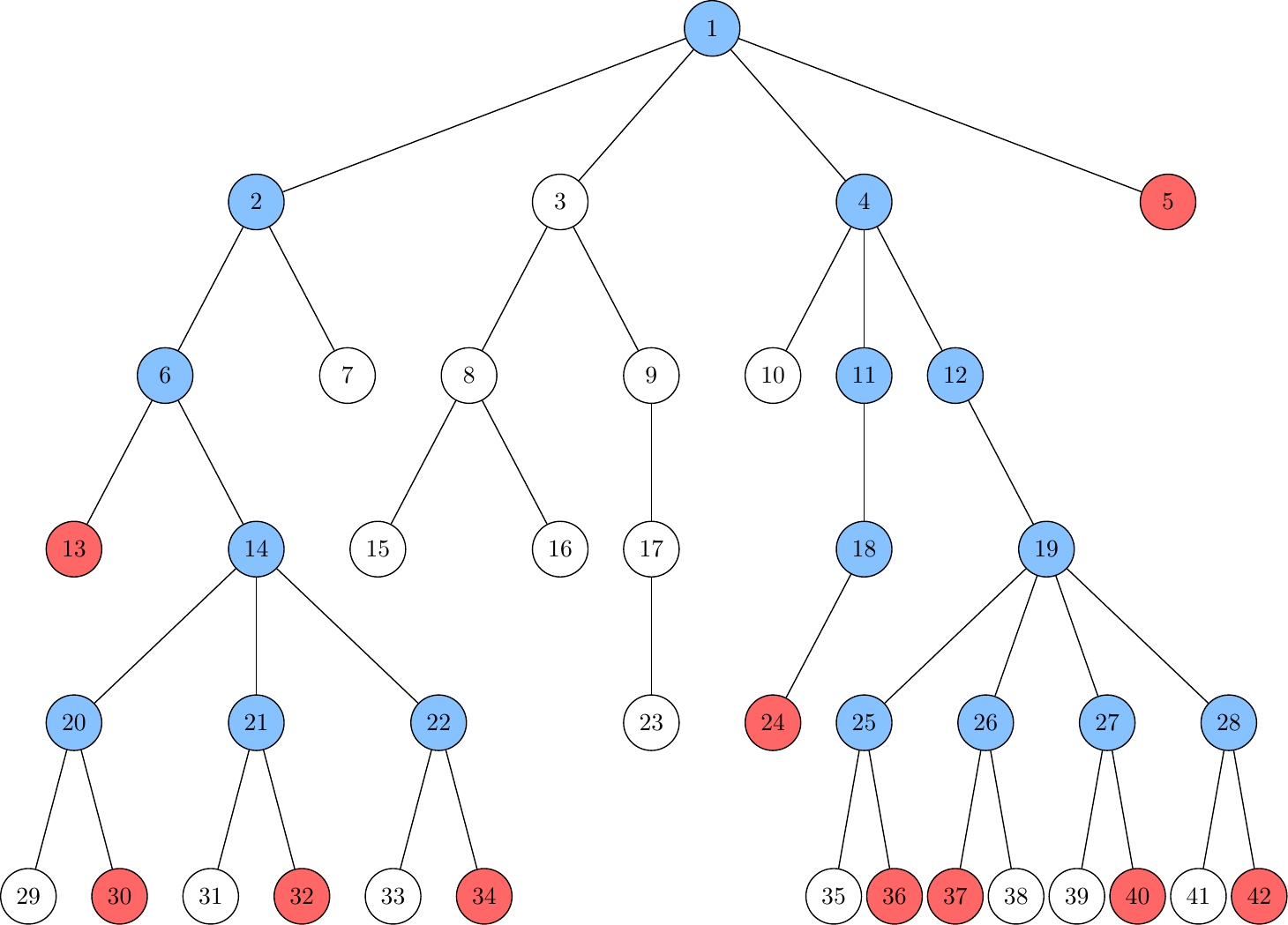}
  \includegraphics[width=0.48\textwidth]{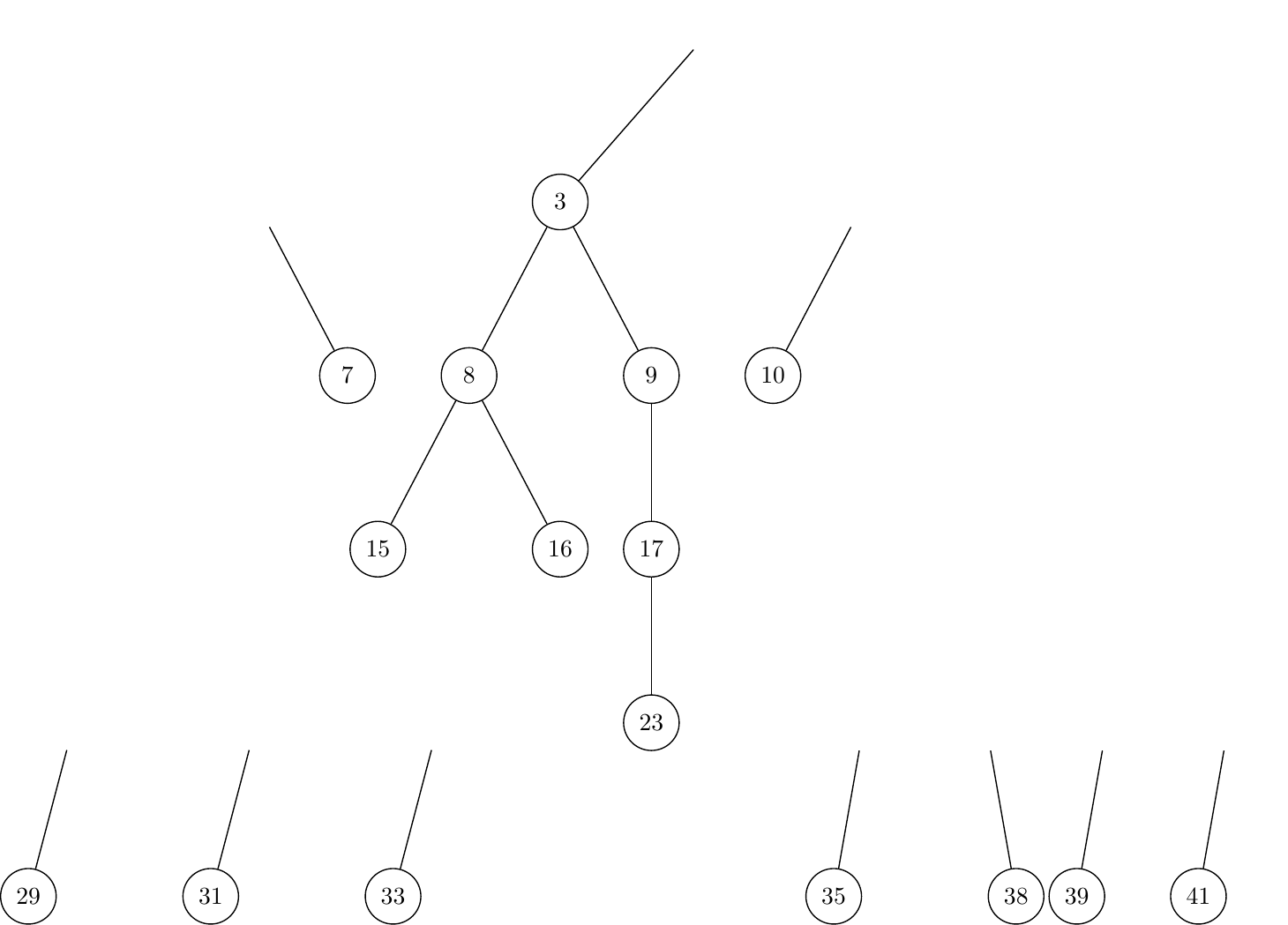}
  \caption{Given a tree that has $42$ vertices (top-left), we label all the vertices from $1$ to $42$. Firstly, we sample some leaves (red vertices, i.e. $\{ 5,13,24,30,32,34,36,37,40,42 \}$) in the tree (top-right tree). Then we find a DFS sequence of the tree (the tree formed by all the blue and red vertices in the bottom-left tree) which only contains all the sampled leaves and their ancestors. Finally, we recursively find the DFS sequences of remaining subtrees(bottom-right).}\label{fig:dfs}
\end{figure}

\begin{algorithm}[h]
\caption{DFS sequence}\label{alg:dfs_sequence}
\begin{algorithmic}[1]
\Procedure{DFS}{$\p:V\rightarrow V,m$} \Comment{Theorem~\ref{thm:dfs_sequence}, Theorem~\ref{thm:success_prob_dfs_sequence}}
\State Output: FAIL or $A=(a_1,a_2,\cdots,a_{2|V|-1}).$
\State $n=|V|,\delta=1/\log_m n.$
\State Let $\p_0=\p.$
\State $(V_0,A_0)=\textsc{SubDFS}(\p_0,m,\delta).$\label{sta:alg_dfs_init_sub_dfs}\Comment{Algorithm~\ref{alg:subdfs}}
\State Let $r=\lceil3/\delta\rceil+2.$
\For{$i=1\rightarrow r$} \Comment{$v\in V_i\Leftrightarrow v$ appears in $A_i$}

\Comment{If $v\in V_i,$ then $v$ appears $|\child_{\p}(v)|+1$ times in $A_i$}
\State Let $V'_i=V\setminus V_{i-1}.$\label{sta:alg_dfs_vprimei}
\State Initialize $\p_i:V'_i\rightarrow V'_i.$
\State For $v\in V'_i,$ if $\p(v)\in V_{i-1},$ let $\p_i(v)=v$; Otherwise, let $\p_i(v)=\p(v).$\label{sta:alg_dfs_pi}
\State $((V''_i,\emptyset),\p_i^{(\infty)})=\textsc{TreeContraction}((V'_i,\emptyset),\p_i).$\label{sta:alg_dfs_tree_contract} \Comment{Algorithm~\ref{alg:tree_contraction}}
\State $V_i\leftarrow V_{i-1}.$
\State $A_i\leftarrow A_{i-1}.$
\For{$v\in V'_i,\p_i(v)=v$}\label{sta:alg_dfs_the_loop}\Comment{The DFS sequence of the subtree of $v$ in $\p$ is missing.}
\State Let $V'_i(v)=\{u\in V'_i\mid \p_i^{(\infty)}(u)=v\}.$
\State Let $\p_{i,v}:V'_i(v)\rightarrow V'_i(v)$ satisfy $\forall u\in V'_i(v),\p_{i,v}(u)=\p_i(u).$
\State Let $(V_{i,v},A_{i,v}) = \textsc{SubDFS}(\p_{i,v},m,\delta).$\label{sta:alg_dfs_inloop_subdfs}\Comment{Algorithm~\ref{alg:subdfs}}
\State $V_i\leftarrow V_i\cup V_{i,v}.$
\State Insert $A_{i,v}$ after the $\rank_{\p}(v)^{\text{th}}$ time appearance of $v$ in $A_i.$\label{sta:alg_dfs_sequence_insertion}
\EndFor
\EndFor
\State If $V_r=V,$ return $A_r$ as $A$. Otherwise, return FAIL.
\EndProcedure
\end{algorithmic}
\end{algorithm}

\begin{theorem}[Correctness of DFS sequence]\label{thm:dfs_sequence}
Let $\p:V\rightarrow V$ be a set of parent pointers (See Definition~\ref{def:parent_pointers}) on a vertex set $V$, and $\p$ has a unique root. Let $n=|V|,m=n^{\delta}$ for some constant $\delta\in(0,1)$. If $A = \textsc{DFS}(\p,m)$ (Algorithm~\ref{alg:dfs_sequence}) does not output {\rm FAIL}, then $A$ is the DFS sequence of $\p.$
\end{theorem}
\begin{proof}
It suffices to prove the following claim.
\begin{claim}\label{cla:Ai_is_good}
Let $i\in\{0\}\cup [r].$ $A_i$ is a subsequence of the DFS sequence of $\p.$ $\forall v\in V_i,$ $\p(v)\in V_i.$ Furthermore, $\forall v\in V_i,$ $v$ appears in $A_i$ exactly $|\child_{\p}(v)|+1$ times, and $\forall v\not\in V_i,$ $v$ does not appear in $A_i.$
\end{claim}
\begin{proof}
Our proof is by induction on $i.$
If $i=0,$ then by Lemma~\ref{lem:A_is_subsequence}, $A_0$ is a subsequence of the DFS sequence of $\p,$ $\forall v\in V_0,$ $v$ appears in $A_0$ exactly $|\child_{\p}(v)|+1$ times, and $\forall v\not\in V_0,$ $v$ does not appear in $A_0.$
By Lemma~\ref{lem:remove_trees}, we have $\forall v\in V_0,\p(v)\in V_0.$

Suppose the claim is true for $i-1.$
Let $u\in V_i.$

If $u\in V_{i-1},$ then since $V_{i-1}\subseteq V_i,$ $\p(u)\in V_i.$
Otherwise $u\in V_{i,v}$ for some $v$ with $\p_i(v)=v.$
If $u=v,$ then $\p(v)\in V_{i-1}\subseteq V_i.$
Otherwise, by Lemma~\ref{lem:remove_trees}, $\p(u)\in V_{i,v}\subseteq V_i.$

Now consider the property of $A_i.$
If $u\in V_{i-1},$ then since $A_{i-1}$ is a subsequence of $A_i,$ and by Lemma~\ref{lem:A_is_subsequence} $u$ cannot appear in any $A_{i,v},$ $u$ must appear in $A_i$ exactly $|\child_{\p}(u)|+1$ times.
Otherwise $u\in V_{i,v}$ for some $v$ with $\p_i(v)=v.$
By Lemma~\ref{lem:A_is_subsequence}, $u$ must appear in $A_{i,v}$ $|\child_{\p_{i,v}}(u)|+1=|\child_{\p}(u)|+1$ times.
Since $u$ cannot appear in $A_{i-1},$ $u$ must appear in $A_i$ exactly $|\child_{\p}(u)|+1$ times.
For $v\in V'_i$ with $\p_i(v)=v,$ according to Fact~\ref{fac:dfs_sequence} and $\forall w\in \{x\in V\mid v\text{~is~an~ancestor~of~}x\},$ $w$ cannot be in $V_{i-1},$ the $\rank_{\p}(v)^{\text{th}}$ time appearance of $v$ and the $(\rank_{\p}(v)+1)^{\text{th}}$ time appearance of $v$ should be adjacent in $A_{i-1}.$
Due to Lemma~\ref{lem:A_is_subsequence}, $A_{i,v}$ is a subsequence of the DFS sequence of the subtree of $v$ in $\p.$
Due to Fact~\ref{fac:dfs_sequence}, $A_i$ is still a subsequence of the DFS sequence of $\p$ after insertion of the sequence $A_{i,v}$.

For any $x\not\in V_i,$ by Lemma~\ref{lem:A_is_subsequence}, $x$ cannot be in any $A_{i,v}.$ By our induction hypothesis, $x$ cannot be in $A_{i-1}.$
Thus, $x$ cannot be in $A_i.$
\end{proof}
If the procedure does not output {\rm FAIL}, then according to the above Claim~\ref{cla:Ai_is_good}, $\forall v\in V_r=V,$ $v$ appears in $A_r=A$ exactly $|\child_{\p}(v)|+1$ times, and $A_r=A$ is a subsequence of the DFS sequence of $\p$.
Due to Fact~\ref{fac:dfs_sequence}, $A=A_r$ is the DFS sequence of $\p$.
\end{proof}

The following lemma claims the success probability of Algorithm~\ref{alg:dfs_sequence}.

\begin{theorem}[Success probability]\label{thm:success_prob_dfs_sequence}
Let $\p:V\rightarrow V$ be a set of parent pointers (See Definition~\ref{def:parent_pointers}) on a vertex set $V$, and $\p$ has a unique root. Let $n=|V|,m=n^{\delta}$ for some constant $\delta\in(0,1)$. With probability at least $1-1/(100n^4),$ $A = \textsc{DFS}(\p,m)$ (Algorithm~\ref{alg:dfs_sequence}) does not output {\rm FAIL}.
\end{theorem}
\begin{proof}
$\forall i\in[r],v\in V'_i$ with $\p_i(v)=v,$ let $\mathcal{E}_{i,v}$ be the event that $\forall u\in V'_i(v)\setminus V_{i,v},$ the number of leaves in the subtree of $u$ in $\p$ is at most $|\leaves(\p_{i,v})|/n^{\delta/3}.$
 Notice that due to Lemma~\ref{lem:remove_trees}, if $\p_i(v)=v,$ then $v$ will be in $V_i.$
 Thus, we use $\mathcal{E}_v$ to denote the event $\mathcal{E}_{i,v}.$
 By Lemma~\ref{lem:remove_trees}, $\mathcal{E}_{v}$ happens with probability at least $1-1/(100n^5).$
 By taking union bound over all $v,$ with probability at least $1-1/(100n^4),$ all the events $\mathcal{E}_v$ will happen.
\begin{claim}\label{cla:no_leaf}
Condition on all the events $\mathcal{E}_v$ happen. $\forall i\in[r],v\in V'_i$ with $\p_i(v)=v,$ we have $|\leaves(\p_{i,v})|\leq n/n^{(i-1)\delta/3}.$
\end{claim}
\begin{proof}
When $i=1,$ the claim is obviously true, since $|\leaves(\p_{i,v})|\leq n.$
Suppose the claim holds for $i-1.$
Let $v\in V'_i$ with $\p_i(v)=v.$
There must be $v'\in V'_{i-1}$ with $\p_{i-1}(v')=v',$ and $v\in V'_{i-1}(v')\setminus V_{i-1,v'}.$
Since $\mathcal{E}_{v'}$ happens, the number of leaves in the subtree of $v$ in $\p$ is at most $|\leaves(\p_{i-1,v'})|/n^{\delta/3}\leq n/n^{(i-1)\delta/3}.$
\end{proof}
If $V'_r\not=\emptyset,$ then $\exists v\in V'_r$ with $\p_i(v)=v$ and $|\leaves(\p_{i,v})|\geq 1.$
If all the events $\mathcal{E}_v$ happens, it will contradict to Claim~\ref{cla:no_leaf}.
Thus, if all the events $\mathcal{E}_v$ happens, $V'_r$ must be $\emptyset,$ and thus $V_r=V$ which implies that the procedure will not fail.
\end{proof}
    \subsection{Range Minimum Query}
Range Minimum Query (RMQ) problem is defined as following: given a sequence of $n$ numbers $a_1,a_2,\cdots,a_n,$ the goal is to preprocess the sequence $a$ to get a data structure such that for any query $(p,q),(p<q)$ we can efficiently find the element which is the minimum in $a_p,a_{p+1},\cdots,a_{q}.$ A classic method is to preprocess a sparse table $f$ in $\log(n)$ number of iterations such that $\forall i\in [n],j\in [\lceil \log n\rceil]\cup\{0\}, $ $f_{i,j}=\arg\min_{i\leq i'\leq \min(n,i+2^j-1)} a_i.$ To answer query for $(p,q),$ it just needs to return $\arg\min_{i\in\{f_{p,j^*},f_{q-2^{j^*}+1,j^*}\}}a_i$ for $j^*=\lfloor\log(q-p+1)\rfloor .$ In this section, we firstly show a modified data structure.
We will compute $\wh{f}_{i,j}=\arg\min_{i\leq i'\leq \min(n,i+\lceil n^{\delta}\rceil^j -1)} a_{i'}$
The Algorithm is shown in Algorithm~\ref{alg:st_rmq}. Then we show how to use $\wh{f}$ to compute $f$ in Algorithm~\ref{alg:st_rmq2}.

\begin{algorithm}[h]
\caption{A Sparser Table for RMQ}\label{alg:st_rmq}
\begin{algorithmic}[1]
\Procedure{SparseTable$^+$}{$a_1,a_2,\cdots,a_n,\delta$} \Comment{ Lemma~\ref{lem:sparser_table} }
\State \Comment{ Output: $\wh{f}_{i,j}$ for $i\in[n],j\in \{0\}\cup[\lceil 1/\delta \rceil]$}
\State Initially, for all $i\in[n]$ let $\wh{f}_{i,0}=i.$ $\forall i>n,j\in\mathbb{Z},$ let $\wh{f}_{i,j}=0,$ and let $a_0=\infty.$ Let $m=\lceil n^\delta\rceil.$
\State For $t\in [\lceil1/\delta\rceil],$ let $S_t=\{x\mid \exists y\in[m-1],x=y\cdot m^t\}.$
\For{$l=1\rightarrow \lceil 1/\delta \rceil$}
    \For{$j=0\rightarrow \lceil n/m \rceil$}
        \State $i\leftarrow j\cdot m+1.$
        \State $z^*_{j,l}\leftarrow \arg\min_{z:t\in[l-1],x\in S_t,z=\wh{f}_{j\cdot m+1+x,t}} a_z.$ \label{sta:rmq_zstar}
        \For{$i'=0\rightarrow \min(m-1,n-i)$}
            \State $T \leftarrow \{x\in\mathbb{Z}\mid i+i' \leq x\leq i+m-1\}\cup\{x\in\mathbb{Z}\mid i+m^l \leq x\leq i+m^l+i'-1\}\cup\{z^*_{j,l}\}$
            \State $\wh{f}_{i+i',l}=\arg\min_{z\in T} a_z.$  \label{sta:assign_value}
        \EndFor
    \EndFor
    \State $l\leftarrow l+1.$
\EndFor
\State \Return $\wh{f}_{i,j}$ for $i\in[n],j\in \{0\}\cup[\lceil 1/\delta \rceil].$
\EndProcedure
\end{algorithmic}
\end{algorithm}

\begin{lemma}\label{lem:sparser_table}
Let $a_1,a_2,\cdots,a_n$ be a sequence of numbers, and $\delta\in(0,1).$ Let $\{\wh{f}_{p,q}\}$ be the output of $\textsc{SparseTable}^+(a_1,a_2,\cdots,a_n,\delta)$ (Algorithm~\ref{alg:st_rmq}). Then $\forall p\in[n],q\in\{0\}\cup [\lceil1/\delta\rceil],$ $\wh{f}_{p,q}=\arg\min_{p\leq p'\leq \min(n,i+\lceil n^{\delta}\rceil^q -1)} a_{p'}.$
\end{lemma}
\begin{proof}
The proof is by induction on $q$. When $q=0,$ the statement obviously holds for all $\wh{f}_{p,0}.$ Suppose all $p\in[n],\wh{f}_{p,0},\wh{f}_{p,1},\cdots,\wh{f}_{p,q-1}$ satisfy the property. The first observation is that the value of $\wh{f}_{p,q}$ will be assigned in the procedure when $l=q,j=\lfloor (p-1)/m\rfloor,i'=(p-1)\mod m.$ Then by line~\ref{sta:rmq_zstar}, $z^*_{j,l}$ will be the position of the minimum value in $a_{j\cdot m+m},a_{j\cdot m+m+2},\cdots,a_{j\cdot m+m^l-1}$ by our induction hypothesis. Then by line~\ref{sta:assign_value}, $\wh{f}_{i+i',l}$ will be the position of the minimum value in $a_{j\cdot m+i'+1},a_{j\cdot m+i'+2},\cdots,a_{j\cdot m+i'+m^l}.$ Thus, Since $j\cdot m+i'+1=p,$ $\wh{f}_{p,q}$ satisfies the property.
\end{proof}

\begin{algorithm}[!t]
\caption{A Sparse Table for RMQ}\label{alg:st_rmq2}
\begin{algorithmic}[1]
\Procedure{SparseTable}{$a_1,a_2,\cdots,a_n,\delta$} \Comment{Lemma~\ref{lem:st_rmq2}}
\State \Comment{ Output: $f_{i,j}$ for $i\in[n],j\in \{0\}\cup[\lceil \log n \rceil].$}
\State Initially, for all $i\in[n]$ let $f_{i,0}=i.$ $\forall i>n,j\in\mathbb{Z},$ let $f_{i,j}=0,$ and let $a_0=\infty.$ Let $m=\lceil n^\delta\rceil.$
\State Let $\{\wh{f}_{p,q}\mid p\in[n],q\in\{0\}\cup\lceil1/\delta\rceil\} = \textsc{SparseTable}^+(a_1,a_2,\cdots,a_n,\delta).$\label{sta:alg_st_rmq2_sttable} \Comment{Algorithm~\ref{alg:st_rmq}}
\State Let all undefined $\wh{f}_{p,q}$ be $0$.
\For{$t\in [\lceil\log n\rceil]$} 
    \If{ $2^t\leq m$}
        \State $k_t \leftarrow -1$
        \State $S_t\leftarrow \emptyset$ 
    \Else
        \State $k_t \leftarrow \lfloor\log_{m}(2^t-m)\rfloor$
        \State {\small $S_t \leftarrow \{x ~|~ x\in [2^t-m-m^{k_t}+1]\text{~s.t.~}x \equiv 1 { \pmod {m^{k_t}} } \text{~or~}(2^t-m-x)\equiv-1 \pmod {m^{k_t}}\}$ }
    \EndIf
\EndFor
\For{$j=0\rightarrow \lceil n/m\rceil$}\label{sta:alg_st_rmq2_loop}
    \For{$t=0\rightarrow \lceil\log n\rceil$}
        \State $i\leftarrow j\cdot m+1.$
        \State $z^*_{j,t}\leftarrow \arg\min_{z:x\in S_t,z=\wh{f}_{j\cdot m+m+x,k_t}} a_z.$
        \For{$i'=0\rightarrow \min(m-1,n-i)$}
            \State $T_1 \leftarrow \{x\in\mathbb{Z} ~|~ i+i' \leq x\leq \min(i+m-1,i+i'+2^t-1)\}$
            \State $T_2 \leftarrow \{x\in\mathbb{Z} ~|~ \max(i+2^t,i+i') \leq x\leq i+2^t+i'-1\}$ 
            \State $T \leftarrow T_1 \cup T_2 \cup \{ z^*_{j,t} \}$
            \State $f_{i+i',t}=\arg\min_{z \in T} a_z.$\label{sta:sparse_table_iiprime}
        \EndFor
    \EndFor
\EndFor
\State \Return $f_{i,j}$ for $i\in[n],j\in \{0\}\cup[\lceil \log n \rceil].$
\EndProcedure
\end{algorithmic}
\end{algorithm}

\begin{lemma}\label{lem:st_rmq2}
Let $a_1,a_2,\cdots,a_n$ be a sequence of numbers, and $\delta\in(0,1).$ Let $\{f_{p,q}\}$ be the output of $\textsc{SparseTable}(a_1,a_2,\cdots,a_n,\delta)$ (Algorithm~\ref{alg:st_rmq2}). Then $\forall p\in[n],q\in\{0\}\cup [\lceil\log n\rceil],$ $f_{p,q}=\arg\min_{p\leq p'\leq \min(n,i+2^q -1)} a_{p'}.$
\end{lemma}
\begin{proof}
Let $m=\lceil n^{\delta}\rceil.$ By Lemma~\ref{lem:sparser_table}, $\forall x\in[n],y\in\{0\}\cup [\lceil1/\delta\rceil],$ $\wh{f}_{x,y}=\arg\min_{x\leq x'\leq \min(n,i+m^y -1)} a_{x'}.$ Thus, by the definition of $S_t,$ we know $z^*_{j,t}=\arg\min_{j\cdot m+m+1\leq z\leq j\cdot m+2^t} a_z.$ An observation is that the value of $f_{p,q}$ will be assigned in the procedure when $t=q,j=\lfloor (p-1)/m\rfloor,i'=(p-1)\mod m.$ By line~\ref{sta:sparse_table_iiprime}, we know $$f_{p,q}=f_{i+i',t}=\arg\min_{z:i+i'+1\leq z\leq i+i'+2^t} a_z=\arg\min_{p\leq p'\leq \min(n,i+2^q -1)} a_{p'}.$$
\end{proof}

\subsection{Applications of DFS Sequence}\label{sec:dfs_app}

In this section, we briefly discuss some applications of the DFS sequence of a tree.

Since the DFS sequence of a subtree should be a continuous subsequence of the DFS sequence of the tree, one direct application of the DFS sequence is to compute the size of each subtree, i.e. for each subtree with root $v$, we can find the first place $v$ appeared and the last place $v$ appeared, and then calculate the vertices between those two appearances. 

Another application of the DFS sequence and the range minimum query is to output a data structure which can answer any LCA query in $O(1)$ time (for both sequential and parallel). This is better than the data structure provided by Section~\ref{sec:lca_multipath} which needs $O(\log D)$ time (for both sequential and parallel) to answer the query.

Since it is easy to output a data structure which can answer the depth of each vertex in $O(1)$ time (in both sequential and parallel), together with the lowest common ancestor data structure, we can answer the query of the tree distance between any two vertices in $O(1)$ time (for both sequential and parallel).                
\section{The $\MPC$ Model}
\label{sec:MPCmodel}
In this section, let us introduce the computational model studied in this paper.
Suppose we have $p$ machines indexed from $1$ to $p$ each with memory size $s$ words, where $n$ is the number of words of the input and $p\cdot s=O(n^{1+\gamma}),s=\Theta(n^{\delta}).$
Here $\delta\in(0,1)$ is a constant, $\gamma\in\mathbb{R}_{\geq 0},$ and a word has $\Theta(\log(s\cdot p))$ bits.
Thus, the total space in the system is only $O(n^{\gamma})$ factor more than the input size $n$, and each machine has local memory size sublinear in $n$.
When $0\leq \gamma\leq O(1/\log n),$ the total space is just linear in the input size.
 The computation proceeds in rounds.
 At the beginning of the computation, the input is distributed on the local memory of $\Theta(n/s)$ input machines.
 Input machines and other machines are identical except that input machine can hold a part of the input in its local memory at the beginning of the computation while each of other machines initially holds nothing.
 In each round, each machine performs computation on the data in its local memory, and sends messages to other machines (including the sender itself when it wants to keep the data) at the end of the round. Although any two machines can communicate directly in any round, the total size of messages (including the self-sent messages) sent or received of a machine in a round is bounded by $s,$ its local memory size. In the next round, each machine only holds the received messages in its local memory.
  At the end of the computation, the output is distributed on the output machines.
  Output machines and other machines are identical except that output machine can hold a part of the output in its local memory at the end of the computation while each of other machines holds nothing.
  We call the above model $(\gamma,\delta)-\MPC$ model.
  The model is exactly the same as the model $\MPC(\varepsilon)$ defined by~\cite{bks13} with $\eps=\gamma/(1+\gamma-\delta)$ and the number of machines $p=O(n^{1+\gamma-\delta})$.
  Since we care more about the total space used by the algorithm, we use $(\gamma,\delta)$ to characterize the model, while in~\cite{bks13} they use parameter $\varepsilon$ to characterize the repetition of the data.
  The main complexity measure is the number of rounds $R$ required to solve the problem.


\subsection{Basic $\MPC$ Algorithms}
\paragraph{Sorting} One of the most important algorithms in $\MPC$ model is sorting. The following theorem shows that there is an efficient sorting algorithm.
\begin{theorem}[\cite{gsz11,g99}] \label{thm:sorting}
 Sorting can be solved in $c/\delta$ rounds in $(0,\delta)-\MPC$ model for any constant $\delta\in(0,1)$, where $c\geq 0$ is a universal constant.  
 Precisely, there is an algorithm $\mathcal{A}$ in $(0,\delta)-\MPC$ model such that for any set $S$ of $n$ comparable items stored $O(n^{\delta})$ per machine on input machines, $\mathcal{A}$ can run in $c/\delta$ rounds and leave the $n$ items sorted on the output machines, i.e. the ouput machine with smaller index holds a smaller part of $O(n^{\delta})$ items.
\end{theorem}
Notice that for any $\delta'\geq \delta,$ $O(1)$ number of machines with $\Theta(n^{\delta'})$ memory can always simulate the computation of $O(n^{\delta'-\delta})$ number of machines with $\Theta(n^{\delta})$ memory.
Thus, if an algorithm $\mathcal{A}$ can solve a problem in $(\gamma,\delta)-\MPC$ model in $R(n)$ rounds, then $\mathcal{A}$ can be simulated in $(\gamma',\delta')-\MPC$ model still in $R(n)$ rounds with all $\gamma'\geq \gamma,\delta'\geq \delta.$

\paragraph{Indexing} In the indexing problem, a set $S=\{x_1,x_2,\cdots,x_n\}$ of $n$ items are stored $O(n^{\delta})$ per machine on input machines. The output is 
\begin{align*}
S'=\{(x,y) ~|~ x\in S,y-1\text{ is the number of items before }x\}
\end{align*} 
of $n$ pairs stored $O(n^{\delta})$ per machine on output machines. Here, ``an item $x'\in S$ is before $x\in S$'' means that $x'$ is held by a input machine with a smaller index, or $x',x$ are stored in the same input machine but $x'$ has a smaller local memory address.
\paragraph{Prefix sum} In the prefix sum problem, a set $S=\{(x_1,y_1),(x_2,y_2),\cdots,(x_n,y_n)\}$ of $n$ (item, number) pairs are stored $O(n^{\delta})$ per machine on input machines. The output is 
\begin{align*}
S'= \left\{(x,y') ~\bigg|~ (x,y)\in S,y'-y=\sum_{(\wt{x},\wt{y})\text{~is~before~}(x,y)}\wt{y} \right\}
\end{align*} of $n$ pairs stored $O(n^{\delta})$ per machine on output machines. Here, ``an pair $(\wt{x},\wt{y})\in S$ is before $(x,y)\in S$'' means that $(\wt{x},\wt{y})$ is held by a input machine with a smaller index, or $(\wt{x},\wt{y}),(x,y)$ are stored in the same input machine but $(\wt{x},\wt{y})$ has a smaller local memory address. Notice that indexing problem is a special case of prefix sum problem.
\begin{theorem}[\cite{gsz11}]\label{thm:indexing}
Indexing/prefix sum problem can be solved in $c/\delta$ rounds in $(0,\delta)-\MPC$ model for any constant $\delta\in(0,1)$, where $c\geq 0$ is a universal constant.
\end{theorem}
Once each item has an index, it is able to reallocate them onto the machines.

\paragraph{Load balance}
Sometimes, local computations of a machine may generate new data.
When some machines are not able to keep the new data generated, we need to do loading balance.
Fortunately, this operation can be done in constant number of rounds of computations.

For arbitrary constant $\delta\in(0,1),$ we are able to spend constant number of rounds to reallocate the data in $(0,\delta)-\MPC$ model such that if a machine is not empty, the size of its local data is at least $n^{\delta}/k$ and is at most $2n^{\delta}/k$ where $k>1$ is an arbitrary constant. The method is very simple, we can use the algorithm mentioned in Theorem~\ref{thm:indexing} to index each data item, and then send them to the corresponding machine.

\paragraph{Predecessor}
In the predecessor problem, a set $S=\{(x_1,y_1),(x_2,y_2),\cdots,(x_n,y_n)\}$ of $n$ (item, $0/1$) pairs are stored $O(n^{\delta})$ per machine on input machines.
The output machines are all input machines.
If an input (also output) machine holds a tuple $(x_i,y_i)\in S$ at the beginning of the computation, then at the end of the computation, that machine should still hold the tuple $(x_i,y_i).$
In addition, if an input (also output) machine holds a tuple $(x,0)\in S$ at the beginning of the computation, then at the end of the computation, that machine should hold a tuple $(x,x')$ such that $(x',1)\in S,$ and $(x',1)$ is the last tuple occurred before $(x,0).$
Here, ``$(x',1)$ is before $(x,0)$'' means that $(x',1)$ is held by a input machine with a smaller index, or $(x',1),(x,0)$ are stored in the same input machine but $(x',1)$ has a smaller local memory address.
\begin{theorem}[\cite{gsz11}]\label{thm:predecessor}
 Predecessor problem can be solved in $c/\delta$ rounds in $(0,\delta)-\MPC$ model for any constant $\delta\in(0,1)$, where $c\geq 0$ is a universal constant.
\end{theorem}
Roughly speaking the algorithm is as the following: firstly, build a $\Theta(n^{\delta})$ branching tree on the machines, then follows by bottom-up stages to collect the last $(x_l,1)$ tuple in each large interval and then follows by top-down stages to compute the predecessors of every prefix.
For completeness, we describe the algorithm for predecessor problem in the following:

\noindent\fbox{
\begin{minipage}{\linewidth}
\small
Predecessor Algorithm:
\begin{itemize}
\item \textbf{Setups:}
    \begin{itemize}

    \item  There are $2p=\Theta(n^{\delta})$ machines indexed from $1$ to $2p$ each with local memory size $s=\Theta(n^\delta)$. The machine with index from $p+1$ to $2p$ are input/output machines.

    \item $(x_1,y_1),\cdots (x_n,y_n)$ are stored on input/output machine $p+1$ to $2p$, where $\forall i\in[n],y_i\in\{0,1\}$.

    \item The goal: If an input machine holds a tuple $(x,y)$ with $y=0,$ then it will create a tuple $(x,x')$ at the end of the computation, where $(x',y')$ is the last tuple with $y'=1$ stored before $(x,y).$

    \end{itemize}

\item \textbf{Bottom-up stage ($O(1/\delta)$ constant rounds):}
    \begin{itemize}

    \item Let $d=s/10$ be the branching factor.

    \item In the $i^{\text{th}}$ round, each machine $j$ with $j$ in the range $\lfloor p/d^{i-1}\rfloor+1$ to $\lfloor (2p-1)/d^{i-1}\rfloor+1$ sends the last $(x_l,y_l)$ tuple with $y_l=1$ in its local memory to machine $\lfloor(j-1)/d\rfloor+1.$ If machine $j$ does not have any tuple with $y_l=1$, it just sends an arbitrary tuple to machine $\lfloor(j-1)/d\rfloor+1.$

    \item Until the end of the computation, machine $j$ sends itself messages to keep the data.
        The stage ends when machine $1$ receives messages.

    \end{itemize}

 \item \textbf{Top-down stage ($O(1/\delta)$ constant rounds):}
    \begin{itemize}

    \item Let $d=s/10$ be the branching factor.

    \item In the $i^{\text{th}}$ round, each machine $j$ with $j$ in the range $\lfloor d^{i-2}\rfloor+1$ to $\min(d^{i-1},p)$ sends to each machine $h$ in the range $(j-1)d+1$ to $\min(j\cdot d,2p)$ a tuple $(x_l,y_l)$ which is the last tuple with $y_l=1$ appeared before machine $h$.

    \item The stage ends when machine $2p$ receives messages.

    \end{itemize}

 \item \textbf{The last round:}
    \begin{itemize}
    \item Machine $i\in\{p+1,\cdots,2p\}$ scans its local memory, for each tuple $(x,y)$ with $y=0,$ create a tuple $(x,x')$ where $(x',y')$ is the last tuple stored before $(x,y)$ with $y'=1.$
    \end{itemize}

\end{itemize}
\end{minipage}

}

\subsection{Data Organization}\label{sec:data_organize}
In this section, we introduce the method to organize the data in the system.
\paragraph{Set} Let $S=\{x_1,x_2,\cdots,x_m\}$ be a set of $m$ items, and each item $x_i$ can be described by $O(1)$ number of words.
 If $x\in S$ is equivalent to that there is a unique machine which holds a pair $(``S",x)$ in its local memory, then we say that $S$ is stored in the system.
 Here $``S"$ is the name of the set $S$ and can be described by $O(1)$ number of words.

 Let $\mathcal{S}=\{S_1,S_2,\cdots,S_m\}$ be a set of $m$ sets, where $\forall i\in [m],$ $S_i$ is stored in the system, and the name $``S_i"$ of each set $S_i$ can be described by $O(1)$ number of words.
 If $S\in\mathcal{S}$ is equivalent to that there is a unique machine which holds a pair $(``\mathcal{S}",``S")$ in its local memory, then we say $\mathcal{S}$ is stored in the system.
 Here $``\mathcal{S}"$ is the name of $\mathcal{S}$ and can be described by $O(1)$ number of words.

 Let $S$ be a set stored in the system. If machine $i$ has a pair $(``S",x),$ then we say that the element $x$ of $S$ is held by the machine $i$. If every element of $S$ is held by a machine with index in $\{i,i+1,\cdots,j\},$ then we say $S$ is stored on the machine $i$ to the machine $j$.

 The total space needed to store $S$ is $\Theta(m).$

 \paragraph{Mapping} Let $f:U\rightarrow H$ be a mapping from a finite set $U$ to a set $H$.
 In the following, we show how to use a set to represent a mapping.
 \begin{definition}[Set representation of a mapping]\label{def:set_for_mapping}
Let $f:U\rightarrow H$ be a mapping from a finite set $U$ to a set $H$.
Let $S=\{(x,y)\mid x\in U,y=f(x)\}.$
then the set $S$ is a set representation of the mapping $f.$
\end{definition}
Let $U$ be a finite set where each element of $U$ can be described by $O(1)$ number of words. Let $S$ be a set representation of the mapping $f:U\rightarrow H.$ If $S$ is stored in the system, then we say $f$ is stored in the system. If $S$ is stored on the machine $i$ to the machine $j$, then $f$ is stored on the machine $i$ to the machine $j$. At any time of the system, there can be at most one set representation $S$ of $f$ stored in the system. Furthermore, the name of $S$ is $``f"$ which is the same as the name of mapping $f$, and can be described by $O(1)$ number of words.

The total space needed to store $f$ is the total space needed to store $S$, and thus is $\Theta(|U|).$

 \paragraph{Sequence} Let $A=(a_1,a_2,\cdots,a_m)$ be a sequence of $m$ elements.
In the following, we show how to use a set to represent a sequence.
\begin{definition}[Set representation of a sequence]\label{def:set_for_sequence}
Let $A=(a_1,a_2,\cdots,a_m)$ be a sequence of $n$ elements. If a set  $S=\{(x_1,y_1),(x_2,y_2),\cdots,(x_m,y_m)\}\subseteq \mathbb{R}\times \{a_1,a_2,\cdots,a_m\}$ satisfies $x_1<x_2<\cdots<x_m,y_1=a_1,y_2=a_2,\cdots,y_m=a_m,$ then the set $S$ is a set representation of the sequence $A.$ Furthermore, if $x_1=1,x_2=2,\cdots,x_m=m,$ then $S$ is a standard set representation of $A$.
\end{definition}
Let $A$ be a sequence of elements where each element can be described by $O(1)$ number of words. Let $S$ be a set representation of the sequence $A.$ If $S$ is stored in the system, then we say $A$ is stored in the system. If $S$ is stored on the machine $i$ to the machine $j$, then $A$ is stored on the machine $i$ to the machine $j$. At any time of the system, there can be at most one set representation $S$ of $A$ stored in the system. Furthermore, the name of $S$ is $``A"$ which is the same as the name of sequence $A$, and can be described by $O(1)$ number of words.

The total space needed to store $A$ is the total space needed to store $S$, and thus is $\Theta(m).$

\subsection{Set Operations}\label{sec:set_operations}
In this section, we introduce some $\MPC$ model operations for sets.

\paragraph{Duplicates removing}
There are $n$ tuples stored in the machines. But there are some duplicates of them. The goal is to remove all the duplicates. To achieve this, we can just sort all the tuples. After sorting, if a tuple is different from its previous tuple, then we keep it. Otherwise, we remove the tuple.

\paragraph{Sizes of sets}
Suppose we have $k$ sets $S_1,S_2,\cdots,S_k$ stored in the system.
Our goal is to get the sizes of all the sets.
We can firstly sort all the tuples such that the tuples from the same set are consecutive.
Then we can calculate the index of each tuple.
Every machine can scan all the tuples in its local memory, if $x$ is an element of set $S_i$ and has the smallest/largest index $y$, then create a pair $(\text{``boundary of }`S_i\text{' ''},y).$
Then we sort all the created pairs, then for each set $S_i,$ there are two pairs $(\text{``boundary of }`S_i\text{' ''},y_1),(\text{``boundary of }`S_i\text{' ''},y_2)$ stored on the same machine.
Each machine can store its local memory. For each pair of tuples $(\text{``boundary of }`S_i\text{' ''},y_1)$, $(\text{``boundary of }`S_i\text{' ''},y_2)$ with $y_1<y_2,$ the machine can generate a new tuple $(``f",(``S_i",y_2-y_1+1)).$
Finally, there will be a mapping $f$ stored in the system, where $f(S_i)=|S_i|.$
Thus, the total number of rounds is a constant.

\paragraph{Copies of sets}
Suppose we have $k$ sets $S_1,S_2,\cdots,S_k$ stored in the system.
Let $s_1,s_2,\cdots,s_k\in \mathbb{Z}_{\geq 1}.$
If a machine holds an element $x\in S_i,$ then the machine knows the value of $s_i$.
Our goal is to create sets $S_{1,1},S_{1,2},\cdots,S_{1,s_1},S_{2,1},S_{2,2},\cdots,S_{2,s_2},\cdots,S_{k,s_k}$ and make them stored in the system, where $S_{i,j}$ is a copy of $S_i.$

The idea is very simple: for an element $x\in S_i,$ we need to make $s_i$ copies $(``S_{i,1}",x),(``S_{i,2}",x),$ $\cdots,(``S_{i,s_i}",x)$ of tuple $(``S_i",x)$.
But the issue is that $s_i$ may be very large such that it is not able to generate all the copies of a tuple on a single machine.
For the above reason, we implement it in three steps: firstly we compute the new ``position'' of each original tuple among all the copies, then send the original tuples to their new ``positions'', and finally filling the gap by generating copies between any two adjacent original tuples.
Precisely, each machine can scan its local memory, and assign each tuple $(``S_i",x)$ a weight $s_i$.
Then we can use prefix sum algorithm (See Theorem~\ref{thm:indexing}) to compute the prefix sum of each tuple $(``S_i",x)$.
The prefix sum $\pos(``S_i",x)$ of a tuple $(``S_i",x)$ denotes the new ``position'' of the last copy of this tuple when all the copies are generated.
Let $n=\sum_{i=1}^k s_i\cdot |S_i|.$
Let machine $1$ to $t$ be $t$ empty machines each maintains $s/10$ ``positions'', i.e. machine $1$ has ``positions'' $1$ to $s/10,$ machine $2$ has ``positions'' $s/10+1$ to $2s/10$, and so on.
Let $t\cdot s/10=\Theta(n).$
The machine which holds tuple $(``S_i",x)$ sends the tuple $(``S_i",x)$ to the ``position'' $\pos(``S_i",x)-s_i+1,$ and sends the tuple $(``S_{i,s_i}",x)$ to the ``position'' $\pos(``S_i",x).$
Then each machine $i\in[t]$ scans its ``positions''.
If a ``position'' received a tuple, the machine marks that ``position'' as ``$1$''.
Otherwise, the machine marks that position as ``$0$''.
Now we can apply the predecessor algorithm (See Theorem~\ref{thm:predecessor})
such that each empty ``position'' learns its predecessor tuple.
If the predecessor tuple of an empty ``position'' $l$ is $(``S_i",x),$ and the predecessor tuple is at ``position'' $l',$ then create a tuple $(``S_{i,l-l'}",x)$ at this empty position.
 Thus, at the end of all the computations, $S_{1,1},S_{1,2},\cdots,S_{1,s_1},S_{2,1},S_{2,2},\cdots,S_{2,s_2},\cdots,S_{k,s_k}$ are stored on the system.

\paragraph{Indexing elements in sets} Suppose we have $k$ sets $S_1,S_2,\cdots,S_k$ stored in the system. The goal is to compute a mapping $f$ such that $\forall i\in[k],x\in S_i,$ $x$ is the $f(S_i,x)^{\text{th}}$ element of $S_i$.

To achieve this goal, we can sort (See Theorem~\ref{thm:sorting}) all the tuples such that the elements from the same set are stored consecutively on several machines.
Then we can run indexing algorithm (See Theorem~\ref{thm:indexing}) to compute the global index of each tuple.
In the next, each machine scans its local data.
If $(``S_i",x)$ is in the local memory, and $x$ is the first element of $S_i$, then the machine marks this tuple as ``$1$''.
For other tuples in the local memory, the machine marks them as ``$0$''.
Then we can invoke predecessor algorithm (See Theorem~\ref{thm:predecessor}) on all the tuples.
At the end of the computation, each machine scans its all tuples.
For a tuple $(``S_i",x)$ with global index $l,$ the machine determine the index of $x$ in $S_i$ based on the global index $l'$ of its predecessor $(``S_i",x)$.
Precisely, the machine creates a tuple $(``f",((``S_i",x),l-l'+1))$ stored in the memory.
Thus at the end of the computation, the desired mapping $f$ is stored in the system.

\paragraph{Set merging} Suppose we need to merge several sets $S_1,S_2,\cdots,S_k$ stored on the system, i.e. create a new set $S=\bigcup_{i=1}^k S_i.$
To implement this operation, each machine scans its local memory. If there is a tuple $(``S_i",x)$ in its memory, then it creates a tuple $(``S",x).$
Finally, we just need to remove all the duplicates.

\paragraph{Set membership} Suppose we have $k$ sets $S_1,S_2,\cdots,S_k$ stored in the system.
There is an another set $Q=\{(x_1,y_1),\cdots,(x_q,y_q)\}$ also stored in the system where $x_i$ is the name of a set $S$, and $y_i$ is an item.
 The goal is to answer whether $y_i$ is in $S$.

 To achieve this, we can firstly sort all the tuples.
 For tuple with form $(``S_i",x),$ the first key is $S_i$, the second key is $x$, and the third key is $-\infty$ which has the highest priority.
 For tuple with form $(``Q",(x,y)),$ the first key is $x$, the second key is $y$, and the third key is $\infty$ which has the lowest priority.
 The comparison in the sorting procedure firstly compare the first key, then the second key, and finally the third key.
 After sorting, for each tuple with form $(``S_i",x),$ we mark it as ``$1$''.
 For each tuple with form $(``Q",(x,y)),$ we mark it as ``$0$''.
 Now we can apply the predecessor algorithm (See Theorem~\ref{thm:predecessor}).
 For each tuple $(``Q",(x,y)),$ if its predecessor is $(``S",y)$ where $x$ is the name of $``S"$, then we create a tuple $(``f",((x,y),1))$; Otherwise, we create a tuple $(``f",((x,y),0)).$
 Thus, at the end of the computation, there is a mapping $f$ stored on the system such that for each $(x,y)\in Q,$ if $x$ is the name of some set $S_i,$ and $y\in S_i,$ then $f(x,y)=1$; Otherwise $f(x,y)=0.$

\subsection{Mapping Operations}\label{sec:mapping_operations}
In this section, we introduce some $\MPC$ model operations for mapping.
The most important operation is called Multiple queries.
\paragraph{Multiple queries}
We have $k$ sets $S_1,S_2,\cdots,S_k$ stored in the system.
Without loss of generality, $S_1,S_2,\cdots,S_t$ $(t\leq k)$ are sets representations of mappings (See Definition~\ref{def:set_for_mapping}) $f_1:U_1\rightarrow H_1,f_2:U_2\rightarrow H_2,\cdots,f_t:U_t\rightarrow H_t$ respectively.
When a machine does local computation, it may need to query some values which are in the form $f_i(u)$ for some $u\in U_i.$
The following lemma shows that we can answer all the such queries simultaneously in constant number of rounds in $(0,\delta)-\MPC$ model for all constant $\delta\in(0,1)$.
It means that we can use constant number of rounds to simulate concurrent read operations on a shared memory where $S_1,\cdots,S_k$ are stored in the shared memory.
\begin{lemma}[Multiple queries]\label{lem:multi_query}
Let $\delta\in(0,1)$ be an arbitrary constant. There is a constant number of rounds algorithm $\mathcal{A}$ in $(0,\delta)-\MPC$ model which satisfies the following properties.
The input of $\mathcal{A}$ contains two parts. The first part are $k$ sets $S_1,S_2,\cdots,S_k$ stored (See Section~\ref{sec:data_organize} for data organization of sets) on the input machines, where $S_1,S_2,\cdots,S_t$ $(t\leq k)$ are sets representations of mappings (See Definition~\ref{def:set_for_mapping}) $f_1:U_1\rightarrow H_1,f_2:U_2\rightarrow H_2,\cdots,f_t:U_t\rightarrow H_t$ respectively. The second part is a set $Q=\{(x_1,y_1,z_1),(x_2,y_2,z_2),\cdots,(x_q,y_q,z_q)\}$ stored on the input machines, where $\forall (x,y,z)\in Q,$ $x$ is the name $``f_i"$ of the mapping $f_i$ for some $i\in[t],$ $y$ is an element in $U_i,$ and $z$ is the index of the input machine which holds the element $(x,y,z)$ of $Q$.
The total input size $n=|Q|+\sum_{i=1}^k|S_i|.$
The output machines are all the input machines.
$\forall i\in[k],x\in S_i,$ if the element $x$ of $S_i$ is held by the input (also output) machine $j$, then at the end of the computation, the element $x$ of $S_i$ should still be held by the output (also input) machine $j$. Let $Q'$ be the set $\{(x,y,z,w)\mid \exists(x,y,z)\in Q,w=f_i(y),\text{where~}x\text{~is the name of~}f_i\}.$ At the end of the computation, $Q'$ is stored on the output (also input) machines such that $\forall (x,y,z,w)\in Q',$ the element $(x,y,z,w)$ of $Q'$ is held by the machine $z$.
\end{lemma}
\begin{proof}
The idea is that we can firstly use sorting (See Theorem~\ref{thm:sorting}) to make queries and the corresponding values be stored consecutively in several machines.
The issue remaining is that there may be many queries queried the same position such that some queries may not be stored in the machine which holds the corresponding value.
In this case, we need to find the predecessor by invoking the algorithm shown in Theorem~\ref{thm:predecessor}.
\end{proof}
The Multiple queries algorithm is shown as the following:

\noindent\fbox{
\begin{minipage}{\linewidth}
\small
Multiple Queries Algorithm:
\begin{itemize}
\item \textbf{Setups:}
    \begin{itemize}

    \item There are $3p=\Theta(n^{\delta})$ machines indexed from $1$ to $3p$ each with local memory size $s=\Theta(n^\delta)$.

    \item The machine with index from $2p+1$ to $3p$ are input/output machines.

    \item Sets $S_1,S_2,\cdots,S_k,Q$ are stored on machine $2p+1$ to $3p$. \Comment{Corresponding to Lemma~\ref{lem:multi_query}}

    \end{itemize}

\item \textbf{The first round:}
    \begin{itemize}

     \item Machine $i\in\{2p+1,\cdots,3p\}$ scans its local memory, and send all the tuples with form $(``f_j",(x,y))$ or $(``Q",(x,y,z))$ to machine $i-p,$ where $``f_j"$ is the name of $f_j$ (also $S_j$) for $j\in[t].$ Until the end of the computation, machine $i$ sends itself messages to keep its local data.

    \end{itemize}

\item \textbf{Using constant number ($O(1/\delta)$) of rounds to sort:}
    \begin{itemize}
    \item Use machine $1$ to $2p$ to sort all the tuples stored on machine $p+1$ to $2p$, and thus at the end of this stage, machine $p+1$ to $2p$ holds sorted tuples.
 For tuple with the form $(``f_j",(x,y))$, the first key value is $``f_j",$ the second key value is $x$ and the third key value is $-\infty$ which is the highest priority. For tuple with form $(``Q",(x,y,z)),$ the first key value is $x$, the second key value is $y$, and the third key value is $\infty$ which is the lowest priority. The comparison in the sorting is: Firstly compare the first key. If they are the same, then compare the second key. If they are still the same, compare the third key.
    \end{itemize}

 \item \textbf{Using constant number ($O(1/\delta)$) of rounds to find predecessors:}
    \begin{itemize}

    \item Machine $p+1$ to $2p$ scans its local memory. For a tuple in the form $(``f_j",(x,y)),$ the machine marked it as ``$1$''. For a tuple in the form $(``Q",(x,y,z)),$  the machine marked it as ``$0$''.

    \item Machine $1$ to $2p$ together invoke the Predecessor algorithm (Theorem~\ref{thm:predecessor}), where the input is on machine $p+1$ to machine $2p$.

    \end{itemize}

 \item \textbf{The last round:}
    \begin{itemize}

    \item Machine $p+1$ to $2p$ scans its local memory. For each tuple with form $(``Q",(x,y,z)),$ it sends machine $z$ a tuple $(``Q'",(x,y,z,w)),$ where $x$ is the name of $f_j$, and $w=f_j(y).$

    \end{itemize}

\end{itemize}
\end{minipage}
}

\subsection{Sequence Operations}
In this section, we introduce some $\MPC$ model operations for sequence.

\paragraph{Sequence standardizing} Suppose there is a sequence $A$, and one of its set representation (see Definition~\ref{def:set_for_sequence}) $S$ is stored in the system. The goal is to modify the set $S$ such that $S$ is a standard set representation of $A$.

We can compute the index (see \textbf{Indexing elements in sets} in Section~\ref{sec:set_operations}) of elements in $S$.
Then for each element $(x,y)\in S,$ we can query (see \textbf{Multiple queries} in Section~\ref{sec:mapping_operations}) the index of $(x,y)$ in $S$.
Suppose the index is $i$, we modify the tuple $(``S",(x,y))$ to $(``S",(i,y)).$

\paragraph{Sequence duplicating} Suppose there is a sequence $A=(a_1,a_2,\cdots,a_s)$, and one of its set representation (see Definition~\ref{def:set_for_sequence}) $S$ is stored in the system. Furthermore, there is a mapping $f:[s]\rightarrow\mathbb{Z}_{\geq 0}$ which is also stored in the system. The goal is to get a set $S'$ stored in the system such that $S'$ is a set representation of the sequence:
\begin{align*}
(\underbrace{a_1,a_1,\cdots,a_1}_{f(1)\text{~times}},\underbrace{a_2,a_2,\cdots ,a_2}_{f(2)\text{~times}},\cdots,\underbrace{a_s,a_s,\cdots,a_s}_{f(s)\text{~times}}).
\end{align*}

Firstly, we can standardize (see the above paragraph \textbf{Sequence standardizing}) the set $S.$
Then for each tuple $(``S",(i,a_i)),$ we create a tuple $(``S_i",a_i),$ and we can query (see \textbf{Multiple queries} in Section~\ref{sec:mapping_operations}) the value of $f(i).$
Then we can copy (see \textbf{Copies of sets} in Section~\ref{sec:set_operations}) set $S_i$ $f(i)$ times.
For each tuple $(``S_{i,j}",a_i),$ we create a tuple $(``S'",((i,j),a_i)).$
Then we can compute the index (see \textbf{Indexing elements in sets} in Section~\ref{sec:set_operations}) of each element in $S'$.
For each tuple $(``S'",((i,j),a_i)),$ we can query (see \textbf{Multiple queries} in Section~\ref{sec:mapping_operations}) the index $i'$ of it, and then modify the tuple as $(``S'",(i',a_i)).$

\paragraph{Sequence insertion} Suppose there are $k+1$ sequences $A=(a_1,a_2,\cdots,a_s),A_1,\cdots,A_k$ which have sets representations (see Definition~\ref{def:set_for_sequence}) $S,S_1,\cdots,S_k$ respectively and stored on the system. There is also a mapping $f:[k]\rightarrow \{0\}\cup[s]$ stored on the system where $\forall i\not=j\in [k],f(i)\not=f(j).$ The goal is to insert each sequence $A_i$ into the sequence $A$, and $A_i$ should be between the element $a_{f(i)}$ and $a_{f(i)+1}.$

Firstly, we can standardize (see \textbf{Sequence standardizing} in Section~\ref{sec:set_operations}) $S$.
Then we can compute the total size (see \textbf{Sizes of sets} in Section~\ref{sec:set_operations}) $N=|S|+|S_1|+\cdots+|S_k|+1.$
For each tuple $(``S",(i,a_i)),$ we can modify it as $(``S",(i\cdot N,a_i)).$
For each tuple $(``S_i",(j,a_{ij})),$ we query (see \textbf{Multiple queries} in Section~\ref{sec:mapping_operations}) the value of $f(i),$ then create a tuple $(``S",(f(i)\cdot N+j,a_{ij})).$

\subsection{Multiple Tasks}\label{sec:multi_tasks}
In this section, we show that if the entire computational tasks consist of some independent small computational tasks, then we are able to schedule the machines such that the small computational tasks can be computed simultaneously.

\paragraph{Task and multiple tasks problem} A computational task here is running a specific algorithm on specific input data.

There are $k$ sets $S_1,S_2,\cdots,S_k$ stored in the system.
Let $n=\sum_{i=1}^k |S_i|$ be the total input size.
There are $h$ independent computational tasks $T_1,T_2,\cdots,T_h.$
Each task $T_i$ needs to take some sets $\mathcal{S}_i\subseteq \{S_1,S_2,\cdots,S_k\}$ as its input, and is running a $(\gamma_i,\delta_i)-\MPC$ algorithm in $r_i$ rounds where $\gamma_i\in\mathbb{R}_{\geq 0},$ constant $\delta_i\in(0,1)$.
$\forall i\in[h],$ let $n_i=\sum_{S\in \mathcal{S}_i}|S|$ be the input size of task $T_i.$
Without loss of generality, we can assume that the input of different tasks are disjoint.
Otherwise we can use sets copying technique (See Section~\ref{sec:set_operations}) to generate different copies of input sets for the tasks shared the same input set.
The goal here is to use the small number of rounds to finish all the tasks.
Since we can always use sorting and indexing to extract the desired input data.
The most naive way is to compute the tasks one-by-one.
This can be trivially done in $r=O(\sum_{i=1}^h r_i)$ rounds in $(\gamma,\delta)-\MPC$ model for $\gamma=\log_{n}(h)+\max_{i\in[h]} \gamma_i,\delta=\max_{i\in[h]} \delta_i.$
Here we show how to compute all the tasks simultaneously in $r=O(\max_{i\in[h]} r_i)$ rounds in $(\gamma,\delta)-\MPC$ model for $\gamma=\log_{n}(m)-1,\delta=\max_{i\in[h]} \delta_i,$ where $m=\Theta(n+\sum_{i=1}^h n_i^{1+\gamma_i}).$

Each machine scans its local memory.
If the machine holds a tuple $(``S_i",x)$, and $S_i$ is a part of input of task $T_j,$ then it creates a tuple $(``W_j", (``S_i",x)).$
Thus, at the end of this step, there are additional $h$ sets $W_1,W_2,\cdots,W_h$ stored in the system.
Here $W_i,i\in[h]$ contains all the information of input data of task $T_i.$
Then we can compute a mapping $f$ such that $\forall i\in[h],f(W_i)=|W_i|$ (see Section~\ref{sec:set_operations}).
Thus, we know the input size of each task.
Then each machine scans its local memory.
If the machine holds a tuple $(``f",(``W_i",|W_i|)),$ then it creates a tuple $(``H_i",|W_i|),$ i.e. a set $H_i=\{|W_i|\}.$
Then for each set $H_i=\{|W_i|\},i\in[h],$ we can copy (see Section~\ref{sec:set_operations}) it $s_i=c\cdot |W_i|^{1+\gamma_i}$ times for a sufficiently large $c$ to get sets $H_{i,1}=H_{i,2}=\cdots=H_{i,s_i}=|W_i|$. Each set $H_{i,j}$ is just a placeholder of one unit working space of the task $T_i.$
Thus, the number of copies of the set $H_i$ is the total space needed for the task $T_i.$
We can sort all the tuples $(``H_{i,j}",|W_i|)$ on machines with index in $I=\{2,5,8,11,\cdots,3p-1\},$ where local memory $s=\Theta(n^{\delta}),$ total required memory $m=\Theta(n+\sum_{i=1}^h n_i^{1+\gamma_i})$, and $p=\Theta(m/s)$
For each machine with index $q\in I,$ the tuples on that machine must be in the following form
\begin{align*}
&(``H_{i,j}",|W_i|),(``H_{i,j+1}",|W_i|),\cdots,(``H_{i,s_i}",|W_i|),(``H_{i+1,1}",|W_{i+1}|),\cdots,(``H_{i+1,s_{i+1}}",|W_{i+1}|),\\
&(``H_{i+2,1}",|W_{i+2}|),\cdots,(``H_{i+2,s_{i+2}}",|W_{i+2}|),\cdots, (``H_{i',1}",|W_{i'}|),\cdots(``H_{i',j'}",|W_{i'}|).
\end{align*}
Then machine $q$ just sends all the tuples
$
(``H_{i,j}",|W_i|),(``H_{i,j+1}",|W_i|),\cdots,(``H_{i,s_i}",|W_i|)
$
to machine $q-1$, and sends all the tuples
$
(``H_{i',1}",|W_{i'}|),(``H_{i',2}",|W_{i'}|),\cdots(``H_{i',j'}",|W_{i'}|)
$
to machine $q+1$.
Thus, $\forall i\in[h],$
\begin{enumerate}
\item either all the $H_{i,1},H_{i,2},\cdots,H_{i,s_i}$ are stored on consecutive machines, machine $q$ to machine $q'$, and any of machine $q$ to machine $q'$ does not hold other tuples,
\item or there is a unique machine $q$ which holds all the sets $H_{i,1},H_{i,2},\cdots,H_{i,s_i}.$
\end{enumerate}
For each machine $q\in[3p],$ if $H_{i,1}$ is held by machine $q$, then it creates a tuple $(``\st",(``T_i",q)).$
If $H_{i,s_i}$ is held by machine $q$, then it creates a tuple $(``\ed",(``T_i",q)).$
The mapping $\st,\ed$ then are stored in the system, where $\st(T_i)$ is the index of the first machine assigned to task $T_i,$ and $\ed(T_i)$ is the index of the last machine assigned to task $T_i.$
Recall that $W_i$ contains all the information of the input data to task $T_i.$
The remaining task is to move the input data of task $T_i$ to the machines with index from $\st(T_i)$ to $\ed(T_i).$
According to Section~\ref{sec:set_operations}, we can compute a mapping $f',$ such that $f'(W_i,x)$ records the index of $x\in W_i$ in set $W_i$.
Now, each machine scans its local memory.
For each tuple $(``W_j", (``S_i",x)),$ the machine needs to query the value of $f'(W_j,(``S_i",x)),$ the value of $\st(T_j)$ and the value of $\ed(T_j).$
By Lemma~\ref{lem:multi_query}, these queries can be handled simultaneously in constant number of rounds.
Then the machine can send the tuple $(``S_i",x)$ to the corresponding machine based on the value of $f'(W_j,(``S_i",x)),$ $\st(T_j),$ and $\ed(T_j).$
Finally, $\forall i\in[h],$ since $\delta\geq \delta_i$ and $(\ed(T_i)-\st(T_i)+1)\cdot s=\Theta(n_i^{1+\gamma_i}),$ the machines with index from $\st(T_i)$ to $\ed(T_i)$ can simulate task $T_i$ in $r_i$ number of rounds.

\section{Implementations in $\MPC$ Model}
\label{sec:implement}
In this section, we show how to implement all the previous batch algorithms in $\MPC$ model.

\subsection{Neighbor Increment Operation}

\begin{lemma}\label{lem:parallel_implement_of_neighbor_incr}
Let graph $G=(V,E),n=|V|,N=|V|+|E|$ and $m=\Theta(N^{\gamma})$ for some arbitrary $\gamma\in[0,2].$
$\textsc{NeighborIncrement}(m,G)$ (Algorithm~\ref{alg:neighbor_increment}) can be implemented in $(\gamma,\delta)-\MPC$ model for any constant $\delta\in(0,1).$ Furthermore, the parallel running time is $O(r),$ where $r$ is the number of iterations (see Definition~\ref{def:neighbor_incr_num_iter}) of $\textsc{NeighborIncrement}(m,G)$.
\end{lemma}
\begin{proof}
To implement line~\ref{sta:init_S2}, we can create a tuple $(``S_v^{(0)}",u)$ for each tuple $(``E",(v,u)).$
Then for each $(``S_v^{(0)}",u)$ we can compute the index (see \textbf{Indexing elements in sets} and \textbf{Multiple queries}) of $u$ in set $S_v^{(0)}.$
If the index of $u$ in set $S_v^{(0)}$ is at least $\lceil(m/n)^{1/2}\rceil,$ then delete $u$ from $S_v^{(0)},$ i.e. delete the tuple $(``S_v^{(0)}",u)$.

Now let us discuss how to implement line~\ref{sta:large_condition} and line~\ref{sta:alreadygood2} in the $i^{\text{th}}$ iteration.
Firstly, we can compute the size of every set stored in the system (see \textbf{Sizes of sets}).
Then for each tuple $(``S_v^{(i-1)}",u),$ the corresponding machine queries (see \textbf{Multiple queries}) the size of $S_u^{(i-1)}.$
If $|S_u^{(i-1)}|\geq \lceil(m/n)^{1/2}\rceil,$ then create a tuple $(``\mathrm{temp}_v^i",u).$
We can index (see \textbf{Indexing elements in sets}) all the elements in set $\mathrm{temp}_v^i,$ and only keep the element with index $1$.
Thus, $\mathrm{temp}_v^i$ has a only element $u$, and we need to create a set $S_v^{(i)}=S_u^{(i-1)}.$
Notice that there may be many $v\in V$ which needs need to implement $S_v^{(i)}=S_u^{(i-1)}.$
Thus, for each tuple $(``\mathrm{temp}_v^i",u),$ we create a tuple $(``\mathrm{target}_u^i",v).$
 $v\in\mathrm{target}_u^i$ means that $S_v^{(i)}$ needs a copy of $S_u^{(i-1)}.$
 Thus, $|\mathrm{target}_u^i|$ means that $S_u^{(i-1)}$ needs to copy $|\mathrm{target}_u^i|$ times.
 For each tuple $(``S_u^{(i-1)}",x),$ the machine queries (see \textbf{Multiple queries}) the size of $\mathrm{target}_u^i.$
 Then each set $S_u^{(i-1)}$ can be copied (see \textbf{Copies of sets}) $|\mathrm{target}_u^i|$ times.
 For each tuple $(``\mathrm{target}_u^i",v),$ we query (see \textbf{Multiple queries}) the index (see \textbf{Indexing elements in sets}) of $v$ in set $\mathrm{target}_u^i,$ and then create a tuple $(``f^i",((``\mathrm{target}_u^i",x),v)),$ where $x$ is the index of $v$ in $\mathrm{target}_u^i.$
 Thus $f^i$ is a mapping such that $f^i(\mathrm{target}_u^i,x)$ is the $x^{\text{th}}$ element in $\mathrm{target}_u^i$.
  For each tuple $(``S_{u,j}^{(i-1)}",x),$ we query (see \textbf{Multiple queries}) the value $v=f^i(\mathrm{target}_u^i,j),$ and then create a tuple $(``S_v^{(i)}",x),$ and a tuple $(``S_v^{(i)}",v).$
  We then remove the duplicates (see \textbf{Duplicates removing}) of elements of for every set $S_v^{(i)}.$
  For each tuple $(``\mathrm{temp}_v^i",u),$ query (see \textbf{Multiple queries}) the size (see \textbf{Sizes of sets}) of $S_v^{(i)}$ and $S_u^{(i-1)}$. If $|S_v^{(i)}|>|S_u^{(i-1)}|,$ then we create a tuple $(``g^i",(v,(u,``\text{delete}")))$; Otherwise, create a tuple $(``g^i",(v,(u,``\text{keep}"))).$
  Finally, for each tuple $(S_v^{(i)},x),$ we query (see \textbf{Multiple queries}) $(u,o)=g^i(v),$ if $u=x$ and $o=\text{``delete"},$ the machine deletes the tuple $(S_v^{(i)},x).$

  Next, let us discuss how to implement line~\ref{sta:stillbad}.
  Similar as before, we can compute the size of every set stored in the system (see \textbf{Sizes of sets}).
Then for each tuple $(``S_v^{(i-1)}",u),$ the corresponding machine queries (see \textbf{Multiple queries}) the size of $S_u^{(i-1)}.$
If $|S_u^{(i-1)}|\geq \lceil(m/n)^{1/2}\rceil,$ then create a tuple $(``\mathrm{temp}_v^i",u).$
For each tuple $(``V",v),$ we can create a tuple $(``\mathrm{temp}_v^i",\mathrm{null}).$
Then for each tuple $(``V",v)$ we can query (see \textbf{Multiple queries}) the size (see \textbf{Sizes of sets}) of $\mathrm{temp}_v^i.$
If $|\mathrm{temp}_v^i|=1,$ then we create a tuple $(``f'^i",1)$; Otherwise, we create a tuple $(``f'^i",0).$
Thus, mapping $f'^i$ is stored in the system, and $f'^i(v)=1$ if and only if $\forall u\in S_v^{(i-1)},|S_u^{(i-1)}|< \lceil(m/n)^{1/2}\rceil.$
For each tuple $(``S^{(i-1)}_v",u),$ we query (see \textbf{Multiple queries}) the value $f'^i(v).$
If $f'^i(v)=1,$ we create a tuple $(``\mathrm{target}_u^i",v).$
Thus, $v\in\mathrm{target}_u^i$ means that $S^{(i-1)}_u$ should be a part of $S^{(i)}_v.$
$|\mathrm{target}_u^i|$ means that $S_u^{(i-1)}$ needs to copy $|\mathrm{target}_u^i|$ times.
 For each tuple $(``S_u^{(i-1)}",v),$ we query (see \textbf{Multiple queries}) the size (see \textbf{Sizes of sets}) of $\mathrm{target}_u^i.$
 Then we can copy (see \textbf{Copies of sets}) each set $S_u^{(i-1)}$ $|\mathrm{target}_u^i|$ times.
 Then for each tuple $(``\mathrm{target}_u^i",v),$ we can query (see \textbf{Multiple queries}) the index $x$ (see \textbf{Indexing elements in sets}) of $v$ in set $\mathrm{target}_u^i$, and then create a tuple $(``f^i",((``\mathrm{target}_u^i",x),v))$ which means that the $x^{\text{th}}$ element of $\mathrm{target}_u^i$ is $f^i(\mathrm{target}_u^i,x)=v.$
 For each tuple $(``S_{u,j}^{(i-1)}",x),$ we query (see \textbf{Multiple queries}) the value $v=f^i(\mathrm{target}_u^i,j),$ and then create a tuple $(``S_v^{(i)}",x).$
  We then remove the duplicates (see \textbf{Duplicates removing}) of elements of for every set $S_v^{(i)}.$

  Finally, let us consider how to implement line~\ref{sta:condition_break}.
  It is very simple, we only need to query the sizes of sets.
  For each tuple $(``V",v),$ query (see \textbf{Multiple queries}) the size (see \textbf{Sizes of sets}) of $S_v^{(i)}$ and $S_v^{(i-1)},$ if $v$ satisfies the condition, create a tuple $(``Done",v).$
  Every machine queries (see \textbf{Multiple queries}) the size (see \textbf{Sizes of sets}) of $\mathrm{Done}$.
  If it is $|V|,$ then all the machines know that they finish the loop.
  In the end, for each tuple $(S_v^{(r)},u)$ we create tuples $(``E'",(u,v)),(``E'",(v,u)),$ and for each tuple $(``E",(u,v))$ we create tuple $(``E'",(u,v)).$
  Then we then remove the duplicates (see \textbf{Duplicates removing}) of elements of $E$.

  In the $i^{\text{th}}$ iteration, we only need to maintain sets $V,E,S_v^{(i-1)}.$
  Since all the copy operation will create at most $ n\cdot (m/n)^{1/2}\cdot (m/n)^{1/2}=m $ tuples, the total space needed is $\Theta(m)$ plus the space needed to maintain $V,E,S_v^{(i)}.$
  By Property~\ref{itm:neighbor_S_pro5} of Lemma~\ref{lem:properties_of_S}, $|S_v^{(i)}|\leq m/n.$
  Thus, the total space is $\Theta(m)+|V|+|E|+\sum_{v\in V}|S_v^{(i)}|=\Theta(m)+N=\Theta(m).$

  The above implementation shows that the parallel time is $O(r),$ where $r$ is the number of iterations (see Definition~\ref{def:neighbor_incr_num_iter}).
\end{proof}

\subsection{Tree Contraction Operation}
In this section, we show how to implement Algorithm~\ref{alg:tree_contraction} in $\MPC$ model.

\begin{lemma}\label{lem:parallel_implement_of_tree_contract}
Let graph $G=(V,E)$ and $\p:V\rightarrow V$ be a set of parent points (see Definition~\ref{def:parent_pointers}) on the vertex set $V$.
$\textsc{TreeContraction}(G,\p)$ (Algorithm~\ref{alg:tree_contraction}) can be implemented in $(0,\delta)-\MPC$ model for any constant $\delta\in(0,1).$ Furthermore, the parallel running time is $O(r),$ where $r$ is the number of iterations (see Definition~\ref{def:num_it_tree_contract}) of $\textsc{TreeContraction}(G,\p)$.
\end{lemma}
\begin{proof}
Let $N=|V|+|E|.$
Then the total space is $\Theta(N).$

Initially, each machine scans its local memory.
If there is a tuple $(``V",v),$ then it queries the value of $\p(v).$
It needs $O(1)$ parallel time to answer all the queries (see Multiple queries in Lemma~\ref{lem:multi_query}).
Then the machine creates a tuple $(``g^{(0)}",(v,\p(v))).$
Thus, in the initialization stage, mapping $g^{(0)},\p$, set $V,E$ are stored in the system.

In the $l^{\text{th}}$ iteration,
Each machine scans its local memory.
If there is a tuple $(``V",v),$ then it queries the value of $g^{(l-1)}(v).$
This can be done by Multiple queries.
Then it queries the value of $\p(g^{(l-1)}(v))$.
This can also be done by Multiple queries.
If $\p(g^{(l-1)}(v))=g^{(l-1)}(v),$ it creates a tuple $(``\mathrm{Done}",v).$
Then the machines can compute the sizes (see Section~\ref{sec:set_operations}) of $V$ and $\mathrm{Done}.$
Each machine queries the size of $V$ and $\mathrm{Done}.$
This can be done by Multiple queries.
Then if $|V|=|\mathrm{Done}|,$ every machine knows that the iterations are finished.
Otherwise, the machine which holds $(``V",v)$ queries the value of $g^{(l-1)}(g^{(l-1)})(v).$
This can be done by Multiple queries.
And then it creates a tuple $(``g^{(l)}",(v,g^{(l-1)}(g^{(l-1)})(v))).$

At the end, if a machine holds a tuple $(``V",v),$ then the queries $\p(v).$
If $v=\p(v),$ it creates a tuple $(``V' ",v).$
If a machine holds a tuple $(``E",(u,v)),$ then it queries $g^{(r)}(u),g^{(r)}(v),$ and creates a tuple $(``E' ",(g^{(r)}(u),g^{(r)}(v))).$

Since at the end of each iteration $l$, the system only stores mappings $\p:V\rightarrow V,g^{(r)}:V\rightarrow V,$ and sets $V,E,$ the total space used is at most $O(N).$
Thus, we can implement the algorithm in $(0,\delta)-\MPC$ model.

The total parallel time is $O(r).$
By Corollary~\ref{cor:tree_contract_conn}, $r=O(\dep(\p)).$
Thus, the total parallel time is $O(\dep(\p))$.
\end{proof}

\subsection{Graph Connectivity}

\begin{theorem}\label{thm:parallel_connectivity}
Let graph $G=(V,E),n=|V|,N=|V|+|E|$ and $m=\Theta(N^{\gamma})$ for some arbitrary $\gamma\in[0,2].$
Let $r>0$ be a round parameter.
$\textsc{Connectivity}(G,m,r)$ (Algorithm~\ref{alg:batch_algorithm2}) can be implemented in $(\gamma,\delta)-\MPC$ model for any constant $\delta\in(0,1).$ Furthermore, the parallel running time is $O(R),$ where $R$ is the total number of iterations (see Definition~\ref{def:total_iter_connectivity}) of $\textsc{Connectivity}(G,m,r).$
\end{theorem}
\begin{proof}
Initially, we store sets $V_0,E_0,V,E$ and mapping $h_0$ in the system.
Now consider the $i^{\text{th}}$ round.
Due to Lemma~\ref{lem:parallel_implement_of_neighbor_incr}, line~\ref{sta:alg2_neighbor_incr} can be implemented in total space $\Theta(m)$ and with $O(k_i)$ parallel time, where $k_i$ is the number of iterations (See Definition~\ref{def:neighbor_incr_num_iter}) of $\textsc{NeighborIncrement}(m,G_{i-1}).$
  To store $V_i'$ and $E_i',$ we need total space $\Theta(m).$
  Line~\ref{sta:alg2_Vdoubleprime_create} can be implemented by operations described in \textbf{Sizes of sets} and \textbf{Multiple queries} (see Section~\ref{sec:MPCmodel}).
  Line~\ref{sta:alg2_Edoubleprime_create} can be implemented by the operations described in \textbf{Set membership} and \textbf{Multiple queries}.
  To implement line~\ref{sta:alg2_sampling}, for each tuple $(``V''_i",v),$ we can create a tuple $(``l_i",(v,x))$ where $x=1$ with probability $p_i$, $x=0$ with probability $1-p_i$.
  To calculate $p_i,$ the machine only needs to know $n_{i-1}.$
  This can be done by the operations described in \textbf{Sizes of sets} and \textbf{Multiple queries}.
  Line~\ref{sta:alg2_leader_set} and line~\ref{sta:alg2_assign_pointer1} can be implemented by operations described in \textbf{Set membership} and \textbf{Multiple queries}.
  For line~\ref{sta:alg2_assign_pointer2}, set $L_i\cap (\Gamma_{G'_i}(v)\cup \{v\})$ can be computed by operations described in \textbf{Set membership} and \textbf{Multiple queries}.
  Then, by operations in \textbf{Indexing elements in sets} and \textbf{Multiple queries}, we can get $\min_{u\in L_i\cap (\Gamma_{G'_i}(v)\cup \{v\})} u$.
  Finally, by operation described in \textbf{Multiple queries}, $\forall v\in V''_i$ with $v\not\in L_i,$  the tuple $(``\p_i",(v,x))$ can be created, where $x=\min_{u\in L_i\cap (\Gamma_{G'_i}(v)\cup \{v\})} u.$
  Due to Lemma~\ref{lem:parallel_implement_of_tree_contract}, line~\ref{sta:alg2_tree_contract} can be implemented in total $\Theta(m)$ space and $O(r_i')$ parallel running time, where $r_i'$ is the number of iterations (see Definition~\ref{def:num_it_tree_contract}) of $\mathrm{TreeContraction}(G_i'',\p_i)$.
  Line~\ref{sta:VspminusVdp} can be implemented by operations in \textbf{Set membership}, \textbf{Indexing elements in sets} and \textbf{Multiple queries}.
  Line~\ref{sta:VdpminusV} can be implemented by operations in \textbf{Set membership} and \textbf{Multiple queries}.
  Line~\ref{sta:hlasttime} can be implemented by \textbf{Multiple queries}.
  For other $v\in V$ with $h_i(v)=\mathrm{null}$ assigned by line~\ref{sta:alg2_init_null}, we can use the operations in \textbf{Set membership} and \textbf{Multiple queries} to find those $v$, and create a tuple $(``h_i",v,\mathrm{null}).$

  Thus, in the $i^{\text{th}}$ round, the parallel time needed is $O(k_i+r'_i).$
  At the end of the $i^{\text{th}}$ round, we only need to keep sets $V_i,E_i,V,E$ and mapping $h_i$ in the system.
  It will take total space at most $O(m).$

  Due to Lemma~\ref{lem:parallel_implement_of_tree_contract}, line~\ref{sta:final_output_alg2} can be implemented in at most $O(m)$ total space and $O(\log r)$ parallel time.

  Thus, the total parallel time is $O(\log r+\sum_{i=1}^r (k_i+r'_i))=O(\sum_{i=1}^r (k_i+r'_i))$.
  By definition~\ref{def:total_iter_connectivity}, the total parallel time is $O(R),$ where $R$ is the total number of iterations of $\textsc{Connectivity}(G,m,r)$.
  The total space in the computation is always at most $\Theta(m)$.
\end{proof}

Here, we are able to conclude the following theorem for graph connectivity problem.
\begin{theorem}\label{thm:formal_statement_of_connectivity}
For any $\gamma\in[0,2]$ and any constant $\delta\in(0,1),$ there is a
randomized $(\gamma,\delta)-\MPC$ algorithm (see Algorithm~\ref{alg:batch_algorithm2}) which can output the
connected components for any graph $G=(V,E)$ in $O(\min(\log D\cdot
\log (1/\gamma'),\log n))$ parallel time, where $D$ is the
diameter of $G,$ $n=|V|,$ $N=|V|+|E|$ and $\gamma'=(1+\gamma)\log_n
\frac{2N}{n^{1/(1+\gamma)}}.$ The success probability is at least
$0.98.$ In addition, if the algorithm fails, then it will return
{\rm FAIL}.
\end{theorem}
\begin{proof}
The implementation of Algorithm~\ref{alg:batch_algorithm2} in $\MPC$ model is shown by Theorem~\ref{thm:parallel_connectivity}.
The correctness of Algorithm~\ref{alg:batch_algorithm2} is proved by Theorem~\ref{thm:correct_rand_leader_alg}.
The total parallel time of Algorithm~\ref{alg:batch_algorithm2} is proved by Theorem~\ref{thm:alg2_time_prob}.
\end{proof}

\subsection{Algorithms for Local Shortest Path Trees}

In this section, we mainly explained how to implement local shortest path tree algorithms described in Section~\ref{sec:local_short_tree} and Section~\ref{sec:multi_local_short_tree}.

\begin{lemma}\label{lem:parallel_implement_of_tree_expand}
Let $G=(V,E)$ be an undirected graph, $s_1,s_2\in\mathbb{Z}_{\geq 0},$ and $v\in V.$ Let $\wt{T}=(V_{\wt{T}},\p_{\wt{T}})$ with root $v$ and radius $s_1$ be a local complete shortest path tree (see Definition~\ref{def:local_complete_tree}) in $G,$ and $\dep_{{\wt{T}}}:V_{\wt{T}}\rightarrow\mathbb{Z}_{\geq 0}$ be the depth of every vertex in $\wt{T}$. $\forall u\in V_{\wt{T}},$ let $T(u)$ with root $u$ and radius $s_2$ be a local complete shortest path tree in $G,$ and $\dep_{{T(u)}}:V_{T(u)}\rightarrow \mathbb{Z}_{\geq 0}$ be the depth of every vertex in $T(u)$. Then $\textsc{TreeExpansion}(\wt{T},\dep_{{\wt{T}}},\{T(u)\mid u\in V_{\wt{T}}\},\{\dep_{{T(u)}}\mid u\in V_{\wt{T}}\})$ (Algorithm~\ref{alg:merge_short_tree}) can be implemented in $(0,\delta)-\MPC$ model for any constant $\delta\in(0,1)$ in $O(1)$ parallel time.
\end{lemma}
\begin{proof}
For line~\ref{sta:lcspt_merge_vertex}, we apply operation shown in \textbf{Copies of sets} to copy each $V_{T(u)},$ then we can merge (see \textbf{Set merging}) all the copies to get $V_{\wh{T}}.$
To implement line~\ref{sta:lcspt_original_pointer} and line~\ref{sta:lcspt_original_depth}, we only need to apply the operation shown in \textbf{Multiple queries}.
To implement line~\ref{sta:lcspt_find_ux}, for each tuple $(``V_{T(u)}",x),$ we can firstly check whether $x\in V_{\wh{T}}\setminus V_{\wt{T}}$ by operations described in \textbf{Set membership} and \textbf{Multiple queries}.
If $x\in V_{\wh{T}}\setminus V_{\wt{T}},$ then we can query the values of $\dep_{{\wt{T}}}(u)$ and $\dep_{{T(u)}}(x)$ by operations shown in \textbf{Multiple queries}.
Then we create a tuple $(``\mathrm{temp}_x",(\dep_{{\wt{T}}}(u)+\dep_{{T(u)}}(x),u)).$
By \textbf{Indexing elements in sets} and \textbf{Multiple queries}, we can find the element with the smallest index in set $\mathrm{temp}_x,$ and thus that element is $(\dep_{{\wt{T}}}(u_x)+\dep_{{T(u_x)}}(x),u_x).$
Finally, the remaining things in line~\ref{sta:lcspt_find_ux} and line~\ref{sta:lcspt_compute_depth} can be done by the operations described by \textbf{Multiple queries}.

For all the operations, the total space is always linear.
The parallel time needed for the above operations is also a constant.
\end{proof}

\begin{lemma}\label{lem:parallel_implement_of_multi_radius}
Let graph $G=(V,E),n=|V|,N=|V|+|E|$ and $m=\Theta(N^{\gamma})$ for some arbitrary $\gamma\in[0,2].$ $\textsc{MultiRadiusLCSPT}(G,m)$ (Algorithm~\ref{alg:doubling_expansion}) can be implemented in $(\gamma,\delta)-\MPC$ model for any constant $\delta\in(0,1).$ Furthermore, the parallel running time is $O(r),$ where $r$ is the number of iterations (see Definition~\ref{def:num_iter_multi_radius_lcspt}) of $\textsc{MultiRadiusLCSPT}(G,m).$
\end{lemma}
\begin{proof}
To implement line~\ref{sta:alg_exp_init1} to line~\ref{sta:alg_exp_init2point5}, we can scan all the tuples $(``E",(u,v)),$ then query the size of $\{v\}\cup\Gamma_G(v)$ and the size of $\{u\}\cup \Gamma_G(u),$ where these operations are described in \textbf{Sizes of sets} and \textbf{Multiple queries.}
Then based on the sizes, we decide whether we need to create the corresponding tuples for $V_{T_0(v)},V_{T_0(u)},\p_{T_0(u)},\p(T_0(v)).$

Now consider the main loop.
We focus on the $i^{\text{th}}$ round.
Line~\ref{sta:alg_exp_set_null1} can be implemented by the operation described in \textbf{Multiple queries}.
To implement line~\ref{sta:alg_exp_set_null2}, for each tuple $(``V_{T_{i-1}(v)}",u),$ we can query (see \textbf{Multiple queries}) whether $T_{i-1}(u)$ is $\mathrm{null}.$
If $T_{i-1}(u)$ is $\mathrm{null},$ then we create a tuple $(``\mathrm{temp}_v",u).$
Then for each tuple $(``V",v),$ we can query the size of $\mathrm{temp}_v$ by operations described in \textbf{Sizes of sets} and \textbf{Multiple queries}.
If the size is not $0$, then $T_i(v)$ must be $\mathrm{null}.$
Line~\ref{sta:alg_exp_double} can be implemented by coping input for different tasks and running tasks in parallel, where it only needs operations shown in \textbf{Copies of sets}, \textbf{Multiple queries} and \textbf{Multiple Tasks} (see Section~\ref{sec:multi_tasks}).
According to Lemma~\ref{lem:parallel_implement_of_tree_expand}, it only needs $O(1)$ parallel time.
Line~\ref{sta:alg_exp_set_null} and line~\ref{sta:alg_exp_condition} only need the operation shown in \textbf{Multiple queries}.

Thus, the total parallel time is $O(r)$ where $r$ is the number of iterations (see Definition~\ref{def:num_iter_multi_radius_lcspt}) of $\textsc{MultiRadiusLCSPT}(G,m).$
For the total space, we stored the sets $V_{T_i(v)}$ for all $i\in[r],v\in V$ and mappings $\p_{T_i(v)},\dep_{T_i(v)}$ for all $i\in[r],v\in V$.
By Lemma~\ref{lem:multi_radius_trees}, the total space to store all of them is at most $O(r\cdot n\cdot (m/n)^{1/4})=O(m).$
In the $i^{\text{th}}$ round of the main loop, line~\ref{sta:alg_exp_double} may make copies of the set.
By Lemma~\ref{lem:multi_radius_trees}, the input size of each task will be at most $O((m/n)^{1/4}\cdot (m/n)^{1/4}).$
Since the there are at most $n$ tasks, the total space needed is at most $O(m).$
\end{proof}

\begin{lemma}\label{lem:parallel_implement_of_multiple_large_tree}
Let graph $G=(V,E),n=|V|,N=|V|+|E|$ and $m=\Theta(N^{\gamma})$ for some arbitrary $\gamma\in[0,2].$ $\textsc{MultipleLargeTrees}(G,m)$ (Algorithm~\ref{alg:maximal_short_tree}) can be implemented in $(\gamma,\delta)-\MPC$ model for any constant $\delta\in(0,1).$ Furthermore, the parallel running time is $O(r),$ where $r$ is the number of iterations (see Definition~\ref{def:num_it_multiplelargetree}) of $\textsc{MultipleLargeTrees}(G,m).$
\end{lemma}
\begin{proof}
By Lemma~\ref{lem:parallel_implement_of_multi_radius}, line~\ref{sta:alg_maximal_invoke_multi} can be implemented in total space $m$ and $O(r)$ parallel time where $r$ is the number of iterations (see Definition~\ref{def:num_iter_multi_radius_lcspt}) of $\textsc{MultiRadiusLCSPT}(G,m).$
Line~\ref{sta:alg_maximal_init0} to line~\ref{sta:alg_maximal_init2} can be implemented by the operation described by \textbf{Multiple queries}.
The implementation of line~\ref{sta:alg_maximal_loop_start} to line~\ref{sta:alg_maximal_loop_end} is similar as the implementation of the main loop of Algorithm~\ref{alg:doubling_expansion} (See Lemma~\ref{lem:parallel_implement_of_multi_radius} for details of the implementation).
The implementation of line~\ref{sta:alg_maximal_create_N1} and line~\ref{sta:alg_maximal_create_N2} only needs the operation described in \textbf{Indexing elements in sets} and \textbf{Multiple queries}.
Line~\ref{sta:alg_maximal_tree_expansion2} can be implemented by copying input sets for different tasks and running multiple tasks in parallel, where the operations needed are described in \textbf{Copies of sets}, \textbf{Multiple queries} and \textbf{Multiple Tasks} (see Section~\ref{sec:multi_tasks}).
Line~\ref{sta:alg_maximal_wT_assignment} to line~\ref{sta:alg_wT_assignment2} can be implemented by the operations described in \textbf{Copies of sets}, \textbf{Set membership}, \textbf{Indexing elements in sets}, and \textbf{Multiple queries}.

The total parallel time of the first loop is $O(r)$ since it has $r$ rounds.
The second loop can be done in one round.
Thus the parallel time of the second loop is $O(1)$.
Then the total parallel time is $O(r).$
Due to Lemma~\ref{lem:parallel_implement_of_multi_radius} and Lemma~\ref{lem:num_it_multiplelargetree}, $r$ is the number of iterations (see Definition~\ref{def:num_it_multiplelargetree}) of $\textsc{MultipleLargeTrees}(G,m).$

We stored all the $V_{T_i(v)},V_{\wt{T}_i(v)},\p_{T_i(v)},\p_{\wt{T}_i(v)},\dep_{T_i(v)},\dep_{\wt{T}_i(v)}$ in the system.
By Lemma~\ref{lem:maximal_large_short_tree} and Lemma~\ref{lem:multi_radius_trees}, the total space needed to store them is at most $O(m).$
Furthermore, at any round, the size of all the input copies for multiple tasks is at most $n\cdot(m/n)^{1/4}\cdot (m/n)^{1/4}=O(m).$
Thus, the total space needed is $O(m).$
\end{proof}

\subsection{Path Generation and Root Changing}
\begin{lemma}\label{lem:implement_of_find_ancestor}
Let $\p:V\rightarrow V$ be a set of parent pointers (See Definition~\ref{def:parent_pointers}) on a vertex set $V.$ Let $n=|V|.$ $\textsc{FindAncestors}(\p)$ (Algorithm~\ref{alg:ancestors}) can be implemented in $(\gamma,\delta)-\MPC$ model for any $\gamma\geq \frac{\log n}{\log \log n}$ and any constant $\delta\in(0,1).$ The parallel running time is $O(r)$, where $r$ is the number of iterations (see Definition~\ref{def:num_iter_ancestors}) of $\textsc{FindAncestors}(\p).$
\end{lemma}
\begin{proof}
The structure of the whole algorithm is the same as the Algorithm~\ref{alg:tree_contraction} (see Lemma~\ref{lem:parallel_implement_of_tree_contract}).
All the steps can be done by operation described in \textbf{Multiple queries}.

Since the number of rounds needed is $r$, the parallel time is $O(r)$.
For the total space, we need to store all the mappings $g_1,\cdots,g_r.$
At the end of the $i^{\text{th}}$ round, we need to store mapping $h_i.$
According to Lemma~\ref{lem:depthandancestor}, $r=O(\log n)$
Thus, the total space is $O(rn)=O(n\log n).$
\end{proof}

\begin{lemma}\label{lem:parallel_implement_of_find_path}
Let $\p:V\rightarrow V$ be a set of parent pointers (See Definition~\ref{def:parent_pointers}) on a vertex set $V.$ Let $q$ be a vertex in $V$, and $n=|V|$. $\textsc{FindPath}(\p,q)$ (Algorithm~\ref{alg:path_find}) can be implemented in $(\gamma,\delta)-\MPC$ model for any $\gamma\geq \frac{\log n}{\log \log n}$ and any constant $\delta\in(0,1).$ The parallel running time is $O(r)$, where $r$ is the number of iterations (see Definition~\ref{def:num_iter_ancestors}) of $\textsc{FindAncestors}(\p)$ (Algorithm~\ref{alg:ancestors}).
\end{lemma}
\begin{proof}
By Lemma~\ref{lem:implement_of_find_ancestor}, $\textsc{FindAncestors}(\p)$ can be implemented in $(\gamma,\delta)-\MPC$ model for $\gamma\geq \frac{\log n}{\log \log n}$ and any constant $\delta\in(0,1).$
All the other other steps in the algorithm can be done by operation described in \textbf{Multiple queries}.
Notice that, after each round, we need to do load balancing which can be done by operation described in \textbf{Load balance}.

The number of rounds must be smaller than $O(r),$ where $r$ should be the number of iterations of $\textsc{FindAncestors}(\p)$ according to Lemma~\ref{lem:implement_of_find_ancestor}.

We store all the mappings $g_i,\dep_{\p}$ in the system.
They need $O(n\log n)$ total space.
In the $i^{\text{th}}$ round, we only need to additionally store set $S_i$ which has size at most $O(n).$
Thus, the total space needed is at most $O(n\log n).$
\end{proof}

\begin{lemma}\label{lem:parallel_implement_of_root_change}
Let $\p:V\rightarrow V$ be a set of parent pointers (See Definition~\ref{def:parent_pointers}) on a vertex set $V.$ Let $q$ be a vertex in $V$. $\textsc{RootChange}(\p,q)$ (Algorithm~\ref{alg:root_change}) can be implemented in $(\gamma,\delta)-\MPC$ model for any $\gamma\geq \frac{\log n}{\log \log n}$ and any constant $\delta\in(0,1).$ The parallel running time is $O(r)$, where $r$ is the number of iterations (see Definition~\ref{def:num_iter_ancestors}) of $\textsc{FindAncestors}(\p)$ (Algorithm~\ref{alg:ancestors}).
\end{lemma}
\begin{proof}
By Lemma~\ref{lem:parallel_implement_of_find_path}, $\textsc{FindPath}(\p,q)$ can be implemented in $(\gamma,\delta)-\MPC$ model.
The remaining steps in the procedure can be implemented by the operation described by \textbf{Multiple queries}, and has $O(1)$ parallel running time.

The total space needed is the total space needed for $\textsc{FindPath}(\p,q)$ plus the space needed to store mapping $h,\wh{\p}.$
Thus the total space needed is $O(n\log n)+O(n)=O(n\log n).$

The parallel running time is linear in the parallel running time of $\textsc{FindPath}(\p,q).$
Then, by Lemma~\ref{lem:parallel_implement_of_find_path}, the parallel running time is $O(r)$ where $r$ is the number of iterations (see Definition~\ref{def:num_iter_ancestors}) of $\textsc{FindAncestors}(\p).$
\end{proof}

\subsection{Spanning Forest Algorithm}

\begin{lemma}\label{lem:parallel_implement_span_exp}
Let $G_2=(V_2,E_2)$ be an undirected graph. Let $\wt{\p}:V_2\rightarrow V_2$ be a set of parent pointers (See Definition~\ref{def:parent_pointers}) which satisfies that $\forall v\in V_2$ with $\wt{\p}(v)\not=v$, $(v,\wt{\p}(v))$ must be in $E_2$. Let $G_1=(V_1,E_1)$ be an undirected graph satisfies $V_1=\{v\in V_2\mid \wt{\p}(v)=v\},E_1=\{(u,v)\in V_1\times V_1\mid u\not=v,\exists (x,y)\in E_2,\wt{\p}^{(\infty)}(x)=u,\wt{\p}^{(\infty)}(y)=v\}.$ Let $\p:V_1\rightarrow V_1$ be a rooted spanning forest (See Definition~\ref{def:spanning_forest}) of $G_1.$ Let $f:V_1\times V_1\rightarrow\{ \nul \}\cup \left(V_2\times V_2\right)$ satisfy the following property:
for $u\not=v\in V_1,$ if $\p(u)=v,$ then $f(u,v)\in\{(x,y)\in E_2\mid \wt{\p}^{(\infty)}(x)=u,\wt{\p}^{(\infty)}(y)=v\},$ and $f(v,u)\in \{(x,y)\in E_2\mid \wt{\p}^{(\infty)}(x)=v,\wt{\p}^{(\infty)}(y)=u\}.$
Let $n=|V_2|.$
Then $\textsc{ForestExpansion}(\p,\wt{\p},f)$ (Algorithm~\ref{alg:spanning_forest_expansion}) can be implemented in $(\gamma,\delta)-\MPC$ model for any $\gamma\geq \log n/\log \log n$ and any constant $\delta\in(0,1)$ in parallel running time $O(R),$ where $R=\log(\dep(\wt{\p})).$
\end{lemma}
\begin{proof}
Due to Lemma~\ref{lem:parallel_implement_of_tree_contract}, line~\ref{sta:alg_span_exp_tree_contract} can be done in $O(R)$ parallel time for $R=\log(\dep(\wt{\p})).$
Line~\ref{sta:alg_spanexp_change} corresponds to multiple tasks, we can implement them parallelly by operations described in \textbf{Multiple queries}, and \textbf{Multiple Tasks} (see Section~\ref{sec:multi_tasks}).
By Lemma~\ref{lem:parallel_implement_of_root_change}, the total space needed is at most $O(n\log n)$ and the parallel running time is at most $O(R)$ where $R=\log(\dep(\wt{\p})).$
\end{proof}

\begin{theorem}\label{thm:parallel_implement_spanning_forest}
Let graph $G=(V,E),n=|V|,N=|V|+|E|$ and $m=\Theta(N^{\gamma})$ for some arbitrary $\gamma\in[0,2].$
Let $r>0$ be a round parameter.
$\textsc{SpanningForest}(G,m,r)$ (Algorithm~\ref{alg:batch_algorithm3}) can be implemented in $(\gamma,\delta)-\MPC$ model for any constant $\delta\in(0,1).$ Furthermore, the parallel running time is $O(R),$ where $R$ is the total number of iterations (see Definition~\ref{def:num_total_it_spanning_tree}) of $\textsc{SpanningForest}(G,m,r).$
\end{theorem}
\begin{proof}
At the beginning of the algorithm, we just store sets $V,E,V_0,E_0$ and mapping $g_0$ in the system.

Consider the $i^{\text{th}}$ round of the loop.
By Lemma~\ref{lem:parallel_implement_of_multiple_large_tree}, line~\ref{sta:alg_span_multilargetree} can be implemented in total space $\Theta(m)$ and in parallel running time $O(k_i)$ where $k_i$ is the number of iterations (see Definition~\ref{def:num_it_multiplelargetree}) of $\textsc{MultipleLargeTrees}(G_i,m).$
Line~\ref{sta:alg_span_create_veprime} can be implemented by operations described in \textbf{Sizes of sets}, \textbf{Set membership}, and \textbf{Multiple queries}.
Line~\ref{sta:alg_span_spanningtree} can be implemented by operations described in \textbf{Indexing elements in sets}, \textbf{Set membership}, and \textbf{Multiple queries}.
In line~\ref{sta:alg_span_samp_prob}, to calculate $\gamma_i,$ we need to query $n_i,$ this can be done by operations described in \textbf{Sizes of sets} and \textbf{Multiple queries}.
In line~\ref{sta:alg_span_leaders}, to compute $L_i,$ we only need operations described in \textbf{Set membership} and \textbf{Multiple queries}.
Line~\ref{sta:alg_span_assign_to_leader} can be implemented by operations shown in \textbf{Set membership}, \textbf{Indexing elements in sets} and \textbf{Multiple queries}.
By Lemma~\ref{lem:parallel_implement_of_find_path}, for line~\ref{sta:alg_span_find_path}, there are multiple tasks each can be implemented in $O(|V_{\wt{T}_i(v)}|\log |V_{\wt{T}_i(v)}|)$ total space, and $O(k_i)$ parallel time.
We can schedule these multiple tasks (see Section~\ref{sec:multi_tasks}) such that we can finish them in parallel in $O(k_i)$ parallel time.
According to Lemma~\ref{lem:parallel_implement_of_tree_contract}, for line~\ref{sta:alg_span_tree_contraction}, we can implement it in $O(n_i)=O(n)$ total space, and in $O(k'_i)$ parallel time, where $k'_i$ is the number of iterations (see Definition~\ref{def:num_it_tree_contract}) of $\textsc{TreeContraction}(G_i',\p_i).$
Line~\ref{sta:alg_span_hi_notnull} can be done by the operation described in \textbf{Multiple queries}.
Line~\ref{sta:alg_span_choose_one_edge} can be done by the operation described in \textbf{Indexing elements in sets} and \textbf{Multiple queries}.

Thus, the parallel time is $O(R),$ where $R=\sum_{i=0}^{r-1}(k_i+k'_i).$
By definition of the total number of iterations (see Definition~\ref{def:num_total_it_spanning_tree}) of $\textsc{SpanningForest}(G,m,r).$
$R$ is the total number of iterations of $\textsc{SpanningForest}(G,m,r).$

For the space, we store all the sets $V,E,V_i,D_i$ and mappings $\p_i,h_i$ in all the rounds.
Notice that $\sum_{i=0}^r |V_i|\leq 40|V|.$
Thus this part takes only $O(N)$ space.
In the $i^{\text{th}}$ round, we additionally store all the sets $V_{\wt{T}_i(v)},V_i',E_i',L_i$ and all the mappings $\p_{\wt{T}_i(v)},\dep_{\wt{T}_i(v)},l_i,z_i.$
The total space for this part is at most $O(m).$
For line~\ref{sta:alg_span_find_path}, it creates multiple tasks.
The input of each task is at most $|V_{\wt{T}_i(v)}|\leq (m/n_i)^{1/2}.$
There are at most $n_i$ tasks, and by Lemma~\ref{lem:parallel_implement_of_find_path}, each task will need space at most $O(|V_{\wt{T}_i(v)}|\log |V_{\wt{T}_i(v)}|).$
Thus, the space for this part is at most $O(m).$
To conclude, the total space needed is at most $O(m).$

\end{proof}

\begin{theorem}\label{thm:parallel_implement_rooted_forest}
Let graph $G=(V,E),n=|V|,N=|V|+|E|$ and $m=\Theta(N^{\gamma})$ for some arbitrary $\gamma\in[0,2].$
Let $r>0$ be a round parameter.
If $\textsc{SpanningForest}(G,m,r)$ (Algorithm~\ref{alg:batch_algorithm3}) does not return {\rm FAIL}, then let the output be the input of $\textsc{Orientate}(\cdot)$ (Algorithm~\ref{alg:batch_algorithm4}), and $\textsc{Orientate}(\cdot)$ can be implemented in $(\gamma,\delta)-\MPC$ model for any constant $\delta\in(0,1).$
Furthermore, the parallel running time is $O(R),$ where $R$ is the total number of iterations (see Definition~\ref{def:num_total_it_spanning_tree}) of $\textsc{SpanningForest}(G,m,r).$
\end{theorem}

\begin{proof}
Line~\ref{sta:batch_alg4_first_loop_st} to line~\ref{sta:batch_alg4_first_loop_ed} can be implemented by operations described in \textbf{Multiple queries}.
Notice that there is a trick here, if $f_i(u,v)=\nul,$ we do not need to store the tuple $(``f_i",((u,v),\nul))$ in the system.
The total space needed to store all the mappings $f_i$ and all the sets $F_{i}$ for $i\in\{0\}\cup[r]$ is at most $\sum_{i=0}^r |V_i|=O(m).$

Line~\ref{batch_alg4_second_loop1} and line~\ref{batch_alg4_second_loop2} can be implemented by operations described in \textbf{Set membership} and \textbf{Multiple queries}.

We now look at the second loop, and focus on round $i$.
Line~\ref{batch_alg4_second_loop3} can be implemented by Lemma~\ref{lem:parallel_implement_span_exp}.
The total space needed is at most $O(|V_i|\cdot (m/|V_i|)^{1/2}\cdot \log (m/|V_i|))=O(m).$
The parallel running time needed is at most $O(k_i),$ where $k_i$ is the number of iterations (see Definition~\ref{def:num_it_multiplelargetree}) of $\textsc{MultipleLargeTrees}(G_i,m),$ $G_i$ is the intermediate graph in the procedure $\textsc{SpanningForest}(G,m,r).$

Thus, the parallel running time is $O(R),$ where $R$ is the total number of iterations (see Definition~\ref{def:num_total_it_spanning_tree}) of $\textsc{SpanningForest}(G,m,r).$
The total space needed is $O(m).$

\end{proof}

Now, we are able to conclude the following theorem for spanning forest problem.
\begin{theorem}\label{thm:formal_statement_of_spanning_tree}
For any $\gamma\in[0,2]$ and any constant $\delta\in(0,1),$ there is a randomized $(\gamma,\delta)-\MPC$ algorithm (see Algorithm~\ref{alg:batch_algorithm3} and Algorithm~\ref{alg:batch_algorithm4}) which can output the rooted spanning forest for any graph $G=(V,E)$ in $O(\min(\log D\cdot \log \frac{1}{\gamma'},\log n))$ parallel time, where $D$ is the diameter of $G,$ $n=|V|,$ $N=|V|+|E|$ and $\gamma'=(1+\gamma)\log_n \frac{2N}{n^{1/(1+\gamma)}}.$ The success probability is at least $0.98.$ In addition, if the algorithm fails, then it will return {\rm FAIL}.
\end{theorem}
\begin{proof}
Algorithm~\ref{alg:batch_algorithm3} outputs all the edges in the spanning forest and all the contraction information. Algorithm~\ref{alg:batch_algorithm4} takes the output of Algorithm~\ref{alg:batch_algorithm3} as its input, and outputs a rooted spanning forest.

The implementation of Algorithm~\ref{alg:batch_algorithm3} and Algorithm~\ref{alg:batch_algorithm4} in $\MPC$ model is shown by Theorem~\ref{thm:parallel_implement_spanning_forest} and Theorem~\ref{thm:parallel_implement_rooted_forest} respectively.
The correctness of Algorithm~\ref{alg:batch_algorithm3} and Algorithm~\ref{alg:batch_algorithm4} is proved by Corollary~\ref{cor:correctness_batch_alg3} and Theorem~\ref{thm:batch_algorithm4} respectively.
The parallel time of Algorithm~\ref{alg:batch_algorithm3} and Algorithm~\ref{alg:batch_algorithm4} is proved by Theorem~\ref{thm:success_prob_batch_alg3}.

\end{proof}

A byproduct of our spanning forest algorithm is an estimator of the diameter of the graph.
\begin{theorem}\label{thm:formal_statement_of_diameter}
For any $\gamma\in[0,2]$ and any constant $\delta\in(0,1),$ there is a
randomized $(\gamma,\delta)-\MPC$ algorithm which can output an
diameter estimator $D'$ for any graph $G=(V,E)$ in $O(\min(\log D\cdot
\log (1/\gamma'),\log n))$ parallel time such that $D\leq D'\leq
D^{O(\log ( 1/\gamma' ))},$ where $D$ is the diameter
of $G,$ $n=|V|,$ $N=|V|+|E|$ and $\gamma'=(1+\gamma)\log_n
\frac{2N}{n^{1/(1+\gamma)}}.$ The success probability is at least
$0.98.$ In addition, if the algorithm fails, then it will return
{\rm FAIL}.
\end{theorem}
\begin{proof}
By Theorem~\ref{thm:formal_statement_of_spanning_tree}, we can find a rooted spanning forest.
By Theorem~\ref{thm:batch_algorithm4}, the depth of that rooted spanning forest is at most $D^{O(\log ( 1/\gamma' ))}$. Then we can implement a doubling algorithm (e.g. Modified Lemma~\ref{lem:implement_of_find_ancestor}, Algorithm~\ref{alg:ancestors} without maintaining useless $g_l$) with log in depth parallel time to output the depth of that spanning forest.
\end{proof}

\subsection{Lowest Common Ancestor and Multi-Paths Generation}

\begin{lemma}\label{lem:parallel_implment_of_lca}
Let $\p:V\rightarrow V$ be a set of parent pointers (See Definition~\ref{def:parent_pointers}) on a vertex set $V$. Let $Q=\{(u_1,v_1),(u_2,v_2),\cdots,(u_q,v_q)\}$ be a set of $q$ pairs of vertices, and $\forall i\in [q],u_i\not=v_i$. Let $n=|V|,N=n+q.$ $\textsc{LCA}(\p,Q)$ (Algorithm~\ref{alg:lca}) can be implemented in $(\gamma,\delta)-\MPC$ model for any $\gamma\geq \log\log N/\log N$ and any constant $\delta\in(0,1)$ in $O(\log(\dep(\p)))$ parallel running time.
\end{lemma}
\begin{proof}
By Lemma~\ref{lem:implement_of_find_ancestor}, line~\ref{sta:alg_lca_find_ancestor} can be implemented in space $O(N\log N)$ and $O(\log(\dep(\p)))$ parallel running time.
It is easy to see that all the other steps in the procedure can be done by the operations shown in \textbf{Multiple queries}.

Thus, the total space needed is $O(N\log N)$ and the parallel running time is $O(\log(\dep(\p))).$
\end{proof}

\begin{lemma}\label{lem:parallel_implement_of_multipath}
Let $\p:V\rightarrow V$ be a set of parent pointers (See Definition~\ref{def:parent_pointers}) on a vertex set $V.$ Let $Q=\{(u_1,v_1),(u_2,v_2),\cdots,(u_q,v_q)\}\subseteq V\times V$ satisfy $\forall j\in[q],$ $v_j$ is an ancestor (See Definition~\ref{def:ancestor}) of $u_j$ in $\p$.
 Let $n=|V|,N=n+q.$
 $\textsc{MultiPath}(\p,Q)$ (Algorithm~\ref{alg:multi_path}) can be implemented in $(\gamma,\delta)-\MPC$ model for any $\gamma$ with $N\log N+\sum_{i=1}^q(\dep_{\p}(u_i)-\dep_{\p}(v_i)+1)=O(N^\gamma)$ and any constant $\delta\in(0,1)$ in $O(\dep(\p))$ parallel running time.
\end{lemma}
\begin{proof}
By Lemma~\ref{lem:implement_of_find_ancestor}, line~\ref{sta:alg_multipath_findancestor} can be implemented in space $O(N\log N)$ and $O(\log(\dep(\p)))$ parallel running time.
It is easy to see that all the other steps in the procedure can be done by the operations shown in \textbf{Multiple queries}.
Notice that after each round, we need to do load balancing (see \textbf{Load balance}) to make each machine have large enough available local memory.
The total space needed is to store all the pathes and the output of line~\ref{sta:alg_multipath_findancestor}.
Notice that in round $i$, we do not need to keep $S_j^{(i')}$ for $i'<i-1,$ thus, the space to keep $S_j^{(i)}$ for all $j\in[q]$ only needs $O(\sum_{j=1}^q(\dep_{\p}(u_j)-\dep_{\p}(v_j)+1))$ space.

Thus, the total space needed is at most $O(N\log N+\sum_{i=1}^q(\dep_{\p}(u_i)-\dep_{\p}(v_i)+1))=O(N^\gamma).$
The parallel running time is then $O(\dep(\p)).$
\end{proof}

\subsection{Leaf Sampling}
\begin{lemma}\label{lem:parallel_implement_of_sample_leaf}
Let $\p:V\rightarrow V$ be a set of parent pointers (See Definition~\ref{def:parent_pointers}) on a vertex set $V$, and $\p$ has a unique root. Let $n=|V|.$ Let $\delta$ be an arbitrary constant in $(0,1),$ and let $m= \lceil n^\delta\rceil.$ Then $\textsc{LeafSampling}(\p,m,\delta)$ (Algorithm~\ref{alg:sample_leaves}) can be implemented in $(\gamma,\delta)-\MPC$ model for any $\gamma\geq \log\log n/\log n$. Furthermore, with probability at least $1-1/(100m^{5/\delta})$, the parallel running time is at most $O(\log\dep(\p))$.
\end{lemma}
\begin{proof}
To implement line~\ref{sta:alg_sample_leaf_compute_leaf}, for each $v\in V,$ we can add $\p(v)$ to a temporary set $X.$
Then each $v$ can check whether $v$ is a leaf by checking whether $v$ is in $X$, and this can be done by the operations shown in \textbf{Set membership} and \textbf{Multiple queries}.

To implement line~\ref{sta:alg_sample_leaf_compute_rank}, for each $v\in V,$ we can add $v$ to the set $\child_{\p}(\p(v)).$
Then $\rank$ can be computed by the operations shown in \textbf{Indexing elements in sets} and \textbf{Multiple queries}.
For line~\ref{sta:when_V_is_small}, we can implement it on a single machine, since a single machine has local memory $\Theta(m).$
For line~\ref{sta:when_L_is_small} to line~\ref{sta:when_L_needs_sample}, for each $x\in L,$ we add $x$ into $S$ with probability $p,$ where $p$ can be computed by querying the size of $L$ (see \textbf{Sizes of sets} and \textbf{Multiple queries}).
Line~\ref{sta:point_to_first_leaf} can be implemented by operation described in \textbf{Indexing elements in sets}, \textbf{Set membership}, and \textbf{Multiple queries}.
By Lemma~\ref{lem:parallel_implement_of_tree_contract}, line~\ref{sta:mid_find_the_first_leaf} can be implemented in total space $O(N\log N)$ and $O(\log\dep(\p))$ parallel time.
By Property~\ref{itm:sampled_leaves_pro3} of Lemma~\ref{lem:properties_of_sampled_leaves}, with probability at least $1-1/(100m^{5/\delta}),$ $|S|^2=O(m).$
Thus, $Q$ can be stored on a single machine.
By Lemma~\ref{lem:parallel_implment_of_lca}, line~\ref{sta:alg_sample_leaf_lca} can be implemented in total space $O(n\log n+|Q|)=O(n\log n)$ and in $O(\log\dep(\p))$ parallel time.
By Lemma~\ref{lem:implement_of_find_ancestor}, line~\ref{sta:alg_smaple_leaf_find_ancestor} can be implemented in total space $O(n\log n)$ and in $O(\log\dep(\p))$ parallel time.
Then line~\ref{sta:alg_sample_leaf_sequential_st} to line~\ref{sta:alg_sample_leaf_sequential_ed} can be implemented on a single machine.

Thus, the total space needed is at most $O(n\log n).$
The parallel time is at most $O(\log\dep(\p))$
\end{proof}

\subsection{DFS Sequence}

\begin{lemma}\label{lem:parallel_implement_sub_dfs}
Let $\p:V\rightarrow V$ be a set of parent pointers (See Definition~\ref{def:parent_pointers}) on a vertex set $V$, and $\p$ has a unique root. Let $n=|V|.$ Let $\delta$ be an arbitrary constant in $(0,1),$ and let $m= \lceil n^\delta\rceil.$ $\textsc{SubDFS}(\p,m,\delta)$ (Algorithm~\ref{alg:subdfs}) can be implemented in $(\gamma,\delta)-\MPC$ model for any $\gamma\geq \log\log n/\log n$. Furthermore, with probability at least $1-1/(100m^{5/\delta})$, the parallel running time is at most $O(\log\dep(\p))$.
\end{lemma}
\begin{proof}
By Lemma~\ref{lem:parallel_implement_of_sample_leaf}, line~\ref{sta:alg_subsequence_sample_leaf} can be implemented in total space $O(n\log n)$ and with probability at least $1-1/(100m^{5/\delta})$ has parallel running time $O(\log\dep(\p)).$
By Lemma~\ref{lem:parallel_implment_of_lca}, line~\ref{sta:alg_subsequence_lca} can be implemented in total space $O(n\log n)$ and in parallel running time $O(\log\dep(\p)).$
Line~\ref{sta:alg_subsequence_Qprime} can be implemented by operation shown in \textbf{Multiple queries}.
By Lemma~\ref{lem:parallel_implement_of_multipath}, since all the pathes are disjoint (except the first path and the last path intersecting on the root) and $V$ has $n$ vertices, line~\ref{sta:alg_subsequence_multi_path} can be implemented in $O(n\log n)$ total space and in $O(\log\dep(\p))$ parallel running time.
Loop in line~\ref{sta:alg_subsequence_loop1} and Loop in line~\ref{sta:alg_subsequence_loop2} can be implemented in parallel, and can be implemented by operations shown in \textbf{Indexing elements in sets} and \textbf{Multiple queries}.
Line~\ref{sta:alg_subsquence_compute_rank} can be implemented by operations shown in \textbf{Indexing elements in sets} and \textbf{Multiple queries}.
Now we describe the implementation of line~\ref{sta:alg_subsequence_pos}.
Firstly, we can standardize (see \textbf{Sequence standardizing}) the sequence $A'.$
For each tuple $(``A' ",(j,u)),$ create a tuple $(``temp_u",j).$
Thus, $``temp_u"$ is a set which contains all the positions that $u$ appeared.
For each tuple $(``temp_u",j),$ we query (see \textbf{Multiple queries}) the index $i$ (see \textbf{Indexing elements in sets}) of $j$ in set $(``temp_u",j),$ and create a tuple $(``\pos",((u,i),j)).$
Thus, the desired mapping $\pos$ is stored in the system.
The loop in line~\ref{sta:alg_subsequence_final_loop} is implemented in parallel.
Line~\ref{sta:just_copy_once} can be implemented by the operations shown in \textbf{Set membership} and \textbf{Multiple queries}.
Line~\ref{sta:the_first_duplicate} to line~\ref{sta:mid_duplicate} can be implemented by the operation shown in \textbf{Multiple queries}.
Finally, line~\ref{sta:duplicate_elements} can be implemented by \textbf{Multiple queries} and \textbf{Sequence duplicating}.

The total space used in the procedure is at most $O(n\log n).$ The parallel running time is $O(\log\dep(\p)).$
\end{proof}

\begin{theorem}\label{thm:parallel_implement_DFS_sequence}
Let $\p:V\rightarrow V$ be a set of parent pointers (See Definition~\ref{def:parent_pointers}) on a vertex set $V$, and $\p$ has a unique root. Let $n=|V|,m=n^{\delta}$ for some arbitrary constant $\delta\in(0,1)$.  $\textsc{DFS}(\p,m)$ (Algorithm~\ref{alg:dfs_sequence}) can be implemented in $(\gamma,\delta)-\MPC$ model for any $\gamma\geq \log \log n/\log n.$ With probability at least $0.99,$ the parallel running time is $O(\log(\dep(\p))).$
\end{theorem}
\begin{proof}
By Lemma~\ref{lem:parallel_implement_sub_dfs}, line~\ref{sta:alg_dfs_init_sub_dfs} can be implemented in total space $O(n\log n).$ With probability at least $1-1/(100n^5),$ the parallel running time is $O(\log(\dep(\p))).$
Line~\ref{sta:alg_dfs_vprimei} to line~\ref{sta:alg_dfs_pi} can be implemented by operations shown in \textbf{Set membership} and \textbf{Multiple queries}.
By Lemma~\ref{lem:parallel_implement_of_tree_contract}, line~\ref{sta:alg_dfs_tree_contract} can be implemented in $O(n)$ total space, and $O(\log\dep(\p))$ parallel running time.
The loop in line~\ref{sta:alg_dfs_the_loop} contains multiple tasks (see Section~\ref{sec:multi_tasks} \textbf{Multiple Tasks}), thus we can implement those tasks in parallel.
By Lemma~\ref{lem:parallel_implement_sub_dfs}, line~\ref{sta:alg_dfs_inloop_subdfs} can be implemented in total space $O(|V_i'(v)|\log|V_i'(v)|).$
Furthermore, with probability at least $1-1/(100n^5),$ the parallel running time is $O(\log(\dep(\p))).$
Thus, the total space needed for those tasks is at most $O(n\log n).$
Line~\ref{sta:alg_dfs_sequence_insertion} can be implemented by operations shown in \textbf{Indexing elements in sets}, \textbf{Sequence insertion} and \textbf{Multiple queries}.

Thus, the total space needed is $O(n\log n).$
By taking union bound over all the task $\textsc{SubDFS},$ with probability at least $0.99,$ the parallel running time is $O(\log\dep(\p)).$
\end{proof}

Now we are able to conclude the following theorem.
\begin{theorem}\label{thm:formal_statement_of_DFS}
For any $\gamma\in[\beta,2]$ and any constant $\delta\in(0,1),$ there is a randomized $(\gamma,\delta)-\MPC$ algorithm (Algorithm~\ref{alg:dfs_sequence}) which can output a Depth-First-Search sequence for any tree graph $G=(V,E)$ in $O(\min(\log D\cdot \log (1/\gamma'),\log n))$ parallel time, where  $n=|V|,$ $\beta=\Theta(\log\log n/\log n),$ $D$ is the diameter of $G,$ and $\gamma'=\gamma+\Theta(1/\log n).$ The success probability is at least $0.98$. In addition, if the algorithm fails, then it will return {\rm FAIL}.
\end{theorem}
\begin{proof}
Firstly, by Theorem~\ref{thm:formal_statement_of_spanning_tree}, we can find a rooted tree.
Algorithm~\ref{alg:dfs_sequence} can output the DFS sequence for a rooted tree.

The implementation and parallel time of Algorithm~\ref{alg:dfs_sequence} is shown by Theorem~\ref{thm:parallel_implement_DFS_sequence}.
The correctness of Algorithm~\ref{alg:dfs_sequence} is proved by Theorem~\ref{thm:dfs_sequence}.
The success probability of Algorithm~\ref{alg:dfs_sequence} is proved by Theorem~\ref{thm:success_prob_dfs_sequence}.
\end{proof}

\subsection{Range Minimum Query}

\begin{lemma}\label{lem:parallel_implement_of_sparser_table}
Let $A=(a_1,a_2,\cdots,a_n)$ be a sequence of numbers. Let $\delta$ be an arbitrary constant in $(0,1).$ $\textsc{SparseTable}^+(a_1,a_2,\cdots,a_n,\delta)$ (Algorithm~\ref{alg:st_rmq}) can be implemented in $(0,\delta)-\MPC$ model with $O(1)$ parallel running time.
\end{lemma}
\begin{proof}
Let $A$ be the sequence $(a_1,a_2,\cdots,a_n)$.
The algorithm takes $O(1/\delta)$ rounds.
$m$ is the local space of a machine.
There are $\Theta(n/m)$ machines each holds a consecutive $\Theta(m)$ elements of sequence $A.$
Now consider the round $l$.
Machine $j\in\{0\}\cup[\lceil n/m\rceil]$ needs to compute $\wh{f}_{j\cdot m+1,l},\wh{f}_{j\cdot m+2,l},\cdots,$ $\wh{f}_{j\cdot m+m-1,l}.$
The number of queries machine $j$ made in line~\ref{sta:rmq_zstar} and line~\ref{sta:assign_value} is at most $\sum_{t=1}^{\lceil1/\delta\rceil}|S_t|+2m\leq O(m/\delta)=O(m).$
Thus, there are total $O(n)$ queries.
These queries can be answered simultaneously by operation shown in \textbf{Multiple queries}.

Thus, the total space needed is $O(n),$ and the parallel running time is $O(1).$
\end{proof}

\begin{lemma}\label{lem:parallel_implement_st_rmq2}
Let $a_1,a_2,\cdots,a_n$ be a sequence of numbers. Let $\delta$ be an arbitrary constant in $(0,1).$  $\textsc{SparseTable}(a_1,a_2,\cdots,a_n,\delta)$ (Algorithm~\ref{alg:st_rmq2}) can be implemented in $(\gamma,\delta)-\MPC$ model for any $\gamma\geq \log\log n/\log n$ in $O(1)$ parallel time.
\end{lemma}
\begin{proof}
By Lemma~\ref{lem:parallel_implement_of_sparser_table}, line~\ref{sta:alg_st_rmq2_sttable} can be implemented in $O(n)$ total space and $O(1)$ parallel time.
The loop in line~\ref{sta:alg_st_rmq2_loop} is similar to Algorithm~\ref{alg:st_rmq}.
Each machine $j$ needs to compute $f_{j\cdot m+1,t},\cdots,f_{j\cdot m+m-1,t}$ for all $t\in[\lceil\log n\rceil]\cup\{0\}.$
The difference from Algorithm~\ref{alg:st_rmq} is that, it can compute for all $t$ at the same time since it only depends on the value of $\wh{f}.$
The number of queries made by each machine is $O(m\log n).$
Thus, the total number of queries is at most $O(n\log n).$
These queries can be answered simultaneously by operation shown in \textbf{Multiple queries}.

Thus, the total space needed is $O(n\log n),$ and the parallel running time is $O(1).$
\end{proof}
\section{Minimum Spanning Forest}\label{sec:mst_bst}
In this section, we discuss how to apply our connectivity/spanning forest algorithm to the Minimum Spanning Forest (MSF) and Bottleneck Spanning Forest (BSF) problem.

The input of MSF/BSF problem is an undirected graph $G=(V,E)$ together with a weight function $w:E\rightarrow\mathbb{Z},$ where $E$ contains $m$ edges $e_1,e_2,\cdots,e_m$ with $w(e_1)\leq w(e_2)\leq \cdots\leq w(e_m).$
The goal of MSF is to output a spanning forest such that the sum of weights of the edges in the forest is minimized.
The goal of BSF is to output a spanning forest such that the maximum weight of the edges in the forest is minimized.
$D$ is the diameter of the minimum spanning forest.
If there are multiple choices of the minimum spanning forest, then let $D$ be the minimum diameter among all the minimum spanning forests.

For simplicity, in all of our proofs, we only discuss the case when all the edges have different weights, i.e. $w(e_1)<w(e_2)<\cdots<w(e_m).$
In this case, the minimum spanning forest is unique.
It is easy to extend our algorithms to the case when there are edges with the same weight.
We omit the proof for this fact.

Firstly, we show that $D$ is an upper bound of the diameter of $G'$ where the vertex set of $G'$ is the vertex set of $G$, and the edge set of $G'$ is $\{e_1,e_2,\cdots,e_i\}$ for some arbitrary $i\in[m].$

\begin{lemma}\label{lem:mst_depth_dimater}
Given a graph $G=(V,E)$ for $E=\{e_1,e_2,\cdots,e_m\}$ together with a weight function $w$ which satisfies $w(e_1)< w(e_2)< \cdots< w(e_m),$ then the diameter of $G'=(V,E')$ is at most $D,$ where $D$ is the diameter of the minimum spanning forest of $G,$ and $E'$ only contains the first $i$ edges of $E$, i.e. $e_1,e_2,\cdots,e_i$ for some arbitrary $i\in[m].$
\end{lemma}
\begin{proof}
The proof follows by Kruskal's algorithm directly.
\end{proof}

Our algorithms is based on the following simple but useful Lemma.
\begin{lemma}\label{lem:mst_divide_and_conquer}
Given a graph $G=(V,E)$ for $E=\{e_1,e_2,\cdots,e_m\}$ together with a weight function $w$ which satisfies $w(e_1)\leq w(e_2)\leq \cdots\leq w(e_m),$ $\forall 1\leq i<j\leq m,$ an edge $e$ from $\{e_i,e_{i+1},\cdots,e_{j}\}$ is in the minimum spanning forest of $G$ if and only if $e'$ from $\{e'_i,e'_{i+1},\cdots,e'_j\}$ is in the minimum spanning forest of $G',$ where the vertices of $G'$ is obtained by contracting all the edges $e_1,e_2,\cdots,e_{i-1}$ of $G,$ and $e',e'_i,e'_{i+1},\cdots,e'_j$ are the edges (or vertices) in $G'$ which corresponds to the edges $e,e_i,e_{i+1},\cdots,e_j$ before contraction.
\end{lemma}
\begin{proof}
The proof follows by Kruskal's algorithm directly.
\end{proof}

A natural way to apply Lemma~\ref{lem:mst_divide_and_conquer} to parallel minimum spanning forest algorithm is that we can divide the edges into several groups, and recursively solve the minimum spanning forest for each group of edges.
More precisely, suppose we have total space $\Theta(km),$ we can divide $E$ into $k$ groups $E_1,E_2,\cdots,E_k,$ where $E_i=\{e_{(i-1)\cdot m/k+1},e_{(i-1)\cdot m/k+2},\cdots,e_{i\cdot m/k}\}.$
We can compute graph $G_1,G_2,\cdots,G_k$ where the vertices of $G_i$ is obtained by contracting all the edges from $e_1$ to $e_{(i-1)\cdot m/k},$ the edges of $G_i$ are corresponding to the edges in $E_i.$
Then by Lemma~\ref{lem:mst_divide_and_conquer}, we can obtain the whole minimum spanning forest by solving these $k$ size $O(m/k)$ minimum spanning forest problems.
For each sub-problem, we can assign it $\Theta(m)$ working space, thus each sub-problem still has $\Theta(k)$ factor more total space.
Therefore, we can recursively apply the above argument.

\begin{theorem}\label{thm:fomal_statement_of_exact_MST}
For any $\gamma\in[0,2]$ and any constant $\delta\in(0,1),$ there is a randomized $(\gamma,\delta)-\MPC$ algorithm which can output a minimum spanning forest for any weighted graph $G=(V,E)$ with weights $w:E\rightarrow \mathbb{Z}$ in $O(\min(\log D\cdot \log (1/\gamma'),\log n)\cdot1/\gamma')$ parallel time, where $n=|V|,$ $\forall e\in E,|w(e)|\leq \poly(n),$ $D$ is the diameter of a minimum spanning forest of $G,$ and $\gamma'=\gamma/2+\Theta(1/\log n).$ The success probability is at least $0.98.$ In addition, if the algorithm fails, then it will return {\rm FAIL}.
\end{theorem}
\begin{proof}
Let $n=|V|,m=|E|.$
Let $E=\{e_1,\cdots,e_m\}$ with $w(e_1)\leq w(e_2)\leq \cdots\leq w(e_m).$
The total space in the system is $\Theta(m^{1+\gamma}).$
Let $k=\Theta(m^{\gamma/2}).$
By our previous discussion, we can divide $E$ into $k$ groups $E_1,E_2,\cdots,E_k,$ where $E_i=\{e_{(i-1)\cdot m/k+1},e_{(i-1)\cdot m/k+2},\cdots,e_{i\cdot m/k}\}.$
By Lemma~\ref{lem:mst_depth_dimater} and Theorem~\ref{thm:formal_statement_of_connectivity}, we can use $O(\min(\log D\cdot\log (1/\gamma'),\log n))$ parallel time and $\Theta(km^{1+\gamma/2})$ total space to compute graph $G_1,G_2,\cdots,G_k$ where the vertices of $G_i$ is obtained by contracting all the edges from $e_1$ to $e_{(i-1)\cdot m/k},$ the edges of $G_i$ are corresponding to the edges in $E_i$ after contraction.

By Lemma~\ref{lem:mst_divide_and_conquer}, it suffices to recursively solve the minimum spanning forest problem for each group $G_i.$
Since each time, we split the edges into $k$ groups, the recursion will have at most $O(1/\gamma')$ levels.
At the end of the recursion, we are able to determine for every edge $e$ whether $e$ is in the minimum spanning forest.

Now let us consider the success probability.
Although Theorem~\ref{thm:formal_statement_of_connectivity} is a randomized algorithm, the parallel time is always bounded by $\min(\log D\cdot\log (1/\gamma'),\log n).$
If we repeat the algorithm until it succeeds, the expectation of number of trials is a constant.
Furthermore, for each level of the recursion, we can regard the graphs in all the tasks composed one large graph.
Thus, in real implementation, in each level of the recursion, we will only invoke one connectivity procedure.
Thus in expectation, the total parallel time is $O(\min(\log D\cdot\log (1/\gamma'),\log n)\cdot 1/\gamma').$
By applying Markov's inequality, we complete the proof.
\end{proof}

In the following theorem, we show that Lemma~\ref{lem:mst_divide_and_conquer} can also be applied in approximate minimum spanning forest problem.

\begin{theorem}\label{thm:formal_statement_of_approximate_MST}
For any $\gamma\in[\beta,2]$ and any constant $\delta\in(0,1),$ there is a randomized $(\gamma,\delta)-\MPC$ algorithm which can output a $(1+\varepsilon)$ approximate minimum spanning forest for any weighted graph $G=(V,E)$ with weights $w:E\rightarrow \mathbb{Z}_{\geq 0}$ in $O(\min(\log D\cdot \log (1/\gamma'),\log n))$ parallel time, where $n=|V|,$ $N=|V|+|E|,$ $\beta=\Theta(\log( \varepsilon^{-1} \log n )/\log n),$ $\forall e\in E,|w(e)|\leq \poly(n),$ $D$ is the diameter of a minimum spanning forest of $G,$ and $\gamma'=(1+\gamma-\beta)\log_n \frac{2N}{n^{1/(1+\gamma-\beta)}}.$ The success probability is at least $0.98.$ In addition, if the algorithm fails, then it will return {\rm FAIL}.
\end{theorem}
\begin{proof}
For each edge $e\in E$, we can round $w(e)$ to $w'(e)$ such that $w'(e)=0$ when $w(e)=0,$ and $w'(e)=(1+\varepsilon)^i$ when $w(e)\not=0,$ and $i$ is the smallest integer such that $w(e)\leq (1+\varepsilon)^i.$

Since $|w(e)|\leq \poly(n)$ for all $e\in E,$ there are only $k=O(\log(n)/\varepsilon)$ different values of $w'(e).$
We can divide $E$ into $k$ groups, where the $i^{\text{th}}$ group $E_i$ contains all edges with the $i^{\text{th}}$ largest weight in $w'$.
By Lemma~\ref{lem:mst_depth_dimater} and Theorem~\ref{thm:formal_statement_of_connectivity}, we can use $O(\min(\log D\cdot\log (1/\gamma'),\log n))$ parallel time and $\Theta(kN^{1+\gamma-\beta})=\Theta(N^{1+\gamma})$ total space to compute graph $G_1,G_2,\cdots,G_k$ where the vertices of $G_i$ is obtained by contracting all the edges from $E_1$ to $E_{i-1},$ the edges of $G_i$ are corresponding to the edges in $E_i$ after contraction.

Then, for each $G_i$, since all the edges have the same $w'$ weight, any spanning forest of $G_i$ is a minimum spanning forest of $G_i.$
By Theorem~\ref{thm:formal_statement_of_spanning_tree}, we can use $O(\min(\log D\cdot\log (1/\gamma'),\log n))$ parallel time and $\Theta(kN^{1+\gamma-\beta})=\Theta(N^{1+\gamma})$ total space to compute the spanning forest for each graph $G_1,G_2,\cdot,G_k.$
By Lemma~\ref{lem:mst_divide_and_conquer}, the union of all the minimum spanning forest with respect to $w'$ must be the minimum spanning forest of $G$ with respect to $w'.$
Since all the weights $w$ are nonnegative integers, $w'$ is a $(1+\varepsilon)$ approximation to $w$.
Therefore, our output minimum spanning forest with respect to $w'$ is a $(1+\varepsilon)$ approximation to the minimum spanning forest with respect to $w$.

For the success probability, we can apply the similar argument made in the proof of Theorem~\ref{thm:fomal_statement_of_exact_MST} to prove that the success probability is at least $0.98.$
\end{proof}

In the following, we show that if we only need to find the largest edge in the minimum spanning tree, then we are able to get a better parallel time. It is an another application of our

\begin{theorem}\label{thm:formal_statement_of_bottleneck_spanning_tree}
For any $\gamma\in[0,2]$ and any constant $\delta\in(0,1),$ there is a randomized $(\gamma,\delta)-\MPC$ algorithm which can output a bottleneck spanning forest for any weighted graph $G=(V,E)$ with weights $w:E\rightarrow \mathbb{Z}$ in $O(\min(\log D\cdot \log (1 / \gamma'),\log n)\cdot \log (1/\gamma'))$ parallel time, where $n=|V|,$ $\forall e\in E,|w(e)|\leq \poly(n),$ $D$ is the diameter of a minimum spanning forest of $G,$ and $\gamma'=\gamma/2+\Theta(1/\log n).$ The success probability is at least $0.98.$ In addition, if the algorithm fails, then it will return {\rm FAIL}.
\end{theorem}
\begin{proof}
Let $n=|V|,m=|E|.$
Let $E=\{e_1,\cdots,e_m\}$ with $w(e_1)\leq w(e_2)\leq \cdots\leq w(e_m).$
The total space in the system is $\Theta(m^{1+\gamma}).$
Let $k=\Theta(m^{\gamma/2}).$
By our previous discussion, we can divide $E$ into $k$ groups $E_1,E_2,\cdots,E_k,$ where $E_i=\{e_{(i-1)\cdot m/k+1},e_{(i-1)\cdot m/k+2},\cdots,e_{i\cdot m/k}\}.$
By Lemma~\ref{lem:mst_depth_dimater} and Theorem~\ref{thm:formal_statement_of_connectivity}, we can use $O(\min(\log D\cdot\log (1/\gamma'),\log n))$ parallel time and $\Theta(km^{1+\gamma/2})$ total space to compute graph $G_1,G_2,\cdots,G_k$ where the vertices of $G_i$ is obtained by contracting all the edges from $e_1$ to $e_{(i-1)\cdot m/k},$ the edges of $G_i$ are corresponding to the edges in $E_i$ after contraction.

By Lemma~\ref{lem:mst_divide_and_conquer}, the edge with largest weight must be in the group $E_i$ for some $i$ with $G_{i+1}=G_{i+2}.$
Thus, we reduce the problem size to $m/k.$
By Remark~\ref{rem:double_exponential_progress}, we can finish the recursion in $O(\log(1/\gamma'))$ phases.

Suppose the bottleneck is $e_i,$ then by Theorem~\ref{thm:formal_statement_of_spanning_tree}, we can find a spanning forest by only using edges from $\{e_1,\cdots,e_i\}$ in $O(\min(\log D\cdot\log (1/\gamma'),\log n))$ parallel time and in $\Theta(m^{1+\gamma/2})$ total space.
Thus, the resulting spanning forest is a bottleneck spanning forest.

For the success probability, we can apply the similar argument made in the proof of Theorem~\ref{thm:fomal_statement_of_exact_MST} to prove that the success probability is at least $0.98.$
\end{proof}
\section{Directed Reachability vs. Boolean Matrix Multiplication}

In this section, we discuss the directed graph reachability problem which is a directed graph problem highly related to the undirected graph connectivity.
In the all-pair directed graph reachability problem, we are given a directed graph $G=(V,E),$ the goal is to answer for every pair $(u,v)\in V\times V$ whether there is a directed path from $u$ to $v$.
There is a simple standard way to reduce Boolean Matrix Multiplication to all-pair directed graph reachability problem. 
In the Boolean Matrix Multiplication problem, we are given two boolean matrices $A,B\in\{0,1\}^{n\times n},$ the goal is to compute $C=A\cdot B,$ where $\forall i,j\in[n],$ $C_{i,j}=\bigvee_{k\in[n]} A_{i,k}\wedge B_{k,j}.$
The reduction is as the following. 
We create $3n$ vertices $u_1,u_2,\cdots,u_n,v_1,v_2,\cdots,v_n,w_1,w_2,\cdots,w_n.$
For every $i,j\in[n],$ if $A_{i,j}=1,$ then we add an edge from $u_i$ to $v_j,$ and if $B_{i,j}=1,$ then we add an edge from $v_i$ to $w_j.$
Thus, $C_{i,j}=1$ is equivalent to there is a path from $u_i$ to $w_j$.
Thus, if we can solve all-pair directed graph reachability problem in $O(T)$ sequential time, then we can solve Boolean Matrix Multiplication in $O(T)$ time.
For the current status of sequential running time of Boolean Matrix Multiplication problem, we refer readers to~\cite{l14} and the references therein.

Now, consider the multi-query directed graph reachability problem.
In this problem, we are given a directed graph $G=(V,E)$ together with $|V|+|E|$ queries where each query queries the reachability from vertex $u$ to vertex $v$. 
The goal is to answer all these queries.
A similar problem in the undirected graph is called multi-query undirected graph connectivity problem.
In this problem, we are given an undirected graph $G=(V,E)$ together with $|V|+|E|$ queries where each query queries the connectivity between vertex $u$ and vertex $v$.

According to Theorem~\ref{thm:formal_statement_of_connectivity} and Lemma~\ref{lem:multi_query}, there is a polynomial local running time fully scalable $\sim\log D$ parallel time $(0,\delta)-\MPC$ algorithm for multi-query undirected graph connectivity problem.
Here polynomial local running time means that there is a constant $c>0$ (independent from $\delta$) such that every machine in one round can only have $O((n^{\delta})^c)$ local computation.

For multi-query directed graph reachability problem, we show that if there is a polynomial local running time fully scalable $(\gamma,\delta)-\MPC$ algorithm which can solve multi-query reachability problem in $O(n^{\alpha})$ parallel time, then we can solve all-pair directed graph reachability problem in $O(n^2\cdot n^{2\gamma+\alpha+\varepsilon})$ sequential running time for any arbitrarily small constant $\varepsilon>0.$ 
Especially, if the algorithm is in $(0,\delta)-\MPC$ model, and the parallel time is $n^{o(1)},$ then we will have an $O(n^{2+\varepsilon+o(1)})$ sequential running time algorithm for Boolean Matrix Multiplication which implies a break through in this field.

Suppose we have a such $\MPC$ algorithm. 
Let the input size be $\Theta(m),$ i.e. the number of edges is $\Theta(m),$ and the number of queries is also $\Theta(m).$ 
Then the total space is $\Theta(m^{1+\gamma}).$
Let $\delta=\varepsilon/(c-2).$
Then the number of machines is $\Theta(m^{1+\gamma-\delta}).$
Now we just simulate this $(\gamma,\delta)-\MPC$ algorithm sequentially, the total running time is $O(m^{1+\gamma-\delta}\cdot m^{c\delta}\cdot n^{\alpha})=O(m\cdot n^{2\gamma+\varepsilon+\alpha}).$
To answer reachability for all pairs, we need total $O(n^2\cdot m\cdot n^{2\gamma+\varepsilon+\alpha}/m)=O(n^2\cdot n^{2\gamma+\alpha+\varepsilon})$ time.
Therefore, we can use this algorithm to solve Boolean Matrix Multiplication in $O(n^2\cdot n^{2\gamma+\alpha+\varepsilon})$ time.

\begin{theorem}\label{thm:formal_statement_of_digraph_reachability}
If there is a polynomial local running time fully scalable $(\gamma,\delta)-\MPC$ algorithm which can answer $|V|+|E|$ pairs of reachability queries simultaneously for any directed graph $G=(V,E)$ in $O(|V|^{\alpha})$ parallel time, then there is a sequential algorithm which can compute the multiplication of two $n\times n$ boolean matrices in $O(n^2\cdot n^{2\gamma+\alpha+\varepsilon})$ time, where $\varepsilon>0$ is a constant which can be arbitrarily small.
\end{theorem}
\begin{proof}
See above discussions.
\end{proof}               
\section{Discussion on a Previous Conjectured Fast Algorithm}\label{sec:rastogi}
In this section, we discuss the hard example for the algorithm described by~\cite{rmcs13}.
In~\cite{rmcs13}, they conjectured that their Hash-to-Min connectivity algorithm can finish in $O(\log D)$ rounds.
The description of their algorithm is as the following:
\begin{enumerate}
\item The input graph is $G=(V,E).$
\item For each vertex $v\in V,$ initialize a set $S_v^{(0)}=v.$
\item in round $i$:
    \begin{enumerate}
    \item Each vertex $v$ find $u\in S_v^{(i-1)}$ which has the minimum label, i.e. $u=\min_{x\in S_v^{(i-1)}}x.$
    \item $v$ sends $u$ the all the vertices in $S_v^{(i-1)}.$
    \item $v$ sends every $x\in S_v^{((i-1))}\setminus\{u\}$ the vertex $u$.
    \item Let $S_v^{(i)}$ be $\{v\}$ union the set of all the vertices received.
    \item If for all $v,$ $S_v^{(i)}$ is the same as $S_v^{(i-1)},$ then finish the procedure.
    \end{enumerate}
\end{enumerate}
The above procedure can be seen as the modification of the graph: in each round, all the vertices together create a new graph.
For each vertex $v$, let $u$ be the neighbor of $v$ with the minimum label, and if $x$ is a neighbor of $v$, then add an edge between $x$ and $u$ in the new graph.
So in each round, each vertex just communicates with its neighbors to update the new minimum neighbor it learned.
At the end of the algorithm, it is obvious that the minimum vertex in each component will have all the other vertices in that component, and for each non minimum vertex, it will have the minimum vertex in the same component.

\begin{figure}[t!]
  \centering
  \includegraphics[width=1.0\textwidth]{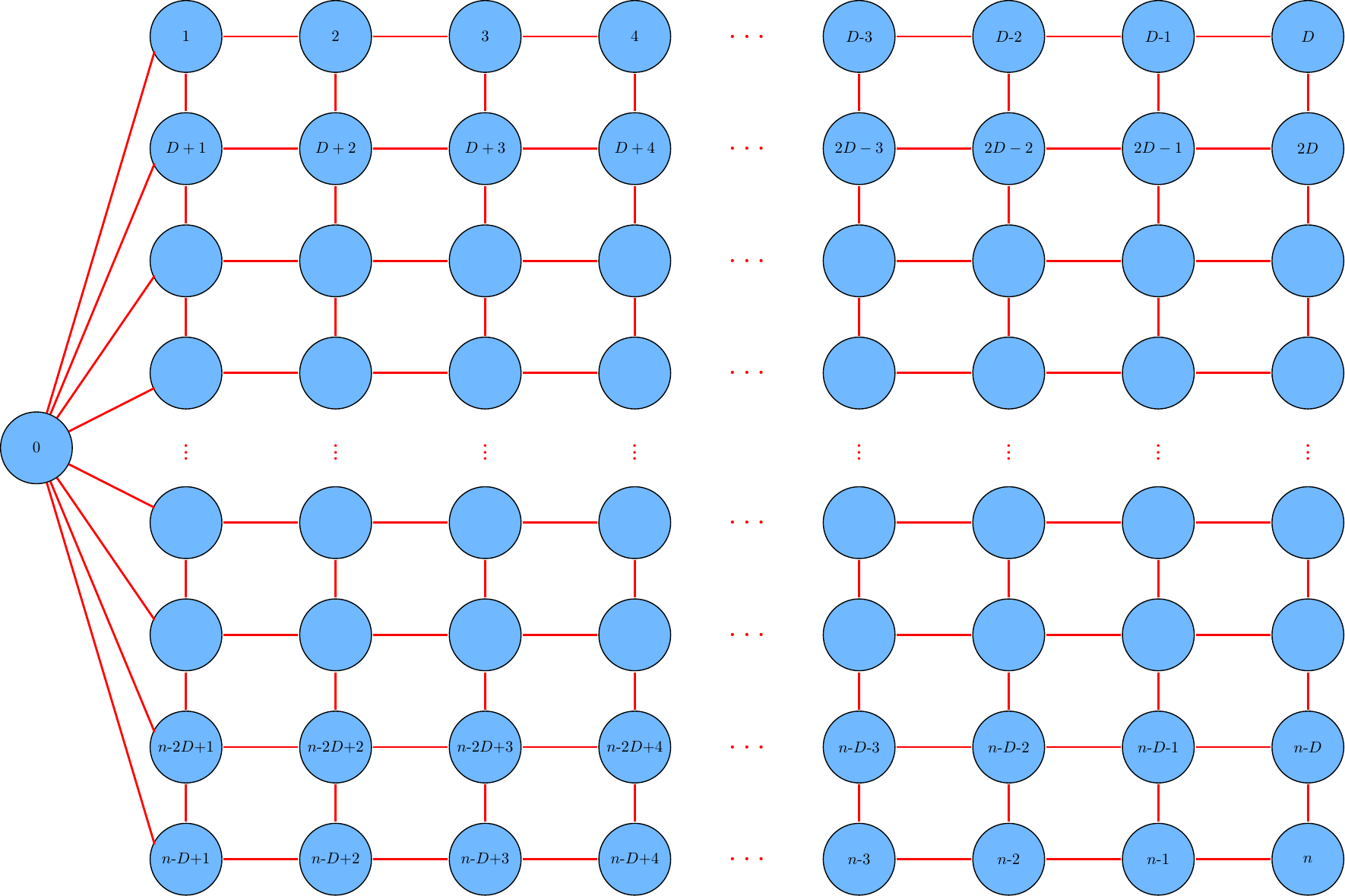}\\
  \caption{A hard example for~\cite{rmcs13}. For each $i \in \{2,3,\cdots, n/D-1\}$ and $j \in \{1,2,\cdots, D-1\}$, node $(i-1)\cdot D + j$ has degree $4$. For node $D$ and $n$, they have degree $2$. Node $0$ has degree $D$. All the other nodes have degree $3$.}\label{fig:hard_example}
\end{figure}

A hard example for this algorithm is shown by Figure~\ref{fig:hard_example}.
The example is a thin and tall grid graph with a vertex connected to all the vertices in the first column.
The total number of vertices is $n$.
The grid graph has $D=\frac{1}{2}\log n$ columns and $n/D$ rows.
We index each column from left to right by $1$ to $D$.
We index each row from top to down by $1$ to $n/D$.
The single large degree vertex has label $0$.
The $i^{\text{th}}$ row has the vertices with label $(i-1)\cdot D+1$ to $i\cdot D$ from the first column to the $D^{\text{th}}$ column.
We claim that if vertex $v$ is the $i^{\text{th}}$ row and $j^{\text{th}}$ column, then before round $k$ for $2^k<i,k<j,$ the neighbors of $v$ will only in column $j-1$, column $j$ and column $j+1$.
Furthermore, the minimum neighbor of $v$ in column $j-1$ will be $v-(2^{k-1}-1)\cdot D -1.$
The minimum neighbor of $v$ in column $j$ will be $v-2^{k-1}\cdot D.$
The minimum neighbor of $v$ in column $j+1$ will be $v-D\cdot(2^{k-1}-1)+1.$
This claim is true when $k=1.$
Then by induction, we can prove the claim.
Thus, it will take at least $\Theta(D)$ rounds to finish the procedure where $D=\Theta(\log n).$

If we randomly label the vertices at the beginning, then consider the case we copy that hard structure at least $n^{n+2}$ times, then with high probability, there is a component which has the labels with the order as the same as described above. In this case, the procedure needs $\Omega(\log\log N)$ rounds, where $N=n^{n+3}$ is the total number of the vertices.

Also notice that, even we give more total space to this algorithm, this algorithm will not preform better.
In our connectivity algorithm, if we have $\Omega(n^{1+\varepsilon})$ total space for some arbitrary constant $\varepsilon>0,$ then our parallel running time is $O(\log D).$

\section{Alternative Approach for Leader Selection}\label{sec:alternative_leader}
In this section, we show that there is a different way to select leaders (see Section~\ref{sec:leader}).
 The number of leaders selected by this approach will depend on the sum of inverse degrees of all the vertices.
Let us first introduce the concept of \textit{Min Parent Forest}.
\subsection{Min Parent Forest}
 Let $G=(V,E)$ be an undirected graph where $V$ denotes the vertex set of $G$, and $E$ denotes the edge set of $G$. Each vertex $v\in V$ has a weight $w(v)\in\mathbb{R},$ and it also has a unique label from $\mathbb{Z}.$ For convenience, for each vertex $v\in V,$ we also use $v$ to denote its label.  Let $\Gamma_G(v)$ denote the set of neighbors of $v$, i.e. $\Gamma_G(v)=\{u\in V\mid (u,v)\in E\}.$ If $G$ is clear in the context, we just use $\Gamma(v)$ to denote $\Gamma_G(v).$  The size of $\Gamma(v),$ $|\Gamma(v)|,$ is called the degree of $v$. Let $f_{G,w}:V\rightarrow V$ be the ``min-weight-parent'' function defined as the following:
 \begin{enumerate}
 \item If $w(v)=\min_{u\in \Gamma(v)\cup \{v\}} w(u),$ then $f_{G,w}(v)=v.$
 \item Otherwise, let $u^*\in \Gamma(v)$ be the vertex which has the smallest weight, i.e. $w(u^*)=\min_{u\in \Gamma(v)} w(u).$ If there is more than one choice of $u^*$, let $u^*$ be the one with the smallest label. And $f_{G,w}(v)$ is defined to be $u^*$.
 \end{enumerate}
 We call $(V,f_{G,w})$ the \textit{min-parent-forest} of graph $G$ with vertex weights $w$. We can then define $i$-step ``min-weight-parent'' function. For $v\in V,$ we define $f^{(0)}_{G,w}(v)=v.$ For $i\in \mathbb{Z}_{>0},$ we can define $f^{(i)}_{G,w}$ as the following:
 \begin{align*}
 \forall v\in V, f^{(i)}_{G,w}(v)=f_{G,w}(f^{(i-1)}_{G,w}(v)).
 \end{align*}

\begin{figure}[!t]
  \centering
    \includegraphics[width=1.0\textwidth]{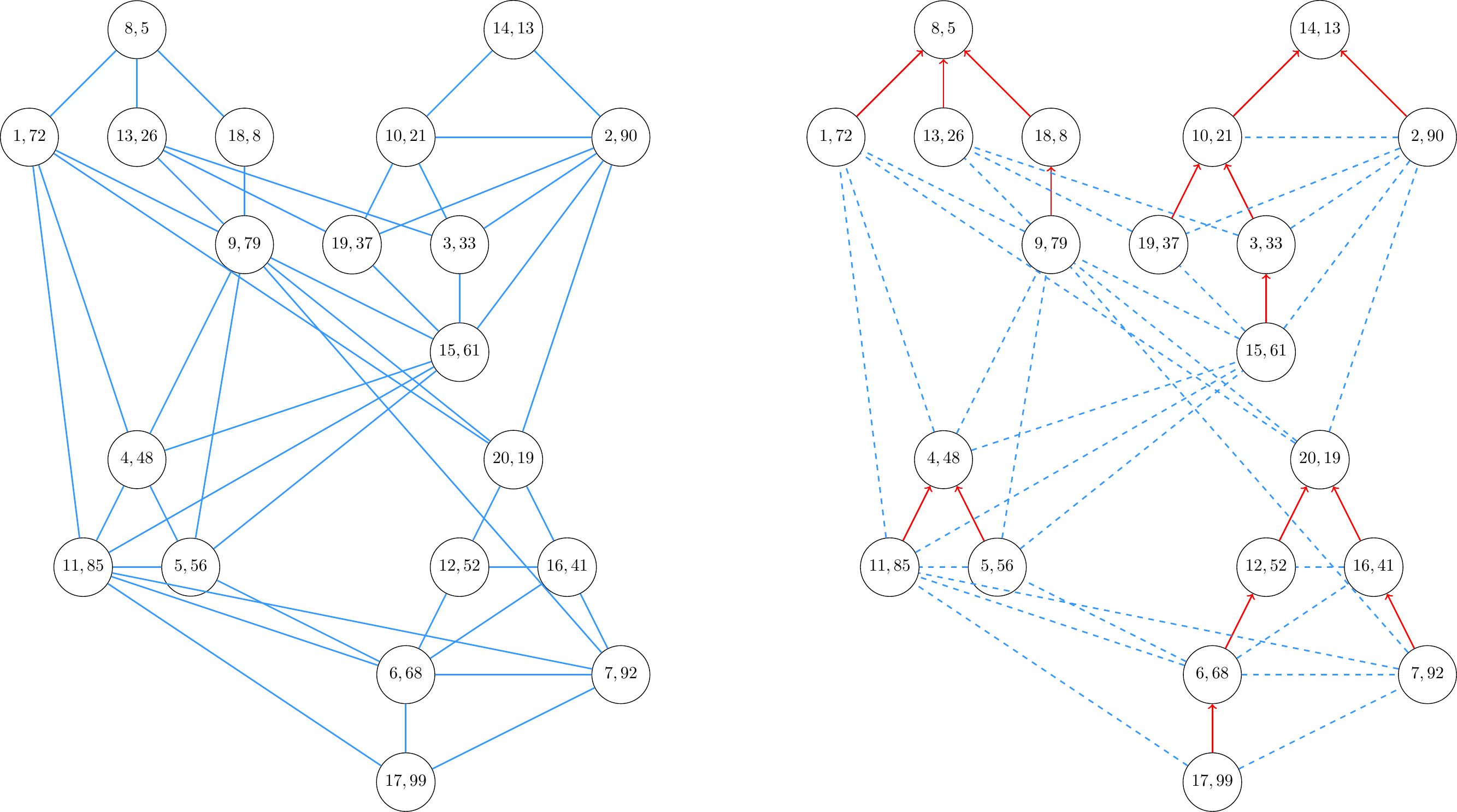}
    \caption{An example where $\#$roots $\approx$ $\sum_{i=1}^{20} 1 /(d(v_i) + 1) $. For each node, it has two numbers, the first number is the ID, and the second number is weight. $\sum_{i=1}^{20} 1/(d(v_i) + 1) = 1/4 + 1/3 + 1/3 + 1/6 + 1/5$ $+ 1/5 + 1/5 + 1/5 + 1/5 + 1/6$ $ + 1/4 + 1/6 + 1/8 + 1/7 + 1/6 $ $+1/9 + 1/8 + 1/7 + 1/6 + 1/4 \approx 3.89$ and $\#$roots$=4$. }
\end{figure}

In the following, we define the concept of roots in the \textit{min-parent-forest}.
\begin{definition}[Roots in the forest]
Let $v\in V,$ and let $(V,f_{G,w})$ be the \textit{min-parent-forest} of graph $G=(V,E)$ with vertex weights $w$. If $f_{G,w}(v)=v,$ then $v$ is a root in the forest $(V,f_{G,w}).$
\end{definition}

The depth of a vertex $v$ is defined as the distance on the tree between $v$ and the corresponding root in the forest.
 \begin{definition}[The depth of $v$]
 Let $v\in V,$ and let $(V,f_{G,w})$ be the \textit{min-parent-forest} of graph $G=(V,E)$ with vertex weights $w$. The depth of $v$ in the forest $(V,f_{G,w})$ is the smallest $i\in \mathbb{Z}_{\geq 0}$ such that $f^{(i)}_{G,w}(v)=f^{(i+1)}_{G,w}(v).$ We use $\dep_{G,w}(v)$ to denote the depth of $v$ in $(V,f_{G,w}).$ We call $f^{(\dep_{G,w}(v))}_{G,w}(v)$ the root of $v.$ For the simplicity of the notation, we also use $f^{(\infty)}_{G,w}(v)$ to denote the root of $v$.
 \end{definition}

 The above definition is well defined since if $f^{(i+1)}_{G,w}(v)\not=f^{(i)}_{G,w}(v)$ then $w(f^{(i+1)}_{G,w}(v))$ should be strictly smaller than $w(f^{(i)}_{G,w}(v))$ by the definition of $f_{G,w}$ and $f^{(j)}_{G,w}$ for all $j\in \mathbb{Z}_{\geq 0}.$ Therefore, there must exist $i$ such that $f^{(i)}_{G,w}(v)=f^{(i+1)}_{G,w}(v).$

The depth of the forest is the largest depth among all the vertices.

 \begin{definition}[The depth of the \textit{min-parent-forest}]
 The depth $\dep(G,w)$ of the forest $(V,f_{G,w})$ is defined as:
 \begin{align*}
 \dep(G,w)=\max_{v\in V} \dep_{G,w}(v).
 \end{align*}
 \end{definition}

 If the weights $w$ of vertices of $G$ are some i.i.d. random variables, then with high probability, the depth of $(V,f_{G,w})$ is only $O(\log |V|).$ Precisely, we have the following Lemma.

 \begin{lemma}[The depth of the random \textit{min-parent-forest}]\label{lem:depth_of_random_forest}
 Let $G=(V,E)$ be an undirected graph with $n$ vertices where $V=\{v_1,v_2,\cdots,v_n\},$ and the labels satisfies $v_1<v_2<\cdots<v_n.$ Let $w(v_1),w(v_2),\cdots,w(v_n)$ be $n$ i.i.d. random variables drawn uniformly from $[N].$ If $N>n^2/\delta$ for some $\delta\in(0,1),$ then for any $t\geq 60\log n,$
 \begin{align*}
 \Pr_{w \sim [N]^n} \left(\dep(G,w)\leq t\right)\geq 1-\delta-e^{-\frac{1}{2}t}.
 \end{align*}
 \end{lemma}

 \begin{proof}
Let $w(v_1),w(v_2),\cdots,w(v_n)$ be $n$ i.i.d. random variables drawn uniformly from $[N].$ Let $(V,f_{G,w})$ be the \textit{min-parent-forest} of $(G,w).$
 For a fixed $s\in V,$ we create a set of random variables $z_1,z_2,\cdots,z_n$ by the following deterministic procedure:
 \begin{enumerate}
 \item Let $z_1=w(s),k=0,S_k=\{s\},u_k=s,i=2,\pos(s)\leftarrow 1.$
 \item Let $S_{k+1}=S_k.$
 \item For $j=1\rightarrow n,$
 \subitem if $v_j\in \Gamma (u_k)$ and $v_j\not\in S_k$ then let $\pos(v_j)\leftarrow i,S_{k+1}\leftarrow S_{k+1}\cup\{v_j\},z_i=w(v_j),$ $i\leftarrow i+1.$
 \item If $f_{G,w}(u_k)\not=u_k,$ then let $u_{k+1}=f_{G,w}(u_{k}),$ $k\leftarrow k+1$ and go to step $2$.
 \item Otherwise, for $j=1\rightarrow n,$
 \subitem if $v_j\not\in S_{k+1}$ then let $\pos(v_j)\leftarrow i,z_i=w(v_j),i\leftarrow i+1.$
 \end{enumerate}

 It is easy to observe that $k$ is exactly $\dep_{G,w}(s)$ at the end of the above procedure. The reason is that $u_0=s=f^{(0)}_{G,w}(s),$  $\forall j\in[k],u_j=f_{G,w}(u_{j-1})=f^{(j)}_{G,w}(s)$ and $f_{G,w}(u_k)=u_k.$

\begin{fact}
$\forall v \in V$, $w(v) = z_{\mathrm{pos}(v)}$, where $\mathrm{pos} : [V] \rightarrow [n]$ and $\mathrm{pos}^{-1} : [n] \rightarrow [V]$.
\end{fact}

 \begin{claim}\label{cla:des_setS}
 $\forall j\in\{0,1,\cdots,k+1\},S_j=\{u_0\}\cup\bigcup_{p=0}^{j-1}\Gamma(u_p).$
 \end{claim}
 \begin{proof}
 We can prove this by induction. The statement is obviously true for $S_0$ since $S_0=\{u_0\}.$ Now suppose the claim is true for $S_{j-1}.$ Then according to the step 3 of the procedure $S_{j}=S_{j-1}\cup(\Gamma(u_{j-1})\setminus S_{j-1})=S_{j-1}\cup\Gamma(u_{j-1})=\{u_0\}\cup\bigcup_{p=0}^{j-1}\Gamma(u_p).$
 \end{proof}

 \begin{claim}\label{cla:minwuj}
 $\forall j\in\{0,1,\cdots,k\},$ $w(u_j)=\min_{v\in S_{j}} w(v).$
 \end{claim}
 \begin{proof}
 Since $\forall j\in[k],u_j=f_{G,w}(u_{j-1}),$ we have $w(u_j)=\min_{v\in\Gamma(u_{j-1})\cup \{u_{j-1}\}}w(v).$ Then we have $w(u_j)=\min_{v\in\{u_0\}\cup\bigcup_{p=0}^{j-1}\Gamma(u_p)}w(v)=\min_{v\in S_{j}} w(v),$ where the last equality follows by Claim~\ref{cla:des_setS}.
 \end{proof}

 We use $\pos^{-1}(i)$ to denote vertex $v$ which satisfies $\pos(v)=i.$ According to the step 3, it is easy to see $\forall j\in\{0,1,\cdots,k+1\},$ we have $\{\pos^{-1}(i)\mid i\in[|S_j|]\}=S_j.$

\begin{claim}\label{cla:smallest_z}
$\forall j\in\{0,1,\cdots,k\},z_{\pos(u_j)}=\min_{p\in[\pos(u_j)]} z_{p}.$
\end{claim}
\begin{proof}
$z_{\pos(u_j)}=w(u_j)=\min_{v\in S_{j}} w(v)=\min_{v\in S_{j}}z_{\pos(v)}=\min_{p\in[|S_j|]}z_{p}=\min_{p\in[\pos(u_j)]}z_{p},$ where the second equality follows by Claim~\ref{cla:minwuj}, and the last equality follows by $u_j\in S_j$, so $\pos(u_j)\leq |S_j|$.
\end{proof}

Now we define an another set of random variables $y_1,y_2,\cdots,y_n,$ where $\forall i\in[n],y_i\in\{0,1\}$ and $y_i=1$ if and only if $z_i=\min_{j\in[i]} z_j.$ According to Claim~\ref{cla:smallest_z}, we have that $\forall i\in\{0,1,\cdots,k\},y_{\pos(u_i)}=1.$ Thus, $\dep_{G,w}(s)=k\leq \sum_{i=1}^n y_i.$ To upper bound $\dep_{G,w}(s),$ it suffices to upper bound $\sum_{i=1}^n y_i.$

Before we look at $y_1,\cdots,y_n,$ we firstly focus on the properties of $z_1,\cdots,z_n:$
\begin{claim}\label{cla:independence_of_z}
$z_1,z_2,\cdots,z_n$ are $n$ i.i.d random variables drawn uniformly from $[N]$.
\end{claim}
\begin{proof}
 A key observation is that if $z_1,z_2,\cdots,z_n$ are given, then we can recover $w(v_1),w(v_2),\cdots,w(v_n)$ exactly by the following deterministic procedure:
 \begin{enumerate}
 \item Let $w(s)=z_1,k=0,S_k=\{s\},u_k=s,i=2.$
 \item Let $S_{k+1}=S_k.$
 \item For $j=1\rightarrow n,$
 \subitem if $v_j\in \Gamma (u_k)$ and $v_j\not\in S_k$ then let $S_{k+1}\leftarrow S_{k+1}\cup\{v_j\},w(v_j)=z_i,$ $i\leftarrow i+1.$
 \item If $f_{G,w}(u_k)\not=u_k,$ then let $u_{k+1}=f_{G,w}(u_{k}),$ $k\leftarrow k+1$ and go to step $2$.
 \item Otherwise, for $j=1\rightarrow n,$
 \subitem if $v_j\not\in S_{k+1}$ then let $\pos(v_j)\leftarrow i,w(v_j)=z_i,i\leftarrow i+1.$
 \end{enumerate}
 Notice that after step 3, $\forall v\in \Gamma(u_k)\cup \{u_k\},$ $w(v)$ is already recovered, thus we can implement step 4. Thus, the above procedure is a valid procedure. Since $z_1,\cdots,z_n$ are generated by $w(v_1),\cdots,w(v_n),$ we can also know $z_1,\cdots,z_n$ by given  $w(v_1),\cdots,w(v_n).$ This means that \begin{align*}
 H(z_1,z_2,\cdots,z_n\mid w(v_1),w(v_2),\cdots,w(v_n))=H( w(v_1),w(v_2),\cdots,w(v_n)\mid z_1,z_2,\cdots,z_n)=0,
 \end{align*}
 where $H(\cdot)$ is the information entropy.
 Notice that
 \begin{align*}
 &I(z_1,z_2,\cdots,z_n; w(v_1),w(v_2),\cdots,w(v_n))\\
 &=H(z_1,z_2,\cdots,z_n)-H(z_1,z_2,\cdots,z_n\mid w(v_1),w(v_2),\cdots,w(v_n))\\
 &=H(w(v_1),w(v_2),\cdots,w(v_n))-H( w(v_1),w(v_2),\cdots,w(v_n)\mid z_1,z_2,\cdots,z_n),
 \end{align*}
 where $I(\cdot)$ is the mutual information.
 Thus, $H(z_1,z_2,\cdots,z_n)=H(w(v_1),w(v_2),\cdots,w(v_n))=n\log N.$ For $i\in[n],$ since the size of the support of $z_i$ is at most $N,$ $H(z_i)\leq \log N$ where the equality holds if and only if $z_i$ is uniformly distributed on $[N].$ Also notice that $H(z_1,z_2,\cdots,z_n)\leq \sum_{i=1}^n H(z_i),$ where the equality holds if and only if $z_i$ are independent. Since $\sum_{i=1}^n H(z_i)\leq n\log N,$ we have $H(z_1,z_2,\cdots,z_n)=\sum_{i=1}^n H(z_i),$ and for each $i\in[n],$ $H(z_i)=\log N.$ Thus, $z_1,z_2,\cdots,z_n$ are i.i.d. random variables drawn uniformly from $[N].$
\end{proof}

\begin{claim}\label{cla:no_collision}
If $N>n^2/\delta$ for some $\delta\in(0,1),$ then with probability at least $1-\delta,$ $\forall i\not=j\in[n],$ we have $w(v_i)\not=w(v_j).$
\end{claim}
\begin{proof}
 Recall that $w(v_1),w(v_2),\cdots,w(v_n)$ are $n$ i.i.d. random variables drawn uniformly from $N$. For any $i\not=j\in [n],$ the $\Pr(w(v_i)\not=w(v_j))=1/N,$ thus $\E(|\{(i,j)\in[n]\times [n]\mid i\not=j,w(v_i)\not=w(v_j)\}|)\leq n^2/N.$ By Markov's inequality,
\begin{align*}
\Pr(|\{(i,j)\in[n]\times [n]\mid i\not=j,w(v_i)\not=w(v_j)\}|\geq 1)\leq n^2/N\leq \delta.
\end{align*}
Thus,
\begin{align*}
\Pr(\forall i\not=j\in[n],z_i\not=z_j)\geq 1-\delta.
\end{align*}
\end{proof}

\begin{claim}\label{cla:upper_bound_of_y}
Let $\mathcal{E}$ be the event that $\forall i\not=j\in[n],w(v_i)\not=w(v_j).$ Then, for any $t\geq 3\sum_{i=1}^n\frac{1}{i},$ we have
\begin{align*}
\Pr_{w \sim [N]^n} \left(\sum_{i=1}^n y_i\geq t+\sum_{i=1}^n\frac{1}{i} ~\bigg|~ \mathcal{E}\right)\leq e^{-\frac{3}{4}t}.
\end{align*}
\end{claim}
\begin{proof}
  Note that $\mathcal{E}$ happened if and only if we have $\forall i\not=j\in[n], z_i\not=z_j.$ Due to Claim~\ref{cla:independence_of_z}, $z_1,z_2,\cdots,z_n$ are i.i.d. random variables drawn uniformly from $[N],$ then conditioned on $\mathcal{E},$ $y_1,y_2, \cdots,$ $y_n$ are independent, and the probability that $y_i=1$ is $1/i.$ Thus, we have:
\begin{align*}
 & ~ \Pr\left(\sum_{i=1}^n y_i\geq\sum_{i=1}^n \frac{1}{i}+t ~\bigg|~ \mathcal{E}\right)\\
= & ~ \Pr \left(\sum_{i=1}^n (y_i-\E\left( y_i\mid \mathcal{E}\right))\geq t ~\bigg|~ \mathcal{E}\right)\\
\leq & ~ \exp \left( {-\frac{\frac{1}{2}t^2}{\sum_{i=1}^n \Var(y_i\mid\mathcal{E})+\frac{1}{3}t}} \right)\\
\leq & ~ \exp \left( {-\frac{\frac{1}{2}t^2}{\sum_{i=1}^n \frac{1}{i}+\frac{1}{3}t}} \right)\\
\leq & ~ \exp \left( {-\frac{\frac{1}{2}t^2}{\frac{2}{3}t}} \right)\\
= & ~ \exp \left( {-\frac{3}{4}t} \right),
\end{align*}
where the first equality follows by $\E(y_i | \mathcal{E})=1/i.$ The first inequality follows by Berinstein inequality. The second inequality follows by
\begin{align*}
\sum_{i=1}^n \Var(y_i\mid \mathcal{E})\leq \sum_{i=1}^n \E(y_i^2\mid \mathcal{E})=\sum_{i=1}^n \E(y_i\mid \mathcal{E})=\sum_{i=1}^n \frac{1}{i}.
\end{align*}
The third inequality follows by $\sum_{i=1}^n \frac{1}{i}\leq \frac{1}{3}t.$
\end{proof}

For a fixed vertex $s\in V,$ due to Claim~\ref{cla:upper_bound_of_y}, for any $t\geq 3\sum_{i=1}^n 1/i,$ we have
\begin{align}\label{eq:upper_bound_of_dep}
\Pr\left(\dep_{G,w}(s)\geq \sum_{i=1}^n 1/i+t ~\bigg|~ \mathcal{E}\right)\leq e^{-\frac{3}{4}t}.
\end{align}
Thus, for any $t\geq 60\log n,$
\begin{align*}
& ~\Pr_{w \sim [N]^n } \left(\exists s\in V \text{~s.t.~} \dep_{G,w}(s)\geq t\right)\\
\leq & ~ \Pr\left(\exists s\in V \text{~s.t.~} \dep_{G,w}(s)\geq 5t/6 + \sum_{i=1}^n 1/i \right)\\
= & ~\Pr\left(\exists s\in V \text{~s.t.~} \dep_{G,w}(s)\geq 5t/6 + \sum_{i=1}^n 1/i  ~\bigg|~ \mathcal{E}\right)\Pr(\mathcal{E}) \\
& ~ +\Pr\left(\exists s\in V \text{~s.t.~} \dep_{G,w}(s)\geq 5t/6 + \sum_{i=1}^n 1/i  ~\bigg|~ \neg\mathcal{E}\right)\Pr(\neg\mathcal{E})\\
\leq & ~ \Pr\left(\exists s\in V \text{~s.t.~} \dep_{G,w}(s)\geq 5t/6 + \sum_{i=1}^n 1/i ~\bigg|~ \mathcal{E}\right)+\Pr(\neg\mathcal{E})\\
\leq & ~ \Pr\left(\exists s\in V \text{~s.t.~} \dep_{G,w}(s)\geq 5t/6 + \sum_{i=1}^n 1/i  ~\bigg|~ \mathcal{E}\right)+\delta\\
\leq & ~ \sum_{s\in V} \Pr\left(\dep_{G,w}(s)\geq 5t/6 + \sum_{i=1}^n 1/i ~\bigg|~ \mathcal{E}\right)+\delta\\
\leq & ~ ne^{-\frac{5}{8}t}+\delta\\
\leq & ~ e^{-\frac{1}{2}t}+\delta
\end{align*}
where the first inequality follows by $\frac{1}{6}t\geq 10\log n\geq \sum_{i=1}^n 1/i.$ The third inequality follows by Claim~\ref{cla:no_collision}. The forth inequality follows by union bound. The fifth inequality follows by Equation~\eqref{eq:upper_bound_of_dep}. The sixth inequality follows by $e^{-\frac{1}{8}t}\leq \frac{1}{n}.$

Thus, we can conclude that for any $t\geq 60\log n,$ we have $\Pr(\dep(G,w)\leq t)\geq 1-\delta-e^{-\frac{1}{2}t}.$
 \end{proof}

\begin{lemma}[The number of roots of the random \textit{min-parent-forest}]\label{lem:number_of_roots}
 Let $G=(V,E)$ be an undirected graph with $n$ vertices where $V=\{v_1,v_2,\cdots,v_n\},$ and the labels satisfies $v_1<v_2<\cdots<v_n.$ Let $w(v_1),w(v_2),\cdots,w(v_n)$ be $n$ i.i.d. random variables drawn uniformly from $[N].$ Let $\delta\in(0,1).$ If $N>n^3,$ then
 \begin{align*}
 \Pr_{ w \sim [N]^n }\left(|\{v\in V ~|~ f_{G,w}(v)=v\}|\geq \frac{2}{\delta}\sum_{v\in V}\frac{1}{|\Gamma(v)|+1}\right)\leq \delta.
 \end{align*}
 \end{lemma}

 \begin{proof}
 Let $w(v_1),w(v_2),\cdots,w(v_n)$ be $n$ i.i.d. random variables drawn uniformly from $[N].$ Let $\mathcal{E}$ be the event that $\forall i\not=j\in [n],w(v_i)\not=w(v_j).$ Notice that for $i\not=j,$ the probability that $w(v_i)=w(v_j)$ is $1/N.$ Thus, $\E(|\{(i,j)\in[n]\times [n] ~|~ i\not=j,w(v_i)=w(v_j)\}|)\leq n^2/N.$ Thus, if $N>n^2,$ then $\Pr(\neg\mathcal{E})=\Pr\left(|\{(i,j)\in[n]\times [n] ~|~ i\not=j,w(v_i)=w(v_j)\}|\geq 1\right)\leq n^2/N\leq \frac{1}{n}.$ Now, we fix a vertex $v\in V,$
 \begin{align*}
 & ~ \Pr\left(f_{G,w}(v)=v\right)\\
 = & ~ \Pr\left(f_{G,w}(v)=v\mid\mathcal{E}\right)\Pr(\mathcal{E})+\Pr\left(f_{G,w}(v)=v\mid\neg\mathcal{E}\right)\Pr(\neg\mathcal{E})\\
 \leq & ~ \Pr\left(f_{G,w}(v)=v\mid\mathcal{E}\right)+\Pr(\neg\mathcal{E})\\
 \leq & ~ \Pr\left(w(v)=\min_{u\in\{v\}\cup\Gamma(v)}w(u)\mid\mathcal{E}\right)+\frac{1}{n}\\
 \leq & ~ \frac{1}{|\Gamma(v)|+1}+\frac{1}{n}\\
 \leq & ~ \frac{2}{|\Gamma(v)|+1}
 \end{align*}
 where the third inequality follows by the symmetry of all the variables $w(u)$ for $u\in\{v\}\cup\Gamma(v)$ so condition on all the $w$ are different, with probability $\frac{1}{1+|\Gamma(v)|},$ $w(v)$ is the smallest one. The last inequality follows by $|\Gamma(v)|+1\leq |V|=n.$

 Thus, $\E(|\{v\in V\mid f_{G,w}(v)=v\}|)\leq \sum_{v\in V}\frac{2}{|\Gamma(v)|+1}.$ Let $\delta\in(0,1),$ then by Markov's inequality,
  \begin{align*}
 \Pr\left(|\{v\in V\mid f_{G,w}(v)=v\}|\geq \frac{2}{\delta}\sum_{v\in V}\frac{1}{|\Gamma(v)|+1}\right)\leq \delta.
 \end{align*}
 \end{proof}

 \subsection{Leader Selection via Min Parent Forest}
 Given a graph, we can randomly assign each vertex a weight, thus we have a \textit{min-parent-forest}, then we select those roots in the \textit{min-parent-forest} as leaders, and try to contract all the vertices to the leaders.
 If we replace line~\ref{sta:alg2_start_to_select_leader} to line~\ref{sta:alg2_assign_pointer2} of Algorithm~\ref{alg:batch_algorithm2} by Algorithm~\ref{alg:alternative_leader}. We can get a new algorithm
 with the following guarantees.

 \begin{algorithm}[h]
 \caption{Leader Selection via Min Parent Forest}\label{alg:alternative_leader}
 \begin{algorithmic}[1]
 \State Let $N = 100rn^{10}.$
 \State $\forall v\in V'_i,$ let $w_i(v)$ be i.i.d. random variables drawn uniformly from $[N].$
\State $\forall v\in V''_i,$ let $\p_i(v)=f_{G'_i,w_i}(v).$ \Comment{$(V'_i,f_{G'_i,w_i})$ is a \textit{min-parent-forest} of $G'_i$ with $w_i$. }\label{sta:ftop}
\end{algorithmic}
\end{algorithm}

\begin{theorem}
Suppose we replace line~\ref{sta:alg2_start_to_select_leader} to line~\ref{sta:alg2_assign_pointer2} of Algorithm~\ref{alg:batch_algorithm2} by Algorithm~\ref{alg:alternative_leader}.

Let $G=(V,E)$ be an undirected graph, $m=\Omega(n),$ and $r\leq n$ be the rounds parameter where $n$ is the number of vertices in $G$. Let $c>0$ be a sufficiently large constant. If $r\geq c\log \log_{m/n} (n)$, then with probability at least $2/3$, the modified $\textsc{Connectivity}(G,m,r)$ (Algorithm~\ref{alg:batch_algorithm2}) will not return FAIL, and the total number of iterations (see Definition~\ref{def:total_iter_connectivity}) of the modified $\textsc{Connectivity}(G,m,r)$ is at most $O(r\cdot(\log D+\log\log n)),$ where $D=\diam(G).$
\end{theorem}

 \begin{proof}
 Let $k_i$ denote the number of iterations (see Definition~\ref{def:neighbor_incr_num_iter}) of $\textsc{NeighborIncrement}(m,G_{i-1})$.
 By Lemma~\ref{lem:properties_of_neighbor_increment}, we have $k_i\leq O(\log D).$
 Thus, $\sum_{i=1}^r k_i=O(r\cdot \log D).$

  According to Lemma~\ref{lem:depth_of_random_forest}, with probability at least $1-\frac{2}{100r},$ $\dep(G''_i,w_i)\leq O(\log n).$ 
  By Lemma~\ref{lem:contraction_properties}, with probability at least $1-\frac{2}{100r},$ the number of iteration of $\textsc{TreeContraction}(G_i'',\p_i)$ (see Definition~\ref{def:num_it_tree_contract}) $r'_i\leq O(\log \log n)$ 
   By taking union bound over all $i\in[r],$ then with probability at least $1-\frac{1}{50},$
   $\sum_{i=1}^r r'_i\leq O(r\cdot \log \log n).$


Due to the Property~\ref{itm:neighbor_incr_pro3} of Lemma~\ref{lem:properties_of_neighbor_increment}, $\forall i\in[r],$ $\forall v\in V''_i,u\in \Gamma_{G'_i}(v),$ we have $u\in V''_i$ which means that $u\in \Gamma_{G''_i}(v).$ Thus, $|\Gamma_{G''_i}(v)|\geq \lceil(m/n_{i-1})^{1/2}\rceil-1.$ Then due to Lemma~\ref{lem:number_of_roots}, we have that with probability at most $\frac{1}{8},$ $n_i\geq 16 n_{i-1}^{3/2}/m^{1/2}.$ Since $m/n\geq m/n_i \geq 1024,$ we have that with probability at most $\frac{1}{8},$ $n_i\geq n_{i-1}^{11/10}/m^{1/10}.$ Let $y_1,y_2,\cdots,y_r$ be $r$ random variables. If $n_i\geq n_{i-1}^{11/10}/m^{1/10},$ then $y_i=1,$ otherwise $y_i=0.$ We have $\E(\sum_{i=1}^r y_i)\leq \frac{r}{8}.$ By Markov's inequality, we have $\Pr(\sum_{i=1}^r y_i\geq \frac{r}{2})\leq \frac{1}{4}.$ Thus, with probability at least $\frac{3}{4},$ $\sum_{i=1}^r y_i\leq \frac{r}{2}.$ Notice that when $y_i=0,$ then $n_i\leq n_{i-1}^{11/10}/m^{1/10},$ and when $y_i=1,$ we have $n_i\leq n_{i-1}.$ So if there are at least $\frac{r}{2}$ number of $y_i$s which are $0$, then
\begin{align*}
n_r&\leq \frac{\left(\frac{\left(\frac{n^{1.1}}{m^{0.1}}\right)^{1.1}}{m^{0.1}}\right)^{\cdots}}{\cdots}&~\text{(Apply $r/2$ times)}\\
&= \frac{n^{1.1^{r/2}}}{m^{1.1^{r/2-1}}}\\
&=n/(m/n)^{1.1^{r/2-1}}\\
&\leq n/(m/n)^{1.1^{r/4}}\\
&\leq \frac{1}{2}
\end{align*}
where 
the last inequality follows by $r\geq \frac{4}{\log 1.1}(\log\log_{m/n}(2n)).$ Since $n_r$ is an integer, when $n_r\leq \frac{1}{2},$ $n_r=0.$ Thus, we can conclude that if $r\geq c\cdot \log\log_{m/n} n$ for a sufficiently large constant $c>0,$ then with probability at least $\frac{3}{4}-\frac{1}{50}\geq \frac{2}{3},$ the modified $\textsc{Connectivity}(G,m,r)$ will not output FAIL.
 \end{proof}

 Notice that though the theoretical guarantees of the \textit{min-parent-forest} leader selection method is worse than the random leader sampling, the merit of \textit{min-parent-forest} leader selection method is that it can have an ``early start''.

 Consider the case when the total space size $m$ is $\Theta(n).$
 In this case, random leader sampling will always sample a half of the vertices as the leaders until the total space $m$ is $\poly(\log n)$ larger than the number of vertices.
 However, \textit{min-parent-forest} leader selection method can make a large progress at the beginning, it will choose the number of leaders to be about the sum of inverse degrees.
 Furthermore, the depth of the \textit{min-parent-forest} may not always have $\log n$ depth.
 Thus, it is an interesting question which leader selection approach has better performance in practice.

\section{Acknowledgments}

We thank Paul Beame, Lijie Chen, Xi Chen, Mika G\"{o}\"{o}s, Rasmus Kyng, Zhengyang Liu, Jelani Nelson, Eric Price, Aviad Rubinstein, Timothy Sun, Omri Weinstein, David P. Woodruff, and Huacheng Yu for helpful discussions and comments. 

\newpage

\addcontentsline{toc}{section}{References}
\bibliographystyle{alpha}
\bibliography{ref,andoni,geom}




\end{document}